\documentclass{article}
\usepackage{lineno}
\usepackage{wrapfig}
\usepackage{graphicx}
\usepackage{floatrow}
\usepackage{sidecap}
\usepackage{caption}
\usepackage{calc} 
\usepackage[numbers,sort&compress]{natbib}

\usepackage{lmodern}

\usepackage[commentsnumbered,ruled,vlined]{algorithm2e}
\usepackage{setspace}
\usepackage[utf8]{inputenc}

\usepackage{tikz}

\usepackage{natbib}
\usepackage[final]{neurips_2024}
\usepackage{xcolor}         
\usepackage{wrapfig}
\usepackage[utf8]{inputenc} 
\usepackage[T1]{fontenc}    
\usepackage{url}            
\usepackage{booktabs}       
\usepackage{tablefootnote}
\usepackage{amsfonts}       
\usepackage{nicefrac}       
\usepackage{microtype}      

\definecolor{myblue}{rgb}{0.255, 0.412, 0.882} 

\definecolor{green1}{RGB}{10,158,10}
\definecolor{blue1}{RGB}{17,85,204}
\definecolor{red1}{RGB}{204,0,0}
\definecolor{mygray}{gray}{0.85}

\definecolor{myblue2}{RGB}{187, 213, 232} 

\usepackage{url}            
\usepackage{booktabs}       
\usepackage{amsfonts}       
\usepackage{nicefrac}       
\usepackage{microtype}      
\usepackage{lipsum}
\usepackage{wrapfig}
\usepackage{siunitx}
\usepackage{sidecap}

\usepackage{amsmath,amsfonts,amssymb,amsthm}
\usepackage{mathtools}
\usepackage{multirow}

\usepackage{bbm}
\usepackage{colortbl}
\usepackage{booktabs}
\usepackage{multirow}
\usepackage{siunitx}
\usepackage{array}
\usepackage{arydshln}

\usepackage{subcaption}

\usepackage{amsmath}
\usepackage{amssymb}
\usepackage{bm} 
\newcommand*{\rom}[1]{\expandafter\@slowromancap\romannumeral #1@}

\usepackage{enumitem}

\newcolumntype{P}[1]{>{\hspace{1ex}}p{#1}<{\hspace{1ex}}}

\newtheorem{theorem}{Theorem}

\newtheorem{lemma}[theorem]{Lemma}

\renewcommand{\mathbf}[1]{\boldsymbol{\mathit{#1}}}

\usepackage{setspace}

   \usepackage[colorlinks=true,
               linkcolor=red,       
               citecolor=green,      
               urlcolor=blue,         
               filecolor=magenta     
               ]{hyperref}

\title{Learning Cortico-Muscular Dependence through Orthonormal Decomposition of Density Ratios}

\author{%
  \textbf{Shihan Ma}$^{1\dagger}$\quad \textbf{Bo Hu}$^{2\dagger}$\quad \textbf{Tianyu Jia}$^{1}$\quad \textbf{Alexander Kenneth Clarke}$^{1}$\\
  \textbf{Blanka Zicher}$^{1}$\quad \textbf{Arnault H. Caillet}$^{1}$\quad \textbf{Dario Farina}$^{1}$\quad \textbf{Jos{\'e}~C.~Pr{\'\i}ncipe}$^{2}$\\$^{\dagger}$\textit{Equal Contribution}\\  $^{1}$Department of Bioengineering, Imperial College London\\
  $^{2}$Department of Electrical and Computer Engineering, University of Florida\\
}

\begin{document}

\maketitle
\begin{abstract}
The cortico-spinal neural pathway is fundamental for motor control and movement execution, and in humans it is typically studied using concurrent electroencephalography (EEG) and electromyography (EMG) recordings. However, current approaches for capturing high-level and contextual connectivity between these recordings have important limitations.
Here, we present a novel application of statistical dependence estimators based on orthonormal decomposition of density ratios to model the relationship between cortical and muscle oscillations. 
Our method extends from traditional scalar-valued measures by learning eigenvalues, eigenfunctions, and projection spaces of density ratios from realizations of the signal, addressing the interpretability, scalability, and local temporal dependence of cortico-muscular connectivity. We experimentally demonstrate that eigenfunctions learned from cortico-muscular connectivity can accurately classify movements and subjects. Moreover, they reveal channel and temporal dependencies that confirm the activation of specific EEG channels during movement. Our code is available at \url{https://github.com/bohu615/corticomuscular-eigen-encoder}. 
\end{abstract}









\section{Introduction}

\begin{wrapfigure}{r}{.4\textwidth}
\vspace{-15pt}
  \centering
\includegraphics[width=1\textwidth]{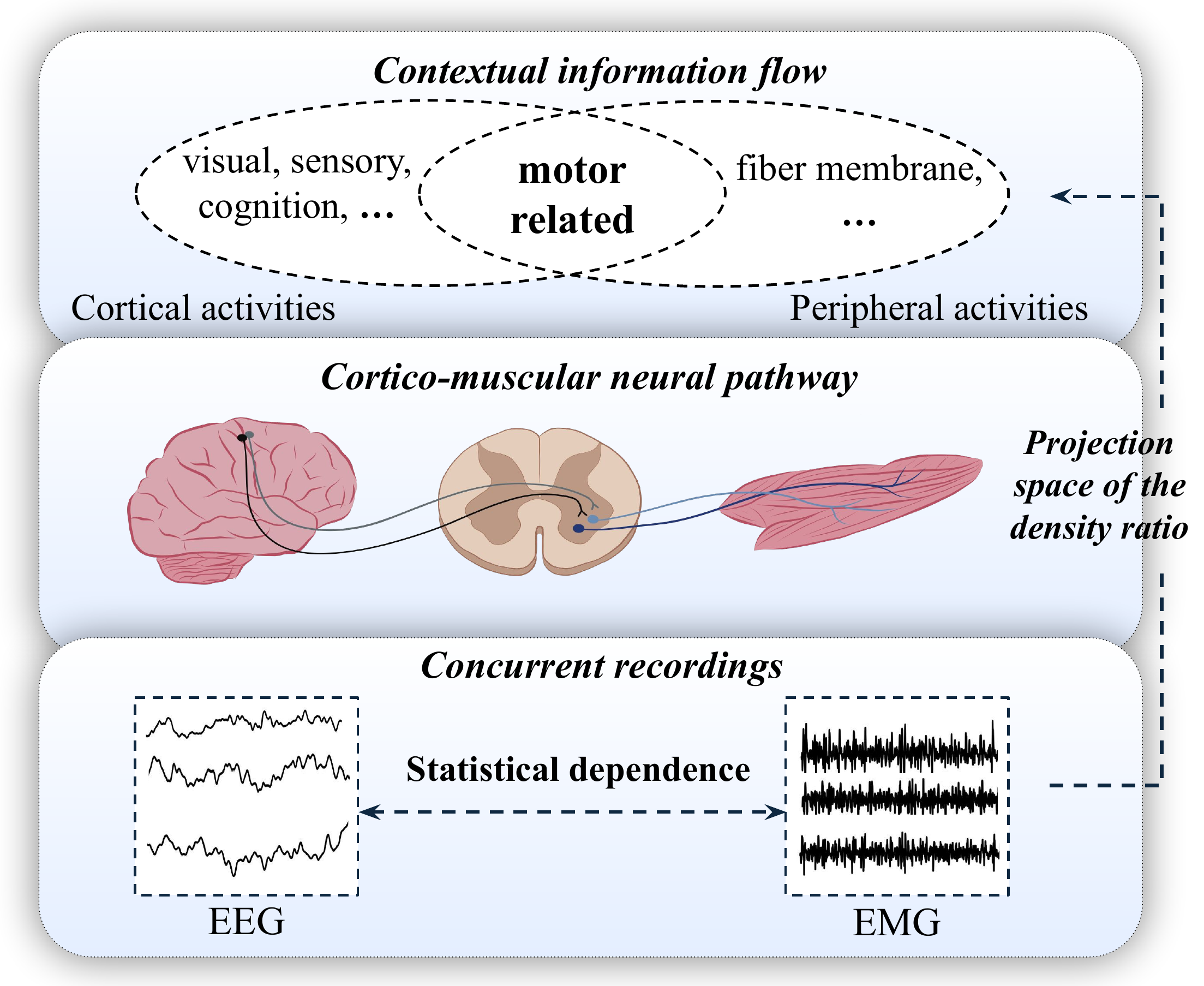}\vspace{-5pt}
  \caption{\small The cortico-muscular pathway allows brain-muscle communication with coherent cortical and peripheral oscillations. This paper models this connectivity through the statistical dependence between their concurrent recordings of EEG and EMG. \vspace{-9pt}}
\label{figure1}
\end{wrapfigure}

The brain communicates with muscles by sending information to the spinal cord.
Part of this information is directly transmitted from the cortex to spinal motor neurons via the cortico-spinal neural pathway, which is vital for motor control and movement execution. Because motor neurons are directly connected to muscles, cortical oscillations traveling through the cortico-spinal pathway are coherent with oscillations in muscle electrical activities, both in humans and non-human primates
~\cite{Baker1999, Baker2007, Schoffelen2005}. This relationship, known as functional cortico-muscular connectivity, is critical in neuroscience and is typically studied using concurrent recordings of neural signals such as electroencephalography (EEG) and electromyography (EMG)~\cite{brambilla2021,grosse2002} (Fig.~\ref{figure1}). Practical applications include diagnosing and monitoring of neuromuscular disorders, such as amyotrophic lateral sclerosis~\cite{Mcmackin2023biomarkers}, stroke~\cite{Gao2024influencing}, and Parkinson's disease~\cite{Zokaei2021}, as well as developing brain-computer interfaces (BCIs) for individuals with motor impairments~\cite{chowdhury2019eeg}. 

Despite several applications, there is still a lack of proper statistical tools to model the relationship between EEG and EMG. The predominant method, Cortico-Muscular Coherence (CMC), measures temporal and spectral coherence by computing the normalized cross-spectrum in time intervals~\cite{mima1999corticomuscular}. From this analysis, it is generally accepted that the EEG's beta band (13-30 Hz) is linked to steady motor control, while the gamma band (>30 Hz) is associated to motor planning and execution
~\cite{grosse2002, pfurtscheller1999event, jensen2007cross, xu2016corticomuscular}.

While CMC provides some relevant information on cortico-muscular connectivity, there remains a lack of generalized and higher-order statistical measures that quantify nonlinear and high-level connectivity. Can high-level contextual information, such as muscle movements and participant identifiers, be directly learned from modeling cortico-muscular connectivity? 
Our paper explores the potential of using statistical dependence estimators to address this problem.

Statistical dependence estimators typically follow a procedure of defining a measure preferably by probabilistic distributions, deriving a variational bound, and optimizing a variational cost of this bound using a function approximator, such as mutual information estimators~\cite{belghazi2018mine, jordan2, jordan1} and Kernel Independent Component Analysis (KICA)~\cite{bach2002kernel, gretton2005measuring}. These measures 
are defined for realizations.

However, statistical dependence estimators above have rarely been successfully applied to cortico-muscular analysis, mainly due to three reasons: instability and poor scalability, lack of spatio-temporal resolution, and lack of practical contextual connections. As these estimators typically quantify dependence at the trial level, they overlook the importance of channel and temporal dependence in cortico-muscular analysis. More importantly, these measures only produce a scalar-valued score, but how this score should be used and its connection to the desired contextual factors are unclear.

This paper successfully applies statistical dependence estimators to EEG-EMG pairs using the concept of \textit{\textbf{orthonormal decomposition of the density ratio}}. Recently, there has been a shift from scalar-valued measures to decomposing density ratio as a positive definite function, and learning its eigenvalues, eigenfunctions, and associated projection spaces through neural network optimization with matrix cost functions such as $\log\det$ and nuclear norm~\cite{hu2024normalized, huang2018gaussian, huang2019universal}. This decomposition, known as the Functional Maximal Correlation Algorithm (FMCA), addresses the fundamental issue of relating dependence to contextual information: Eigenvalues define a multivariate dependence measure, and eigenfunctions span a feature projection space that captures the contextual factors affecting dependence.

This paper expands on this idea, addressing interpretability, scalability, and local-level dependence that are missing in existing dependence analyses. Sec.~\ref{density_ratio_decompose} explains why eigenfunctions, learned from cortico-muscular connectivity, can capture contextual factors for motor control and participant identification. Sec.~\ref{algorithm} introduces FMCA-T, optimizing a new matrix trace cost for the theory, which demonstrates greater efficiency and stability than the $\log\det$ cost. Sec.~\ref{temporal_resolution} shows that while the objective estimates global dependence from trial realizations, localized channel-level and temporal-level dependencies can also be formed in a top-down manner, which are important in EEG as they indicate channel activations and synchronization of activities. Our framework is illustrated in Fig.~\ref{fig:schematic}. 

Our main experiment demonstrates that the learned eigenfunctions, without labels, effectively capture factors such as movements and subjects that contribute to high-level cortico-muscular connectivity. After training, using EEG's eigenfunction as a feature projector noticeably improves classification accuracy over various baselines. Additionally, channel-level and temporal-level dependencies indicate that specific EEG channels are selectively activated during movements, corroborating neuroscientific findings. Simulated data further confirm that our proposed measure is invariant to nonstationary noise, including pink noise and random delays.

\vspace{-5pt}
\section{Methods} \label{sec:method}

\vspace{-2pt}

\subsection{Density ratio decomposition for EEG-EMG signal pairs}
\label{density_ratio_decompose}

\setcounter{figure}{1}
\begin{figure}[t]
    \centering
\includegraphics[width=1\columnwidth]{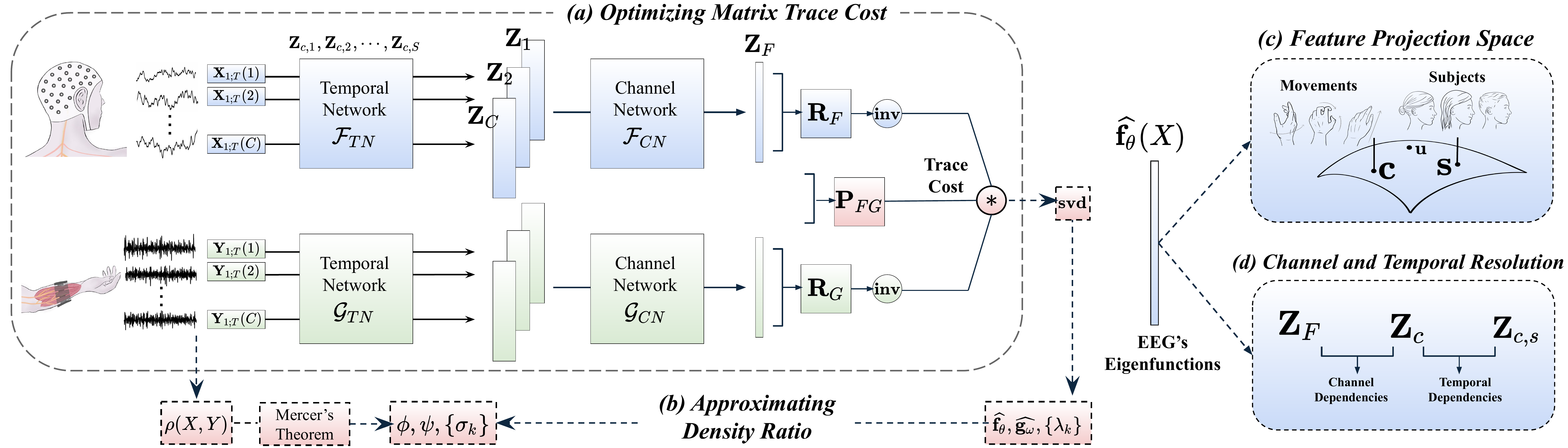}
\caption{\small Diagram for learning cortico-muscular dependence by decomposing density ratios: (a) Network $\mathbf{f}_\theta$ is applied to EEG $\mathbf{X}_{1:T}$ and $\mathbf{g}_\omega$ to EMG $\mathbf{Y}_{1:T}$ to minimize a matrix trace cost. (b) EEG-EMG pairs are sampled from a joint distribution, from which a density ratio $\rho(X,Y)$ is defined and considered a positive definite function. Its linear operator has a spectral decomposition of eigenfunctions $\{\mathbf{\phi}$, $\mathbf{\psi}\}$ and eigenvalues $\{\sigma_k\}$. The networks provably approximate the dominant eigenvalues and eigenfunctions of this decomposition with network outputs $\{\widehat{\mathbf{f}_\theta}$, $\widehat{\mathbf{g}_\omega}\}$, and SVD results $\{\lambda_k\}$. Eigenvalues here measure multivariate statistical dependence; eigenfunctions are optimal feature projectors. (c) After training, the eigenfunctions, specifically those from EEG, form a projection space containing contextual information for motor control and participant identification. (d) To provide channel activation and activity synchronization for cortico-muscular analysis, we compute density ratios between channel-level $\mathbf{Z}_c$ and temporal-level features $\mathbf{Z}_{c, s}$ against global features $\mathbf{Z}_F$ to quantify channel-level and temporal-level dependencies.\vspace{-0pt}} 
\label{fig:schematic}
\end{figure}

\textbf{Problem formulation.} Consider EEG signals $\mathbf{X}:=\mathbf{X}_{1:T}$ and EMG signals $\mathbf{Y}:=\mathbf{Y}_{1:T}$. Denote $\mathbf{s}$ as the subject, $\mathbf{c}$ as the type of movement, and $\mathbf{u}$ as other auxiliary contextual factors. These are factors that could potentially affect the statistical dependence between EEG and EMG signals. Each signal is conditioned on these parameters. Denote these factors as $\mathbf{z} := \{\mathbf{s}, \mathbf{c}, \mathbf{u}\}$ with distribution $\mathbb{P}(\mathbf{z})$. Distributions for EEG and EMG given these conditions are $p(\mathbf{X}=X|z)$ and $p(\mathbf{Y}=Y|z)$, respectively. Their joint distribution is given by $p(X, Y) = \int p(X|z) p(Y|z) p(z) dz$. Similarly, the marginal distributions are given by $p(X) = \int p(X|z) p(z)dz$, and likewise $p(Y) = \int p(Y|z) p(z) dz$. 

Our goal is to extract factors $\mathbf{s}, \mathbf{c}, \mathbf{u}$ that affect the dependence between two modalities, from available sample pairs of $\mathbf{X}$ and $\mathbf{Y}$, even when $\mathbf{s}$, $\mathbf{c}$, and $\mathbf{u}$ are not given. We propose that this can be achieved by decomposing the density ratio of this probabilistic system. \vspace{5pt}

\textbf{Decomposition of EEG and EMG density ratios.} Following the work on FMCA, we propose an orthonormal decomposition of the density ratio to measure the dependence between EEG and EMG:

\begin{equation}
\begin{gathered}
\rho := \frac{p(X, Y)}{p(X) p(Y)} = \sum_{k=1}^\infty \sqrt{\sigma_k} \, \phi_k(X) \psi_k(Y), \\
\int_{\mathcal{X}} \phi_i(X) \phi_j(X) \, d\mathbb{P}(X) =
\begin{cases} 
1, & \text{if } i = j \\ 
0, & \text{if } i \neq j 
\end{cases}, \quad
\int_{\mathcal{Y}} \psi_i(Y) \psi_j(Y) \, d\mathbb{P}(Y) =
\begin{cases} 
1, & \text{if } i = j \\ 
0, & \text{if } i \neq j 
\end{cases},
\end{gathered}
\end{equation}

for any $i, j = 1, 2, \cdots$. The density ratio $\rho(X, Y)$ is treated as a positive definite function associated with a linear operator $\mathbf{L} f := \int \rho(X, \cdot) f(X) \, dX$ for any measurable scalar function $f$. According to Mercer's theorem, this operator has a spectral decomposition with eigenvalues $\sigma_1, \sigma_2, \cdots$, and orthonormal basis functions $\phi_1, \phi_2, \cdots$ and $\psi_1, \psi_2, \cdots$. In scenarios where $\mathbf{X}$ and $\mathbf{Y}$ are statistically independent, all eigenvalues are zero. Conversely, larger eigenvalues suggest stronger dependence. \vspace{5pt} 

Such eigenfunctions form a linear span. Our hypothesis is that this span captures shared contextual factors such as $\mathbf{c}$ and $\mathbf{s}$ across two modalities, stated as follows.  

\begin{lemma} 
Assuming conditional independence given $\mathbf{z} := \{\mathbf{s}, \mathbf{c}, \mathbf{u}\}$, we have $p(X, Y|z) = p(X|z) p(Y|z)$. Hence, the ratio $\rho(X, Y) := \frac{p(X, Y)}{p(X)p(Y)}$ decomposes as $\rho(X,Y) = \int \frac{p(X|z) p(z)}{p(X)p(z)} \cdot \frac{p(Y|z) p(z)}{p(Y)p(z)} \cdot p(z)dz$. Assuming $\mathbf{z}$ is discrete (e.g., movement patterns $\mathbf{c}$ and participant identities $\mathbf{s}$), the information of $\mathbf{z}$ is contained in the span of the basis functions for the density ratio $\rho(X, Y)$.
\end{lemma}%
\begin{proof}
Define the ratios \(\rho_X(X, z) = \frac{p(X, z)}{p(X) p(z)}\) and \(\rho_Y(Y, z) = \frac{p(Y, z)}{p(Y) p(z)}\). Considering the sample space \(\mathcal{Z}\), the sets \(\rho_X(X, z)\) and \(\rho_Y(Y, z)\) for \(z \in \mathcal{Z}\) are discrete. Under the conditional independence assumption, these ratios satisfy \(\rho(X, Y) = \sum_{z \in \mathcal{Z}} p(z)\rho_X(X, z) \rho_Y(Y, z)\).

This indicates that the sets \(\rho_X(X, z)\) and \(\rho_Y(Y, z)\) decompose \(\rho(X, Y)\), similar to the eigenfunctions \(\phi_k\) and \(\psi_k\). Since both decompositions represent \(\rho(X, Y)\), the functions \(\rho_X(X, z)\) and \(\rho_Y(Y, z)\) must lie within the span of these basis functions. Hence, there exist coefficients \(\alpha_{z, k}\) and \(\beta_{z, k}\) such that \(\rho_X(X, z) = \sum_{k} \alpha_{z, k} \phi_k(X)\) and \(\rho_Y(Y, z) = \sum_{k} \beta_{z, k} \psi_k(Y)\) for each \(z\). Thus, learning the dependence between $\mathbf{X}$ and $\mathbf{Y}$ is implicitly learning the dependence between each of them relative to the factors $\mathbf{z}$, even when $\mathbf{z}$ is not observed. 
\end{proof}

\subsection{FMCA-T: Learning decomposition for the matrix trace}
\label{algorithm}

When probability densities are unavailable, we approximate eigenvalues and eigenfunctions using a learning system with two neural networks and a cost function, typically a matrix cost like \(\log\det\) or nuclear norm~\cite{hu2024normalized, huang2018gaussian, huang2019universal}. These costs optimize an aggregation of eigenvalues. The networks learn the dominant eigenvalues and eigenfunctions when optimized.

\vspace{5pt}

\textbf{Aggregation of eigenvalues.} To measure the total power of the eigenspectrum, define a scalar-valued measure using a convex function \(\mathbf{\xi}:\mathbb{R}\rightarrow\mathbb{R}\) with \(\mathbf{\xi}(0) = 0\). Assume the eigenvalues are ranked \(\sigma_1 \geq \sigma_2 \geq \cdots\). The truncated total statistical dependence measure of the top \(K\) eigenvalues is defined by \(T_{\mathbf{\xi}} := \sum_{k=1}^K \mathbf{\xi}(\sigma_k)\). Function \(\mathbf{\xi}(x) = -\log (1 - x)\) corresponds to the $\log\det$ cost.

\vspace{5pt}


\textbf{Prior work: log-determinant cost.} Consider two networks, $\mathbf{f}_\theta:\mathcal{X} \rightarrow \mathbb{R}^K$ and $\mathbf{g}_\omega:\mathcal{Y} \rightarrow \mathbb{R}^K$, mapping realizations of $\mathbf{X}$ and $\mathbf{Y}$ to $K$-dimensional outputs, respectively. Assume $\mathbf{f}_\theta$ is for EEG and $\mathbf{g}_\omega$ for EMG. The autocorrelation (ACFs) and cross-correlation functions (CCFs) are defined as:
\begin{equation}
\resizebox{1\linewidth}{!}{$
\begin{gathered}
\mathbf{R}_F = \mathbb{E}_{\mathbf{X}}[\mathbf{f}_\theta(\mathbf{X})  \mathbf{f}_\theta^\intercal(\mathbf{X}) ],\; \mathbf{R}_G = \mathbb{E}_{\mathbf{Y}}[\mathbf{g}_\omega(\mathbf{Y})  \mathbf{g}_\omega^\intercal(\mathbf{Y}) ], \mathbf{P}_{FG} = \mathbb{E}_{\mathbf{X}, \mathbf{Y}}[\mathbf{f}_\theta(\mathbf{X})  \mathbf{g}_\omega^\intercal(\mathbf{Y}) ],\; \mathbf{R}_{FG} = \begin{bmatrix}\mathbf{R}_F & \mathbf{P}_{FG} \\
\mathbf{P}^\intercal_{FG} & \mathbf{R}_G
\end{bmatrix}.
\end{gathered}$}
\label{fmca_def}
\end{equation}

FMCA minimizes a $\log\det$ cost, which reaches the negative value of the total measure \(T_{\mathbf{\xi}}\) of \(\mathbf{\xi}(x) = -\log(1 - x)\) when minimized. The cost is defined by:
\begin{equation}
\resizebox{.85\linewidth}{!}{$
\begin{gathered}
\min_{\theta, \omega} \;r_L(\theta, \omega) =  \log\det \mathbf{R}_{FG}-\log \det \mathbf{R}_{F} - \log\det \mathbf{R}_{G}, \;\;r^*_L =  \sum_{k=1}^K\log(1-\sigma_k).
\end{gathered}$}
\end{equation}
\textbf{Normalization trick.} After training, normalizations are needed to obtain eigenfunctions. The first step is to ensure orthonormality: $\overline{\mathbf{f}_\theta} = {\mathbf{R}}_F^{-\frac{1}{2}} \mathbf{f}_\theta, \overline{\mathbf{g}_\omega} = {\mathbf{R}}_G^{-\frac{1}{2}} \mathbf{g}_\omega$. The second step is a singular value decomposition: $\overline{{\mathbf{P}}_{FG}} = \mathbb{E}[ \overline{\mathbf{f}_\theta}(\mathbf{X}) \overline{\mathbf{g}_\omega}^\intercal(\mathbf{X})] = \mathbf{U}\mathbf{S}^{\frac{1}{2}}\mathbf{V}$, where $\mathbf{S} = \text{diag}(\lambda_1, \cdots, \lambda_K)$. The third step is to normalize functions such that they are invariant to the linear operator: $\widehat{\mathbf{f}_\theta} = \mathbf{U}^\intercal \overline{\mathbf{f}_\theta},\;\; \widehat{\mathbf{g}_\omega} = \mathbf{V}^\intercal \overline{\mathbf{g}_\omega}$. Functions $\widehat{\mathbf{f}_\theta}, \widehat{\mathbf{g}_\omega}$ are the top eigenfunctions of the density ratio, and $\lambda_1, \lambda_2, \cdots$ are the top eigenvalues. An approximation of the density ratio is given by $\widehat{\rho} = \widehat{\mathbf{f}_\theta}{}^\intercal \mathbf{S}^{\frac{1}{2}}\widehat{\mathbf{g}_\omega}\approx \rho$.

\vspace{5pt}

\textbf{Newly proposed: matrix trace cost.} This paper explores alternative convex functions, specifically the simplest case \({\mathbf{\xi}}(x) = x\), The cost, in the form of a matrix trace, is described below.
\begin{lemma}
Denote $\mathbf{P}:=\mathbf{P}_{FG}$. Given neural nets $\mathbf{f}_\theta$ and $\mathbf{g}_\omega$, minimizing the matrix trace

\begin{equation}
\begin{gathered}
\min_{\theta,\omega} \; r_T(\theta,\omega) = -\text{Trace}(\mathbf{R}_{F}^{-1}\mathbf{P}\mathbf{R}_G^{-1}\mathbf{P}^\intercal), 
\end{gathered}
\label{fmca_maximal}
\end{equation}

yields $r_T^*(\theta,\omega) = -\sum_{k=1}^K \sigma_k$, where $r_T^*(\theta, \omega)$ is the optimal cost, reaching the sum of the top $K$ eigenvalues of the density ratio when minimized. We name this algorithm FMCA-T.
\end{lemma}

\begin{proof}
Applying the Schur complement to $r_L$, we obtain $r_L = \log\det (\mathbf{I} - \mathbf{R}_F^{-\frac{1}{2}} \mathbf{P}\mathbf{R}_G^{-1} \mathbf{P}^\intercal \mathbf{R}_F^{-\frac{1}{2}})$. Denoting eigenvalues of a matrix as $\lambda_1(\cdot),\cdots,\lambda_K(\cdot)$, the cost becomes $r_L = \sum_k \log(1-\lambda_k(\mathbf{M}))$, where $\mathbf{M} = \mathbf{R}_F^{-\frac{1}{2}} \mathbf{P}\mathbf{R}_G^{-1} \mathbf{P}^\intercal \mathbf{R}_F^{-\frac{1}{2}}$. Optimizing the sum of eigenvalues instead, we use $Trace(\mathbf{M})$ and, based on the trace property $Trace(\mathbf{A}\mathbf{B}) = Trace(\mathbf{B}\mathbf{A})$, derive the trace cost for learning multivariate statistical dependence as $r_T = -{Trace}(\mathbf{R}_{F}^{-1}\mathbf{P}\mathbf{R}_G^{-1}\mathbf{P}^\intercal)$.
\end{proof}

FMCA-T is more computationally efficient as it uses only matrix operations of dimension \(K\). Directly optimizing the sum of the eigenvalues is also more stable than optimizing their logarithm.

\subsection{Channel-level and temporal-level dependencies}
\label{temporal_resolution}

\textbf{Motivations.} For EEG \(\mathbf{X}_{1:T}\) and EMG \(\mathbf{Y}_{1:T}\), FMCA-T applies two networks, \(\mathbf{f}_\theta\) and \(\mathbf{g}_\omega\), to minimize the matrix trace cost. Dependence is measured at two levels: \textbf{\textit{random-process level}}, measured by eigenvalues for the overall dataset dependence, and \textbf{\textit{trial level}}, measured by the density ratio—the higher the ratio, the greater the contribution of this pair of realizations to overall dependence. In cortico-muscular analysis, it is vital to understand how individual channels and time steps contribute to connectivity, especially in EEG signals, as they represent the temporal and spatial dynamics of brain. Hence, we propose localized density ratios to measure \textbf{\textit{temporal-level dependence}} and \textbf{\textit{channel-level dependence}}. The core idea is computing density ratios between channel-level and temporal-level features against the global trial-level features.

\vspace{10pt}

\textbf{Channel-level features.} We design a specialized network topology to generate features for individual channels and time intervals, ensuring that the internal layers of this network quantify channel-level and temporal-level features, similar to \cite{lawhern2018eegnet, zhang2024brant, yuan2024brant}. 

Given \(\mathbf{X}_{1:T} = [\mathbf{X}_{1:T}(1), \cdots, \mathbf{X}_{1:T}(C)]^\intercal \) for channels \(c = 1, \cdots, C\), we define a temporal network \(\mathcal{F}_{TN}: \mathbb{R}^{T} \rightarrow \mathbb{R}^K\) that maps single-channel signals to a \(K\)-dimensional feature space, and a channel network \(\mathcal{F}_{CN}: \mathbb{R}^{L \times K} \rightarrow \mathbb{R}^K\) that maps concatenated channel features to global features:

\begin{equation}
\begin{gathered}
\mathbf{Z}_c = \mathcal{F}_{TN}\left(\mathbf{X}_{1:T}(c)\right), \quad c = 1,\cdots, C; \quad \mathbf{Z}_F = \mathcal{F}_{CN}\left([\mathbf{Z}_1, \mathbf{Z}_2, \cdots, \mathbf{Z}_C]^\intercal\right),
\end{gathered}
\end{equation}

where $\mathbf{Z}_1, \mathbf{Z}_2, \cdots, \mathbf{Z}_C$ are channel-level features, and $\mathbf{Z}_F$ is global trial-level features.

\vspace{5pt}

\textbf{Channel-level dependence $\widehat{\rho_{C,F}}(c)$.} The density ratio of $\mathbf{Z}_1, \mathbf{Z}_2, \cdots, \mathbf{Z}_C$ relative to $\mathbf{Z}_F$ measures channel-level dependence. Post-training and with fixed parameters, we compute the ACF of the channel features $\mathbf{R}_C = \frac{1}{C} \mathbb{E}[\sum_{c=1}^C \mathbf{Z}_c  \mathbf{Z}_c{}^\intercal]$, the ACF of the global features $\mathbf{R}_F = \mathbb{E}[\mathbf{Z}_F  \mathbf{Z}_F^\intercal]$, and the CCF between them $\mathbf{P}_{C,F} = \frac{1}{C} \mathbb{E}[\sum_{c=1}^C \mathbf{Z}_c  \mathbf{Z}_F^\intercal]$. 

Next, the features are normalized as in the Sec.~\ref{algorithm}: \(\mathbf{Z}_c\) and \(\mathbf{Z}_F\) are normalized to \(\overline{\mathbf{Z}_c} = \mathbf{R}_C^{-\frac{1}{2}}\mathbf{Z}_c\) and \(\overline{\mathbf{Z}_F} = \mathbf{R}_F^{-\frac{1}{2}}\mathbf{Z}_F\) for orthonormality. The SVD of \(\mathbf{R}_C^{-\frac{1}{2}}\mathbf{P}_{C,F}\mathbf{R}_F^{-\frac{1}{2}} = \mathbf{U}\mathbf{S}^{\frac{1}{2}}\mathbf{V}\) is computed. The outputs are further normalized to \(\widehat{\mathbf{Z}_c} = \mathbf{U}^\intercal \overline{\mathbf{Z}_c}\) and \(\widehat{\mathbf{Z}_F} = \mathbf{V}^\intercal \overline{\mathbf{Z}_F}\) to guarantee invariance in the linear operator. Finally, the density ratio can be constructed as \(\widehat{\rho_{C,F}}(c) = \widehat{\mathbf{Z}_c}{}^\intercal \mathbf{S} \widehat{\mathbf{Z}_F}\).

This ratio \(\widehat{\rho_{C,F}}(c)\) is a function of channel $c$ and trial $\mathbf{X}$, implying the dependence between channel and global features. The greater the value, the stronger the activation of the channels, showing which channels contribute the most to the cortico-muscular connectivity.

\vspace{10pt}

\textbf{Temporal-level features and dependence.} To measure time-domain dependence, we compute density ratios between the internal features of the temporal network $\mathcal{F}_{TN}$ and the global features, in two steps: first, computing density ratios between \textbf{\textit{adjacent}} network layers; second, aggregating these ratios to consider all layers.

\textbf{Step 1: Construct density ratios $\widehat{\rho_{s-1,s, c}}(\tau_1, \tau_2)$ between \textit{adjacent} layers.} Fix a channel $c$ and feature $\mathbf{Z}_c$. Consider a simple temporal network with $S$ convolution layers with nonlinaer activation functions: $\mathcal{F}_{TN}^{(1)}, \mathcal{F}_{TN}^{(2)}, \cdots, \mathcal{F}_{TN}^{(S)}$, with kernel sizes $\Delta_1, \Delta_2, \cdots, \Delta_S$, and their outputs $\mathbf{Z}_{c, 1}, \mathbf{Z}_{c, 2}, \cdots, \mathbf{Z}_{c, S}$. Suppose the time dimensions of these layers are $T_1, T_2, \cdots, T_S$. Fix any layer $s$. The $\tau$-th element of $\mathbf{Z}_{c, s}$, denoted as $\mathbf{Z}_{c, s}(\tau)$, is obtained by applying a nonlinear operation to a segment of the previous layer’s output:
\begin{equation}
\begin{gathered}
\mathbf{Z}_{c, s}(\tau) = \mathcal{F}_{TN}^{(s-1)}\left(\mathbf{Z}_{c, s-1}(\tau:\tau+\Delta_{s-1})\right). 
\end{gathered}
\label{equation_convolution}
\end{equation}

\begin{itemize}[leftmargin=*]
\item Define the ACF of layer $s-1$: $\mathbf{R}_{c, s-1} = \frac{1}{T_{s-1}} \mathbb{E}[\sum_{\tau} \mathbf{Z}_{c, s-1}(\tau) \mathbf{Z}_{c, s-1}^\intercal(\tau)]$ 
\item Define the ACF of layer $s$: $\mathbf{R}_{c, s} = \frac{1}{T_s} \mathbb{E}[\sum_{\tau} \mathbf{Z}_{c, s}(\tau) \mathbf{Z}_{c, s}^\intercal(\tau)]$ 
\item Define the CCF between them: $\mathbf{P}_{c, s-1,s} = \frac{1}{T_{s}} \mathbb{E}[\sum_{\tau}\sum_{\delta=1}^{\Delta_s} \mathbf{Z}_{c, s-1}(\tau+\delta) \mathbf{Z}_{c, s}^\intercal(\tau)]$.
\end{itemize}
Normalization with $\mathbf{R}_{c, s-1}$, $\mathbf{R}_{c, s}$, and $\mathbf{P}_{c, s-1,s}$ yields density ratios $\widehat{\rho_{s-1,s, c}}(\tau_1, \tau_2)$, which quantify the dependence between time $\tau_1$ and $\tau_2$ across two layers $s-1$ and $s$, for a given trial and channel $c$. A higher value indicates a stronger dependence between adjacent network layers.

\textbf{Step 2: Aggregate layer-wise ratios for localized responses $\widehat{\varrho_{s, c}}(\tau)$.} While $\widehat{\rho_{s-1,s, c}}(\tau_1, \tau_2)$ quantifies dependence between two layers, we aggregate these ratios to account for all network layers.

Again, fix the element $\mathbf{Z}_{c, s}(\tau)$ in layer $s$. We focus on its mapping to the next layer $s+1$. Based on Eq.~\eqref{equation_convolution}, elements in layer $s+1$ that are mapped from $\mathbf{Z}_{c, s}(\tau)$ by $\mathcal{F}_{FP}^{(s)}$ are within a window of size $\Delta_s$ (the kernel size). To ensure this window stays within the feature vector's boundary, we define coordinates $ \mathcal{I}_s = \left[ \max\left(0, \tau-\Delta_s+1\right), \min\left(i, T_{s+1}-1\right) \right] $. Feature elements in layer $s+1$ that are mapped from $\mathbf{Z}_{c,s}(\tau)$ fall within these coordinates.

We then create a series of functions $\widehat{\varrho_{s, c}}(\tau)$ with $\tau \in [1, T_s]$ for each layer $s$ as the aggregations of the ratios. Starting from $\widehat{\varrho_{S, c}}(\tau) = \widehat{\rho_{C,F}}(c)$ (channel-level density ratio), compute recursively 

\begin{equation}
\begin{gathered}
\widehat{\varrho_{s, c}}(\tau_1) = \sum_{\tau_2 \in I_s} \widehat{\varrho_{s+1, c}} (\tau_2) \widehat{\rho_{s, s+1, c}}(\tau_1, \tau_2), \;\;\tau_1 \in [1, T_s]
\end{gathered}
\label{telescoping_post_training}
\end{equation}

That is, starting from the top layer of the network, we aggregate the density ratios within the window $\mathcal{I}_s$, layer-by-layer, until we generate a localized measurement for each element of the function $\widehat{\varrho_{s, c}}(\tau)$ at layer $s$, channel $c$, and time $\tau \in [1, T_s]$, considering all neural network layers. \vspace{5pt}

The final localized responses of the density ratio, ${\widehat{\varrho_{s, c}}}(\tau)$, obtained in a top-down manner, are functions of the EEG trial $\mathbf{X}_{1:T}$, time step $\tau$, and channel $c$, providing both temporal and channel-level resolution. The same analysis applies to EMG signals, differing only in the number of channels.

\section{Experiments}\label{sec:exps}

Our experiments have three key findings: (1) Dependence measured by FMCA-T is stable against nonstationary noises and delays in simulated dataset; (2) Learning from unlabeled EEG-EMG pairs extracts movement and subject information from EEG's eigenfunctions; (3) EEG’s spatio-temporal dependencies are consistent with ground truth brain activations in simulated dataset and match theoretical evidence in experimental dataset.

\subsection{Datasets and baselines}

\label{Sec_dataset}

\textbf{Dataset 1: SinWav.} We construct a simulated dataset where each data pair \(\{\mathbf{X}_{1:T}, \mathbf{Y}_{1:T}\}\) has a clean sinusoidal signal \(\mathbf{X}_t = A\sin(\omega t)\) and a noisy counterpart \(\mathbf{Y}_t\) superimposed with various types of noise: stationary white noise \(\epsilon_t \sim \mathcal{N}(0, \sigma^2)\), nonstationary Gaussian noise \(\sigma_t^2 \propto |\mathbf{X}_t|\), and nonstationary pink noise \(S(f) \propto 1/f^\alpha\). Random delays are added by padding the start and end of the signal with noise such that \(\mathbf{Y}_t = \epsilon_t\) is white for \(1 \leq t \leq \tau_1\) and \(T - \tau_2 < t \leq T\), and \(\mathbf{Y}_t = \mathbf{X}_{t - \tau_1}\) is sinusoidal for \(\tau_1 < t \leq T - \tau_2\), where \(\tau_1 + \tau_2 = \tau\) is fixed. Since these noises do not change the underlying sinusoid, the signal pairs are statistically independent when conditioned on the sinusoid. Thus, we expect the dependence measure to be unaffected by the noise level.

\vspace{5pt}

\textbf{Dataset 2: EEG-EMG-Fusion.} We use a public dataset~\cite{jeong2020multimodal} (approved by the Institutional Review Board at Korea University, 1040548-KU-IRB-17-181-A-2) with paired 60-channel EEG and 7-channel EMG recordings from 25 subjects. The subjects perform three \textit{\textbf{main movements}}: arm-reaching, hand-grasping, and wrist-twisting. Each main movement contains \textit{\textbf{sub-movements}}: arm-reaching along six directions, hand-grasping three objects, and wrist-twisting with two motions, and thus $11$ movements in total. Subjects perform one sub-movement per trial, and 50 trials are collected per sub-movement. The same recordings are repeated for three sessions at one-week intervals. Both EEG and EMG are recorded at 2,500 Hz and downsampled to 1,000 Hz. The dataset is cleaned by removing eye-blinking artifacts and baseline wandering, and segmented into 4-second intervals, creating 41,250 paired samples of complete movement cycles. \vspace{5pt}

\textbf{Dataset 3: Simulated EEG-EMG Dataset.} We simulate 128-channel EEG signals and 7-channel EMG signals for left/right motor and sensory activations from 20 subjects using EEGSourceSim \cite{barzegaran2019eegsourcesim}. Motor sources are used to simulate the corresponding EMG signals. Both EEG and EMG are sampled at 1,000 Hz. White Gaussian noise at 5 dB is added to the EEG and EMG signals. FMCA-T is trained on 16 subjects and tested on 4 subjects to compare FMCA-T's spatial-level dependence representation with the ground truth activations, as shown in Fig.~\ref{EEG_EMG_SIMULATED}.

\vspace{5pt}

\textbf{Classification tasks.} We conduct three classification experiments: \textbf{\textit{3-class}} (three main movements), \textbf{\textit{11-class}} (11 sub-movements), and \textbf{\textit{Subj}} (25 subjects). We also compare \textit{\textbf{inter-subject}} and \textit{\textbf{cross-subject}} classifications. Cross-subject means the test set contains only unseen subjects. Inter-subject uses an 80-20 split of trials from all subjects for training and testing, while cross-subject uses 20 subjects for training and 5 for testing with five-fold cross-validation.

\vspace{5pt}

\textbf{Statistical dependence baselines.} We compare our proposed dependence measure with established baselines: (1) KICA~\cite{bach2002kernel} and HSIC~\cite{gretton2005measuring}, which solve the generalized eigenvalue problem of two kernel Gram matrices, using $\sum_{i} \lambda_i$ for HSIC and $-\sum_{i} \log(1-\lambda_i)$ for KICA; (2) MINE~\cite{belghazi2018mine}, which optimizes the Donsker-Varadhan representation with a three-layer MLP; (3) CC: Pearson correlation coefficient averaged over time; (4) MIR (KNN estimator~\cite{kraskov2004estimating}), which optimizes entropies using k-nearest neighbor distances, and then computes mutual information; (5) Our method uses density ratios for trial-level dependence and eigenvalue aggregations $T_{\mathbf{\xi}}$ for random-process-level dependence.

For the EEG-EMG dependence study, we compare with CMC, the correlation coefficient between EEG and EMG spectra of windowed data on the alpha band of channel C4~\cite{jeong2020multimodal}. We extend CMC by replacing linear correlation with nonlinear measures, computing CMC-KICA and CMC-MIR.

\vspace{5pt}

\textbf{EEG feature projector baselines.} After training ${\mathbf{f}_\theta}$ and ${\mathbf{g}_\omega}$ networks for dependence estimation, with parameters fixed, we train a three-layer MLP on EEG's eigenfunctions ($\widehat{\mathbf{f}_\theta}(X)$) for classification. This is compared with baseline EEG classifiers trained from scratch: (1) {Supervised}: \textit{Vanilla Classifier}, with the same topology as ours but using a standard log-likelihood cost; \textit{EEG-Net}~\cite{lawhern2018eegnet}, a specialized network for EEG-based BCI;
\textit{EEG-Conformer}~\cite{song2022eeg}, a compact convolutional transformer for EEG decoding and visualization; \textit{Deep4}~\cite{schirrmeister2017deep}, a deep ConvNet for classification using raw EEG;
and \textit{CSP-RLDA}~\cite{ang2012filter}, using common spatial pattern (CSP) for feature extraction and regularized linear discriminant analysis (RLDA), adapted for multi-class classification with majority voting. (2) {Self-Supervised}: contrastive costs using 1-second windows from the same signal as positive pairs and from different signals as negative pairs, including \textit{Barlow Twins}~\cite{zbontar2021barlow}, \textit{SimCLR}~\cite{chen2020simple}, and \textit{VicReg}~\cite{bardes2021vicreg}. Experimental and implementation details can refer to the App.~\ref{appendix:implementation}.
 
\vspace{5pt}

\subsection{Main results}\label{sec:main_results}

\textbf{Robustness of FMCA-T.} Fig.~\ref{fig:sine_comp} shows the robustness of our proposed measure in the SinWav dataset, when there are increasing levels of nonstationary noise and delays. Since EEG and EMG signals are often damaged and distorted by environmental noise and the functional coupling occurs with time delays, an effective measure should maintain its robustness against these factors. \vspace{-5pt}

\begin{figure}[h]

    \centering

\includegraphics[width=.9\columnwidth]{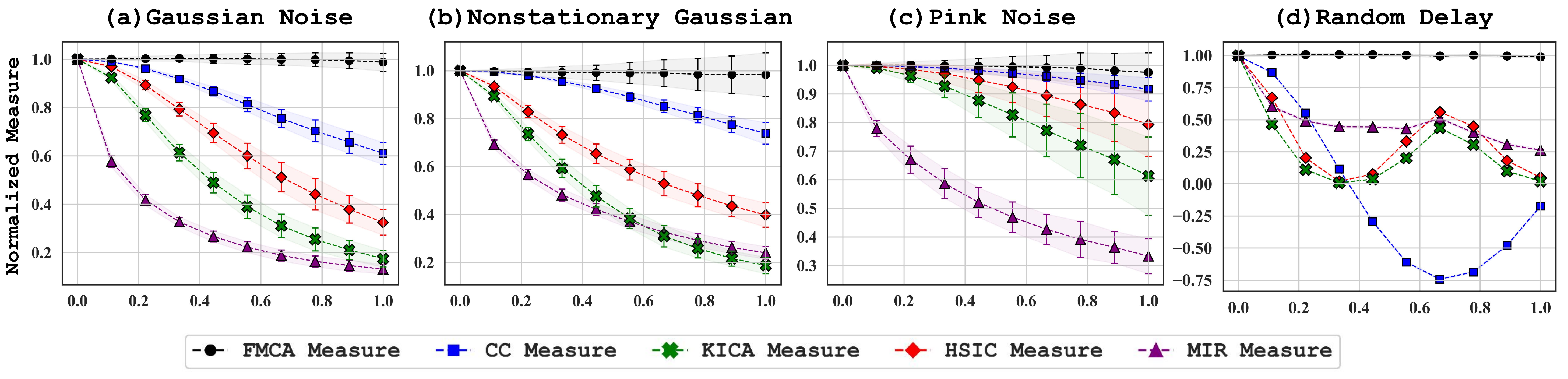}\vspace{-5pt}

    \caption{\small Density ratios from FMCA-T are robust to various noise types: (a) stationary white Gaussian noise, (b) nonstationary Gaussian noise, (c) nonstationary pink noise, and (d) random delays. FMCA-T proves the most robust estimations across all noise types and outperforms all linear and nonlinear baselines. Note that as delays increase, estimations using CC produce negative values given the opposite phase between the paired sinusoids. }

    \label{fig:sine_comp}

\end{figure} \vspace{-5pt}

In Fig.~\ref{fig:sine_comp}, FMCA-T is first trained on noisy data pairs with all four noise types and magnitudes (Sec.~\ref{Sec_dataset}). Using the trained models, we measure the dependence between a clean sinusoid and its noisy counterpart. A noise level of $1.0$ means the noise magnitude matches the sinusoid. The delay level determines the extent to which the clean sinusoid is shifted from its original position. A delay level of $1.0$ shifts the sinusoid to have no intersection with the original one. FMCA-T consistently shows invariance to noise and delays, as the dependence is determined by their frequency but not by the noise and phase shift. MINE fails to converge and produce stable results for this dataset.

\vspace{12pt}

\textbf{Applying FMCA-T to EEG-EMG-Fusion.} We confirm our primary hypothesis that the projection space defined by EEG eigenfunctions, derived from modeling the statistical dependence between EEG-EMG recordings, captures essential contextual factors like movements and subjects without requiring labels. We visualize the learned eigenfunctions and density ratios in Fig.~\ref{CLUSTERFIGURE}.

\textbf{Fig.~\hyperref[CLUSTERFIGURE]{4(a)}, Fig.~\hyperref[CLUSTERFIGURE]{4(b)}: eigenfunctions $\widehat{\mathbf{f}_\theta}$.} EEG's eigenfunctions effectively capture relevant contextual information. After training eigenfunctions using the entire dataset, we extract a subset of eigenfunctions that belong to a specific subject or a movement and apply t-SNE to visualize them.

Fig.~\hyperref[CLUSTERFIGURE]{4(a)} visualizes the eigenfunctions of all trials for one subject (\textbf{\texttt{SUB1}}). Each trial is color-coded by the type of movement (\textbf{\texttt{MOV1}}$\sim$\textbf{\texttt{MOV3}}). Notably, the eigenfunctions form nine clusters, which are verified to correspond to the three movements recorded over three sessions. This demonstrates that the eigenfunctions contain motion-related information. The consistent clustering patterns across all 25 subjects are detailed in the App.~\ref{appendix:exp}.

\vspace{-5pt}

Fig.~\hyperref[CLUSTERFIGURE]{4(b)} visualizes the eigenfunctions from a single type of movement (reaching, labeled as \textbf{\texttt{MOV1}}) across ten different subjects (\textbf{\texttt{SUB1}}$\sim$\textbf{\texttt{SUB10}}). Each color represents a subject. Distinct clustering patterns are observed, showing that the eigenfunctions also contain subject information which could be useful for participant identification.

\textbf{Fig.~\hyperref[CLUSTERFIGURE]{4(c)},~\hyperref[CLUSTERFIGURE]{4(d)}: density ratio $\widehat{\rho(X, Y)}$.} Based on the t-SNE plot for \textbf{\texttt{SUB1}} trials, we plot the estimated density ratio values between each EEG-EMG pair in Fig.~\hyperref[CLUSTERFIGURE]{4(c)}. In Fig.~\hyperref[CLUSTERFIGURE]{4(d)}, we extract the mean of density ratios for trials in each cluster (\textbf{\texttt{C1}}$\sim$\textbf{\texttt{C9}}), rank them from smallest to largest, and plot them alongside the standard deviation. These figures show that the density ratios remain consistent within each cluster (a movement during a session) while vary across different clusters. We find the highest dependence in reaching, followed by grasping, and the lowest in twisting. Our results are consistent with existing literature that links cortico-muscular connectivity with movement types~\cite{guerrero2022coherence, ye2022investigation}. The results are consistent across all subjects, detailed in the App.~\ref{appendix:exp}.

\textbf{Fig.~\hyperref[CLUSTERFIGURE]{4(e)}$\sim$\hyperref[CLUSTERFIGURE]{(h)}} presents the results of MINE and CMC measures for \textbf{\texttt{SUB1}} trials. Only MINE produces a comparable measure that shows difference across clusters, but with higher variance and instability.

\begin{figure}[h]

  \centering

    \phantomsection

    \includegraphics[width=.9\linewidth]{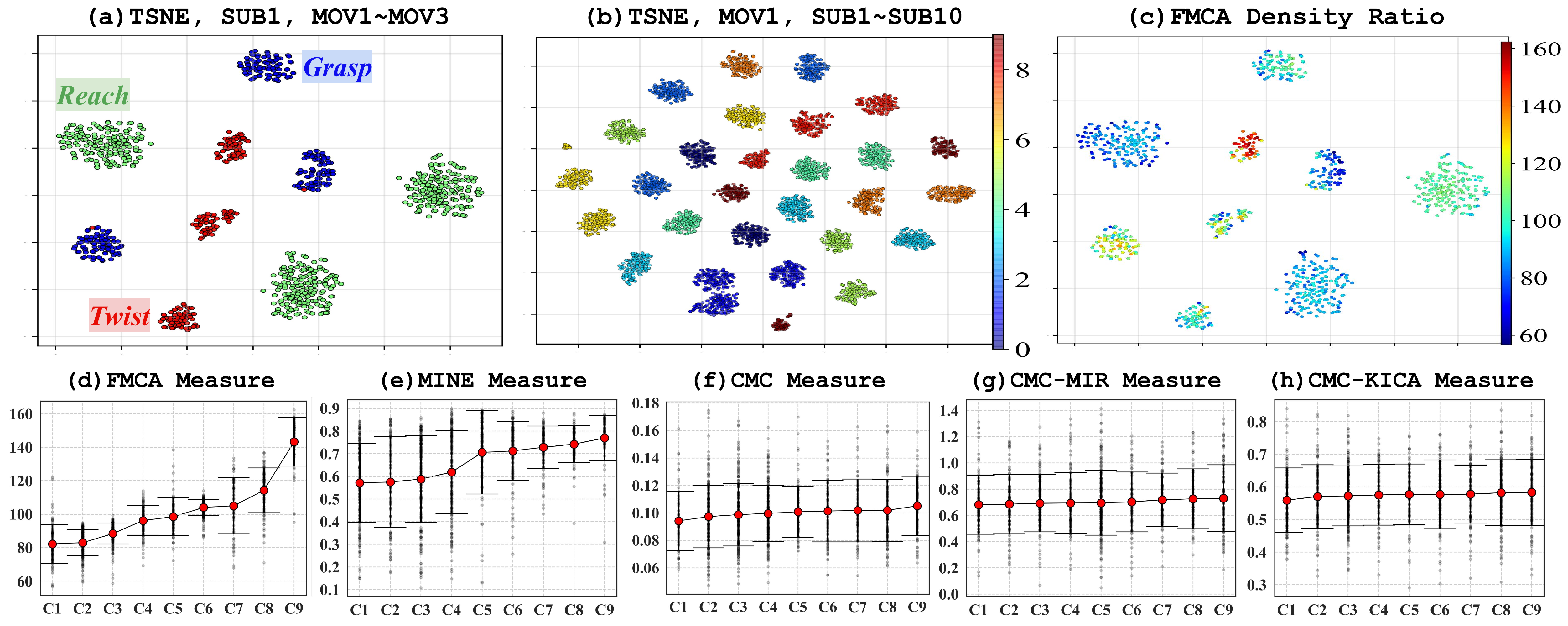}


\caption{\small Visualizing eigenfunctions and density ratios in EEG-EMG fusion with FMCA-T: (a) t-SNE of EEG's eigenfunctions for a single subject (\textbf{\texttt{SUB1}}) show nine clusters specific to three movements (\textbf{\texttt{MOV1}}$\sim$\textbf{\texttt{MOV3}}) across three sessions. (b) t-SNE of EEG's eigenfunctions for reaching movement (\textbf{\texttt{MOV1}}) across 10 subjects (\textbf{\texttt{SUB1}}$\sim$\textbf{\texttt{SUB10}}) shows clusters specific to subjects, where each color is a subject. (c) Density ratios and (d) their mean and std of each cluster (\textbf{\texttt{C1}}$\sim$ \textbf{\texttt{C9}}) demonstrate intra-cluster consistency and inter-cluster separability. (e-h) Comparison of baseline measures, where only MINE is comparable but with higher variance and instability. }\label{CLUSTERFIGURE}

\end{figure}

\textbf{EEG's eigenfunctions as optimal feature projector.} Table~\ref{tab:sctable} validates the claims in Fig.~\ref{CLUSTERFIGURE} with classification accuracy comparisons. We extract EEG eigenfunctions from the training set after the networks are trained on EEG-EMG pairs. The eigenfunctions are used to train a three-layer MLP for classification. This classifier predicts the class of any EEG test samples using its eigenfunctions (output of $\widehat{\mathbf{f}_\theta}$), without requiring EMG samples. As detailed in the Sec.~\ref{Sec_dataset}, scores for $\textbf{\textit{3-class}}$, $\textbf{\textit{11-class}}$, and $\textbf{\textit{subj}}$ are presented in both $\textbf{\textit{inter-subject}}$ and $\textbf{\textit{cross-subject}}$ settings.
\vspace{2pt}

In inter-subject 3-class classification (80-20 split across trials from 25 subjects), FMCA-T exceeds the supervised baseline (EEGNet) by 7.2\%, the classical EEG decoding method (CSP-RLDA) by 1.0\%, and self-supervised methods by over 9.5\%. Notably, CSP-RLDA is trained and tested on each individual subject, with the accuracy averaged across 25 subjects, thereby representing an upper bound for classical methods. All other methods are trained and tested on the combined subject data. For the 11-class classification task, FMCA-T surpasses all baselines with a classification accuracy of 0.32, significantly higher than the chance level of 0.09. Since CSP-RLDA uses binary classification with majority voting, it is computationally infeasible for 11 sub-movements classification.  

\vspace{2pt}

In the more challenging cross-subject classification (trained on 20 subjects, tested on 5), FMCA-T with trace loss outperforms all baselines by over 10\%, achieving an accuracy of 0.54 in the 3-class task. The highest scores from 10,000 iterations are recorded, and experiments are repeated with five-fold validation. The superior performance of FMCA-T in the cross-subject setting suggests that learning EEG-EMG dependence is robust against distribution shifts and nonstationary noise, which is consistent with the observation that self-supervised methods outperform supervised ones.

Comparing FMCA-LD ($\log\det$) and FMCA-T (matrix trace), we find that trace cost has greater stability and reduced variance. The sum of eigenvalues, especially during prolonged training, is more stable. While both costs show similar performance at the initial training stages, FMCA-T has notably reduced variance during the convergence stage of training. 

\vspace{5pt}

\begin{SCtable}[][h]
\centering
\resizebox{.75\columnwidth}{!}{%
\begin{tabular}{>{\raggedright\arraybackslash}p{2.7cm}ccc>{\raggedright\arraybackslash}p{2cm}cc}
\toprule
\textbf{Methods} & \multicolumn{3}{c}{\textbf{(a) Inter-Subject Acc.}} & \multicolumn{2}{c}{\textbf{(b) Cross-Subject Acc.}} \\
\cmidrule(lr){2-4} \cmidrule(l){5-7}
                 & \textit{\textbf{3-Class}} & \textit{\textbf{11-Class}} & \textit{\textbf{Subj}} & \textbf{\textit{\;\;\;\;3-Class}} & \textbf{\textit{11-Class}} \\ \midrule
\rowcolor{myblue!0}\multicolumn{6}{l}{\scriptsize \textit{Supervised}} \\
\rowcolor{myblue!0}\quad {Vanilla} & 0.907±0.020 & 0.220±0.015 & 0.980±0.010 & 0.427±0.021 & 0.110±0.005 \\
\rowcolor{myblue!0}\quad {EEGNet} & 0.904±0.015 & 0.246±0.028 & 0.988±0.007 & 0.405±0.019 & 0.095±0.021 \\
\rowcolor{myblue!0}\quad {EEG-Conformer} & 0.949±0.001 & 0.268±0.001 & 0.976±0.002 & 0.415±0.002 & 0.105±0.001 \\
\rowcolor{myblue!0}\quad {Deep4} & 0.901±0.001 & 0.274±0.001 & 0.941±0.001 & 0.429±0.001 & \textbf{0.140±0.000} \\
\rowcolor{myblue!0}\quad {CSP-RLDA} & 0.985±0.019 & / & / & 0.408±0.018 & / \\ \midrule
\multicolumn{6}{l}{\scriptsize \textit{Self-Supervised}} \\
\rowcolor{myblue!0}\quad {Barlow Twins} & 0.893±0.018 & 0.269±0.012 & 0.987±0.008 & 0.437±0.018 & 0.115±0.004 \\
\rowcolor{myblue!0}\quad {SimCLR} & 0.890±0.019 & 0.257±0.013 & 0.979±0.011 & 0.441±0.020 & 0.117±0.006 \\
\rowcolor{myblue!0}\quad {VicReg} & 0.899±0.016 & 0.274±0.014 & 0.980±0.009 & 0.449±0.016 & 0.115±0.005 \\ \midrule
\rowcolor{myblue!30}\multicolumn{6}{l}{\scriptsize \textit{EEG-EMG Dependence}} \\
\rowcolor{myblue!30}\quad {FMCA-LD} & 0.985±0.003 & 0.257±0.011 & 0.989±0.007 & 0.509±0.014 & 0.115±0.003 \\
\rowcolor{myblue!30}\quad {FMCA-T} & \textbf{0.994±0.002} & \textbf{0.320±0.009} & \textbf{0.998±0.004} & \textbf{0.540±0.012} & \textbf{0.121±0.002} \\
\hline\hline
\end{tabular}
}
\caption{\sloppy \small Comparison of classification accuracies: supervised, self-supervised, and our EEG-EMG dependence learning. FMCA-T's eigenfunctions, trained with trace cost without labels, are optimal feature projectors for EEG. EMG is not required for testing, but only used for training.}
\label{tab:sctable}
\end{SCtable}\vspace{9pt}


\textbf{Spatio-temporal dependencies - real data.} We visualize the local density ratio responses  of cluster \texttt{SUB3-C1} (reaching movement) in both spatial \((\widehat{\rho_{C,F}}(c))\) and temporal domains \((\widehat{\varrho_{s, c}}(\tau))\) in Fig.~\ref{EEG_EIG_PROJ}. The channel-level dependence is averaged across all trials and displayed in Fig.~\hyperref[EEG_EIG_PROJ]{5(a)}. We also randomly select nine trials from the same cluster and visualize them in Fig.~\hyperref[EEG_EIG_PROJ]{5(b)}. The temporal-level dependence for the first trial \texttt{T1} in \texttt{SUB3-C1} is shown in Fig.~\hyperref[EEG_EIG_PROJ]{5(c)}. Consistent activations are observed in other subjects, details in the App.~\ref{appendix:exp}.

As illustrated in Fig.~\ref{EEG_EIG_PROJ}, the localized density ratio remains stable throughout the 4-second movement, corroborating the consistency of brain-muscle connectivity during stable states \cite{liu2019corticomuscular}. We also find that in channel-level dependence, the density ratio activates the fronto-central (FC) areas. The sensorimotor area is crucial for movement control, with EEG data from these regions often used to decode motor intentions. However, motor control also relies on cognitive processes \cite{gallivan2018decision}, especially during movement planning, complex tasks, and collaborations \cite{jia2024enhancing}. Thus, the region of Brodmann area 6, well acknowledged to playing a role in movement planning may contribute differently during various movement tasks \cite{dum2004motor}. Our findings show that the features extracted from these fronto-central areas play an important role in classification.

\vspace{5pt}

\textbf{Spatio-temporal dependencies - simulated data.} We implement FMCA-T on paired EEG-EMG samples from the simulated dataset, using 16 subjects for training and 4 subjects for testing. We compare FMCA-T's spatial-level dependence maps for these 4 testing subjects with their ground truth brain activations computed by the motor ROI and forward matrices, shown in Fig.~\ref{EEG_EMG_SIMULATED}. FMCA-T's spatial-level dependencies are highly similar to the ground truth activations, indicating that the learned density ratios effectively captured the real brain activations.
\begin{figure}[H]

  \centering

    \includegraphics[width=1\linewidth]{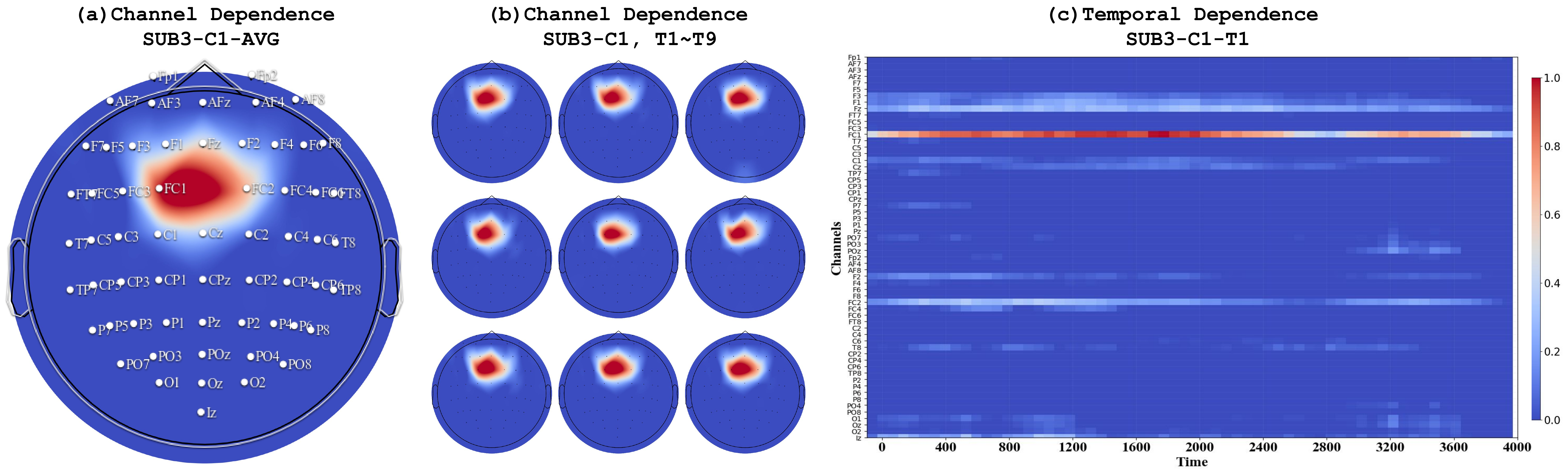}

\caption{\small Localized density ratio for real data. Topographies for \texttt{SUB3}, \texttt{C1} (reaching session) are normalized to the range $[0, 1]$. (a) displays the averaged spatial distribution of \texttt{C1} across all $50$ trials; (b) presents nine random trials \texttt{T1-T9} from this session. Dark red indicates prominent activations around channel FC1. (c) shows 
stable temporal-level dependence over the 4-second movement.
}
\label{EEG_EIG_PROJ} 
\vspace{-3pt}
\end{figure} \vspace{-20pt}

\begin{figure}[H]

  \centering
    \includegraphics[width=1\linewidth]{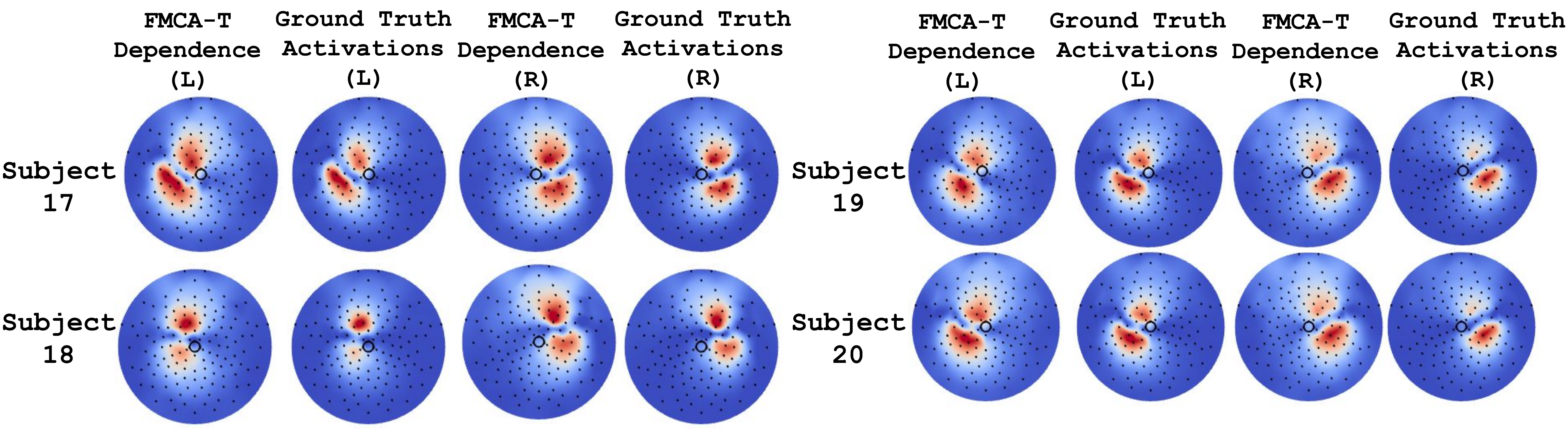}
\vspace{-15pt}
\caption{\small Localized density ratio for simulated data, showing that FMCA-T's distribution is consistent with ground truth brain activations when evaluated on four test subjects. Topographies are normalized to the range $[0, 1]$. L indicates left brain activities and R indicates right brain activities.}
\label{EEG_EMG_SIMULATED} 
\end{figure}

\vspace{-10pt}
\section{Conclusion}
This paper introduces a novel approach for estimating cortico-muscular dependence through the orthonormal decomposition of the density ratio. By treating the density ratio as a positive definite function and learning its projection space from EEG and EMG, we unveil the relationship between statistical dependence, contextual factors impacting connectivity, and the spatio-temporal information shared between the brain and muscles. While our method shows promising results, challenges remain. For example, performance drops in cross-subject classification, likely due to the limited dataset of 25 participants. Future work will focus on applying our framework to larger datasets and incorporating additional bio-signal modalities to model a broader common space in neural data.

{\setlength{\bibsep}{6pt}
\fontsize{9}{10.5}\selectfont

\bibliography{reference}}

\newpage

\appendix  

\section{Additional Implementation Details of FMCA-T}\label{appendix:algorithm}

\textbf{Adaptive estimators.} Since autocorrelation functions (ACF) and crosscorrelation functions (CCF) are expectations of inner products, we prefer to compute these expectations and the cost gradient using an adaptive filter rather than a batch of data. For any parameter $\theta$, the partial derivative of the cost $r_T(\mathbf{f}_\theta, \mathbf{g}_\omega) := Trace(\mathbf{R}_F^{-1} \mathbf{P}\mathbf{R}_G^{-1} \mathbf{P}^\intercal)$ has the following form:\begin{equation}
\resizebox{.7\linewidth}{!}{
$\begin{aligned}
\frac{\partial r_T}{\partial \theta} = -Trace(\mathbf{R}_F^{-1}\frac{\partial \mathbf{R}_{F}}{\partial \theta}\mathbf{R}_F^{-1} \mathbf{P}\mathbf{R}_G^{-1} \mathbf{P}^\intercal)+ Trace(\mathbf{R}_F^{-1} \frac{\partial \mathbf{P}}{\partial \theta}\mathbf{R}_G^{-1} \mathbf{P}^\intercal) \\
-Trace(\mathbf{R}_F^{-1} \mathbf{P}\mathbf{R}_G^{-1}\frac{\partial \mathbf{R}_{G}}{\partial \theta}\mathbf{R}_G^{-1} \mathbf{P}^\intercal) + Trace(\mathbf{R}_F^{-1} \mathbf{P} \mathbf{R}_G^{-1} \frac{\partial \mathbf{P}^\intercal}{\partial \theta}). 
\end{aligned}$}
\label{equation_gradient}
\end{equation}
Observe that terms needed for the gradient have two classes: terms $\mathbf{R}_F, \mathbf{R}_G, \mathbf{P}$, as well as their inverse, and terms of the derivatives $\frac{\partial \mathbf{R}_{F}}{\partial \theta}, \frac{\partial \mathbf{R}_{G}}{\partial \theta}, \frac{\partial \mathbf{P}}{\partial \theta}$. We use a smoothing window over iterations to estimate the $\mathbf{R}_F, \mathbf{R}_G, \mathbf{P}$, and then use the estimated values of $\widetilde{\mathbf{R}_F}, \widetilde{\mathbf{R}_G}, \widetilde{\mathbf{P}}$ as the matrices and their inverse in the cost. For derivative terms $\frac{\partial \mathbf{R}_{F}}{\partial \theta}, \frac{\partial \mathbf{R}_{G}}{\partial \theta}, \frac{\partial \mathbf{P}}{\partial \theta}$, we use the derivatives of the batch. We choose smoothing coefficient $\beta = 0.9$ across all experiments. 

\textbf{Important parameters.} When computing the matrix inverse of ACFs ($\widetilde{\mathbf{R}_F}{}^{-1}$ and $\widetilde{\mathbf{R}_G}{}^{-1}$), a small diagonal matrix scaled by a regularization parameter, $\epsilon \mathbf{I}$, is added. This constant is important for training stability. We choose $\epsilon = 10^{-5}$ across all experiments. The output dimension of the network is chosen as $K=128$, which is also the number of eigenfunctions. 

\textbf{Pseudo-code of the algorithm.} We also explain our algorithm with pseudo-code. Algorithm~\hyperref[algorithm_fmca_ae]{1} shows the pseudo-code for FMCA-T with adaptive estimators and computing eigenfunctions. Algorithm~\hyperref[algorithm_channel_temporal]{2} shows how to generate channel-level and temporal-level dependencies. 



\begin{algorithm}

\setstretch{0.9}

  \SetAlgoLined\DontPrintSemicolon

  \SetKwProg{myalg}{Alg.}{}{}

\vspace{3pt}
\myalg{\textbf{Adaptive matrix estimators}\vspace{3pt}}{

\nl \KwIn{Batch estimation $\mathbf{A}$; Tracking matrix $\acute{\mathbf{A}}$; Iteration $i$;}

\nl\textbf{\small Step 1.} $\acute{\mathbf{A}} \leftarrow \beta \cdot \acute{\mathbf{A}}$ + (1-$\beta$) $\cdot$  $\mathbf{A}$

\nl\textbf{\small Step 2.} $\widetilde{\mathbf{A}}  = \acute{\mathbf{A}}/(1-\beta^i)$

\nl\Return{Tracking matrix $\acute{\mathbf{A}}$; Smoothed estimation $\widetilde{\mathbf{A}}$}\;

\caption{FMCA-T w/ Adaptive Estimators.}

}

\vspace{6pt}

  \myalg{\textbf{FMCA-T: Optimize matrix trace cost\vspace{3pt}}}{

\nl \textbf{Initialize:} Neural networks $\{\mathbf{f}_\theta$, $\mathbf{g}_\omega$\}

\nl \textbf{Initialize:} Tracking matrices $\acute{\mathbf{R_F}}$, $\acute{\mathbf{R_G}}$, $\acute{\mathbf{P}}$

\nl\For{$i=1,2,\cdots$}{ 

\nl Sample a batch of signals $\{X_n, Y_n\}_{n=1}^{bs}$;

\nl Compute the ACF and CCF $\{\mathbf{R_F}, \mathbf{R}_G, \mathbf{P}\}$ with this batch;

\nl Apply adaptive estimators, obtain $\{\widetilde{\mathbf{R_F}}$, $\widetilde{\mathbf{R_G}}$, $\widetilde{\mathbf{P}}\}$, and update $\{\acute{\mathbf{R_F}}$, $\acute{\mathbf{R_G}}, \acute{\mathbf{P}}\}$.

\nl Estimate gradients with smoothed estimator, update networks.

}

\nl\Return{$\theta, \omega$}\;

}

\vspace{6pt}

\myalg{\textbf{Retrieve eigenfunctions and density ratios\vspace{3pt}}}{

\KwIn{Trained networks $\mathbf{f}_\theta$ and $\mathbf{g}_\omega$. Dataset $\{{X}_{n}, Y_{n}\}_{n=1}^N$. }

\nl \textbf{\small Step 1.} Obtain outputs of networks. $\{\mathbf{\mathbf{f}_\theta}(X_n), \mathbf{\mathbf{g}_\theta}(Y_n)\}$.

\nl \textbf{\small Step 2.} Re-estimate $\mathbf{R}_F$, $\mathbf{R}_G$, $\mathbf{P}$ using these outputs.

\nl \textbf{\small Step 3.} Normalization for orthonormality: $\overline{\mathbf{f}_\theta} = {\mathbf{R}}_F^{-\frac{1}{2}} \mathbf{f}_\theta, \overline{\mathbf{g}_\omega} = {\mathbf{R}}_G^{-\frac{1}{2}} \mathbf{g}_\omega$

\nl \textbf{\small Step 4.} Compute SVD: $\overline{{\mathbf{P}}} = \mathbb{E}[ \overline{\mathbf{f}_\theta}(\mathbf{X}) \overline{\mathbf{g}_\omega}^\intercal(\mathbf{X})] = \mathbf{U}\mathbf{S}^{\frac{1}{2}}\mathbf{V}$

\nl \textbf{\small Step 5.} Normalization for invariance: $\widehat{\mathbf{f}_\theta} = \mathbf{U}^\intercal \overline{\mathbf{f}_\theta},\;\; \widehat{\mathbf{g}_\omega} = \mathbf{V}^\intercal \overline{\mathbf{g}_\omega}$

\nl \textbf{\small Step 6.} Construct the density ratio: $\widehat{\rho} = \widehat{\mathbf{f}_\theta}{}^\intercal \mathbf{S}^{\frac{1}{2}}\widehat{\mathbf{g}_\omega}\approx \rho$.

\nl \Return{Eigenfunctions $\widehat{\mathbf{f}_\theta}$, $\widehat{\mathbf{g}_\omega}$, eigenvalues $\mathbf{S}=\text{diag}(\lambda_1,\cdots, \lambda_K)$, density ratio $\widehat{\rho}$}\;

\vspace{3pt}

}
\label{algorithm_fmca_ae}
\end{algorithm}

\begin{algorithm}
\setstretch{1.1}

  \SetAlgoLined\DontPrintSemicolon

  \SetKwProg{myalgfd}{}{}{}

  \SetKwProg{myalg}{Alg.}{}{}

\vspace{3pt}

\myalg{\textbf{Generate temporal and channel features:\vspace{3pt}}}{

\nl \KwIn{Trained network $\mathbf{f}_\theta$ for EEG, consisting of a temporal network $\mathcal{F}_{TN}$ and a channel network $\mathcal{F}_{CN}$; An arbitrary EEG trial $X$}

\nl \textbf{\small Step 1.} Apply $\mathcal{F}_{TN}$ to each channel: $Z_c = \mathcal{F}_{TN} (X(c)),\; c=1,\cdots, C$

\nl \textbf{\small Step 2.} Apply $\mathcal{F}_{CN}$ to all channel features: $Z_F = \mathcal{F}_{CN} ([Z_1, \cdots, Z_C]^\intercal)$

\nl \textbf{\small Step 3.} Extract internal features of $\mathcal{F}_{CN}$: $\mathbf{Z}_{c, 1},\cdots, \mathbf{Z}_{c, S}$ from $S$ convolution blocks

\nl \Return{Temporal features $\mathbf{Z}_{c, s}$, channel features $Z_{c}$, global features $Z_F$}\;}

\vspace{15pt}

\myalg{\textbf{Generate channel-level dependencies}\vspace{3pt}}{

\nl  \KwIn{Temporal, channel, global features from the dataset}

\nl  \ForEach{channel $c$}{

\nl  Compute ACF and CCF for global and channel features, $Z_F$ and $Z_c$

\nl  Use ACF and CCF to compute density ratio $\widehat{\rho_{C,F}}(c)$ with normalization

}} \vspace{1pt}

\myalg{\textbf{Generate temporal-level dependencies:\vspace{3pt}}}{

\nl  \ForEach{channel $c$ and layer $s$}{

\nl  Compute ACF and CCF for temporal features between two layers, $\mathbf{Z}_{c, s}, \mathbf{Z}_{c, s+1}$

\nl  Use ACF and CCF to compute density ratio $\widehat{\rho_{s-1,s, c}}(\tau_1, \tau_2)$ between two layers. }


\nl  Initialize $\widehat{\varrho_{S, c}}(\tau) = \widehat{\rho_{C,F}}(c)$; Initialize each ${\widehat{\varrho_{s, c}}}(\tau)$ to have length $T_s$.

\nl  \For{$s = S-1, \cdots, 1$}{

\nl  For every $\mathbf{Z}_{s}(\tau)$, find elements in layer $s+1$ that are mapped from it (set $I_s$)

\nl  Aggregate density ratios $\widehat{\varrho_{s, c}}(\tau_1) = \sum_{\tau_2 \in I_s} \widehat{\varrho_{s+1, c}} (\tau_2) \widehat{\rho_{s, s+1, c}}(\tau_1, \tau_2), $

}

\nl  \Return{Channel dependence $\widehat{\rho_{C,F}}$ and spatio-temporal dependence ${\widehat{\varrho_{s, c}}}(\tau)$ that can be applied to any arbitrary EEG trial $X$.}

}

\caption{Generate channel-level and temporal-level dependencies.}\vspace{3pt}
\label{algorithm_channel_temporal}

\end{algorithm}

\clearpage

\section{Additional Experiments}\label{appendix:exp}


\textbf{Visualization of EEG's eigenfunctions.} Building on {Fig.~\ref{CLUSTERFIGURE}} from the main paper, we further present results for the visualization of EEG's eigenfunctions on all $25$ participants. In Fig.~\ref{density_ratio_file}, each trial is color-coded by the estimated density ratios for all participants (\textbf{\texttt{SUB1}} to \textbf{\texttt{SUB25}}). In Fig.~\ref{clustering_figure}, each trial is color-coded by the label of their movements.

It can be observed that for each participant, the t-SNE projections form individual clusters, each corresponding to one of the three main movements during a session. The density ratios are highly similar within each cluster (intra-cluster) while varying across clusters (inter-cluster). This indicates that using density ratio as a dependence measure effectively captures each movement individually. \vspace{10pt}



\begin{figure}[h]
{\large\textbf{\texttt{FMCA-T Density Ratio, SUB1$\sim$SUB25}}}
\centering
\begin{subfigure}[b]{0.2\textwidth}
  \centering
  \includegraphics[width=\textwidth]{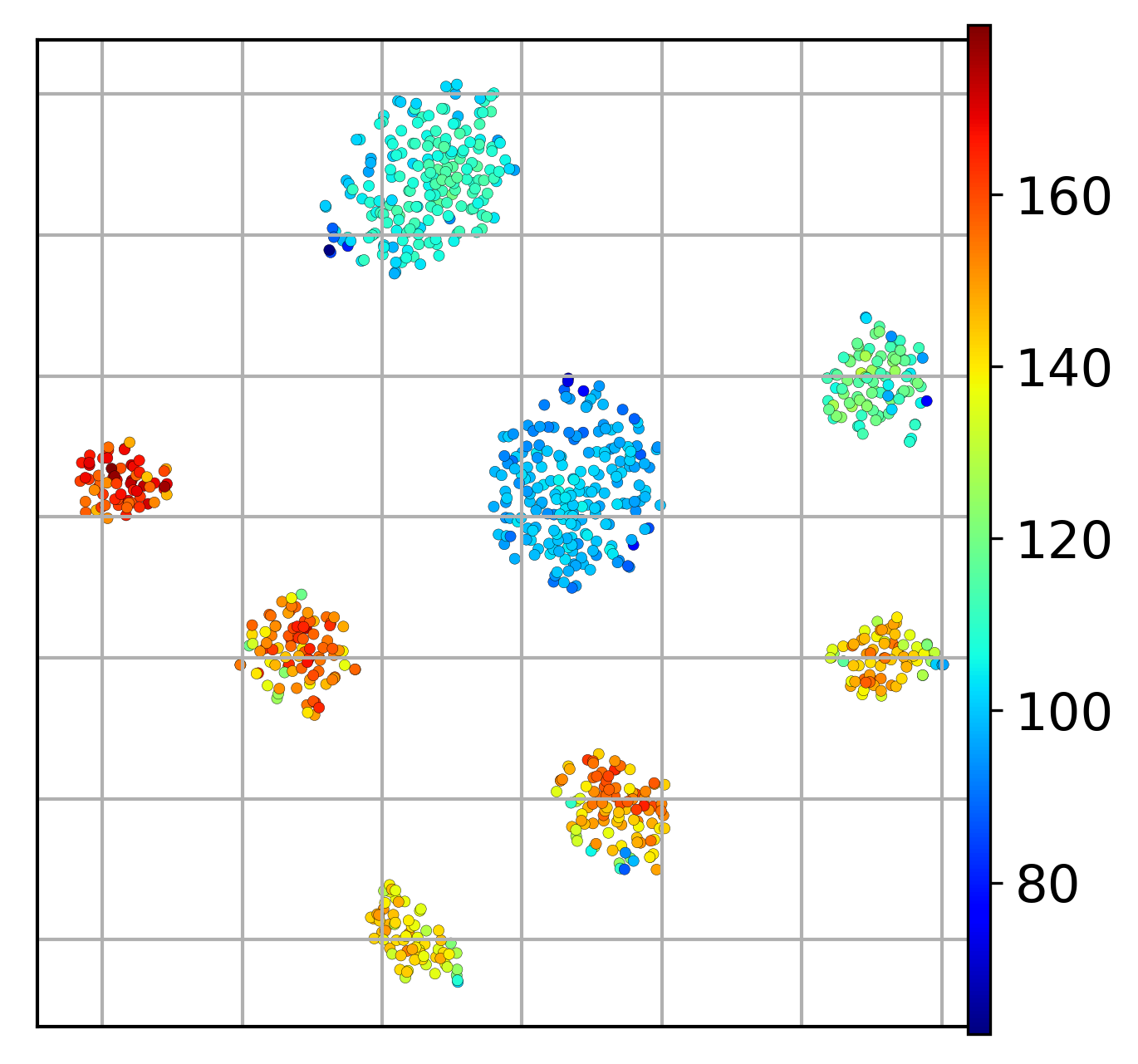}
  \vspace{-5pt}\phantomsection
\end{subfigure}%
\begin{subfigure}[b]{0.2\textwidth}
  \centering
  \includegraphics[width=\textwidth]{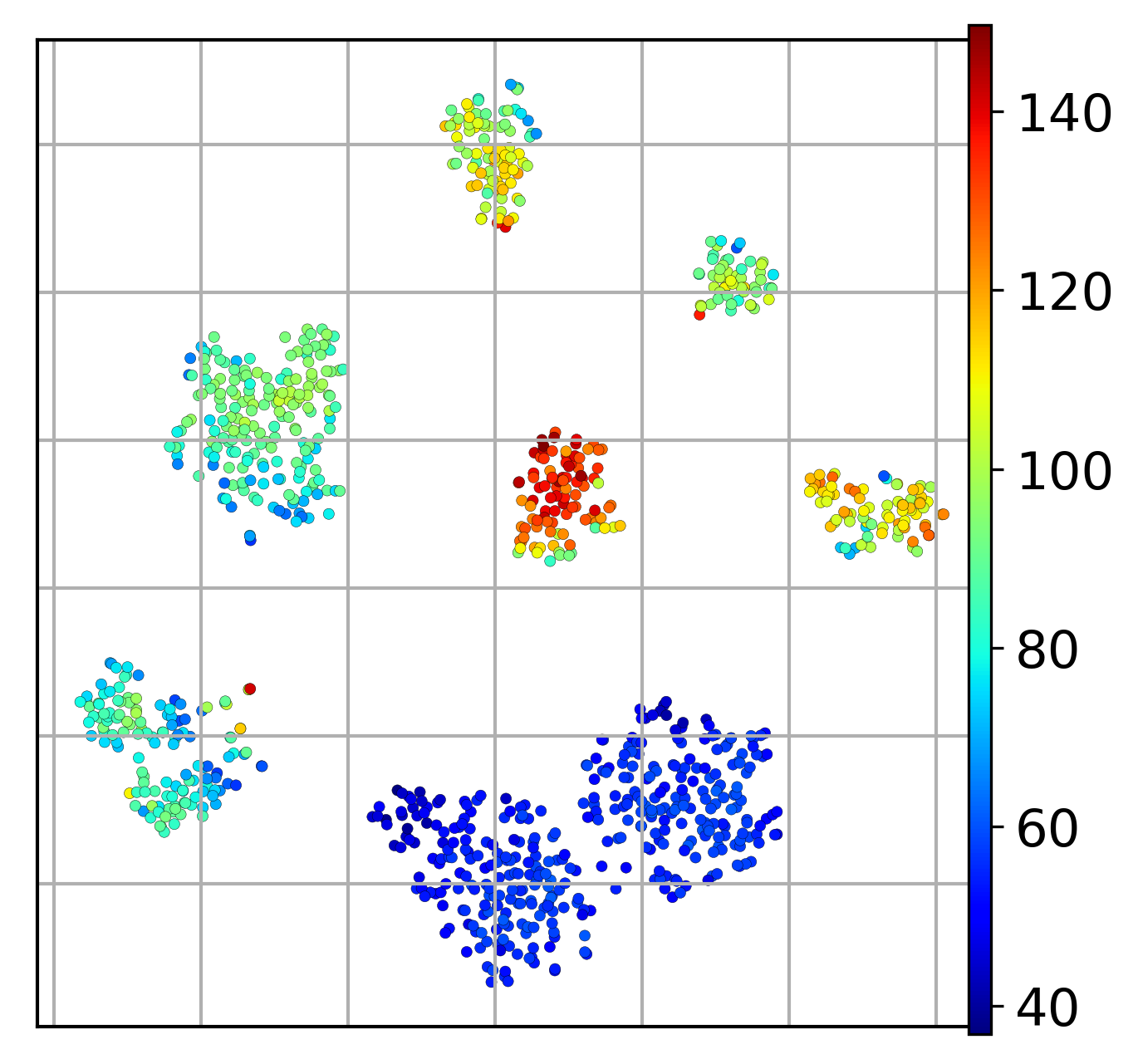}
  \vspace{-5pt}\phantomsection
\end{subfigure}%
\begin{subfigure}[b]{0.2\textwidth}
  \centering
  \includegraphics[width=\textwidth]{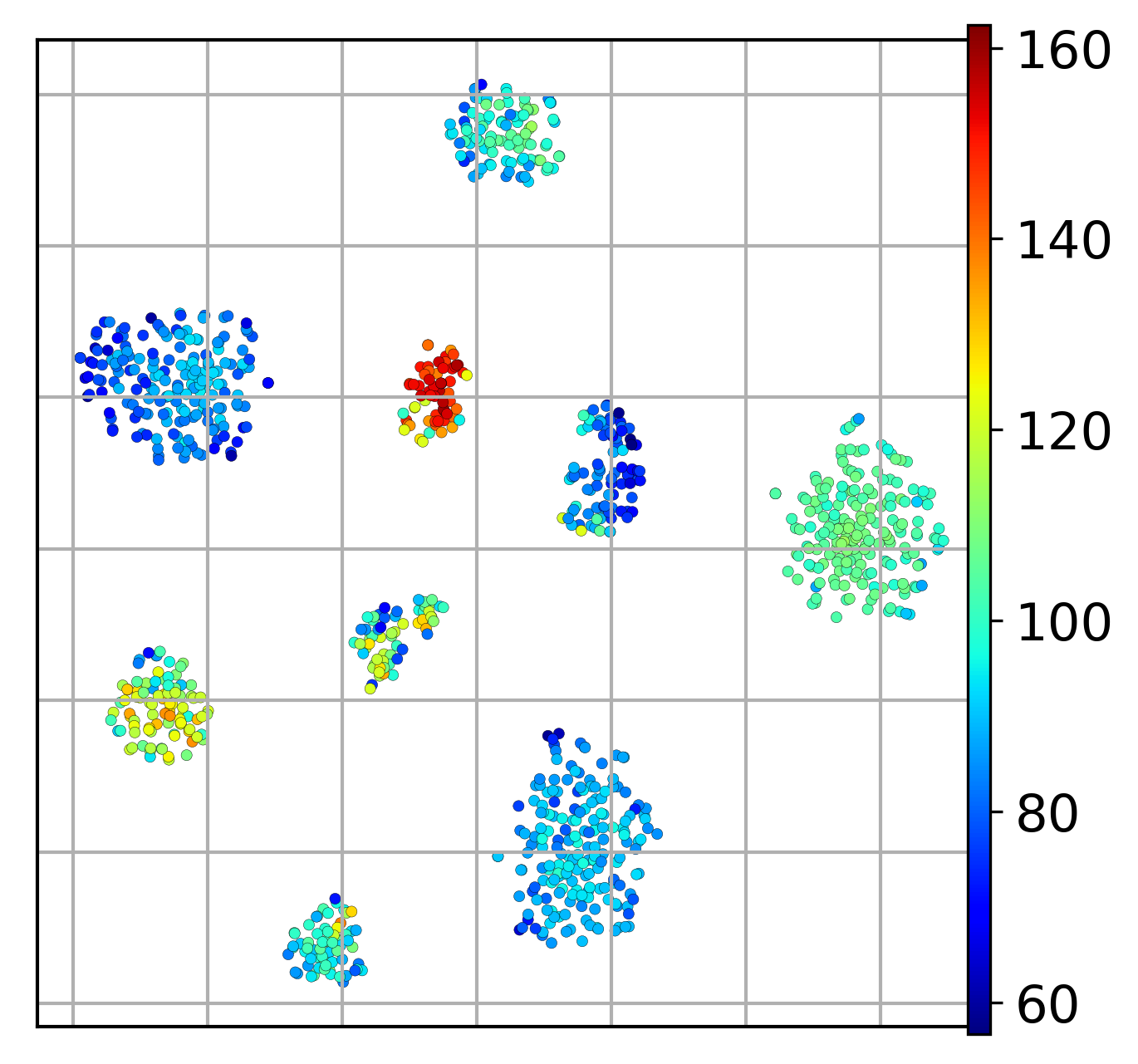}
  \vspace{-5pt}\phantomsection
\end{subfigure}%
\begin{subfigure}[b]{0.2\textwidth}
  \centering
  \includegraphics[width=\textwidth]{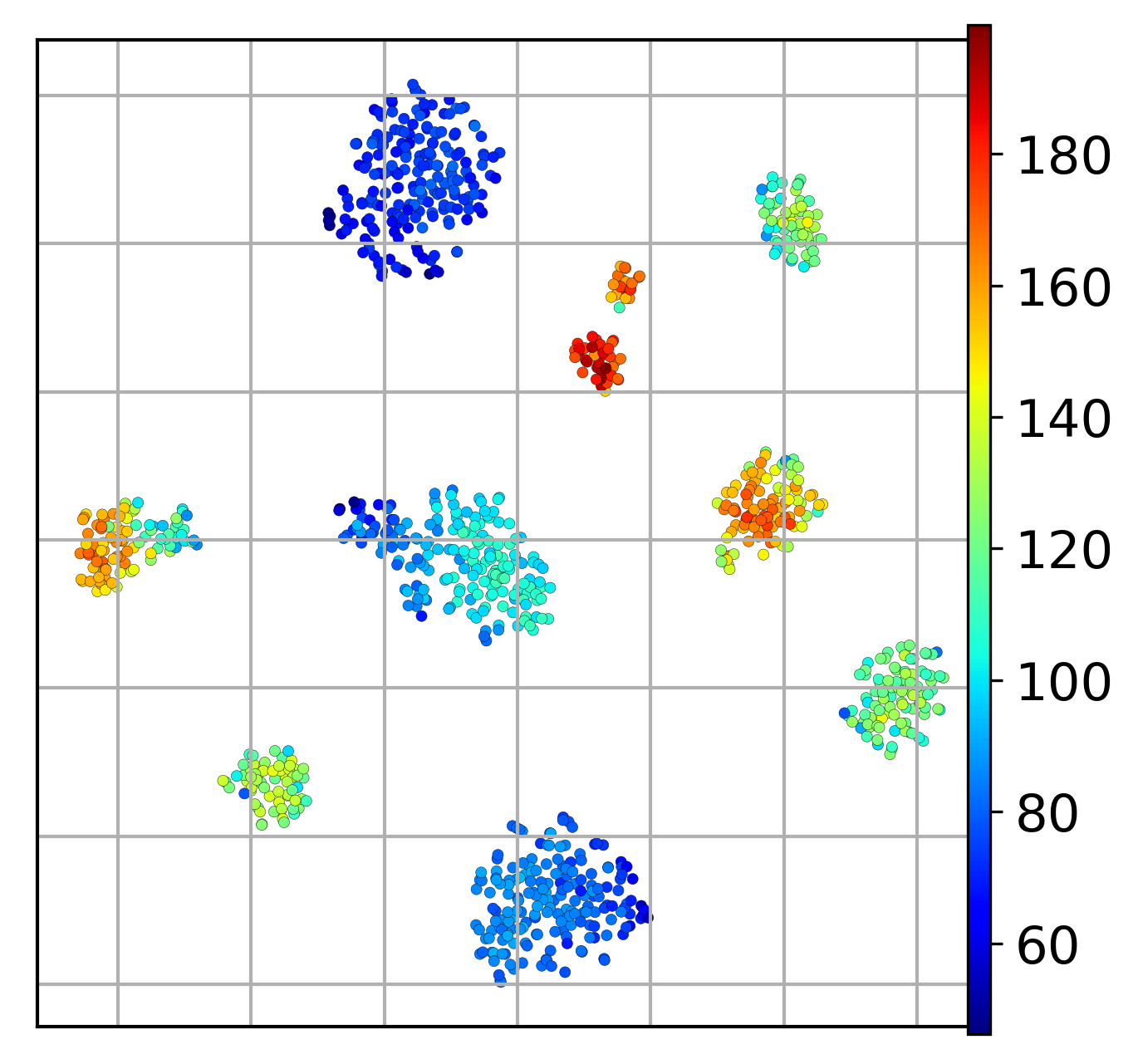}
  \vspace{-5pt}\phantomsection
\end{subfigure}%
\begin{subfigure}[b]{0.2\textwidth}
  \centering
  \includegraphics[width=\textwidth]{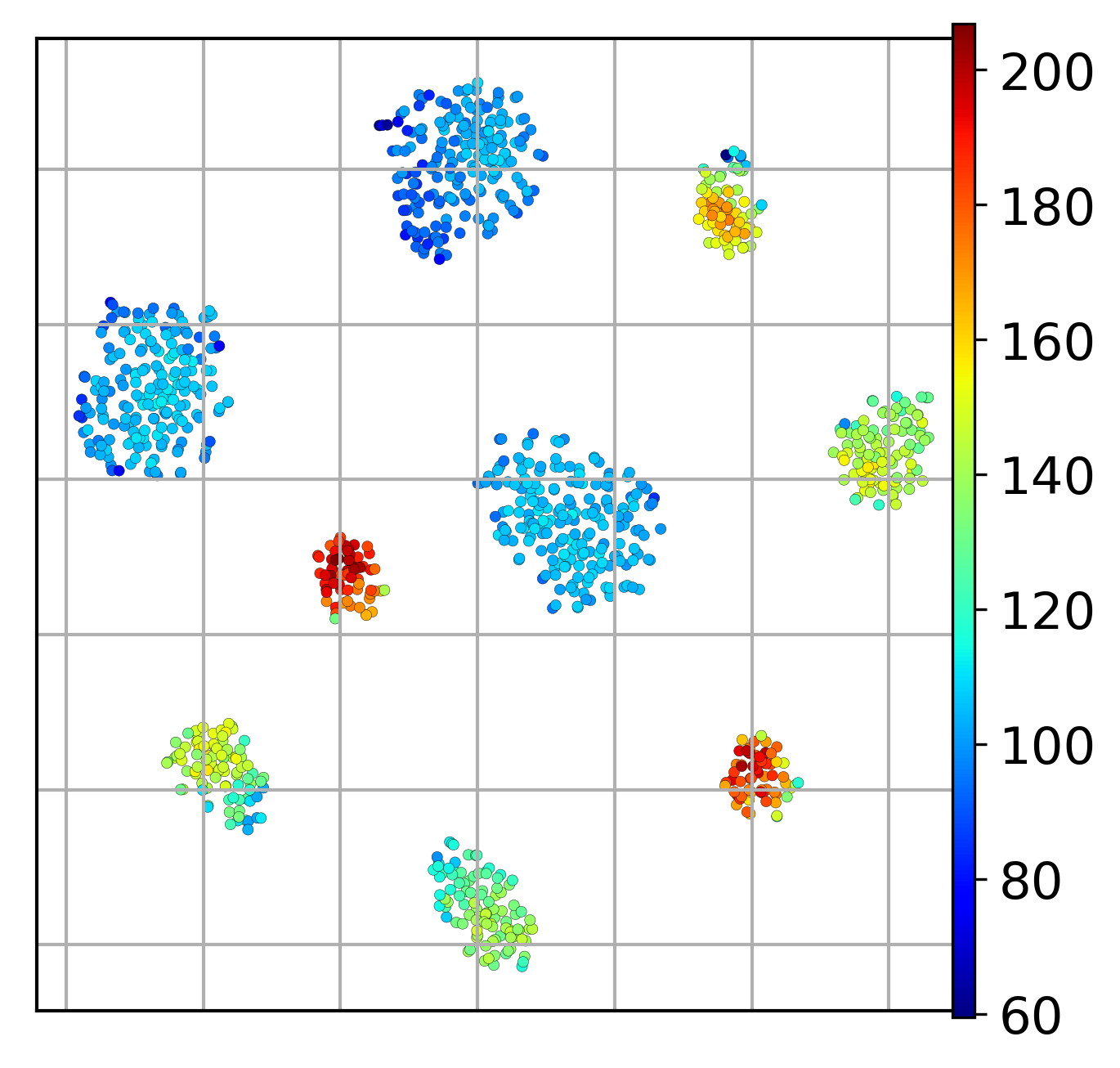}
  \vspace{-5pt}\phantomsection
\end{subfigure}

\begin{subfigure}[b]{0.2\textwidth}
  \centering
  \includegraphics[width=\textwidth]{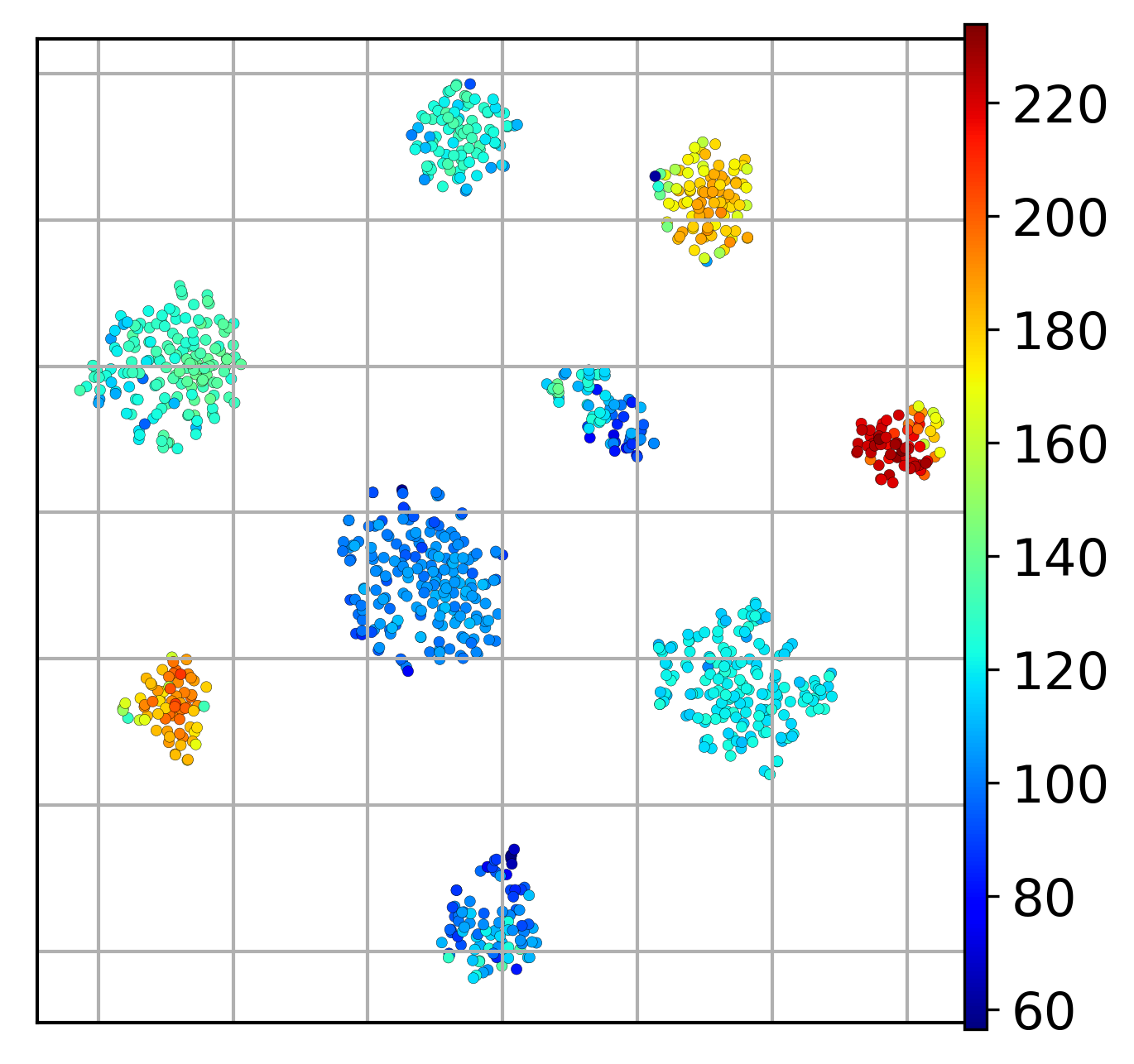}
  \vspace{-5pt}\phantomsection
\end{subfigure}%
\begin{subfigure}[b]{0.2\textwidth}
  \centering
  \includegraphics[width=\textwidth]{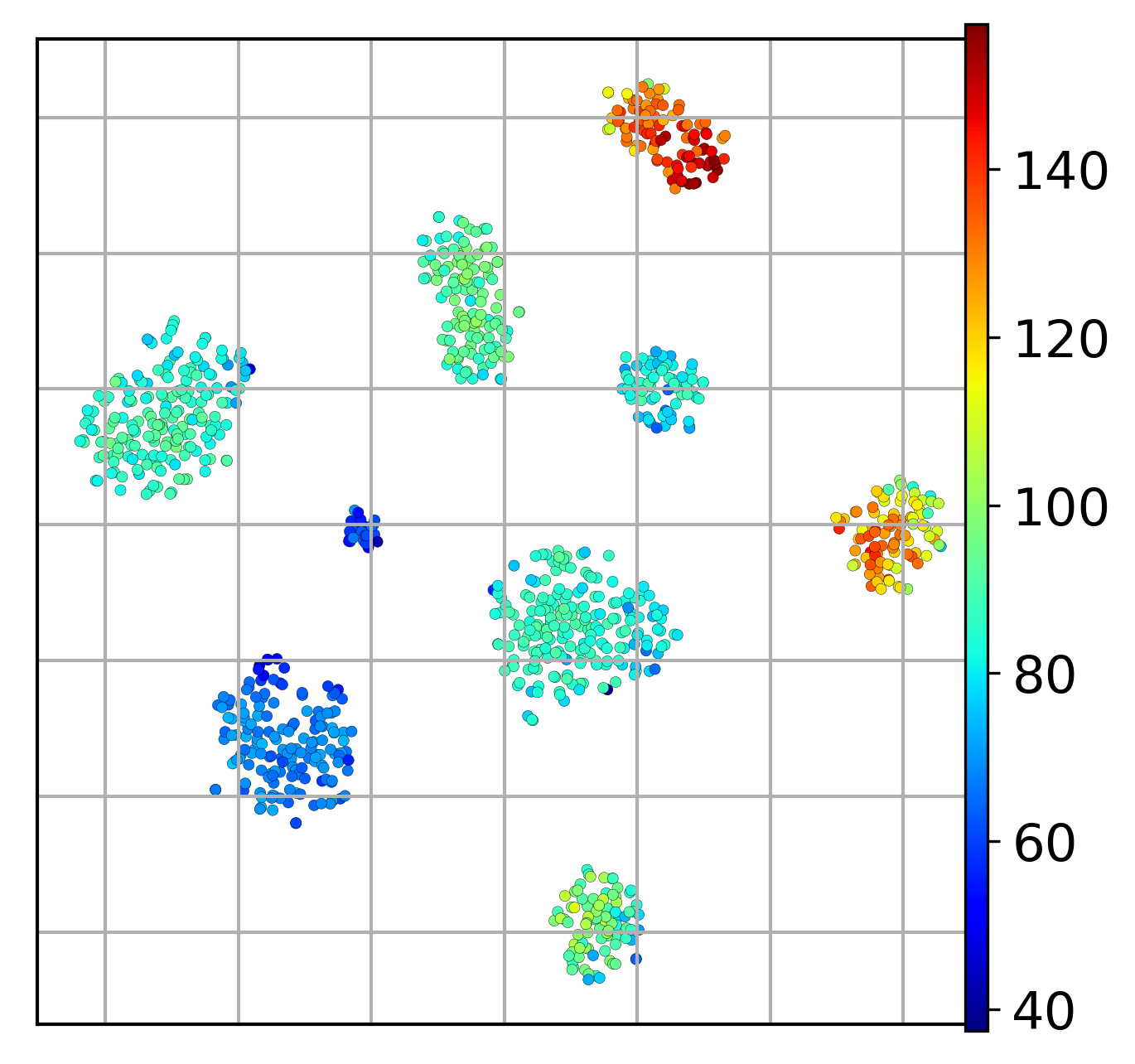}
  \vspace{-5pt}\phantomsection
\end{subfigure}%
\begin{subfigure}[b]{0.2\textwidth}
  \centering
  \includegraphics[width=\textwidth]{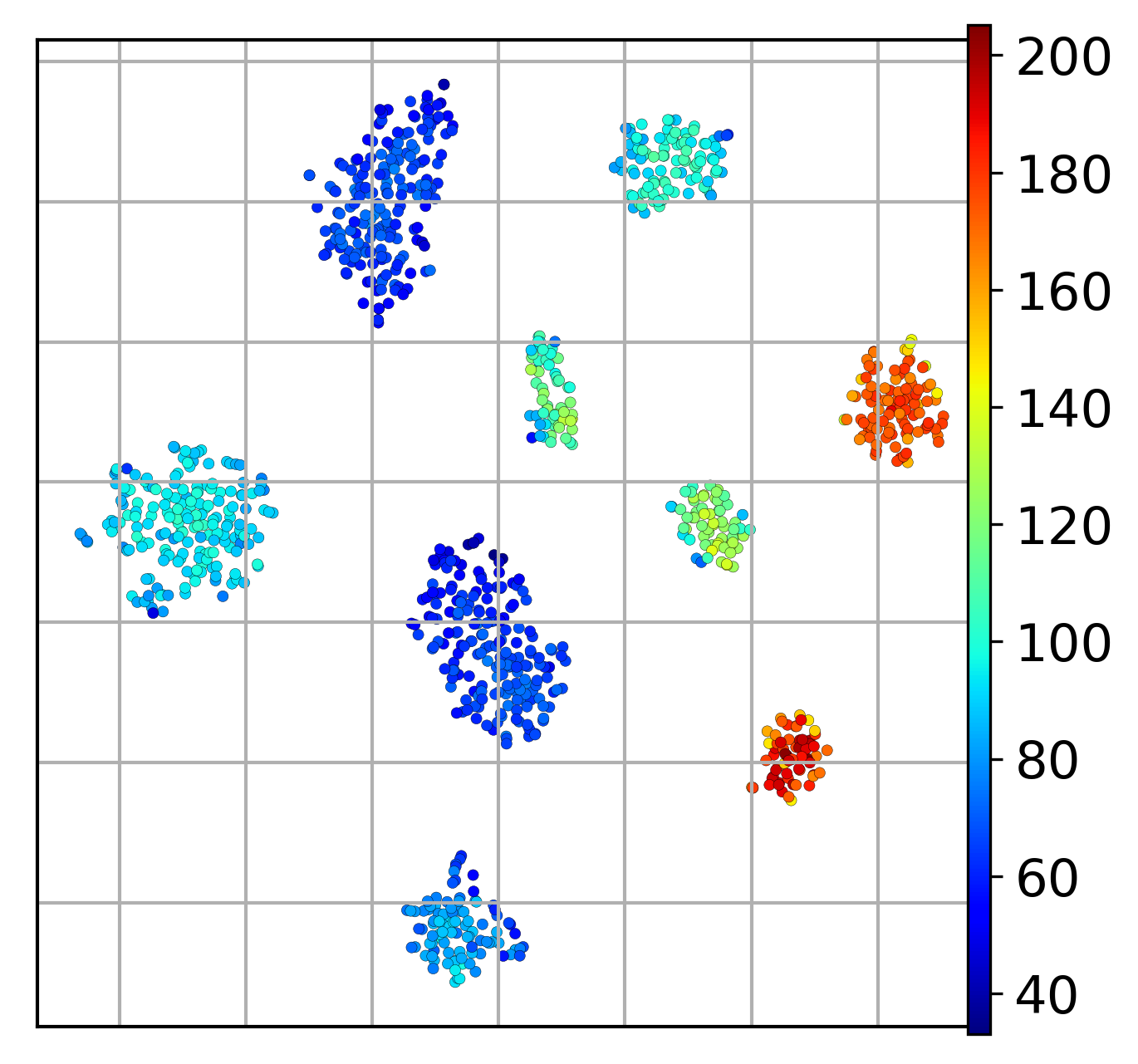}
  \vspace{-5pt}\phantomsection
\end{subfigure}%
\begin{subfigure}[b]{0.2\textwidth}
  \centering
  \includegraphics[width=\textwidth]{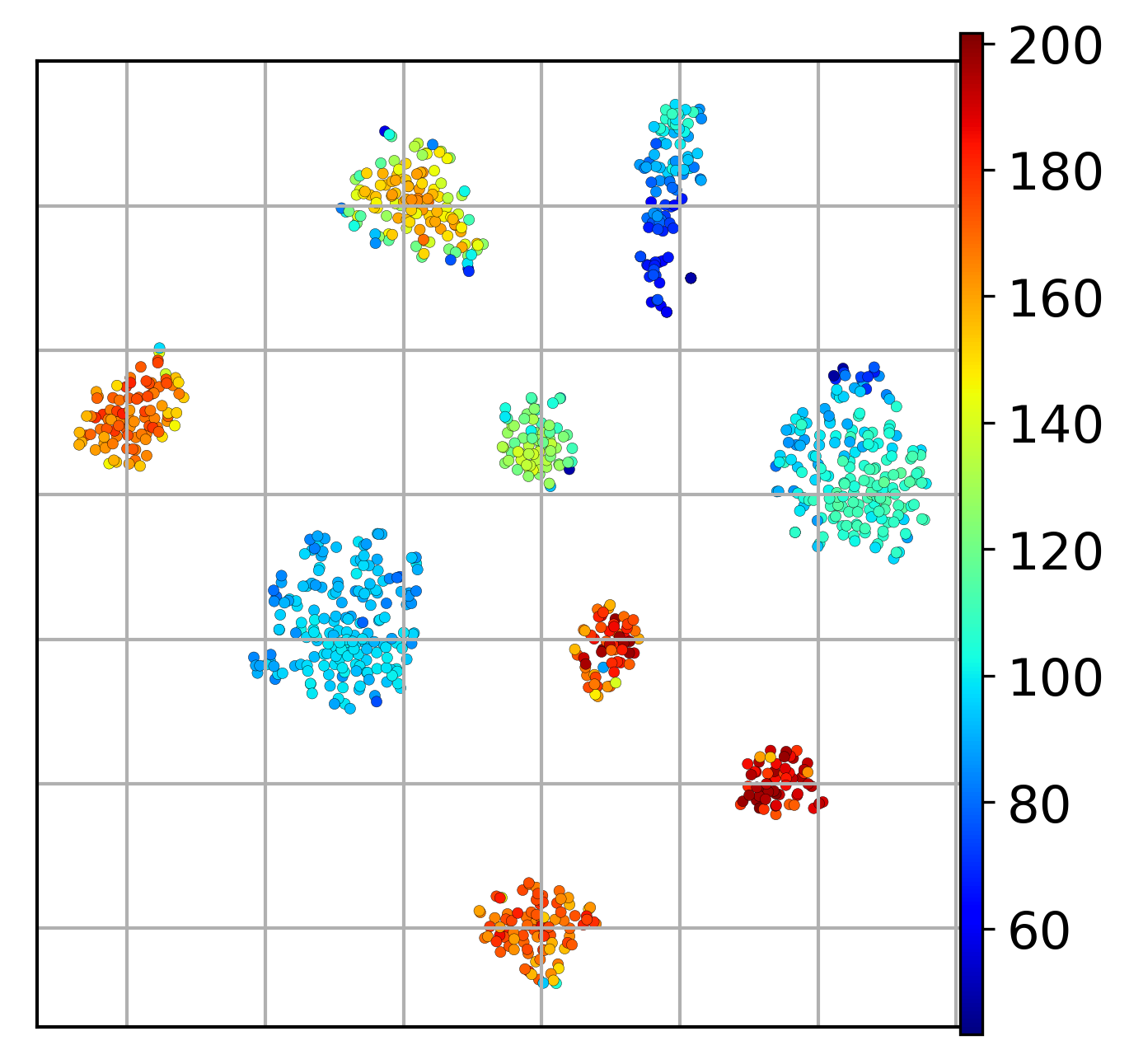}
  \vspace{-5pt}\phantomsection
\end{subfigure}%
\begin{subfigure}[b]{0.2\textwidth}
  \centering
  \includegraphics[width=\textwidth]{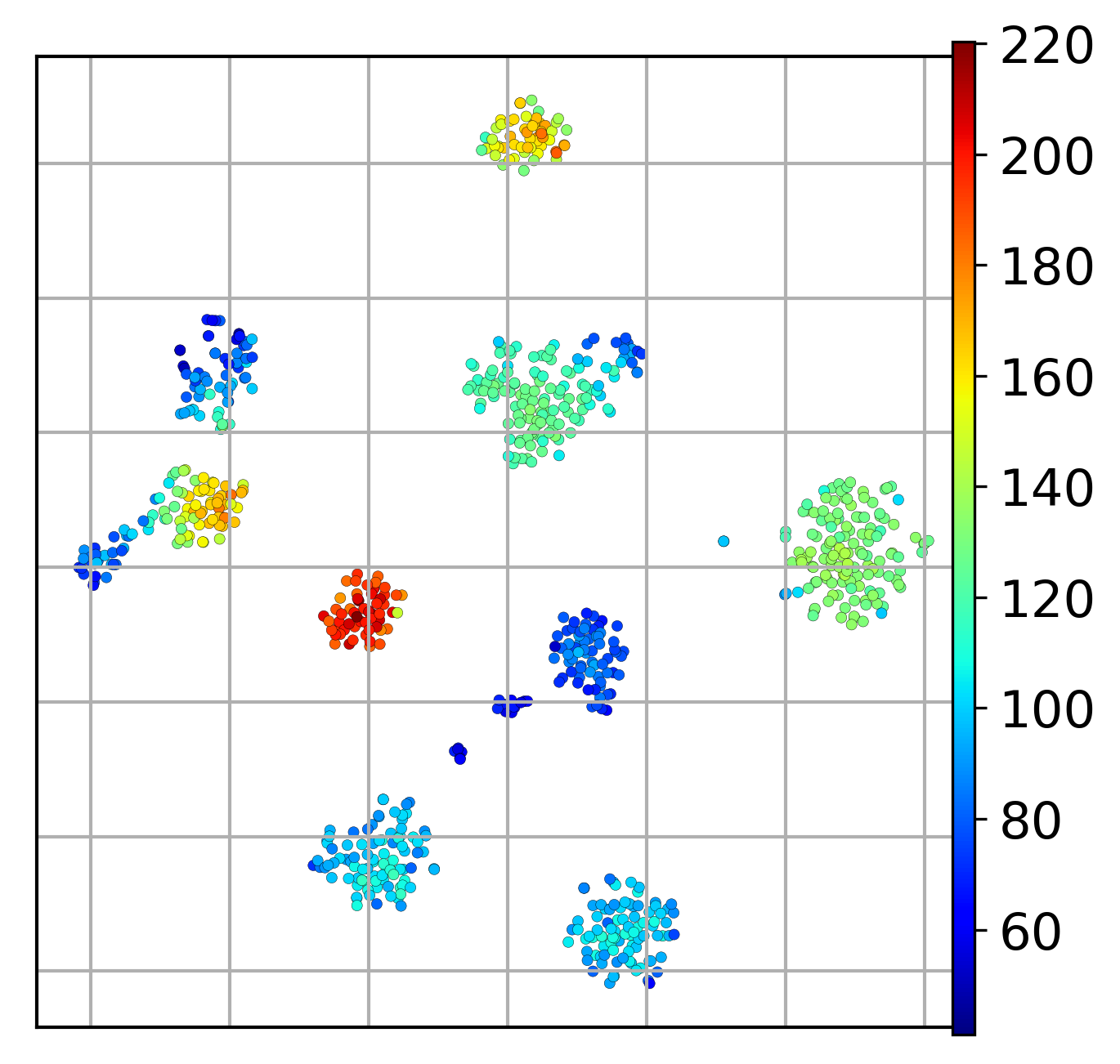}
  \vspace{-5pt}\phantomsection
\end{subfigure}

\begin{subfigure}[b]{0.2\textwidth}
  \centering
  \includegraphics[width=\textwidth]{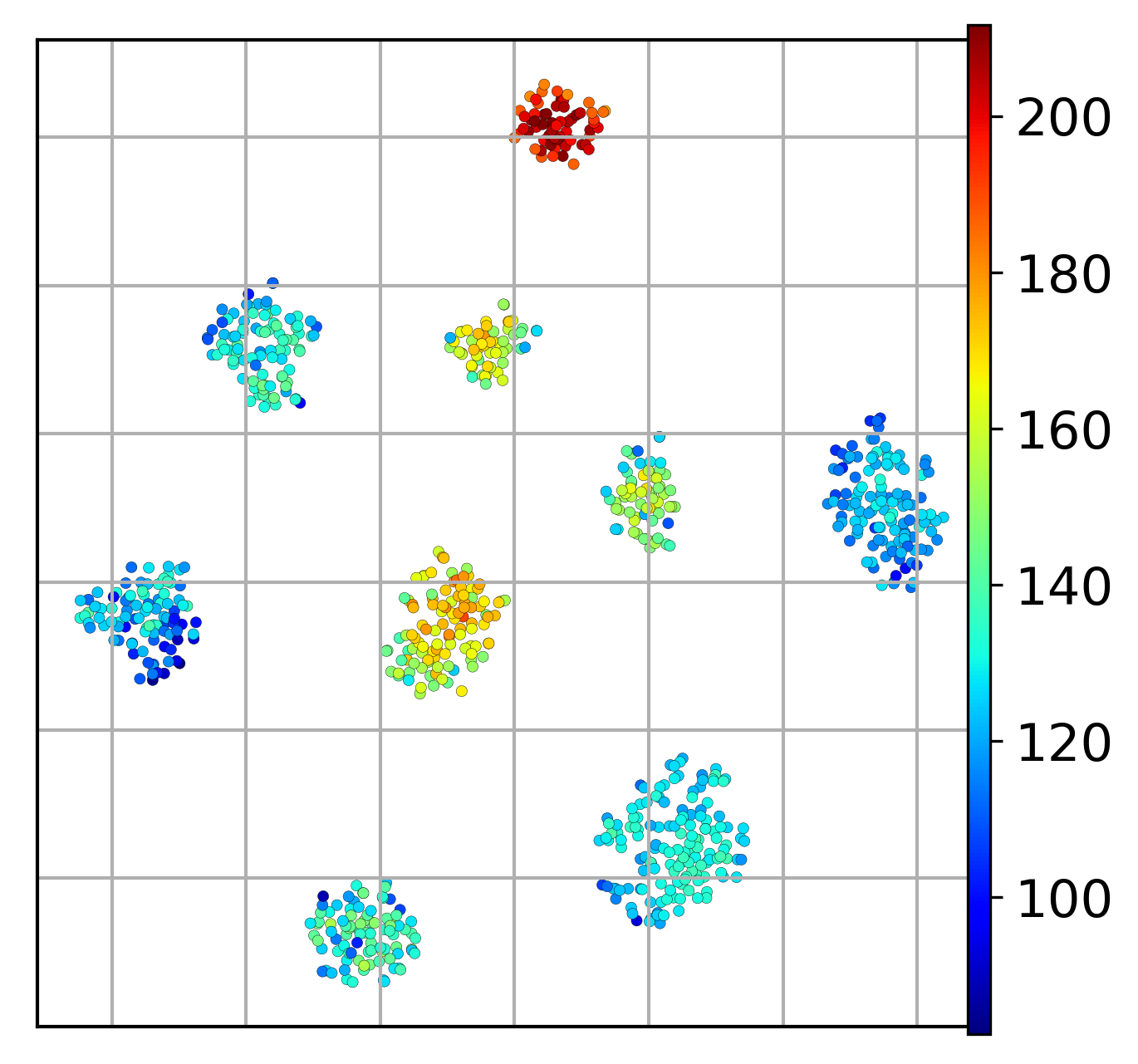}
  \vspace{-5pt}\phantomsection
\end{subfigure}%
\begin{subfigure}[b]{0.2\textwidth}
  \centering
  \includegraphics[width=\textwidth]{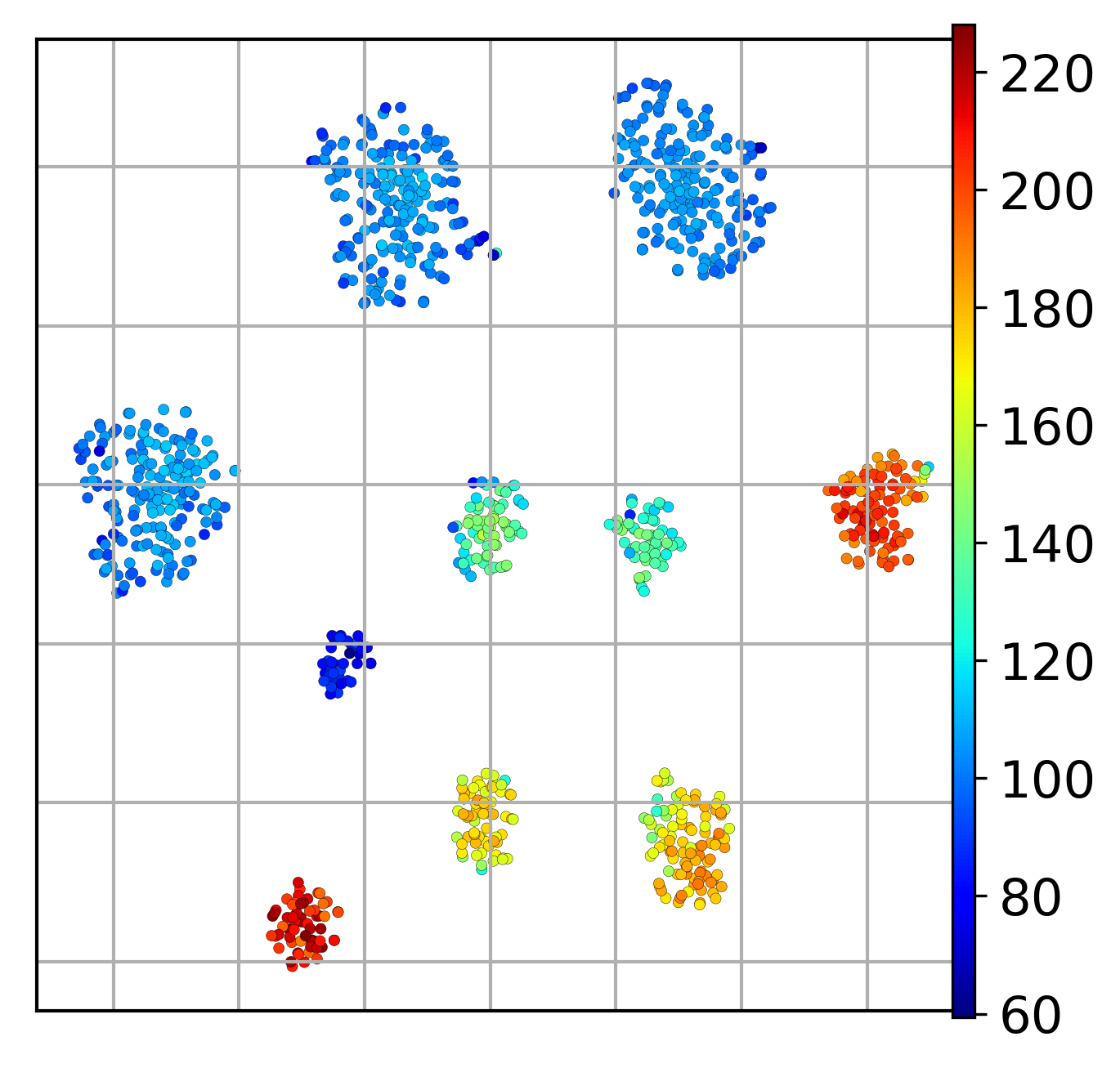}
  \vspace{-5pt}\phantomsection
\end{subfigure}%
\begin{subfigure}[b]{0.2\textwidth}
  \centering
  \includegraphics[width=\textwidth]{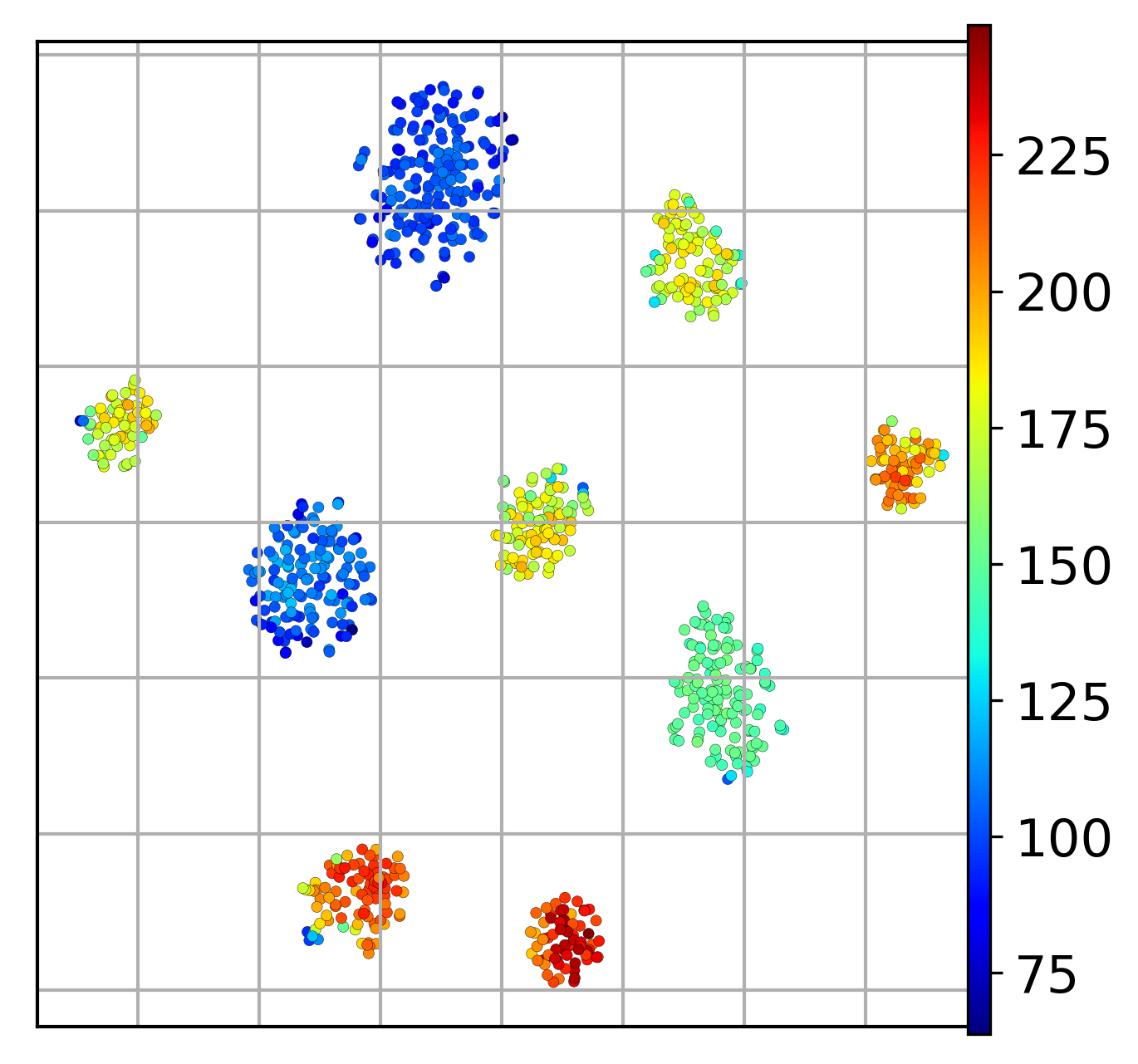}
  \vspace{-5pt}\phantomsection
\end{subfigure}%
\begin{subfigure}[b]{0.2\textwidth}
  \centering
  \includegraphics[width=\textwidth]{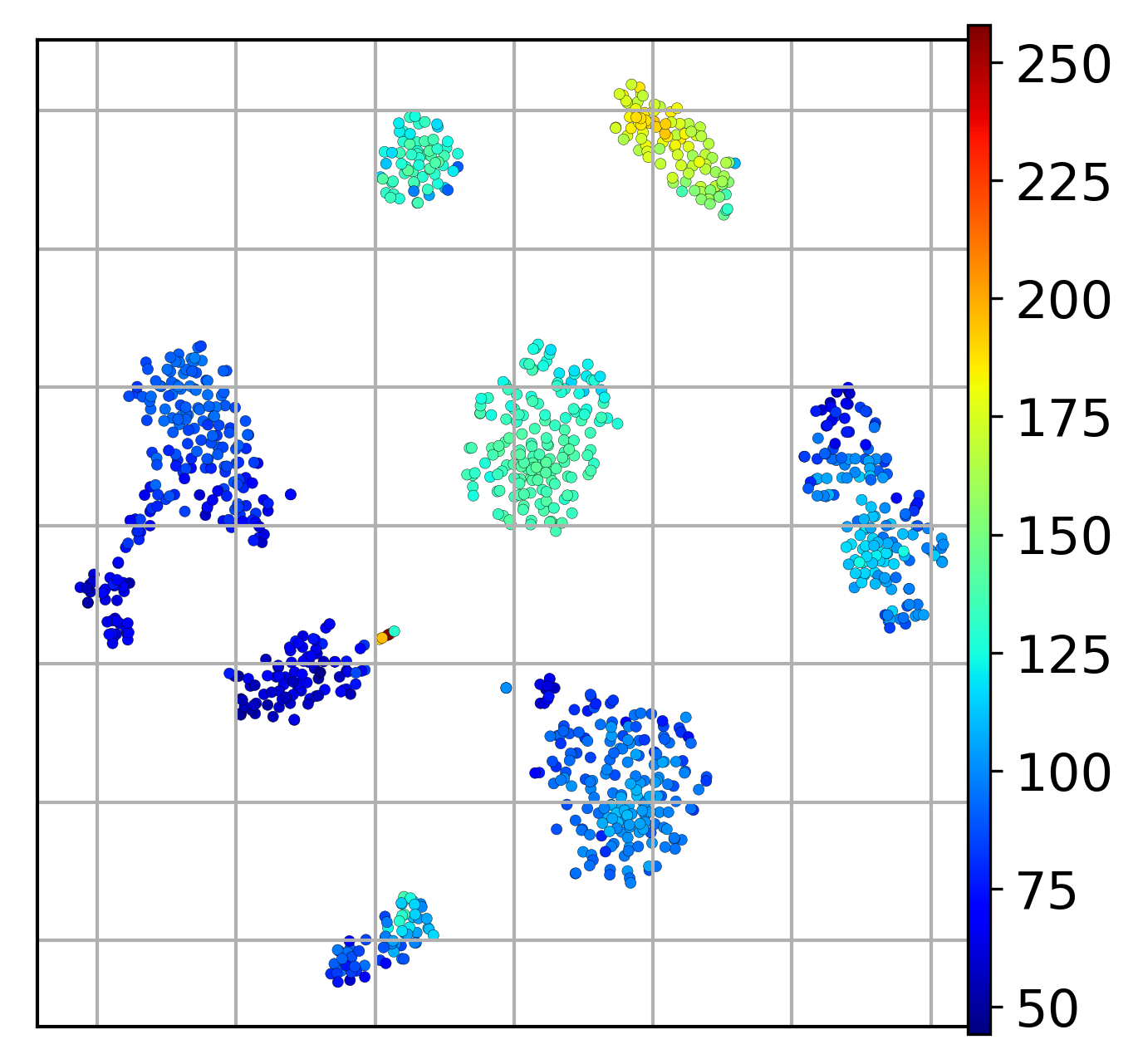}
  \vspace{-5pt}\phantomsection
\end{subfigure}%
\begin{subfigure}[b]{0.2\textwidth}
  \centering
  \includegraphics[width=\textwidth]{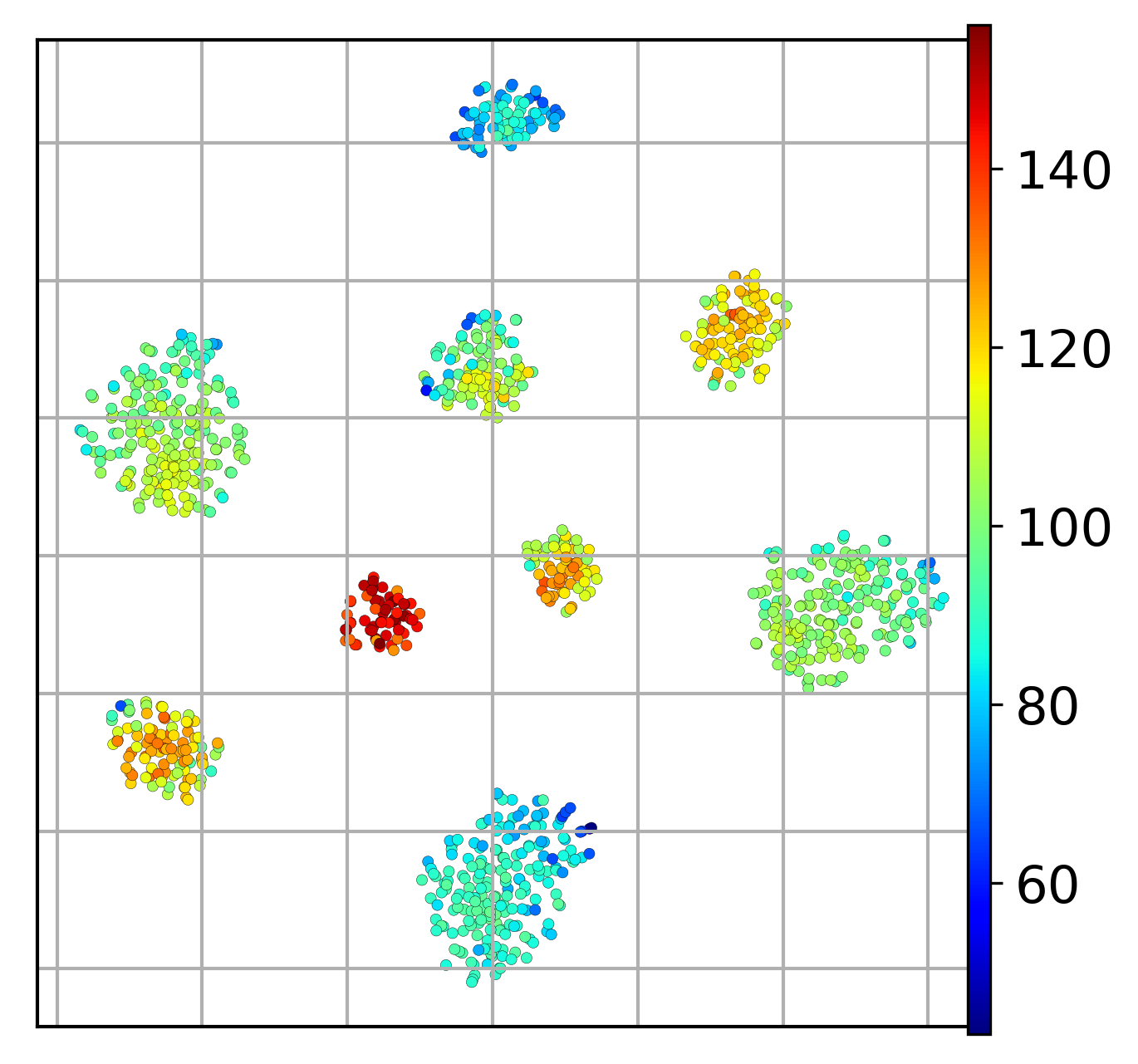}
  \vspace{-5pt}\phantomsection
\end{subfigure}

\begin{subfigure}[b]{0.2\textwidth}
  \centering
  \includegraphics[width=\textwidth]{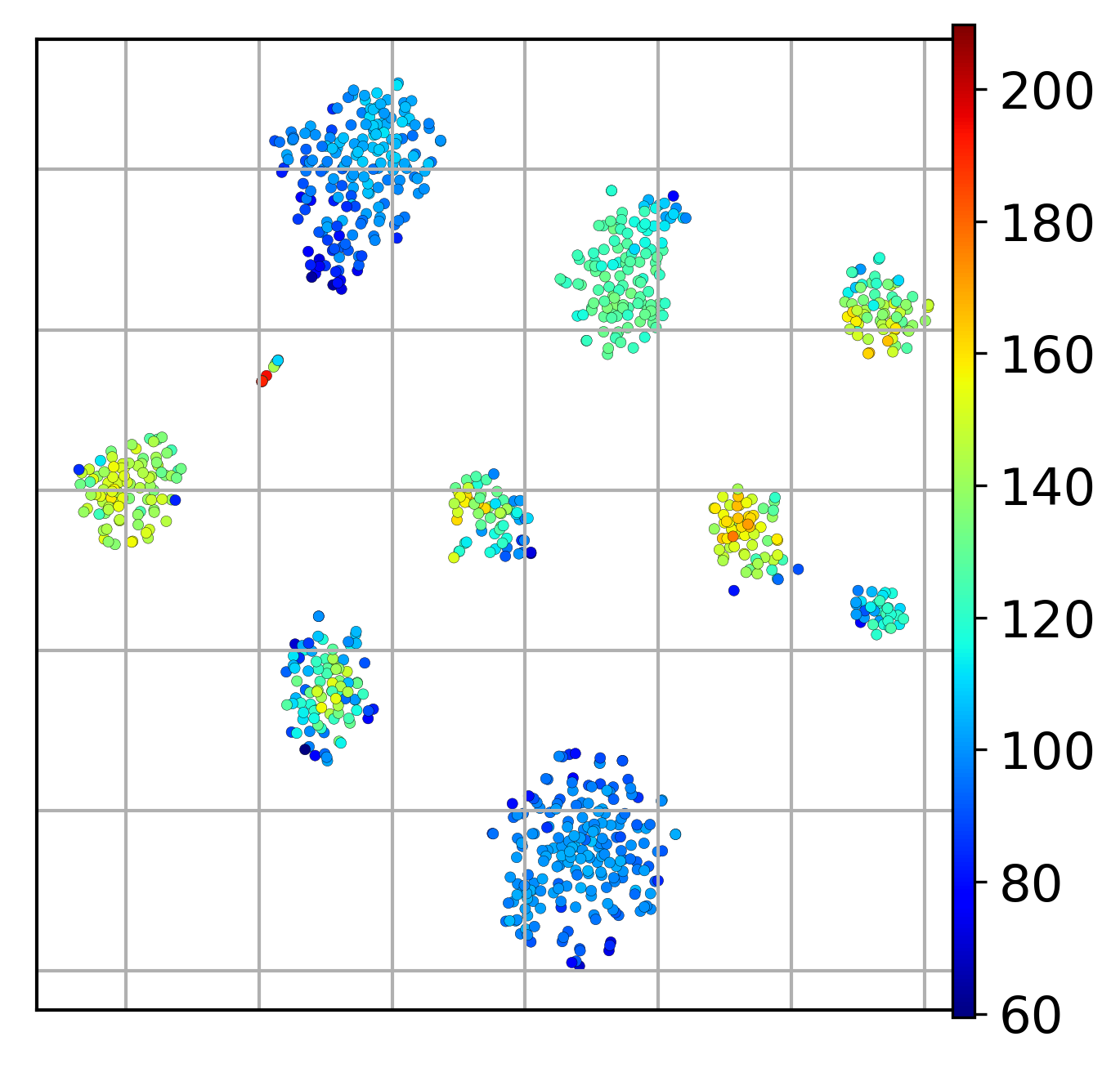}
  \vspace{-5pt}\phantomsection
\end{subfigure}%
\begin{subfigure}[b]{0.2\textwidth}
  \centering
  \includegraphics[width=\textwidth]{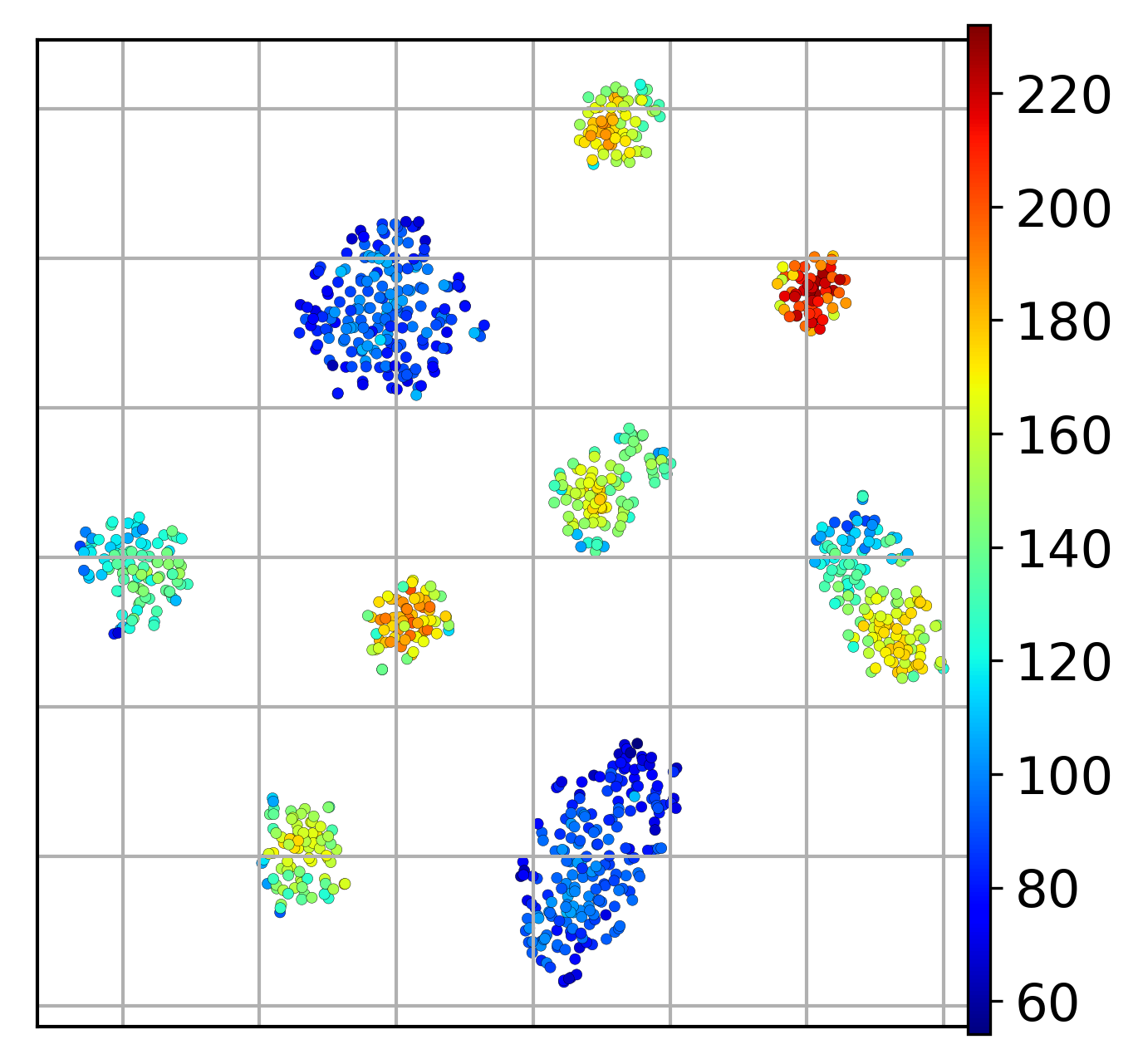}
  \vspace{-5pt}\phantomsection
\end{subfigure}%
\begin{subfigure}[b]{0.2\textwidth}
  \centering
  \includegraphics[width=\textwidth]{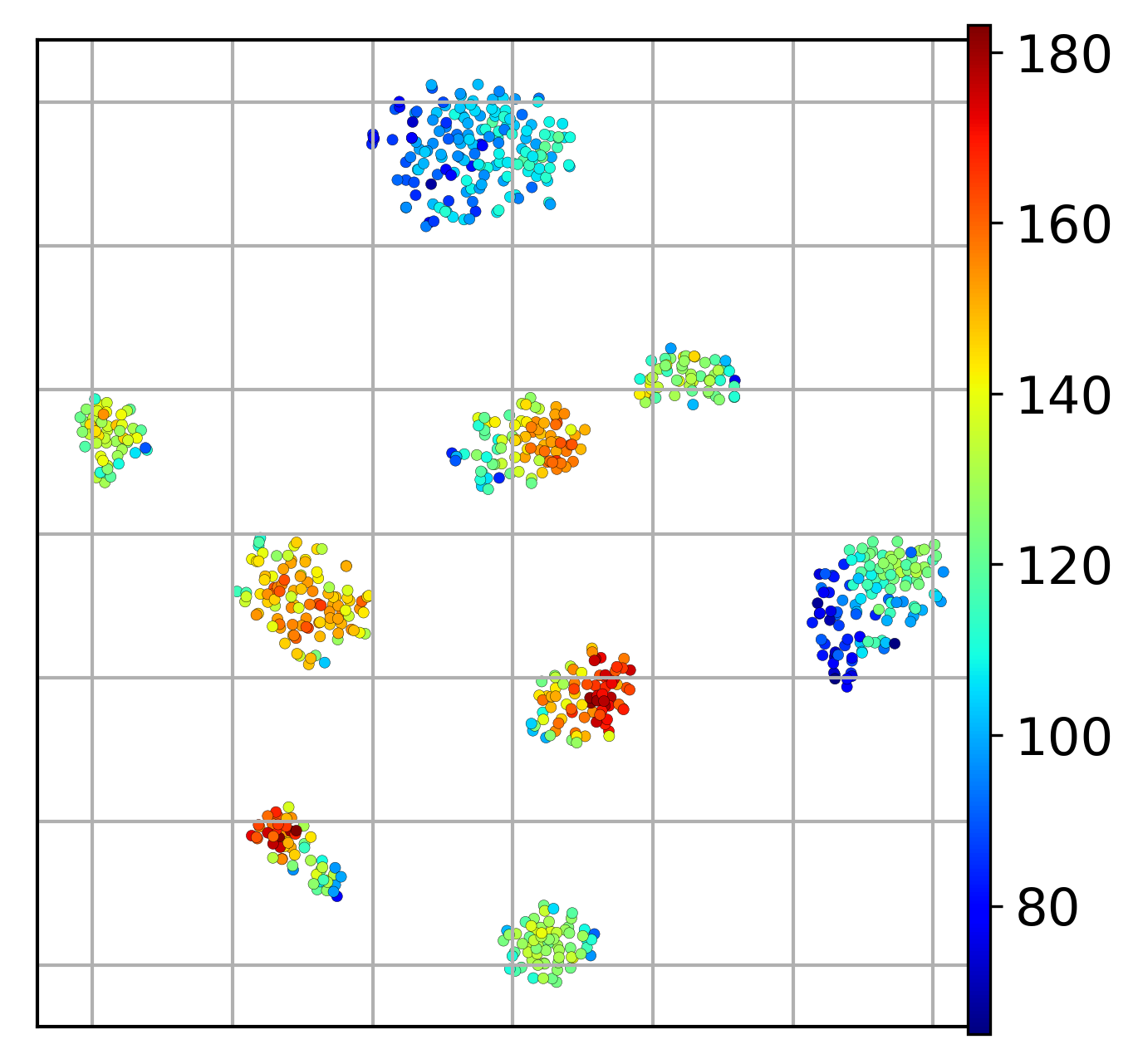}
  \vspace{-5pt}\phantomsection
\end{subfigure}%
\begin{subfigure}[b]{0.2\textwidth}
  \centering
  \includegraphics[width=\textwidth]{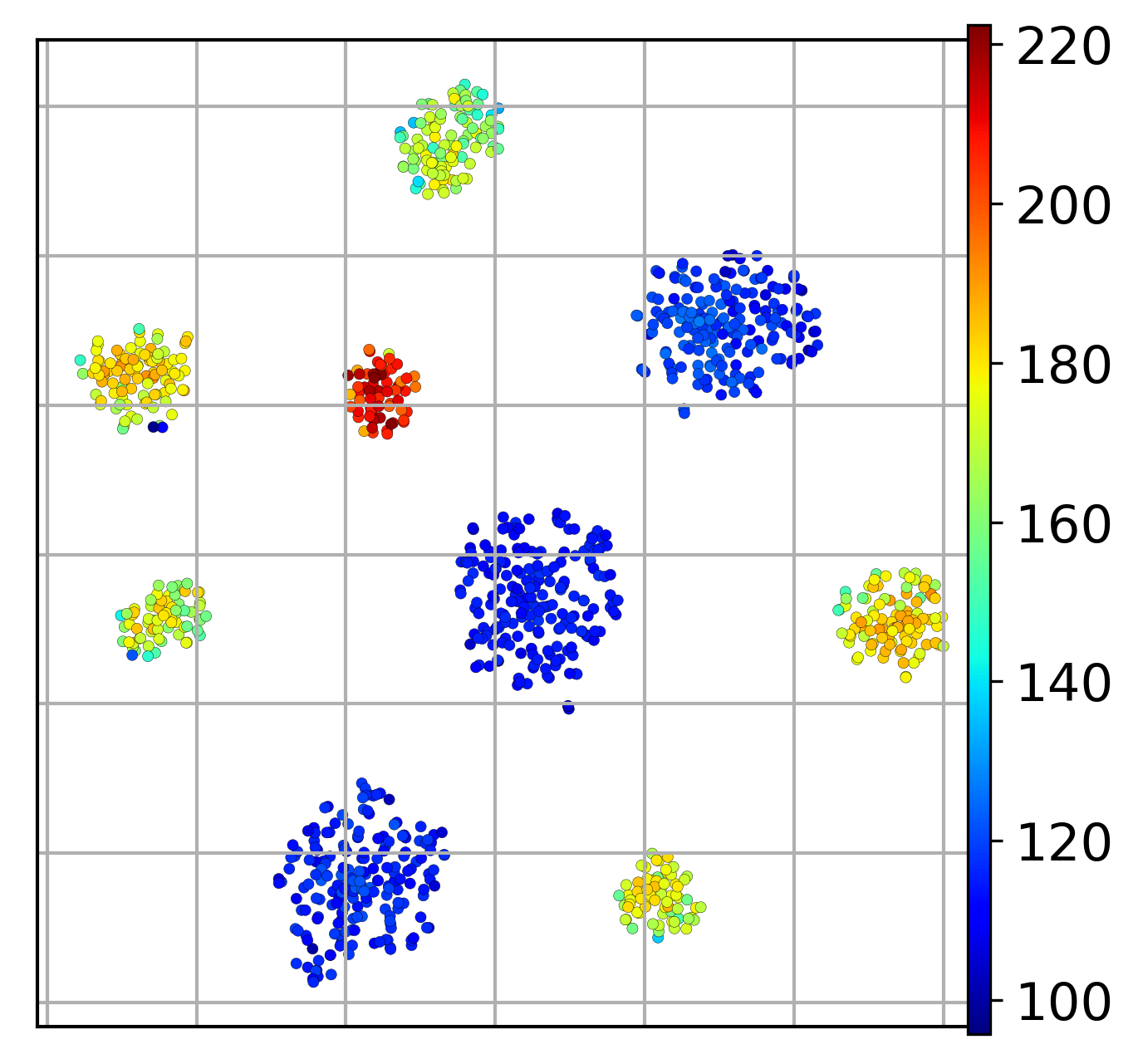}
  \vspace{-5pt}\phantomsection
\end{subfigure}%
\begin{subfigure}[b]{0.2\textwidth}
  \centering
  \includegraphics[width=\textwidth]{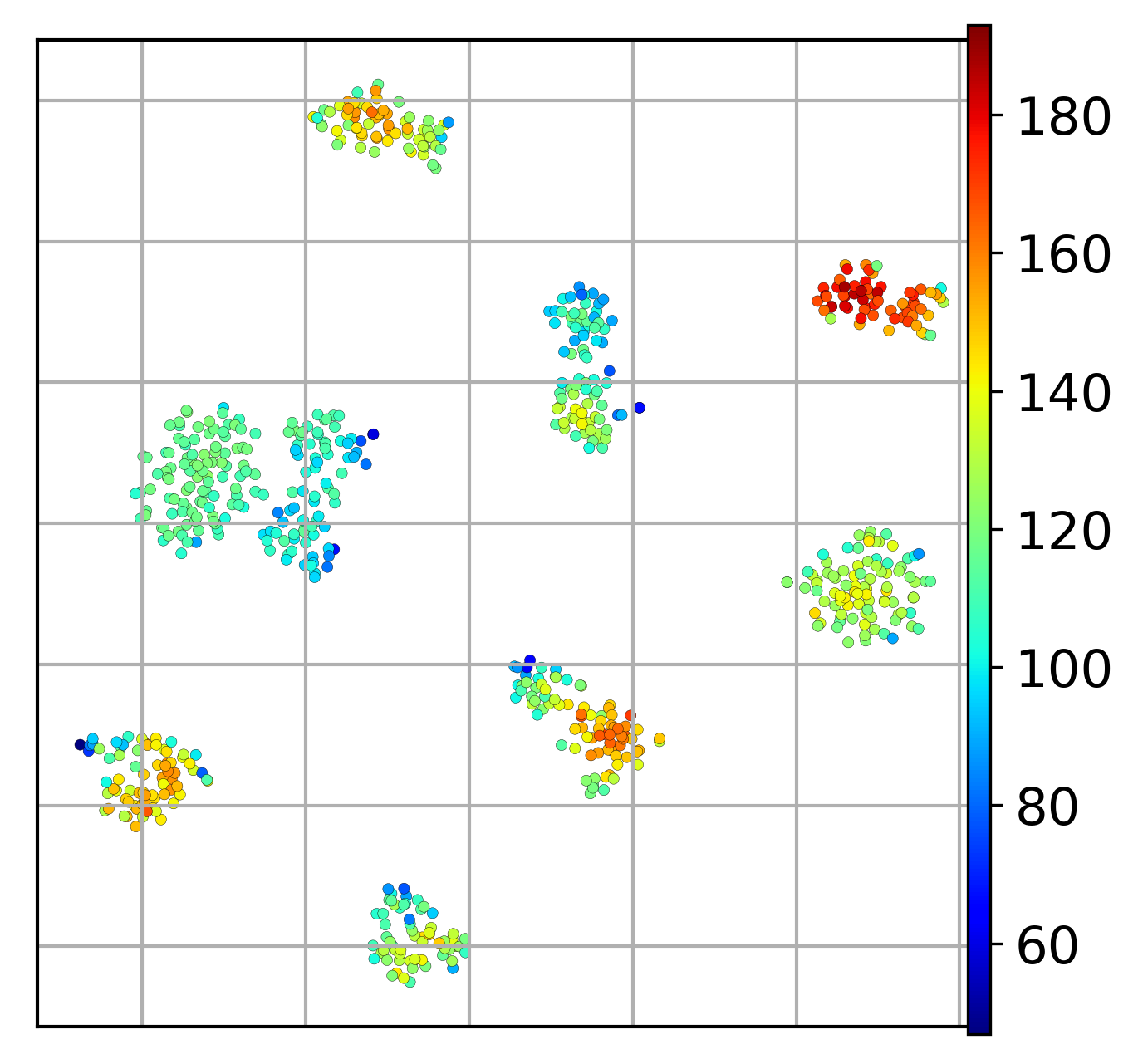}
  \vspace{-5pt}\phantomsection
\end{subfigure}

\begin{subfigure}[b]{0.2\textwidth}
  \centering
  \includegraphics[width=\textwidth]{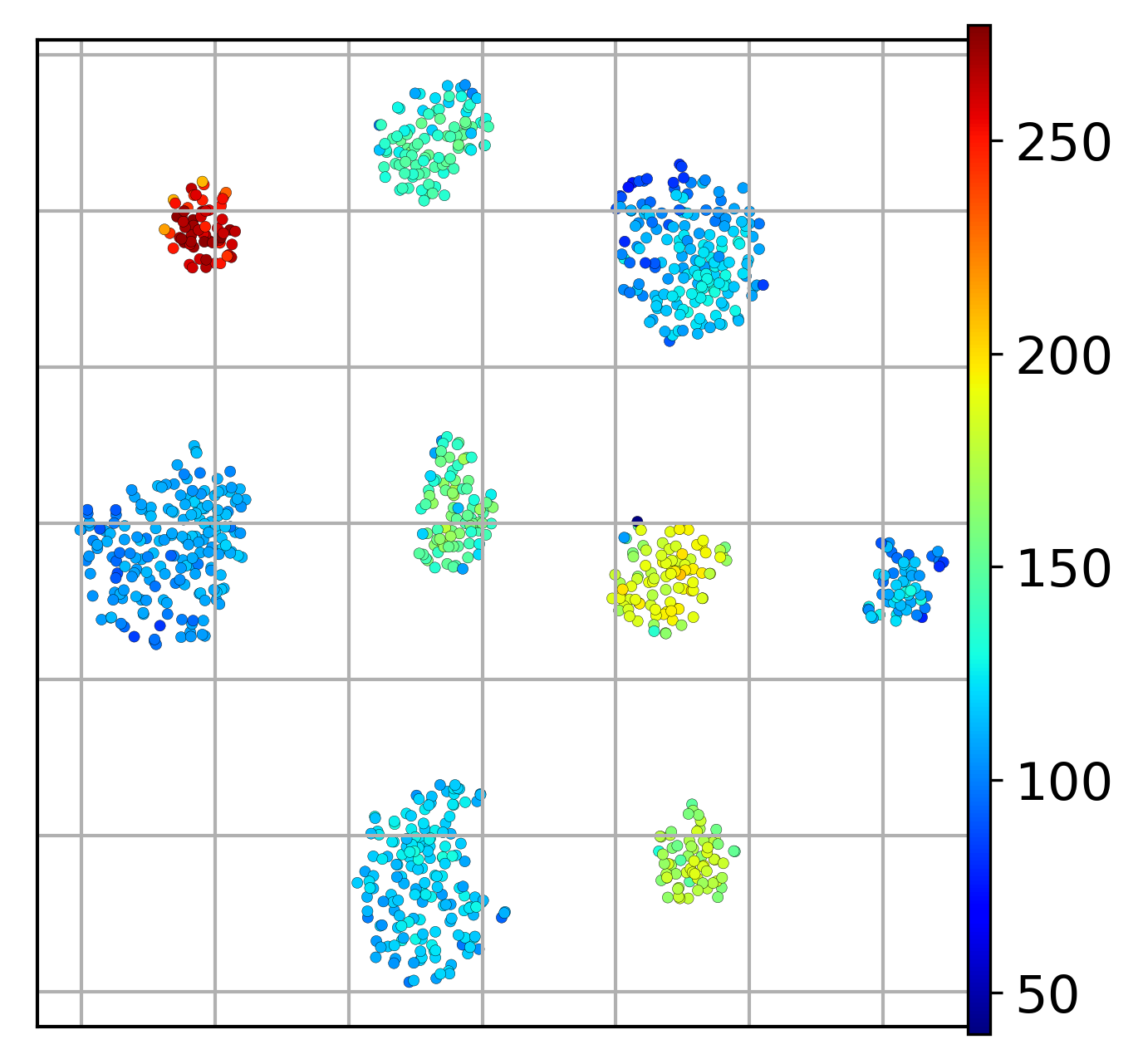}
  \vspace{-5pt}\phantomsection
\end{subfigure}%
\begin{subfigure}[b]{0.2\textwidth}
  \centering
  \includegraphics[width=\textwidth]{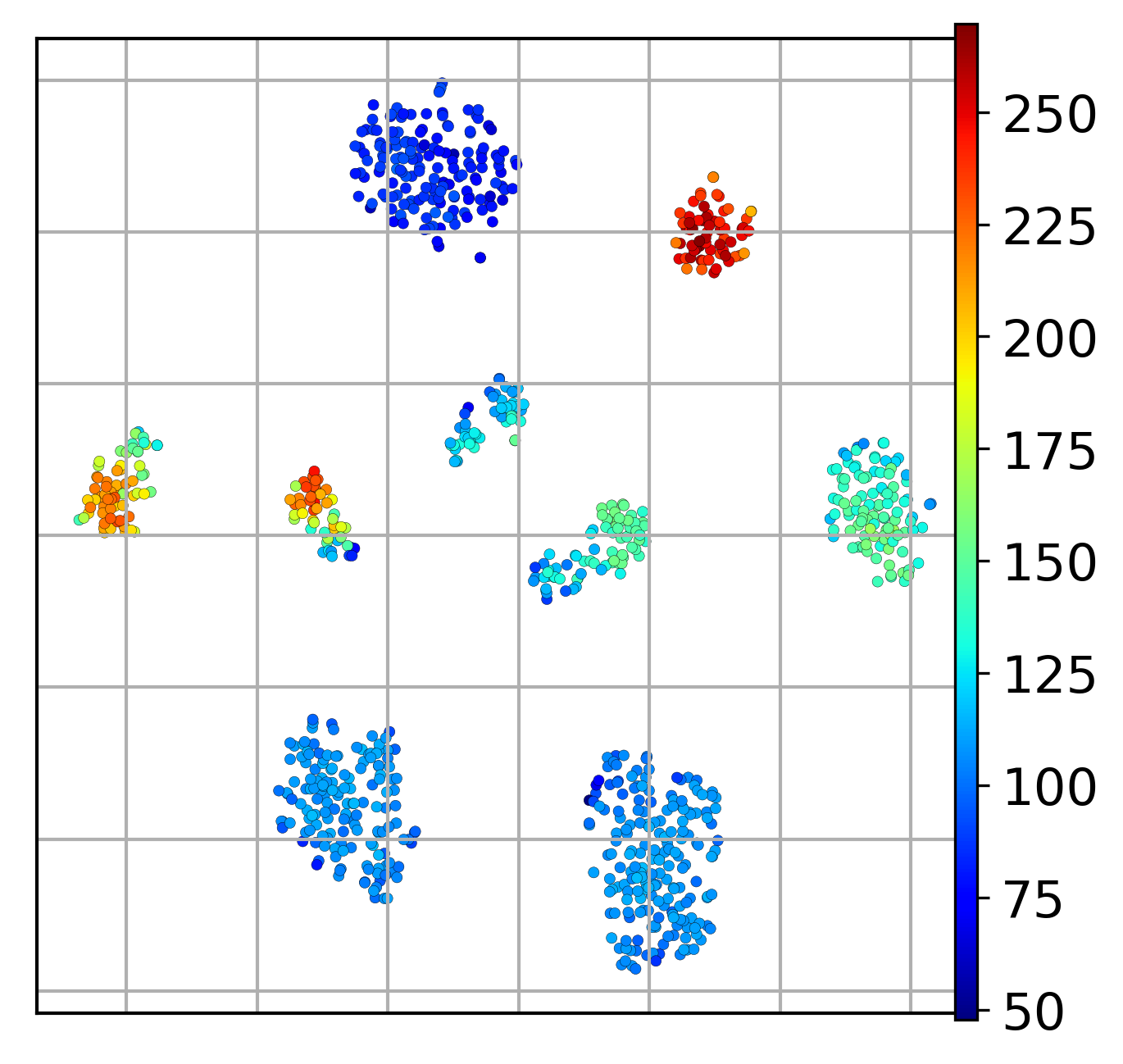}
  \vspace{-5pt}\phantomsection
\end{subfigure}%
\begin{subfigure}[b]{0.2\textwidth}
  \centering
  \includegraphics[width=\textwidth]{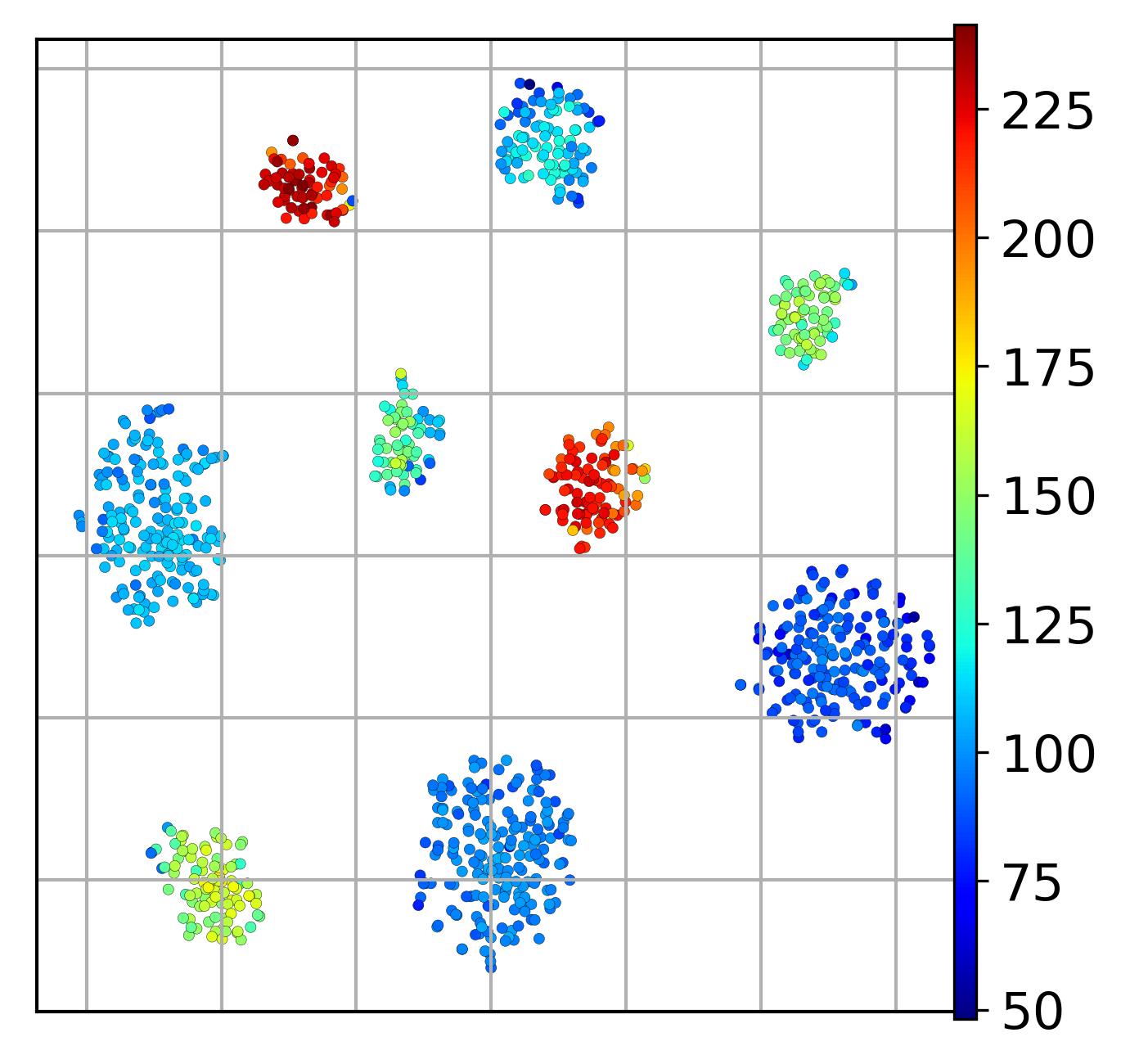}
  \vspace{-5pt}\phantomsection
\end{subfigure}%
\begin{subfigure}[b]{0.2\textwidth}
  \centering
  \includegraphics[width=\textwidth]{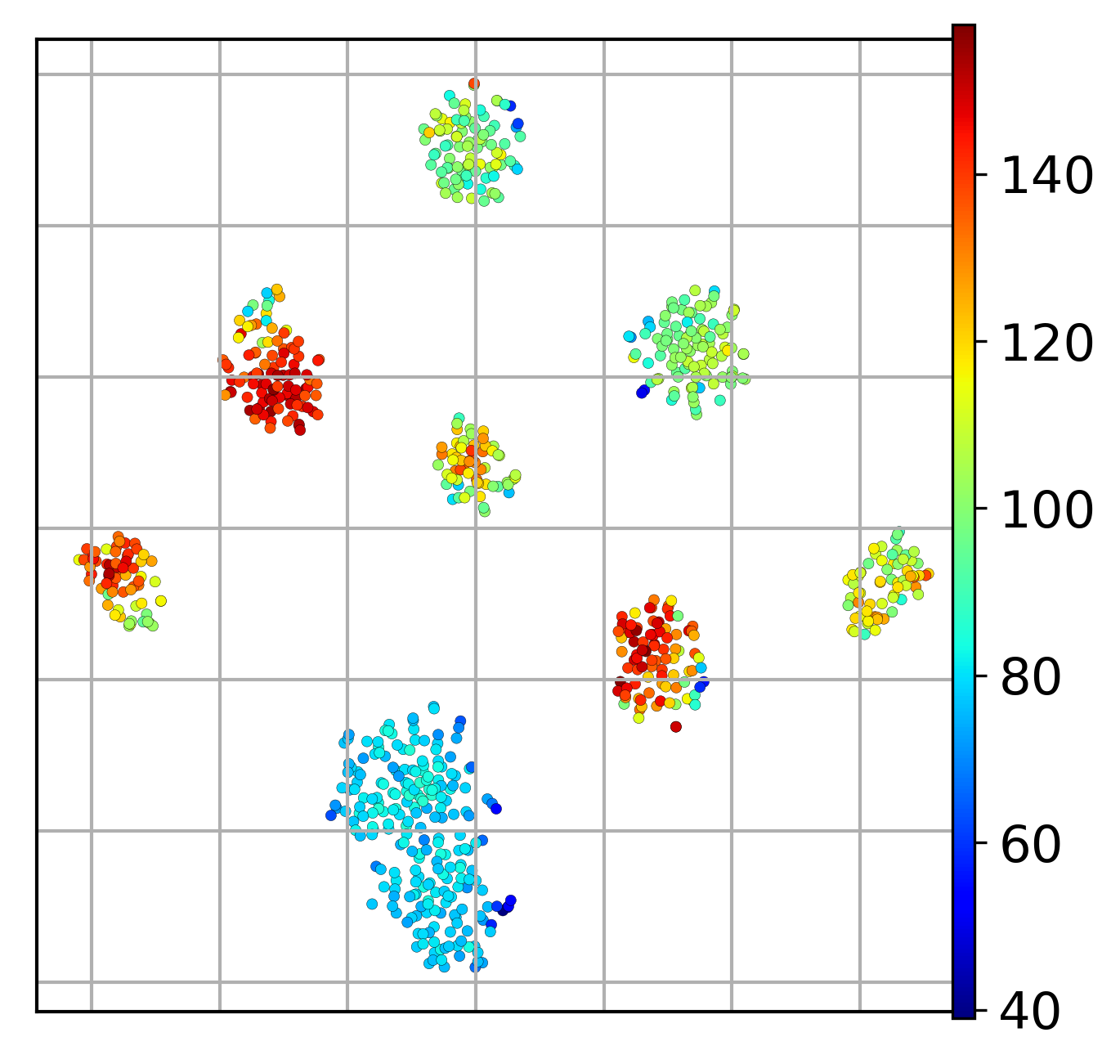}
  \vspace{-5pt}\phantomsection
\end{subfigure}%
\begin{subfigure}[b]{0.2\textwidth}
  \centering
  \includegraphics[width=\textwidth]{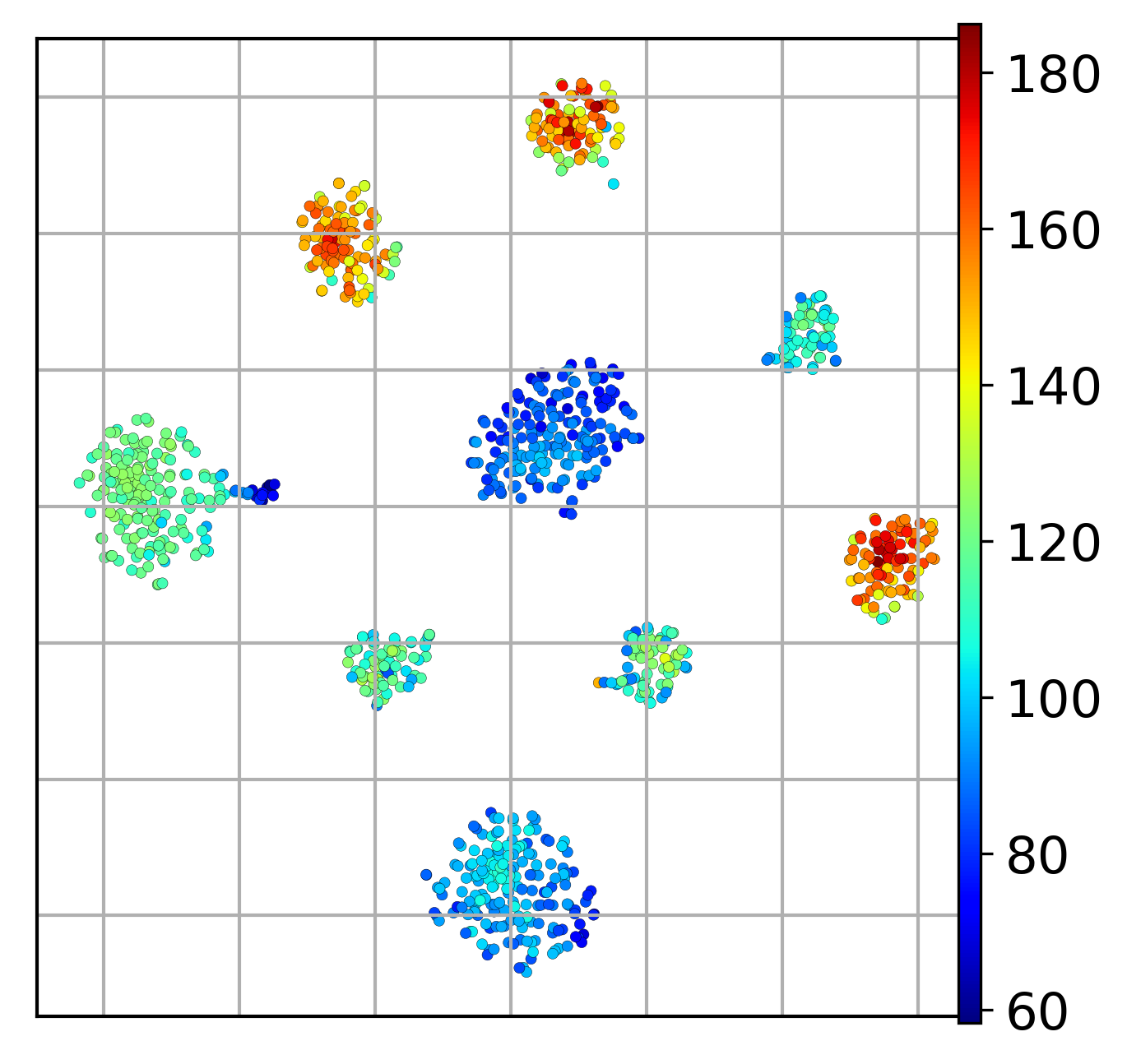}
  \vspace{-5pt}\phantomsection
\end{subfigure}
\caption{t-SNE visualization of eigenfunctions from all trials of 25 participants, color-coded by their corresponding density ratios. Each subplot represents one participant. The density ratios for the three movements (reaching, grasping, and twisting) show consistent values: reaching is the lowest, grasping is in the middle, and twisting is the highest. For cluster labels, refer to Fig.~\ref{clustering_figure}.}\label{density_ratio_file}\end{figure}

\clearpage
\begin{figure}[t]
\centering
{\large\textbf{\texttt{TSNE, SUB1$\sim$SUB25, MOV1$\sim$MOV3}}
}\begin{subfigure}[b]{0.2\textwidth}
  \centering
  \includegraphics[width=\textwidth]{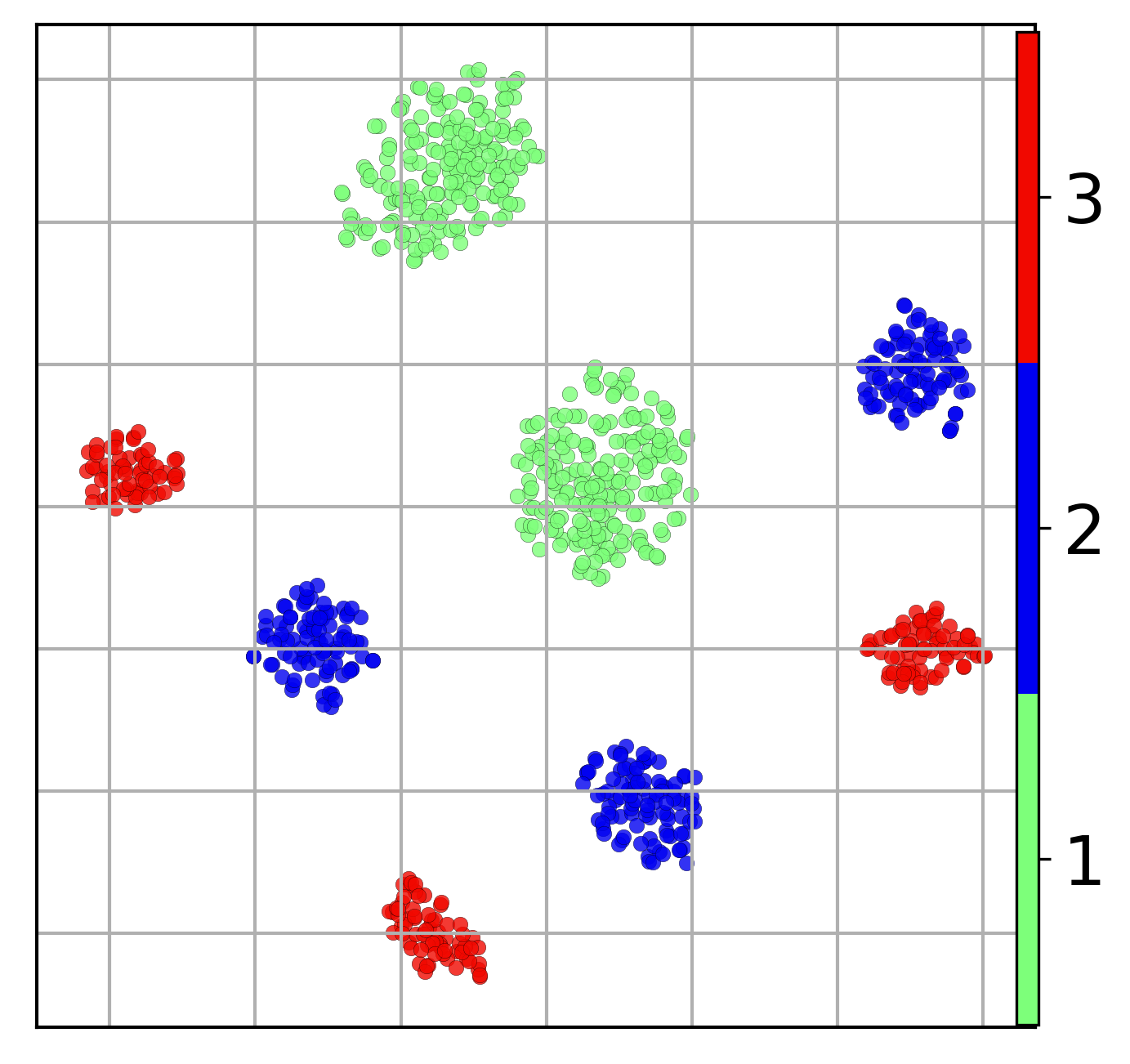}
  \vspace{-5pt}\phantomsection
\end{subfigure}%
\begin{subfigure}[b]{0.2\textwidth}
  \centering
  \includegraphics[width=\textwidth]{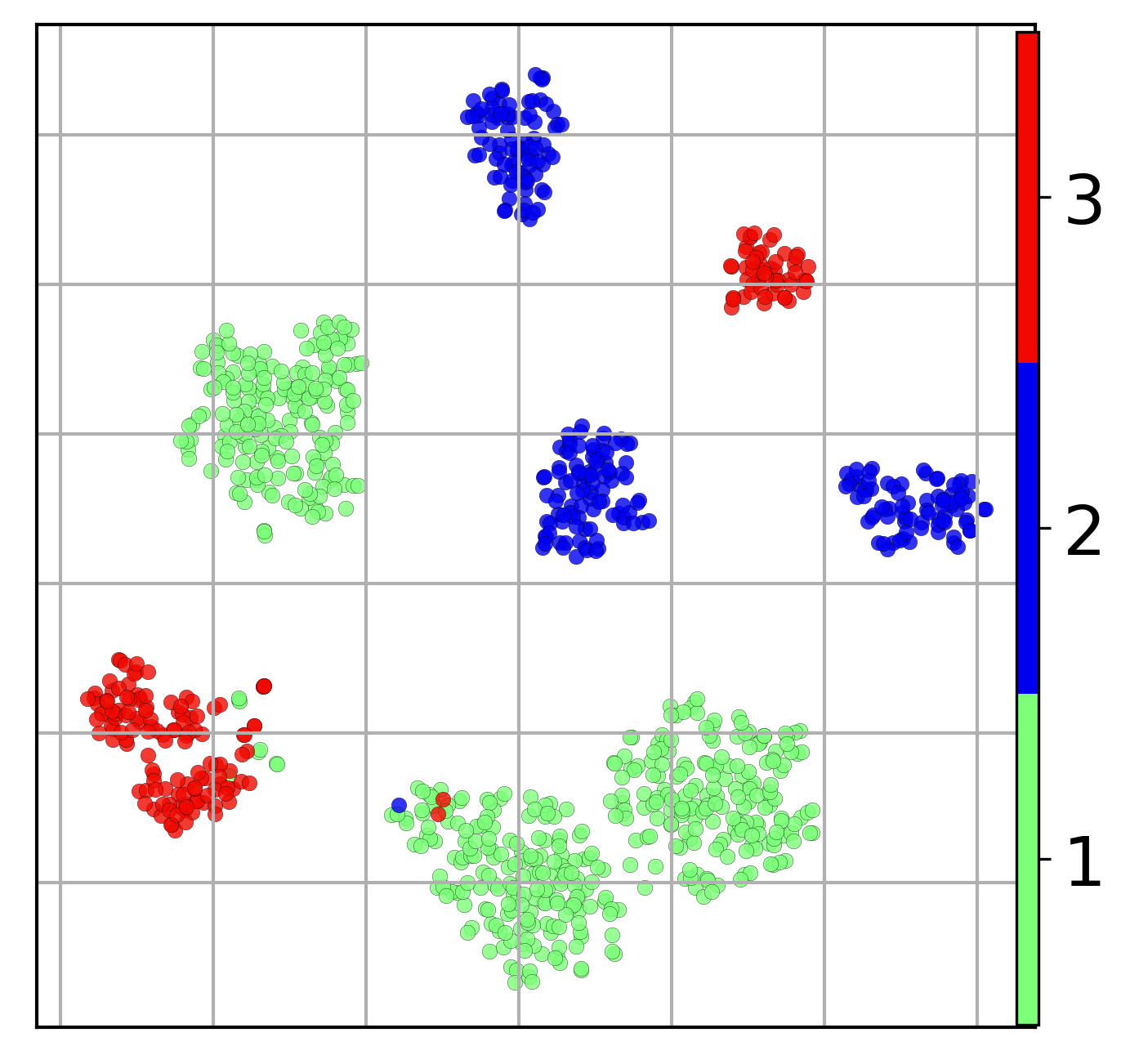}
  \vspace{-5pt}\phantomsection
\end{subfigure}%
\begin{subfigure}[b]{0.2\textwidth}
  \centering
  \includegraphics[width=\textwidth]{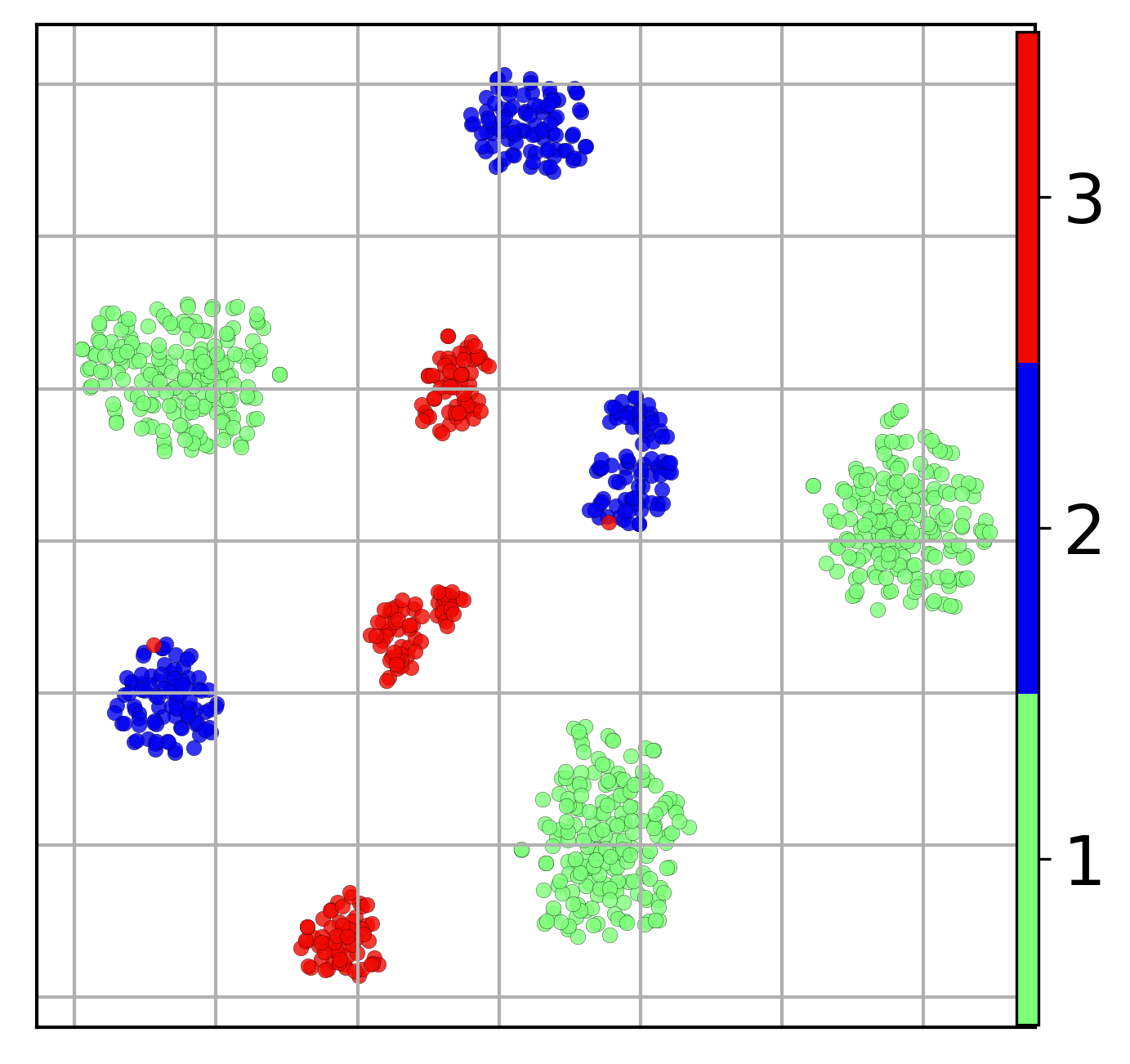}
  \vspace{-5pt}\phantomsection
\end{subfigure}%
\begin{subfigure}[b]{0.2\textwidth}
  \centering
  \includegraphics[width=\textwidth]{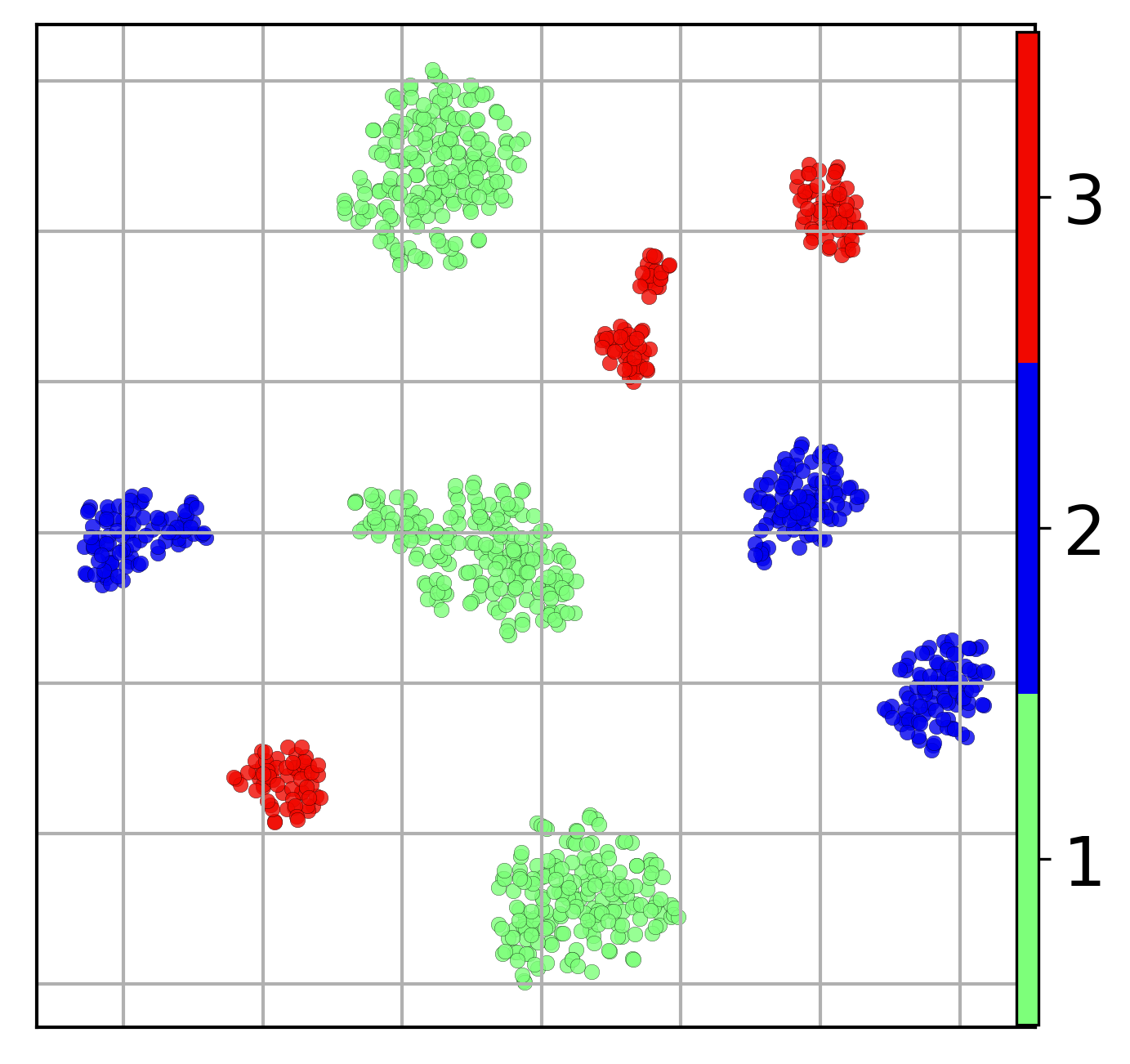}
  \vspace{-5pt}\phantomsection
\end{subfigure}%
\begin{subfigure}[b]{0.2\textwidth}
  \centering
  \includegraphics[width=\textwidth]{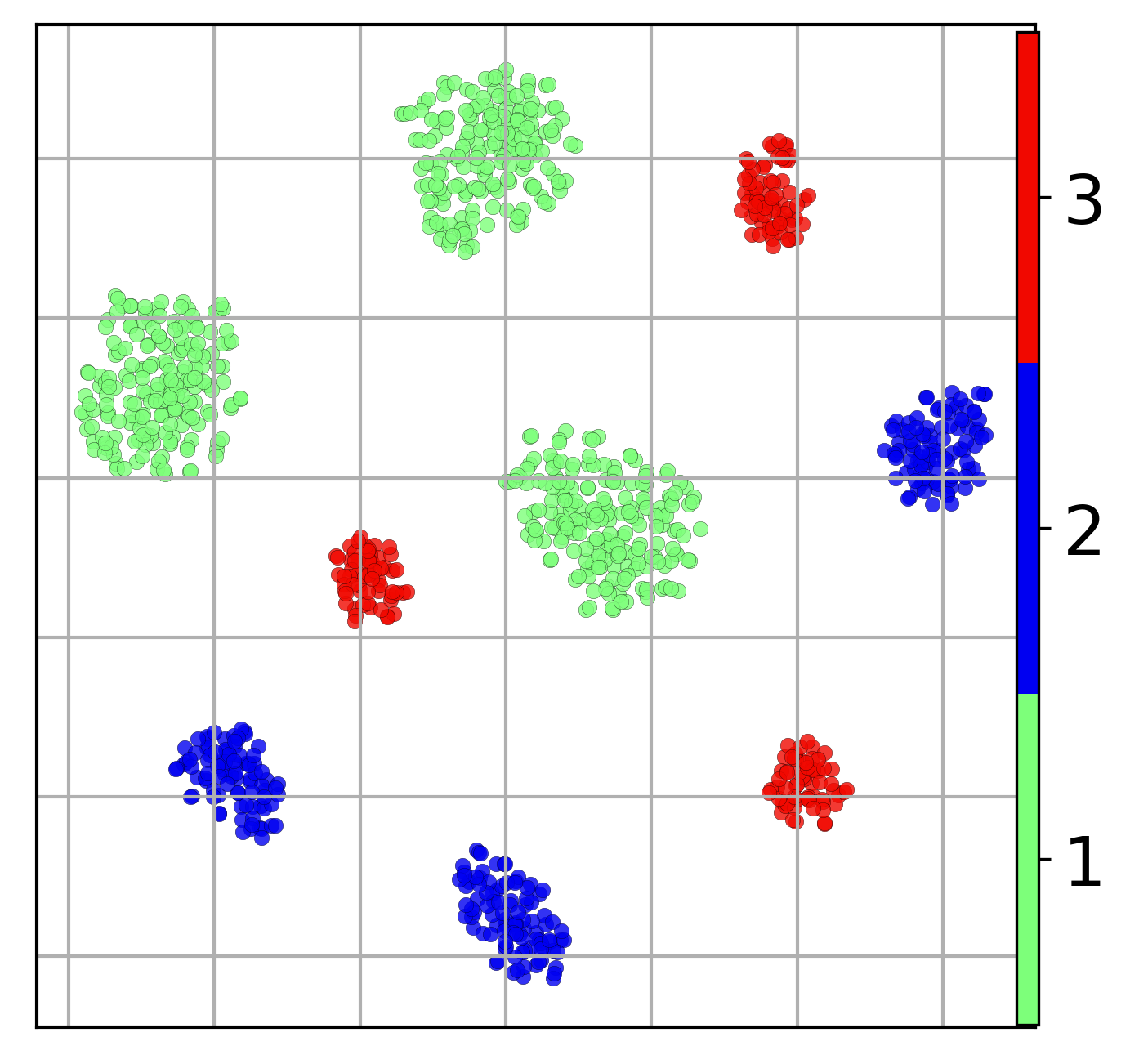}
  \vspace{-5pt}\phantomsection
\end{subfigure}

\begin{subfigure}[b]{0.2\textwidth}
  \centering
  \includegraphics[width=\textwidth]{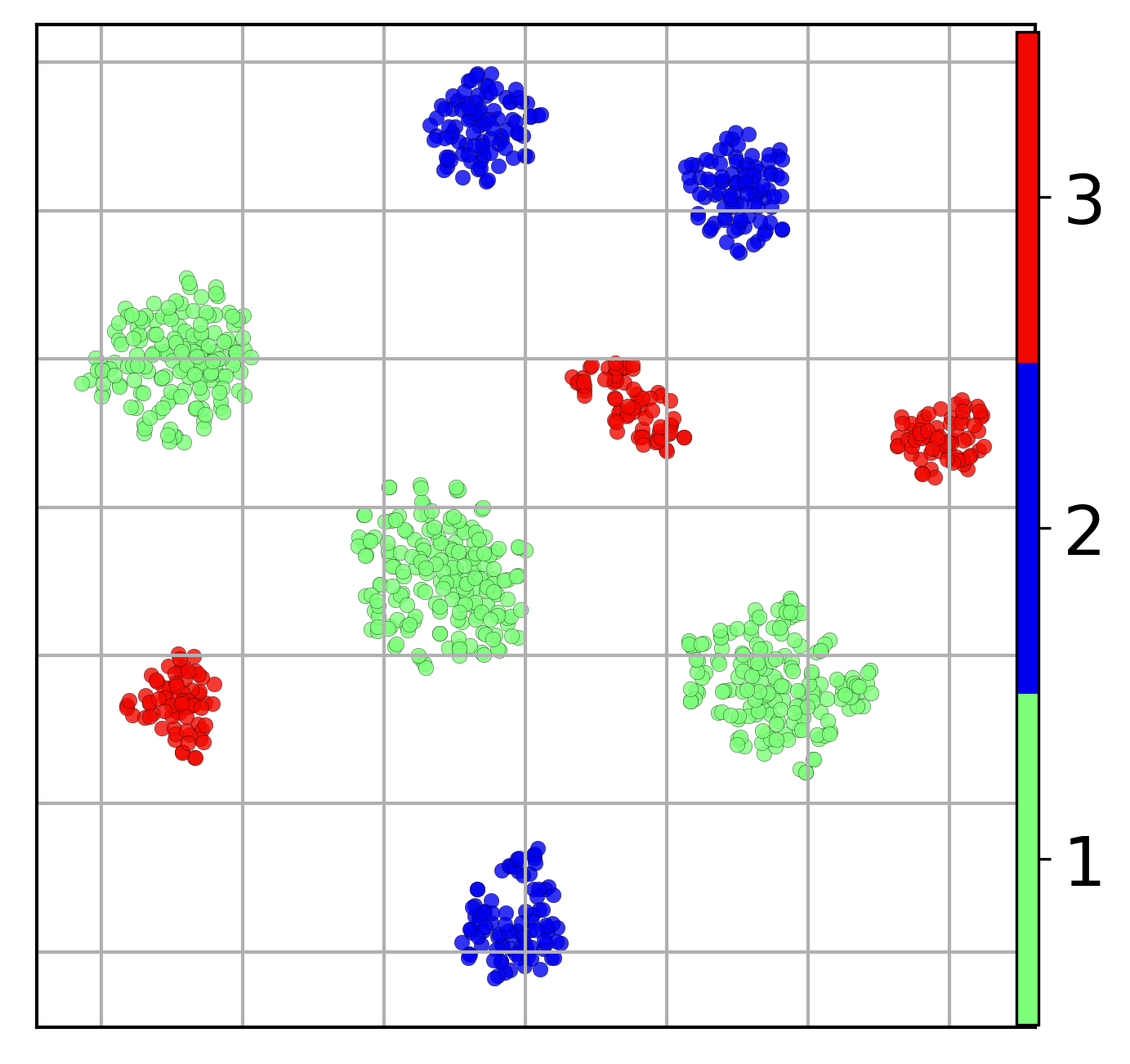}
  \vspace{-5pt}\phantomsection
\end{subfigure}%
\begin{subfigure}[b]{0.2\textwidth}
  \centering
  \includegraphics[width=\textwidth]{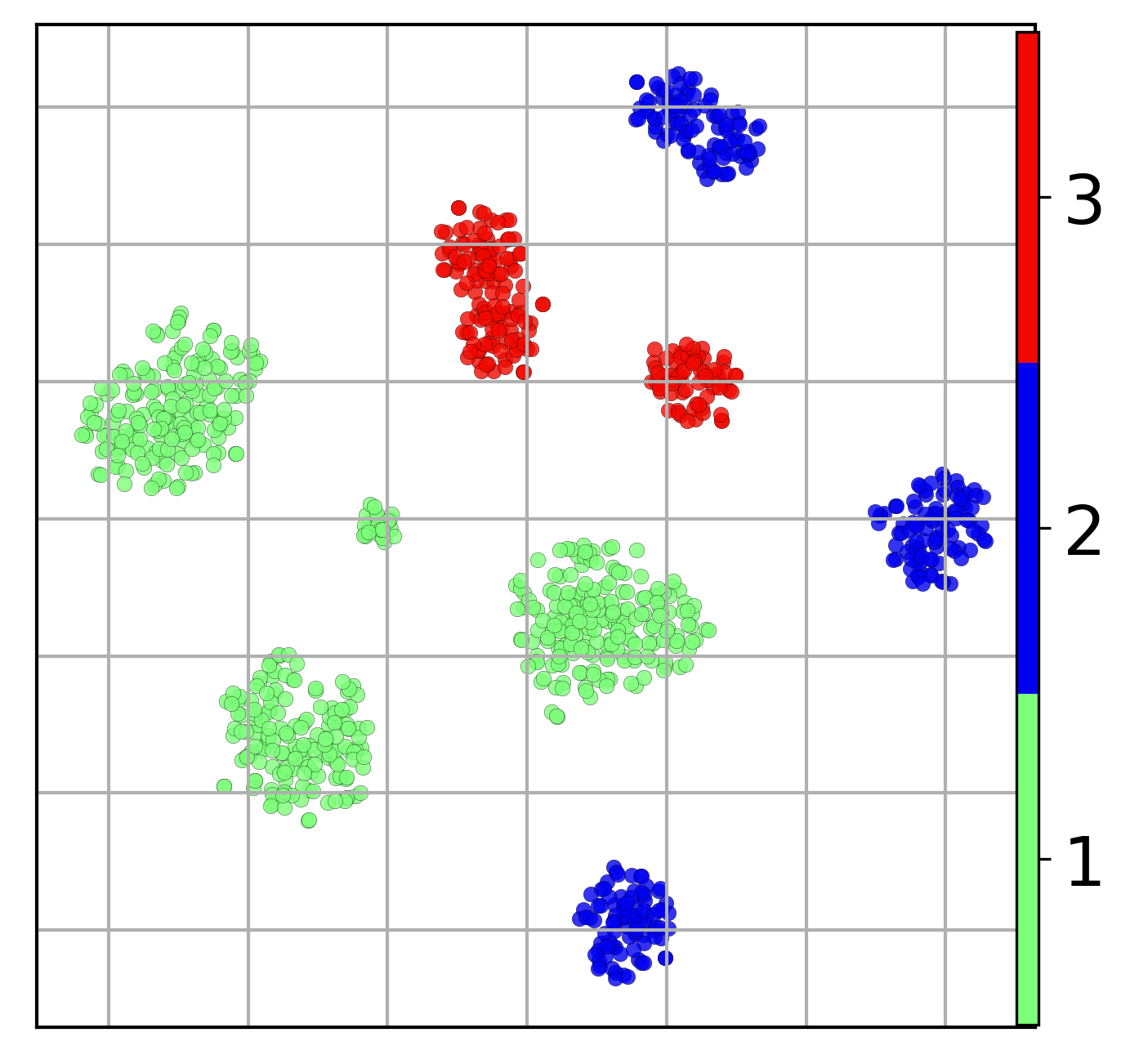}
  \vspace{-5pt}\phantomsection
\end{subfigure}%
\begin{subfigure}[b]{0.2\textwidth}
  \centering
  \includegraphics[width=\textwidth]{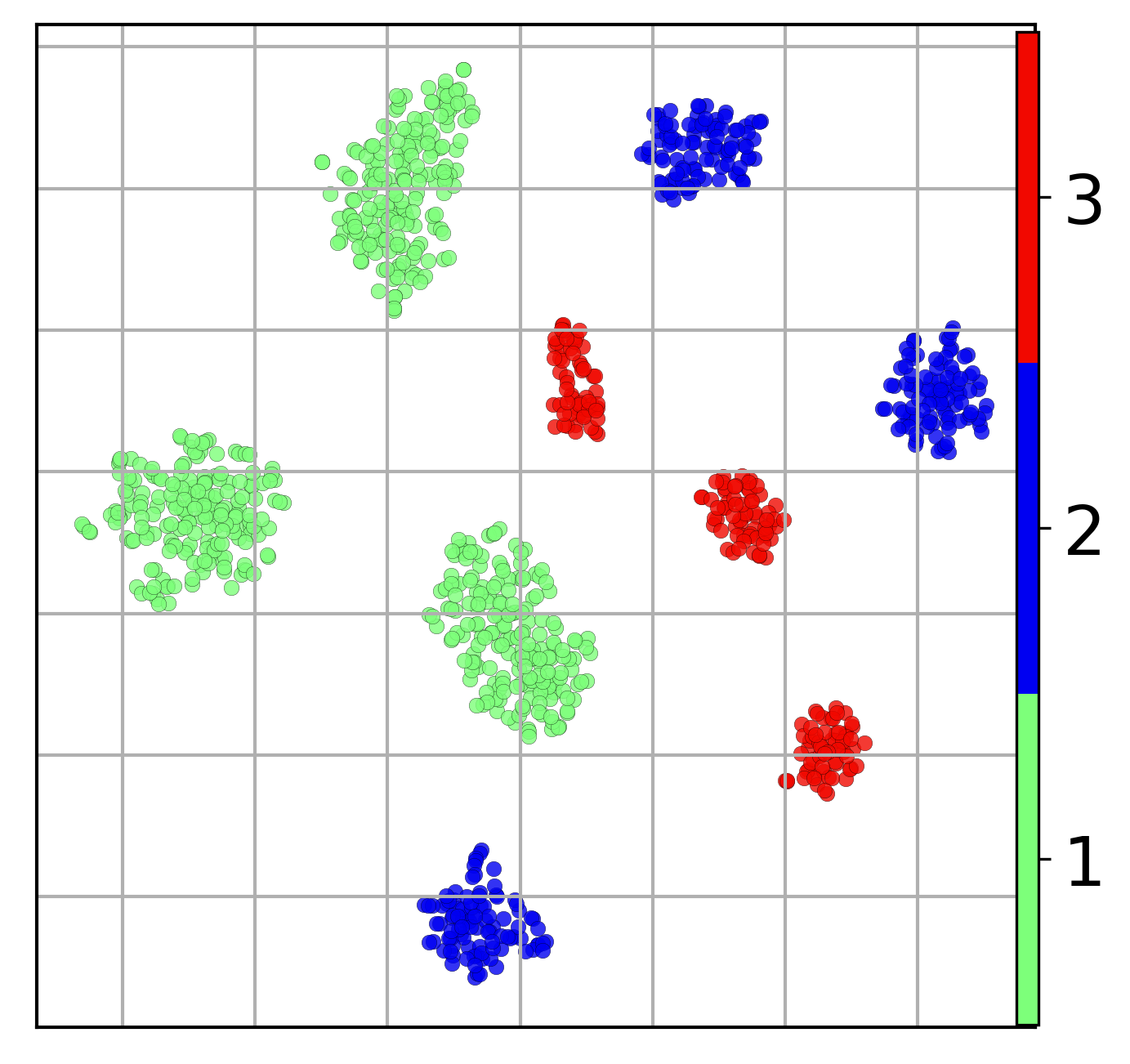}
  \vspace{-5pt}\phantomsection
\end{subfigure}%
\begin{subfigure}[b]{0.2\textwidth}
  \centering
  \includegraphics[width=\textwidth]{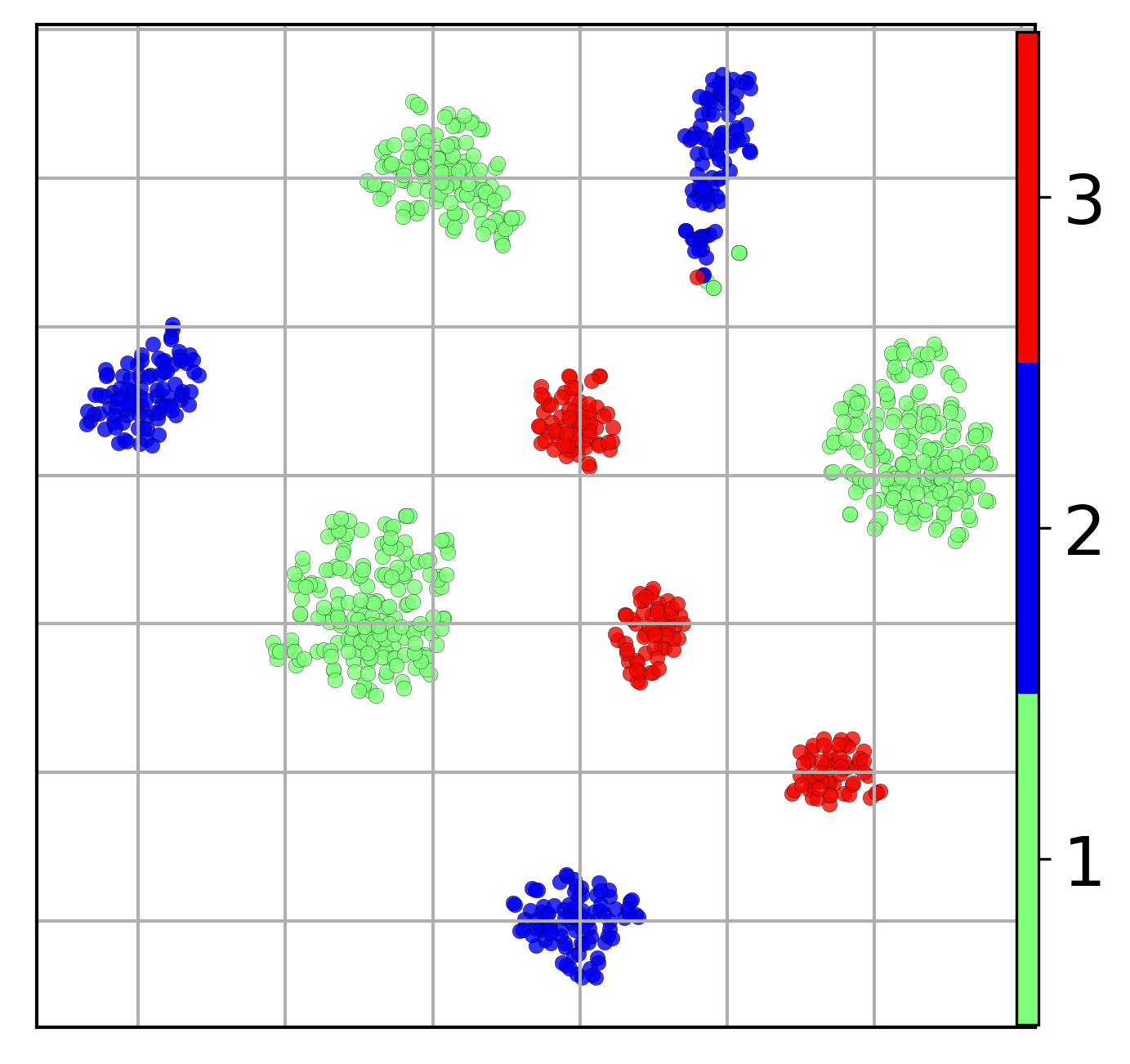}
  \vspace{-5pt}\phantomsection
\end{subfigure}%
\begin{subfigure}[b]{0.2\textwidth}
  \centering
  \includegraphics[width=\textwidth]{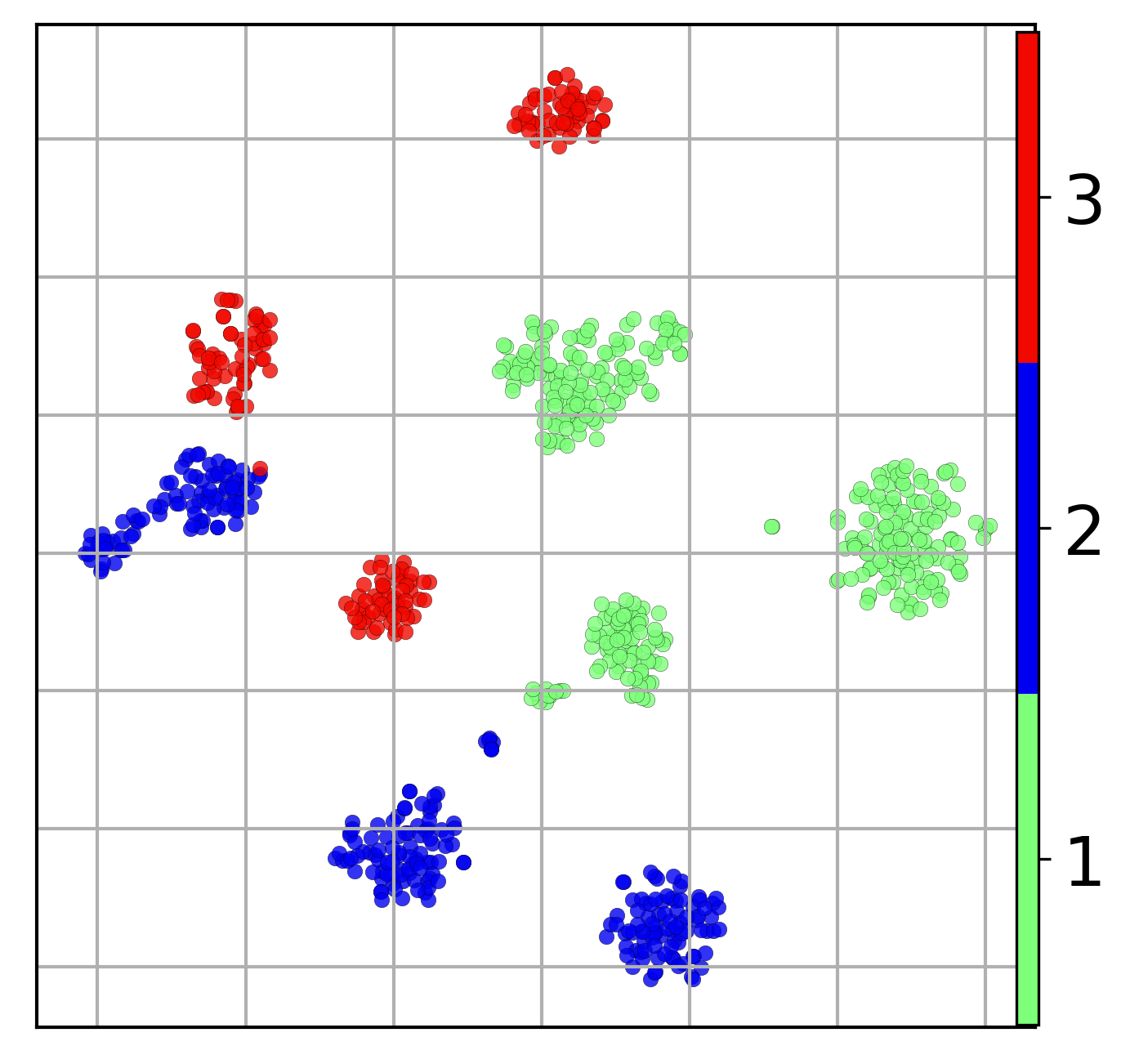}
  \vspace{-5pt}\phantomsection
\end{subfigure}

\begin{subfigure}[b]{0.2\textwidth}
  \centering
  \includegraphics[width=\textwidth]{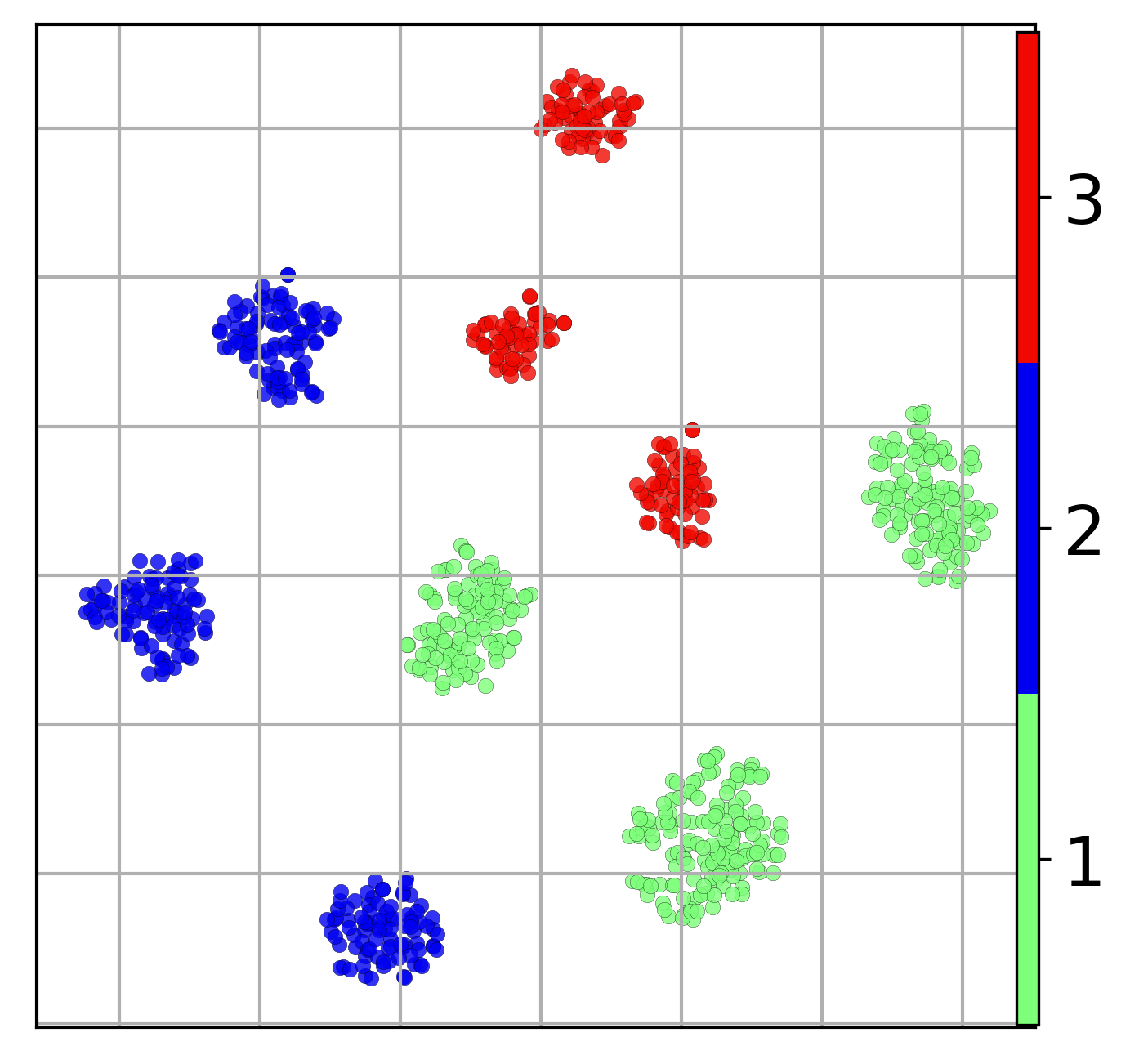}
  \vspace{-5pt}\phantomsection
\end{subfigure}%
\begin{subfigure}[b]{0.2\textwidth}
  \centering
  \includegraphics[width=\textwidth]{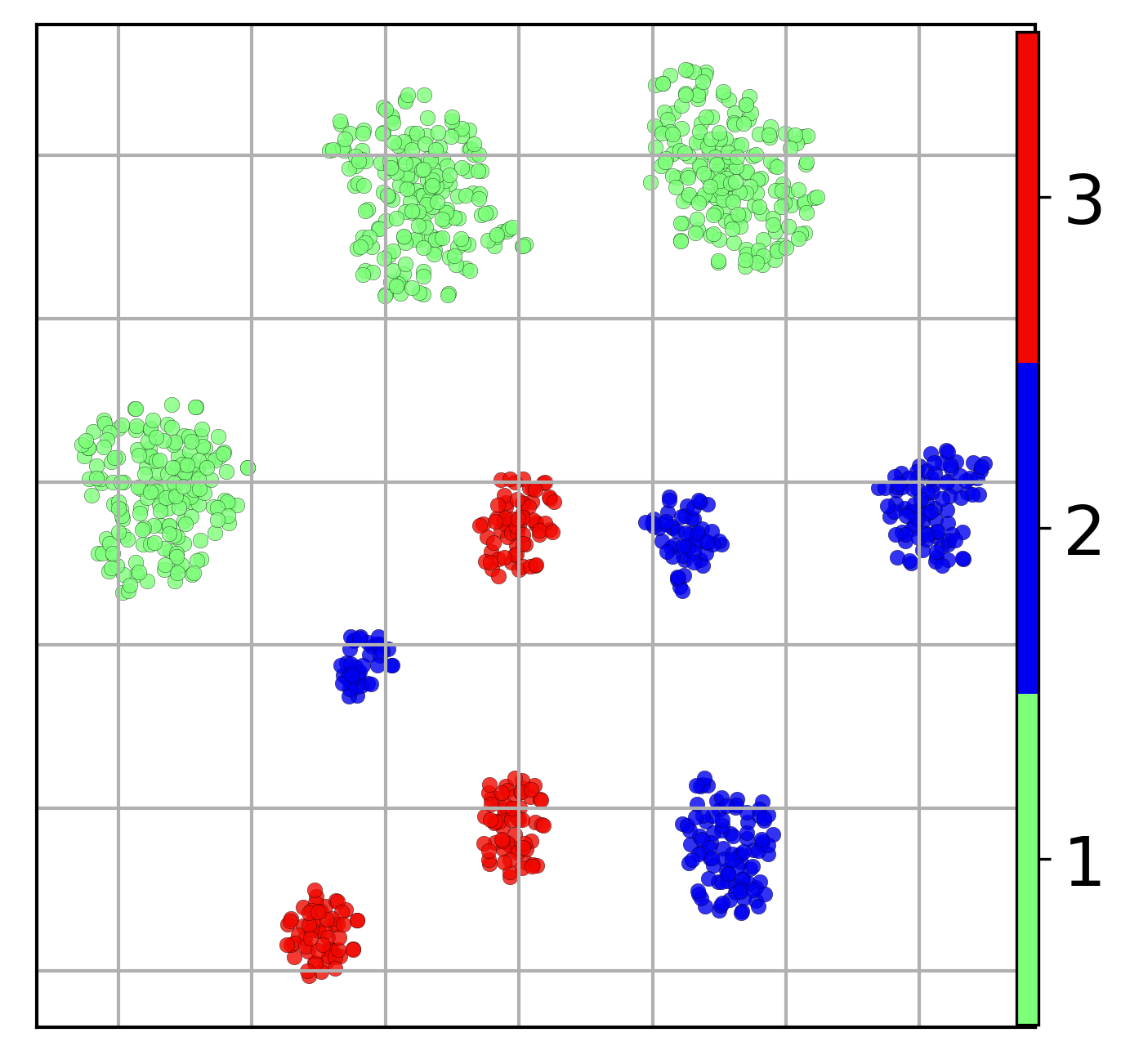}
  \vspace{-5pt}\phantomsection
\end{subfigure}%
\begin{subfigure}[b]{0.2\textwidth}
  \centering
  \includegraphics[width=\textwidth]{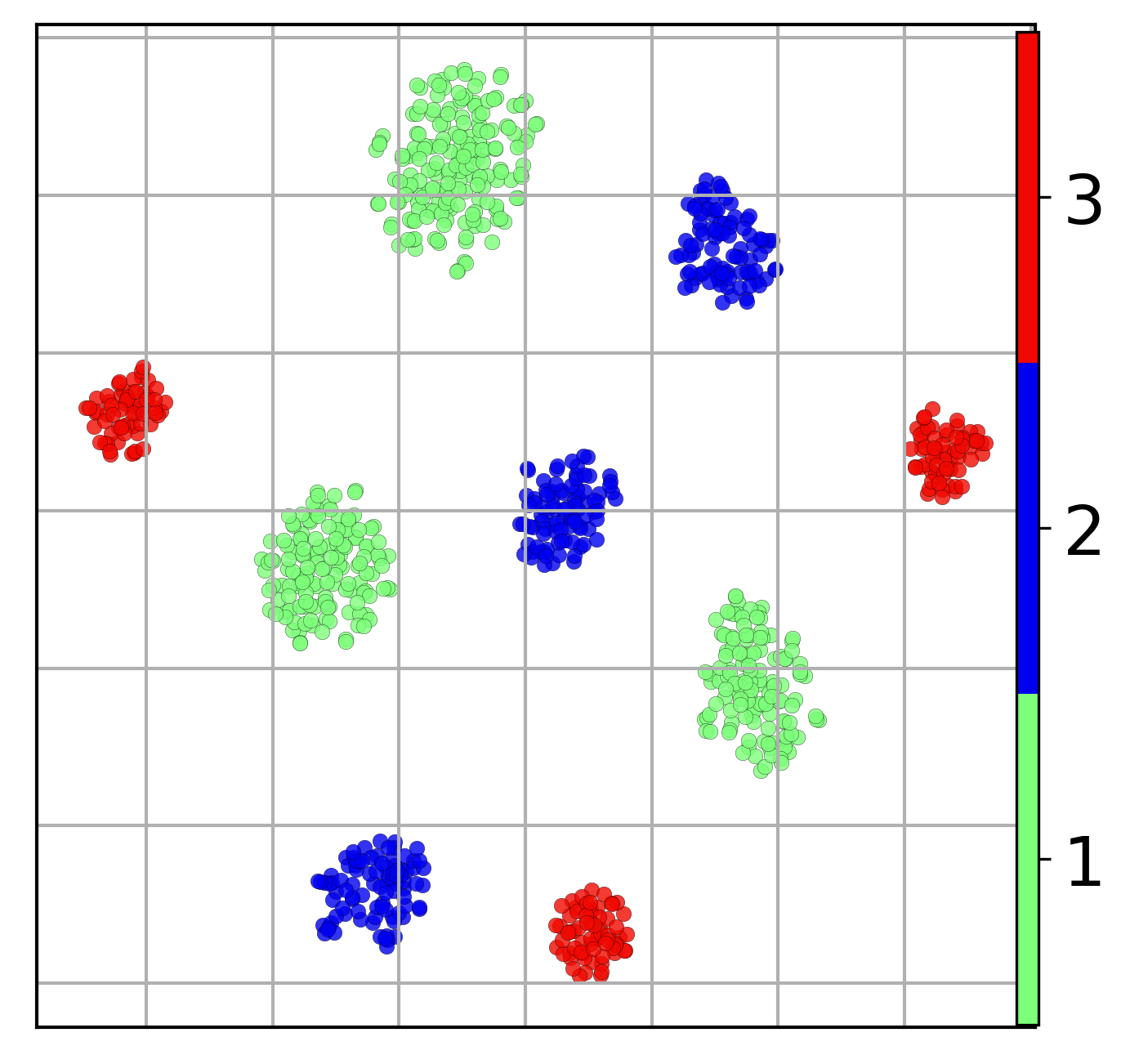}
  \vspace{-5pt}\phantomsection
\end{subfigure}%
\begin{subfigure}[b]{0.2\textwidth}
  \centering
  \includegraphics[width=\textwidth]{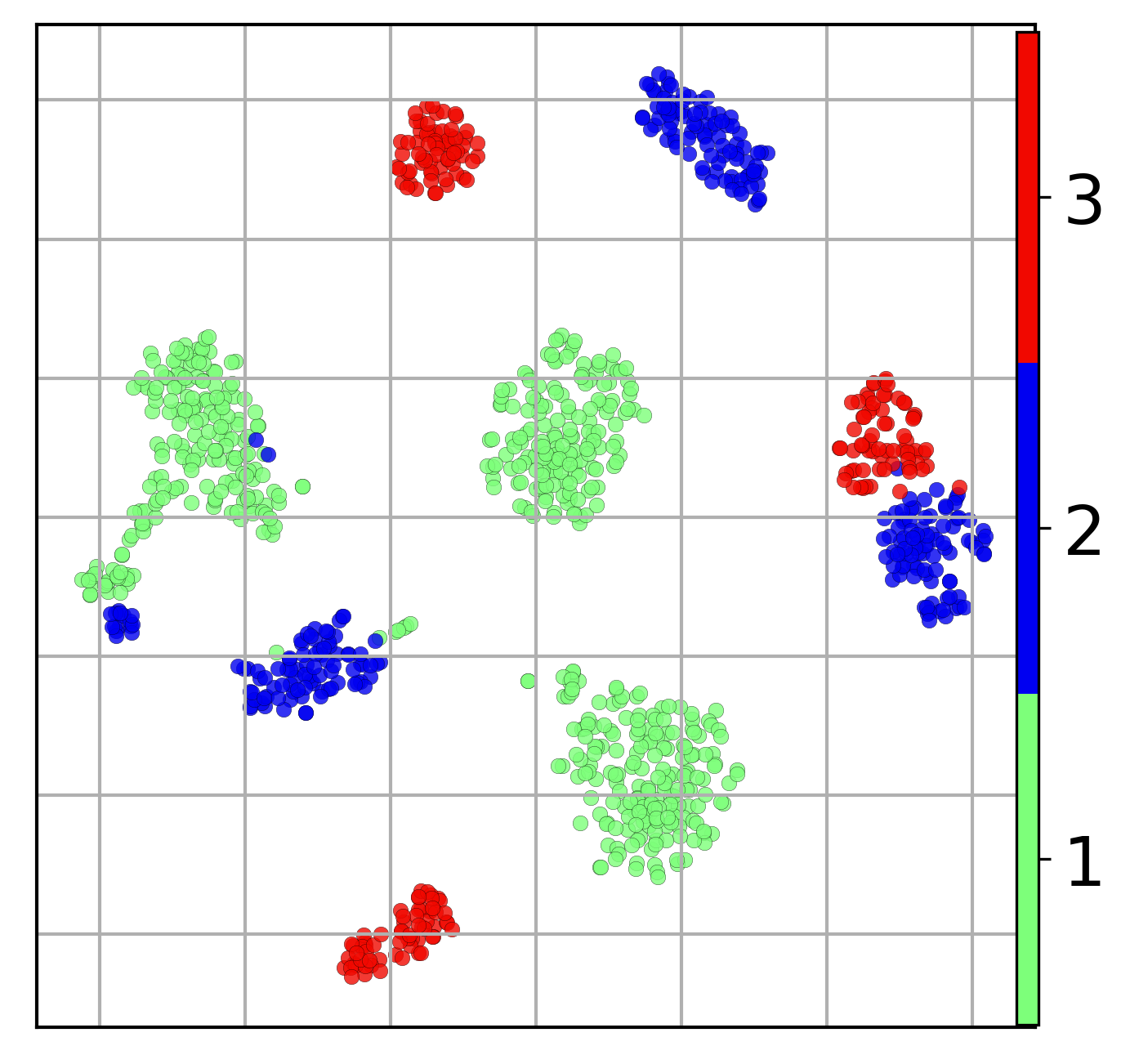}
  \vspace{-5pt}\phantomsection
\end{subfigure}%
\begin{subfigure}[b]{0.2\textwidth}
  \centering
  \includegraphics[width=\textwidth]{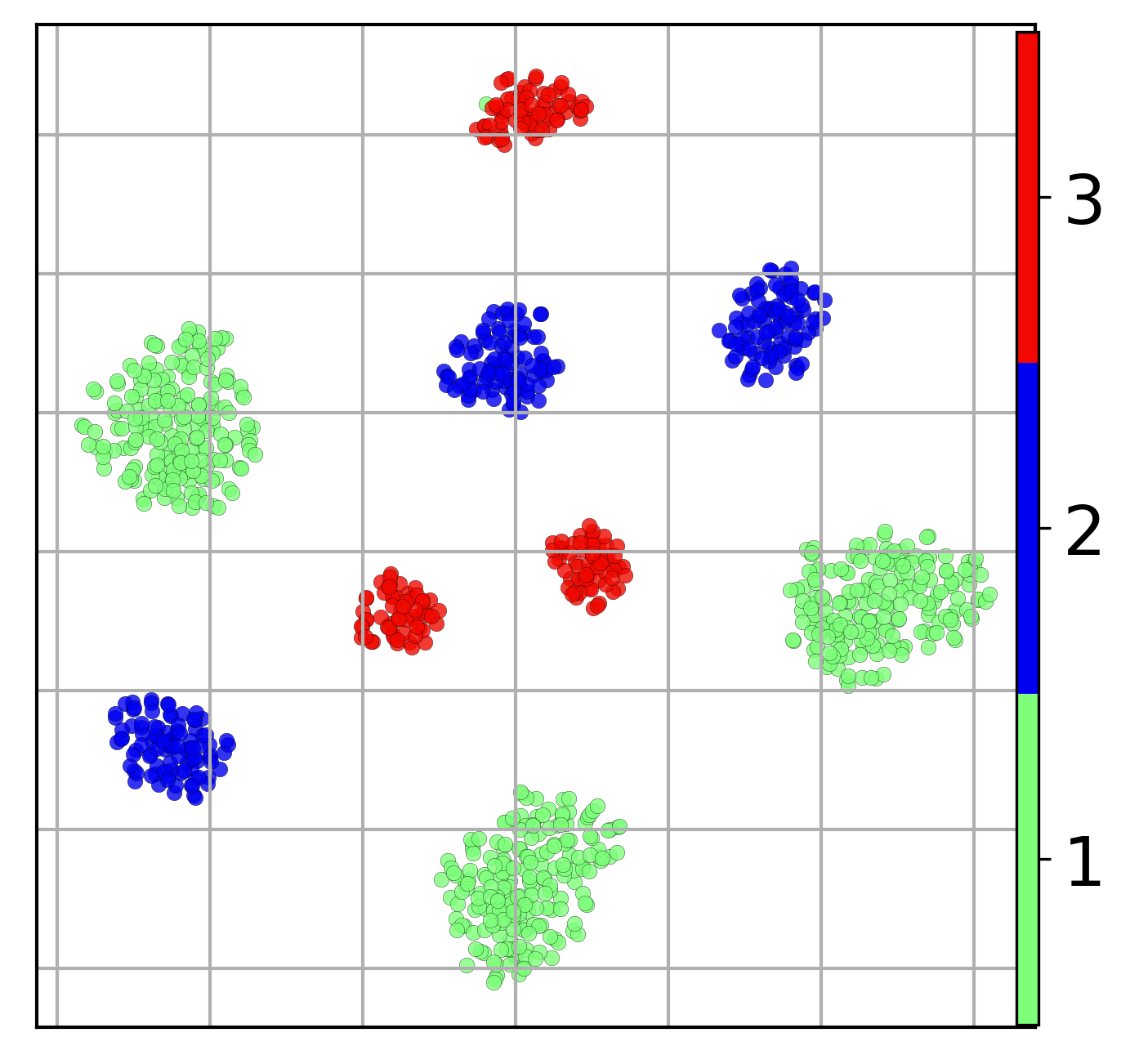}
  \vspace{-5pt}\phantomsection
\end{subfigure}

\begin{subfigure}[b]{0.2\textwidth}
  \centering
  \includegraphics[width=\textwidth]{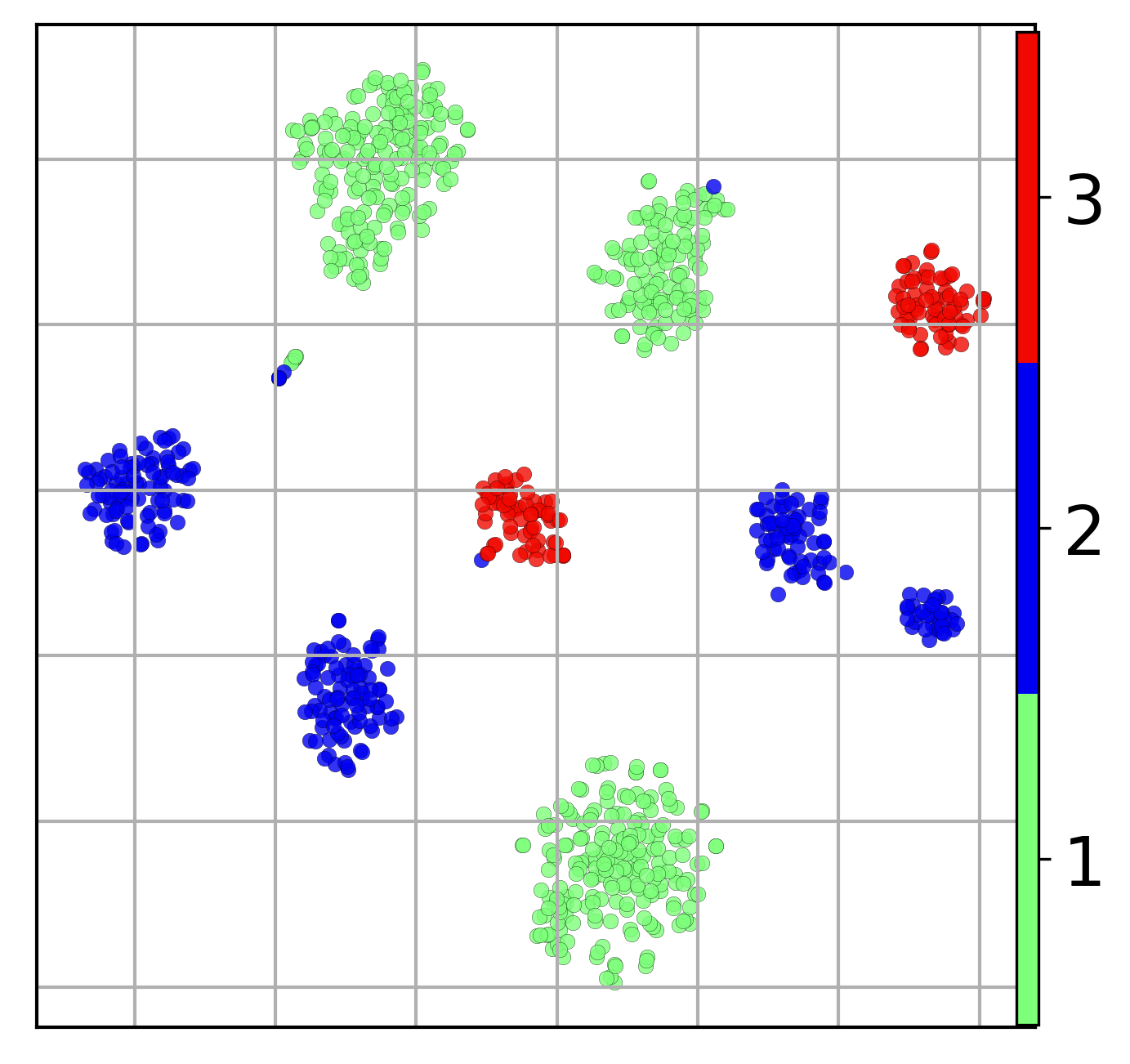}
  \vspace{-5pt}\phantomsection
\end{subfigure}%
\begin{subfigure}[b]{0.2\textwidth}
  \centering
  \includegraphics[width=\textwidth]{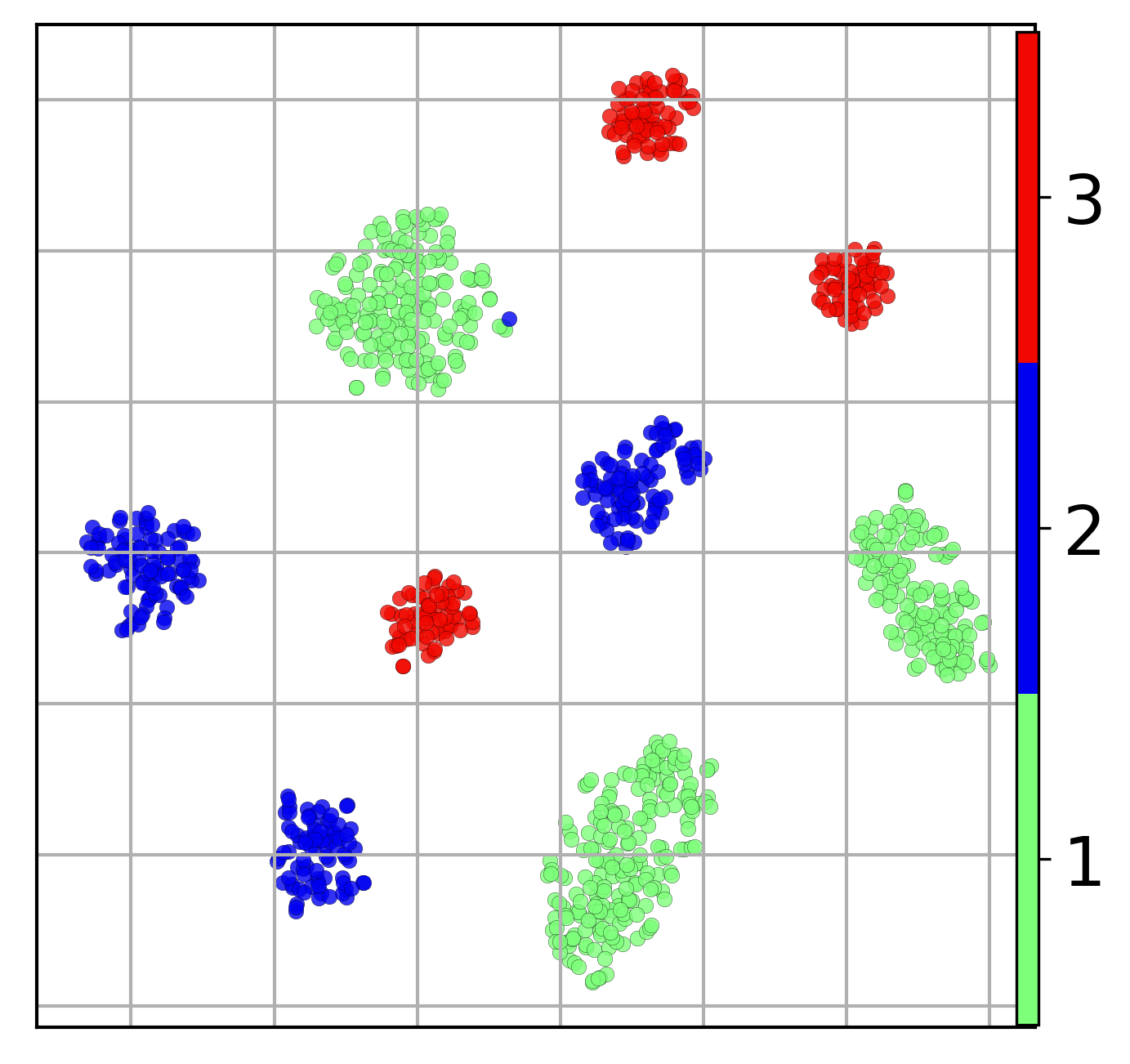}
  \vspace{-5pt}\phantomsection
\end{subfigure}%
\begin{subfigure}[b]{0.2\textwidth}
  \centering
  \includegraphics[width=\textwidth]{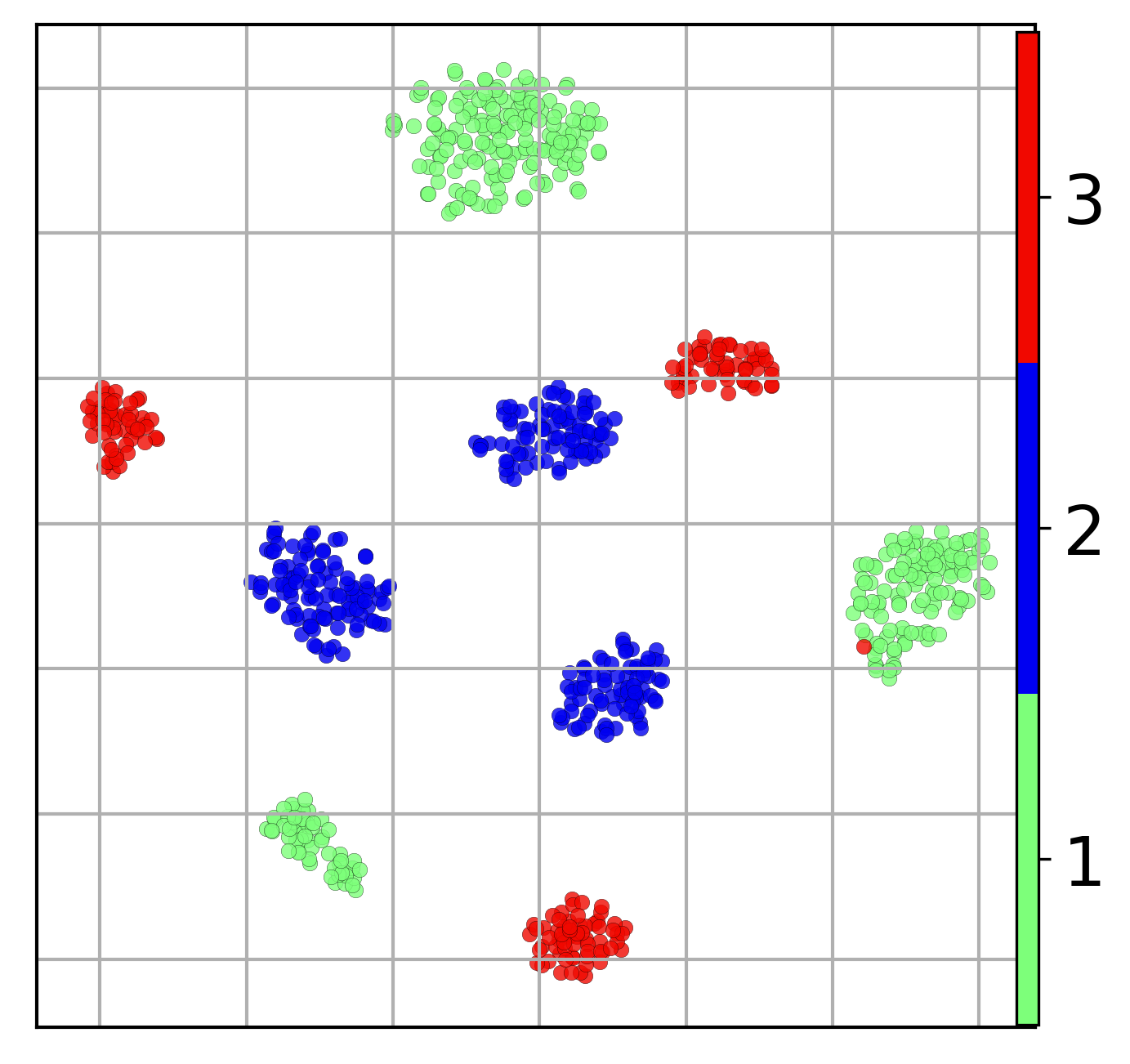}
  \vspace{-5pt}\phantomsection
\end{subfigure}%
\begin{subfigure}[b]{0.2\textwidth}
  \centering
  \includegraphics[width=\textwidth]{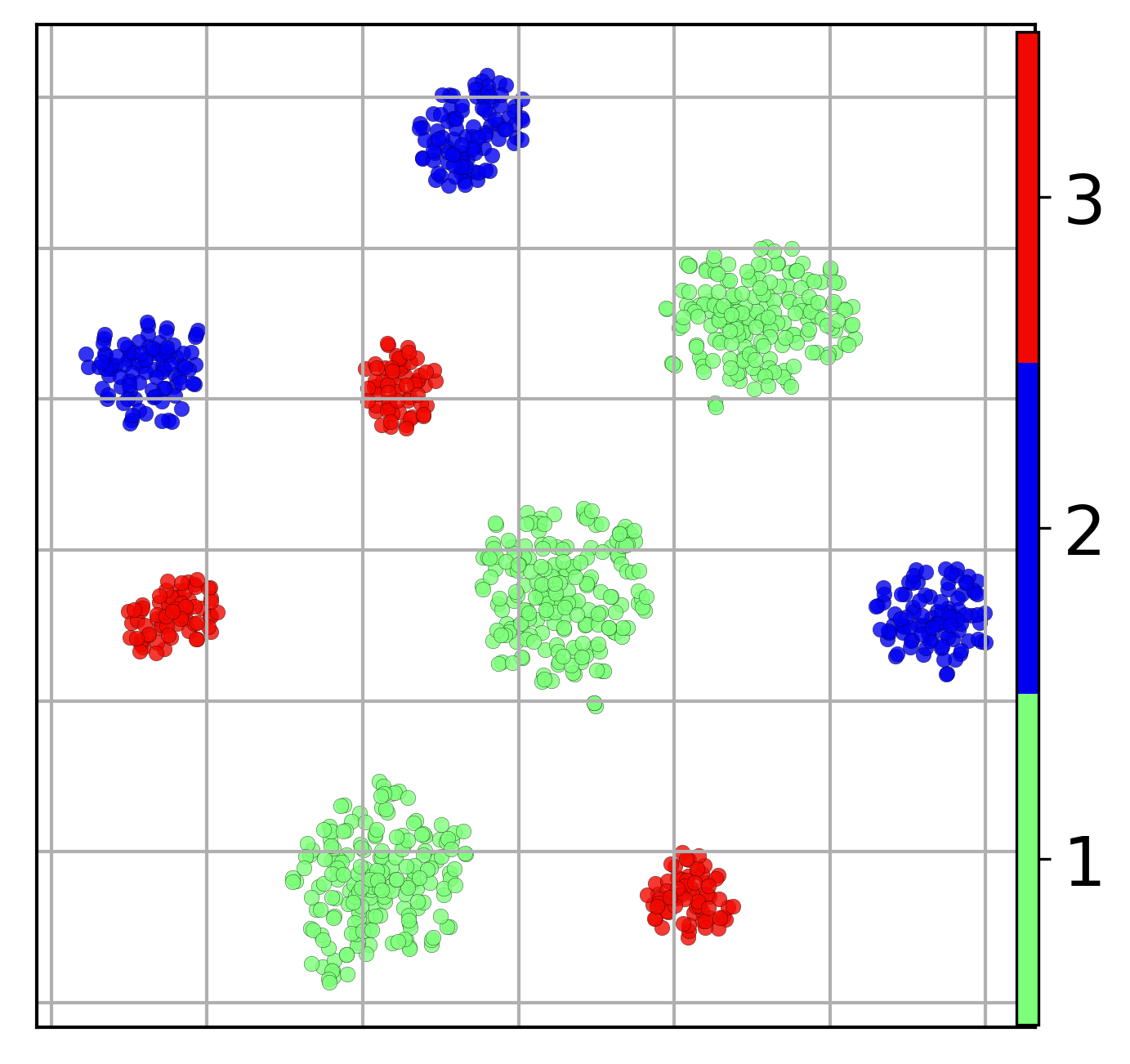}
  \vspace{-5pt}\phantomsection
\end{subfigure}%
\begin{subfigure}[b]{0.2\textwidth}
  \centering
  \includegraphics[width=\textwidth]{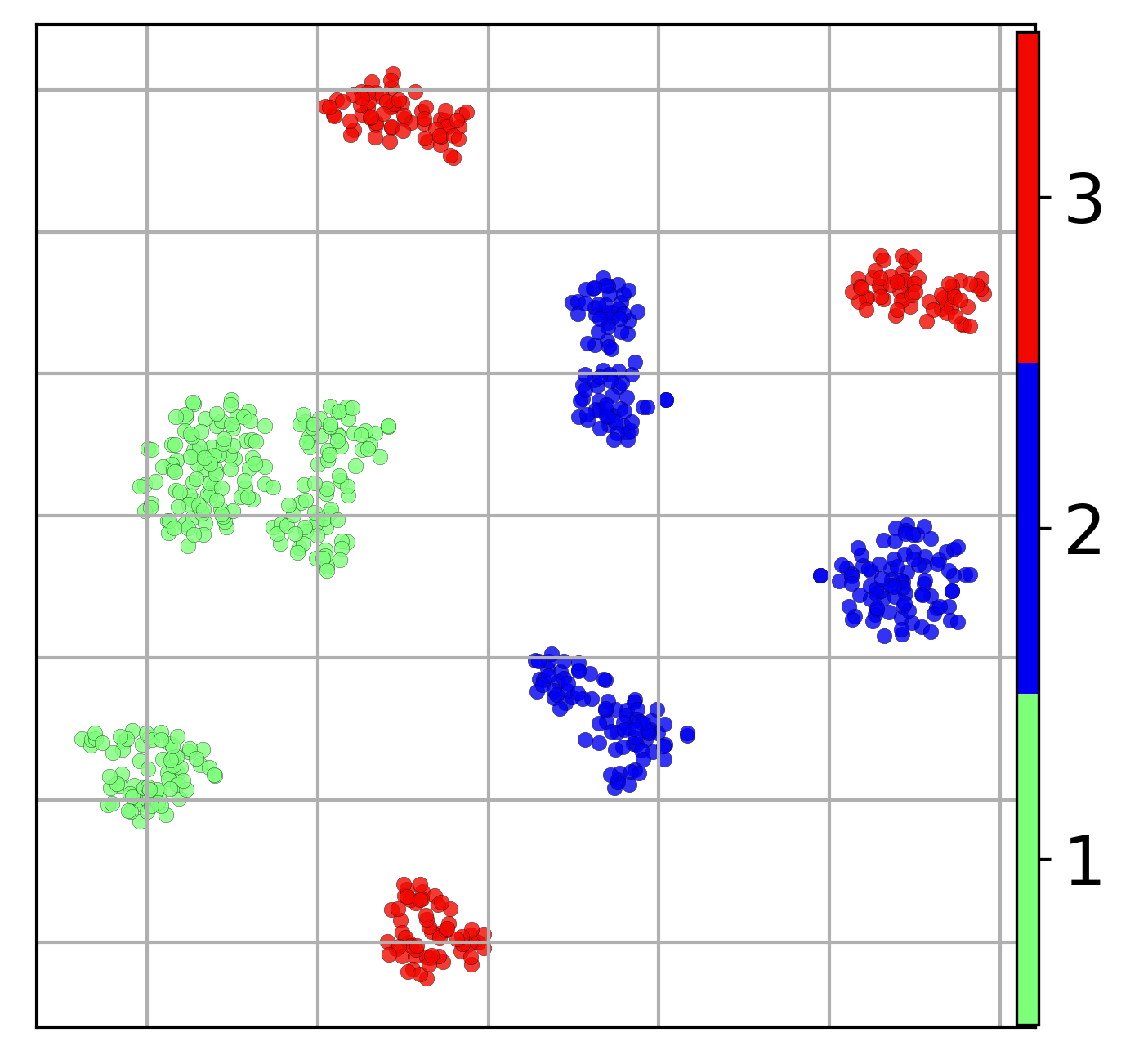}
  \vspace{-5pt}\phantomsection
\end{subfigure}

\begin{subfigure}[b]{0.2\textwidth}
  \centering
  \includegraphics[width=\textwidth]{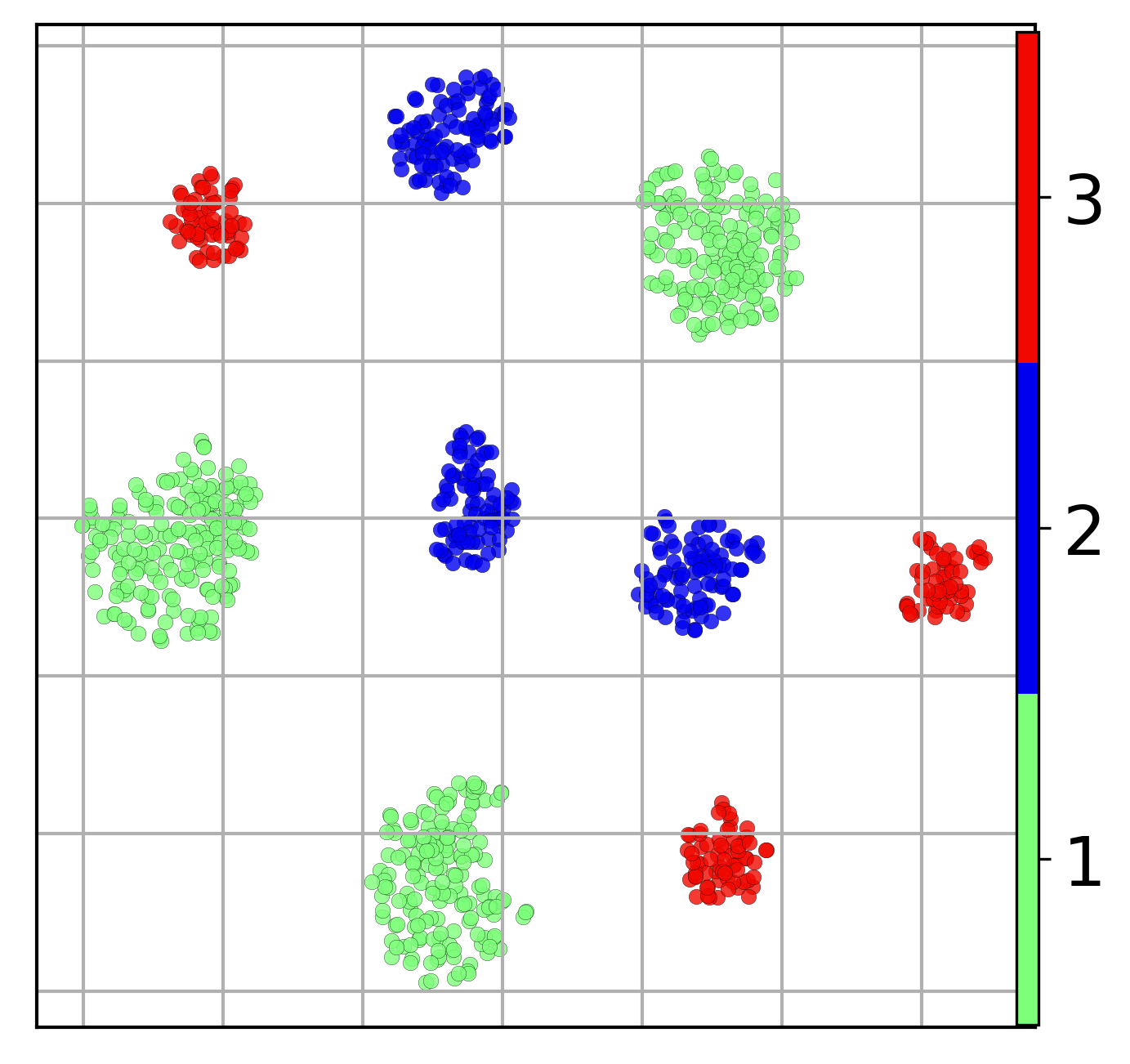}
  \vspace{-5pt}\phantomsection
\end{subfigure}%
\begin{subfigure}[b]{0.2\textwidth}
  \centering
  \includegraphics[width=\textwidth]{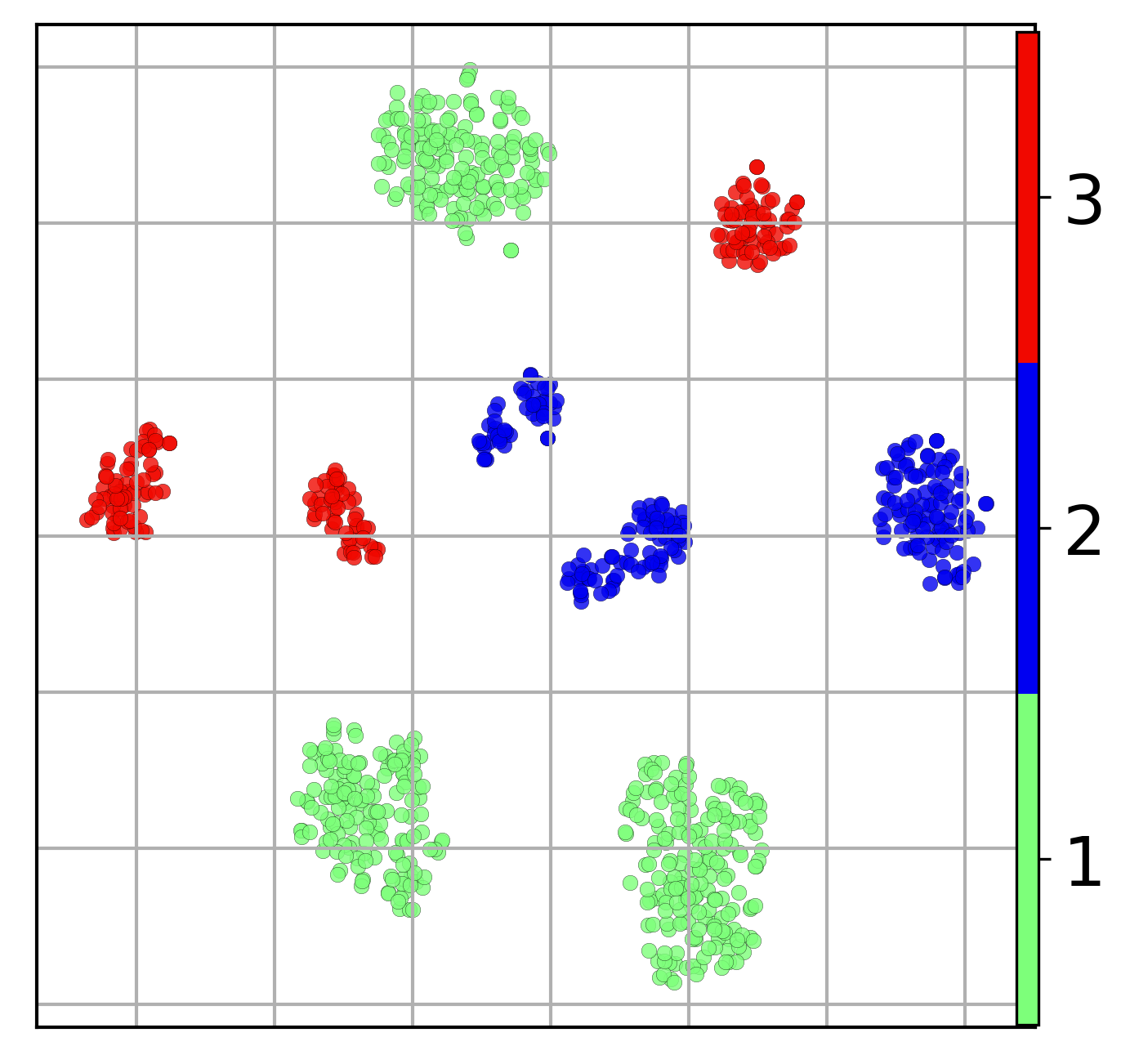}
  \vspace{-5pt}\phantomsection
\end{subfigure}%
\begin{subfigure}[b]{0.2\textwidth}
  \centering
  \includegraphics[width=\textwidth]{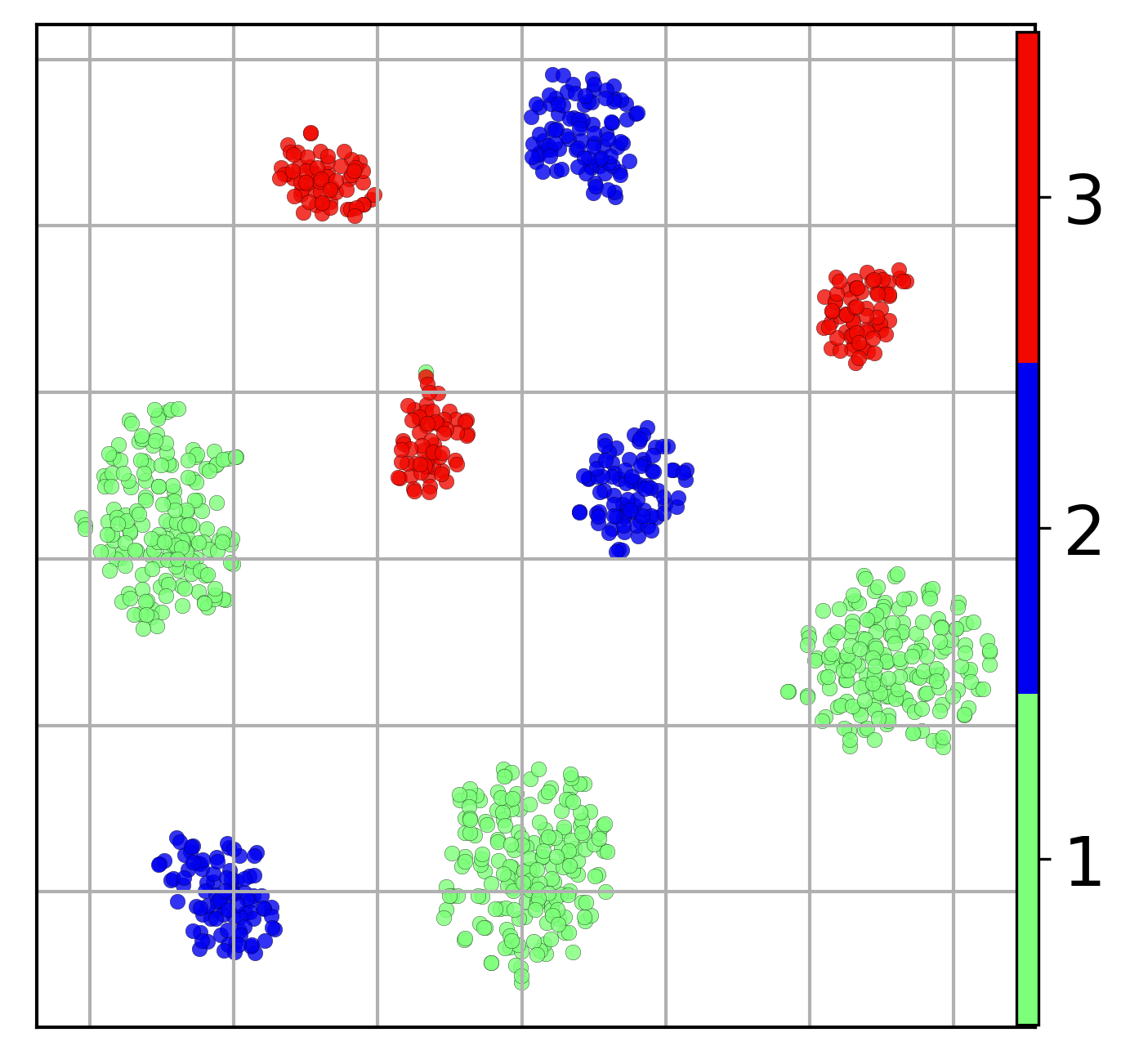}
  \vspace{-5pt}\phantomsection
\end{subfigure}%
\begin{subfigure}[b]{0.2\textwidth}
  \centering
  \includegraphics[width=\textwidth]{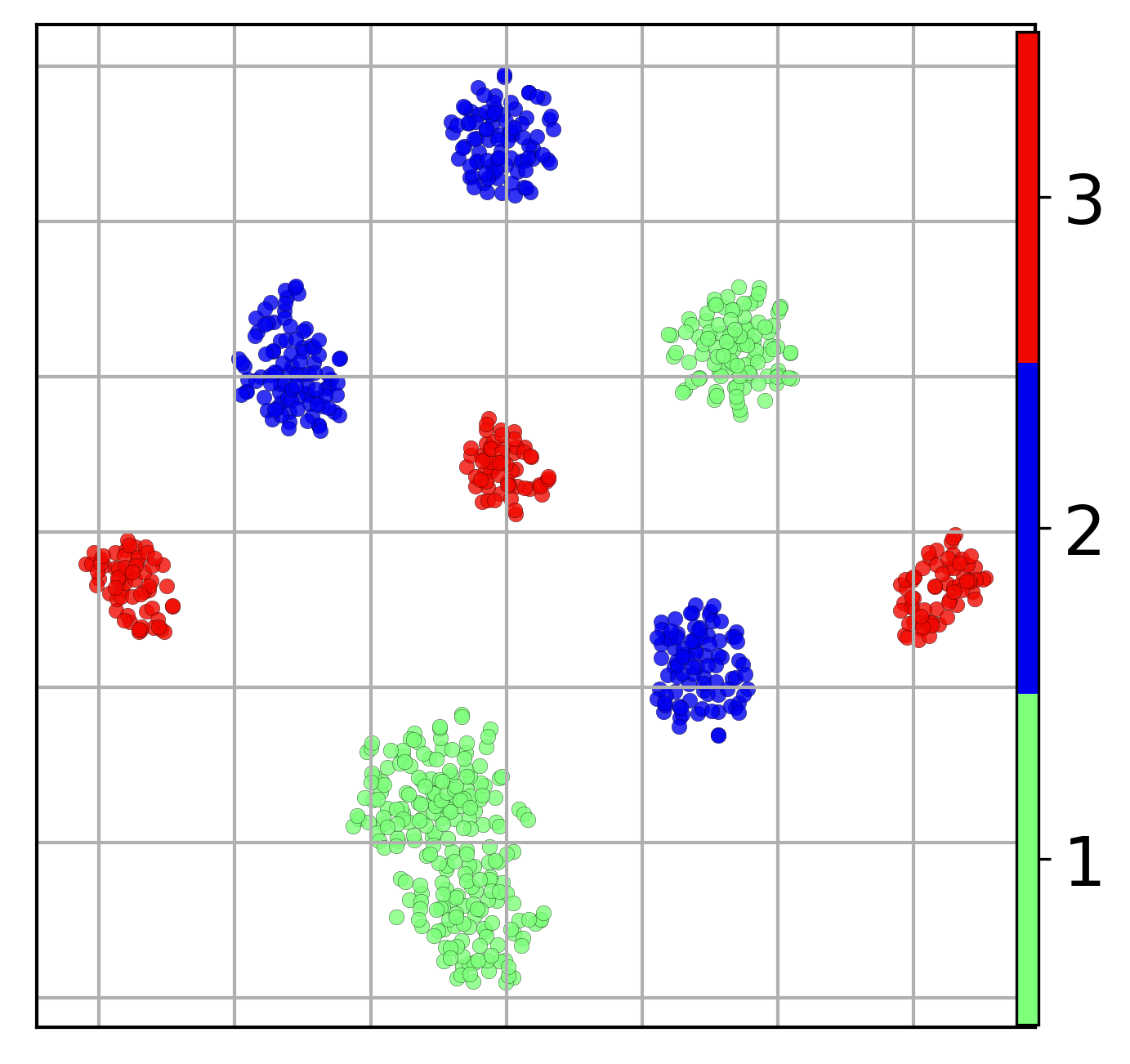}
  \vspace{-5pt}\phantomsection
\end{subfigure}%
\begin{subfigure}[b]{0.2\textwidth}
  \centering
  \includegraphics[width=\textwidth]{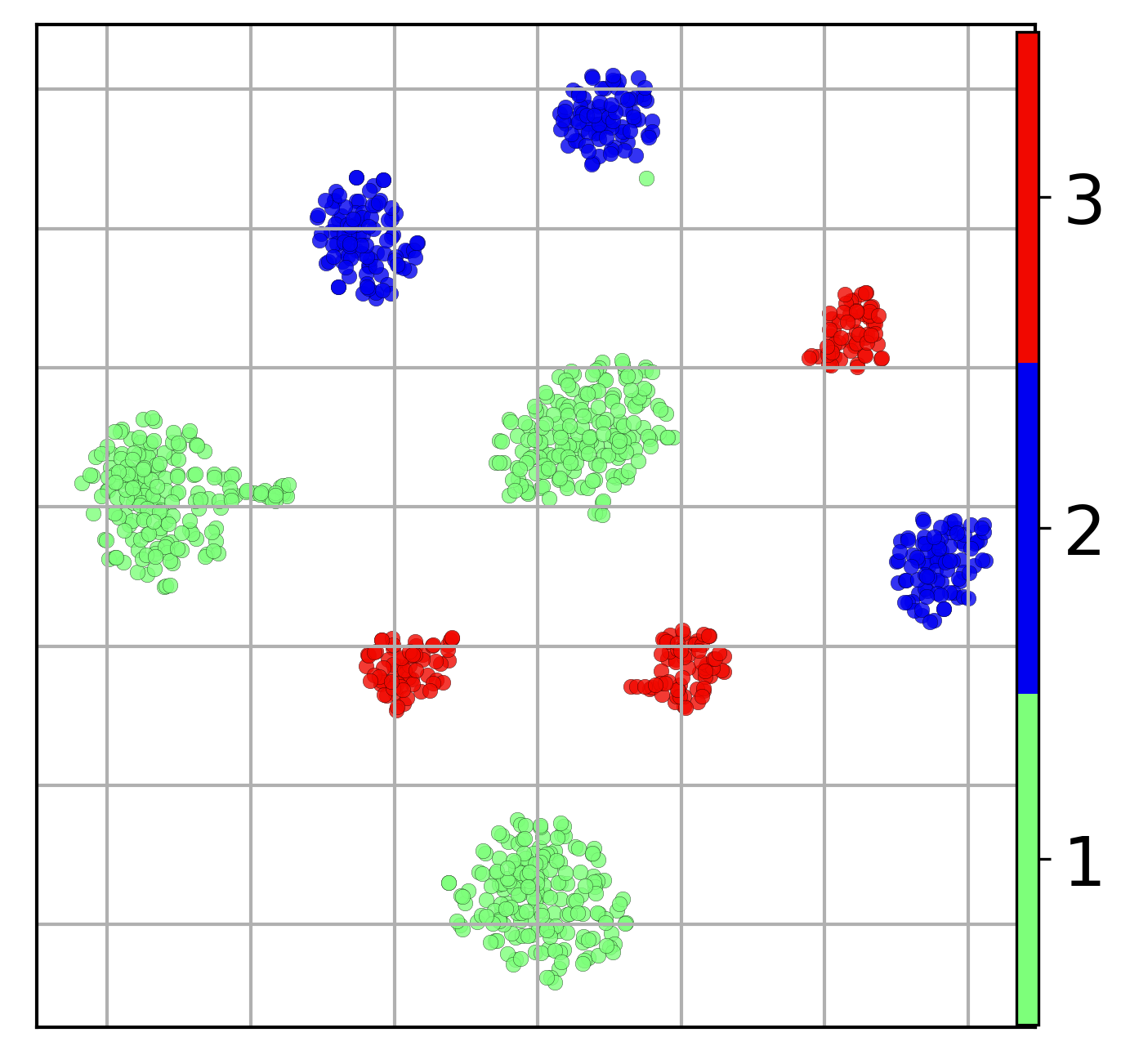}
  \vspace{-5pt}\phantomsection
\end{subfigure}
  
\caption{t-SNE visualization of all trials from 25 participants, color-coded by the type of movement: reaching (green), grasping (blue), and twisting (red). Each participant has distinct clusters, with each cluster corresponding to one of the three types of movements observed during the sessions.}\label{clustering_figure}
\end{figure} \vspace{-15pt}

From the figure, we observe that the reaching movement consistently has the lowest dependence, twisting the highest, and grasping falling in between. This pattern matches the differences and complexity of different movement tasks \cite{guerrero2022coherence,ye2022investigation}. This variability reflects the differences in connectivity between the brain and target muscles, which is the basis of movement recognition and measurement by using CMC \cite{sun2023enhancement}. \vspace{5pt}


\textbf{Mean and variance of density ratios in clusters.} Same as in Fig.~\ref{CLUSTERFIGURE} in the main paper, we extract the nine clusters shown in the t-SNE projections ({Fig.~\ref{clustering_figure}}) and compute the mean and standard deviation of the density ratios for each cluster (\texttt{C1}$\sim$\texttt{C9}). This process is repeated for all 25 participants (\texttt{SUB1}$\sim$\texttt{SUB25}), with results shown in {Fig.~\ref{ratio_figure_variation}}. The results further demonstrate the effectiveness of using density ratio as a dependence measure, as the density ratios are similar within each cluster but vary across different clusters (movements and sessions). \vspace{5pt}

\textbf{Participant identification.} Similar to Fig.~\ref{CLUSTERFIGURE} in the main paper, where we visualize the projected EEG eigenfunctions from 10 subjects during the reaching movement, we extend this analysis to include all three movements: reaching, grasping, and twisting. The projections shown in Fig.~\ref{ratio_figure_variation} illustrate that participant information is consistently contained in the projection space of the eigenfunctions across all movements, not just limited to reaching.

\begin{figure}[H]
\centering
\begin{minipage}{\textwidth}
\centering
{\large\textbf{\texttt{FMCA Density Ratio Variability, SUB1$\sim$SUB25}}}
\end{minipage}
\begin{subfigure}[b]{0.2\textwidth}
  \centering
  \includegraphics[width=\textwidth]{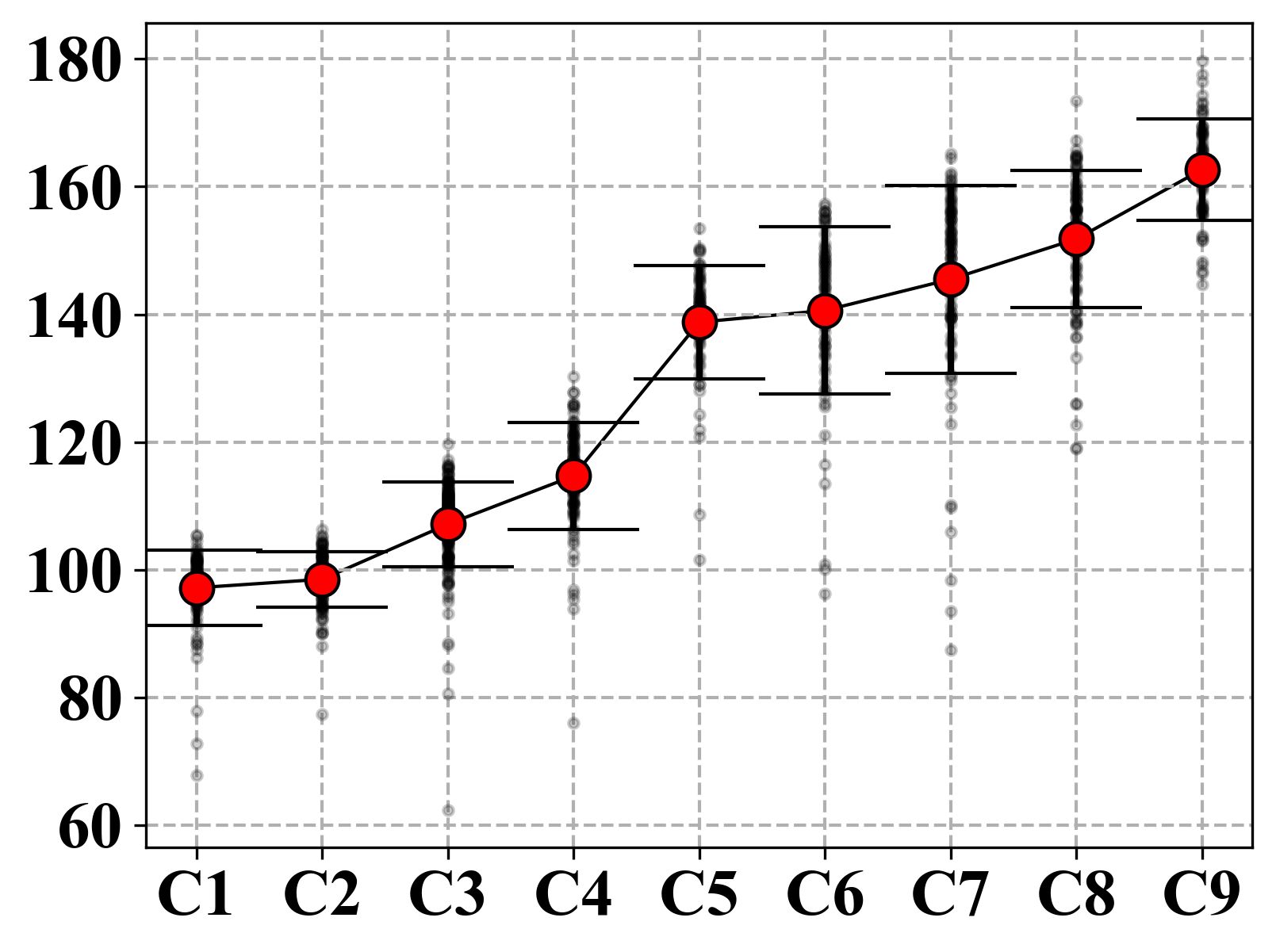}
  \vspace{-5pt}\phantomsection
\end{subfigure}%
\begin{subfigure}[b]{0.2\textwidth}
  \centering
  \includegraphics[width=\textwidth]{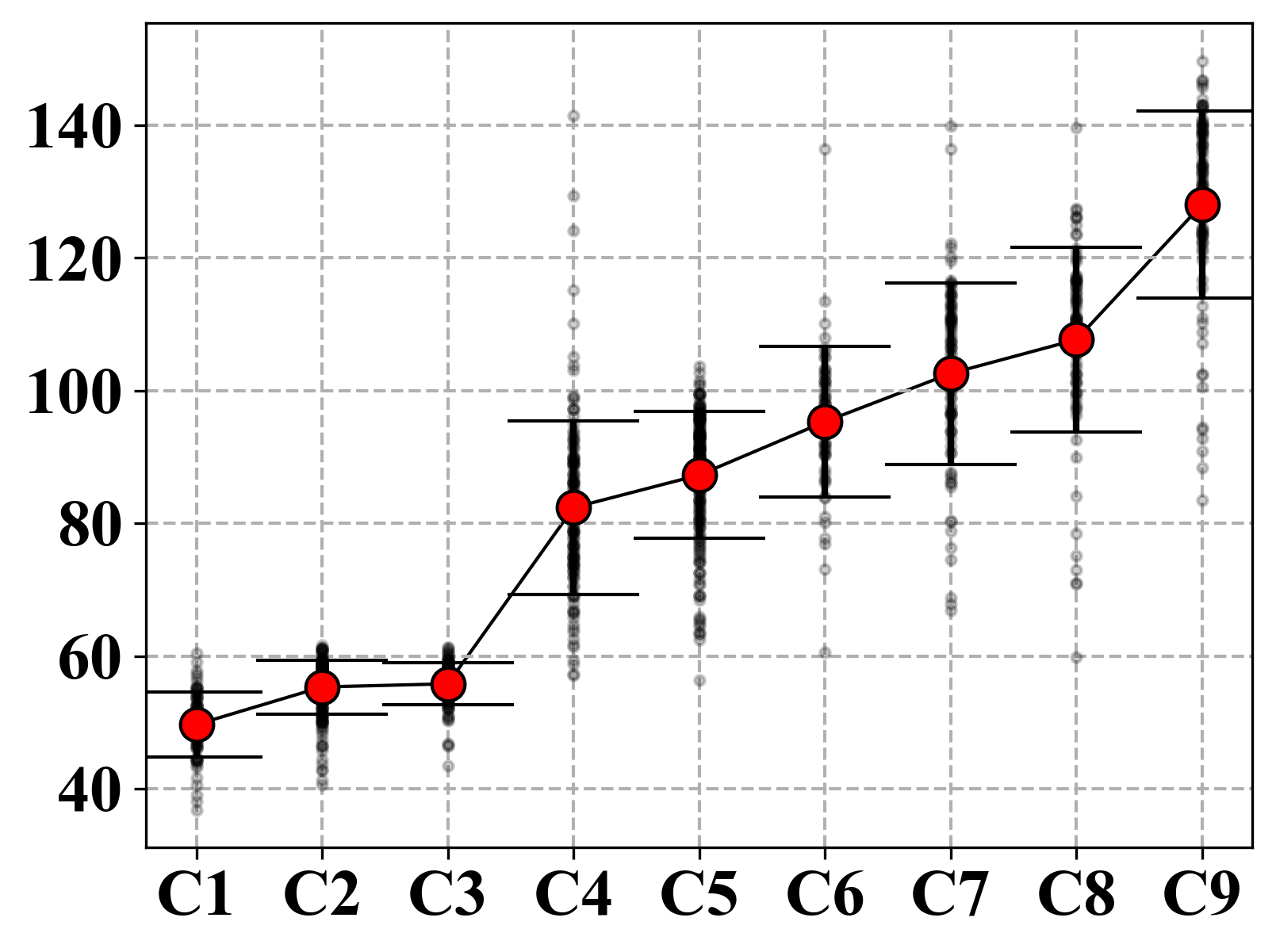}
  \vspace{-5pt}\phantomsection
\end{subfigure}%
\begin{subfigure}[b]{0.2\textwidth}
  \centering
  \includegraphics[width=\textwidth]{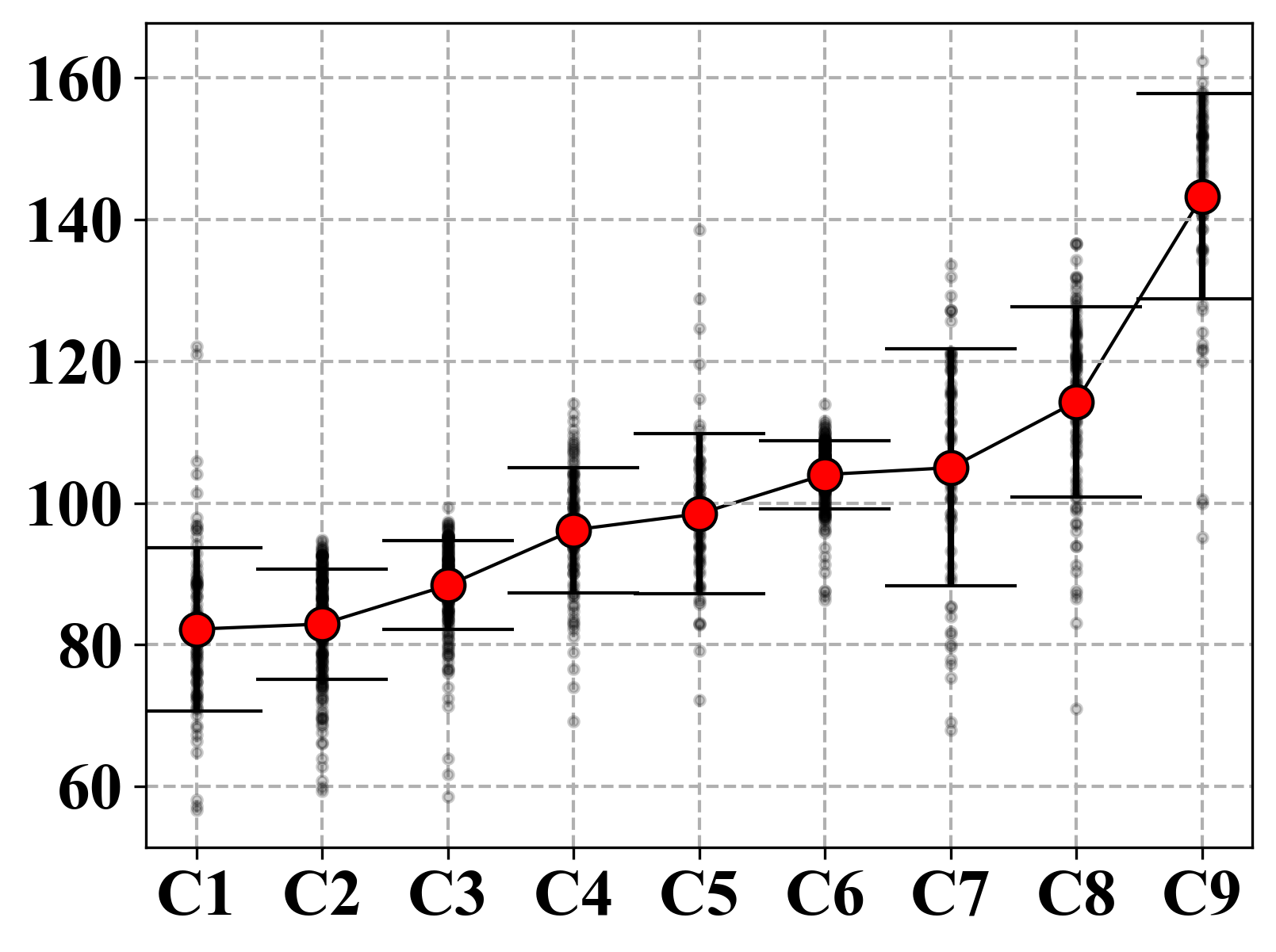}
  \vspace{-5pt}\phantomsection
\end{subfigure}%
\begin{subfigure}[b]{0.2\textwidth}
  \centering
  \includegraphics[width=\textwidth]{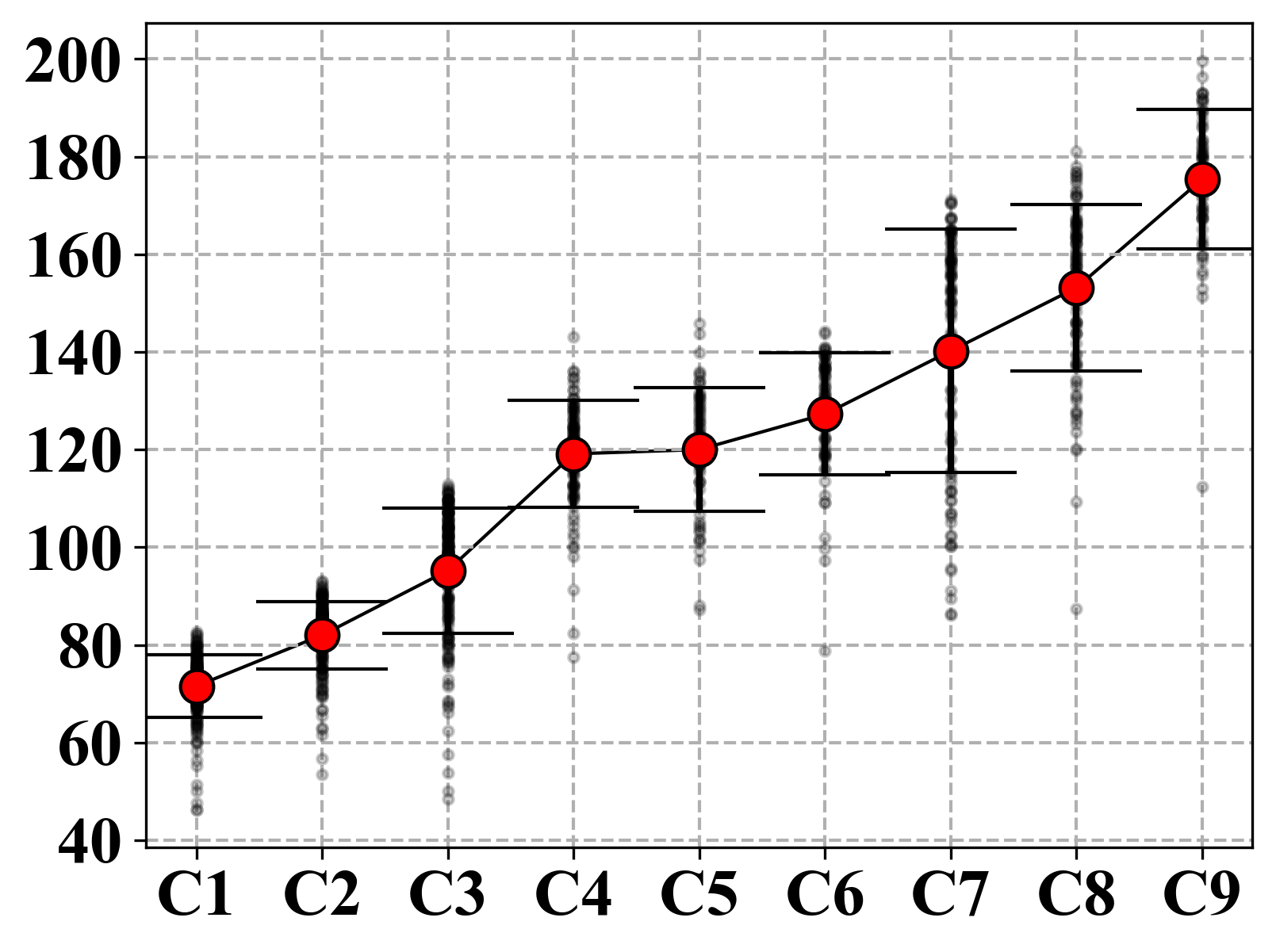}
  \vspace{-5pt}\phantomsection
\end{subfigure}%
\begin{subfigure}[b]{0.2\textwidth}
  \centering
  \includegraphics[width=\textwidth]{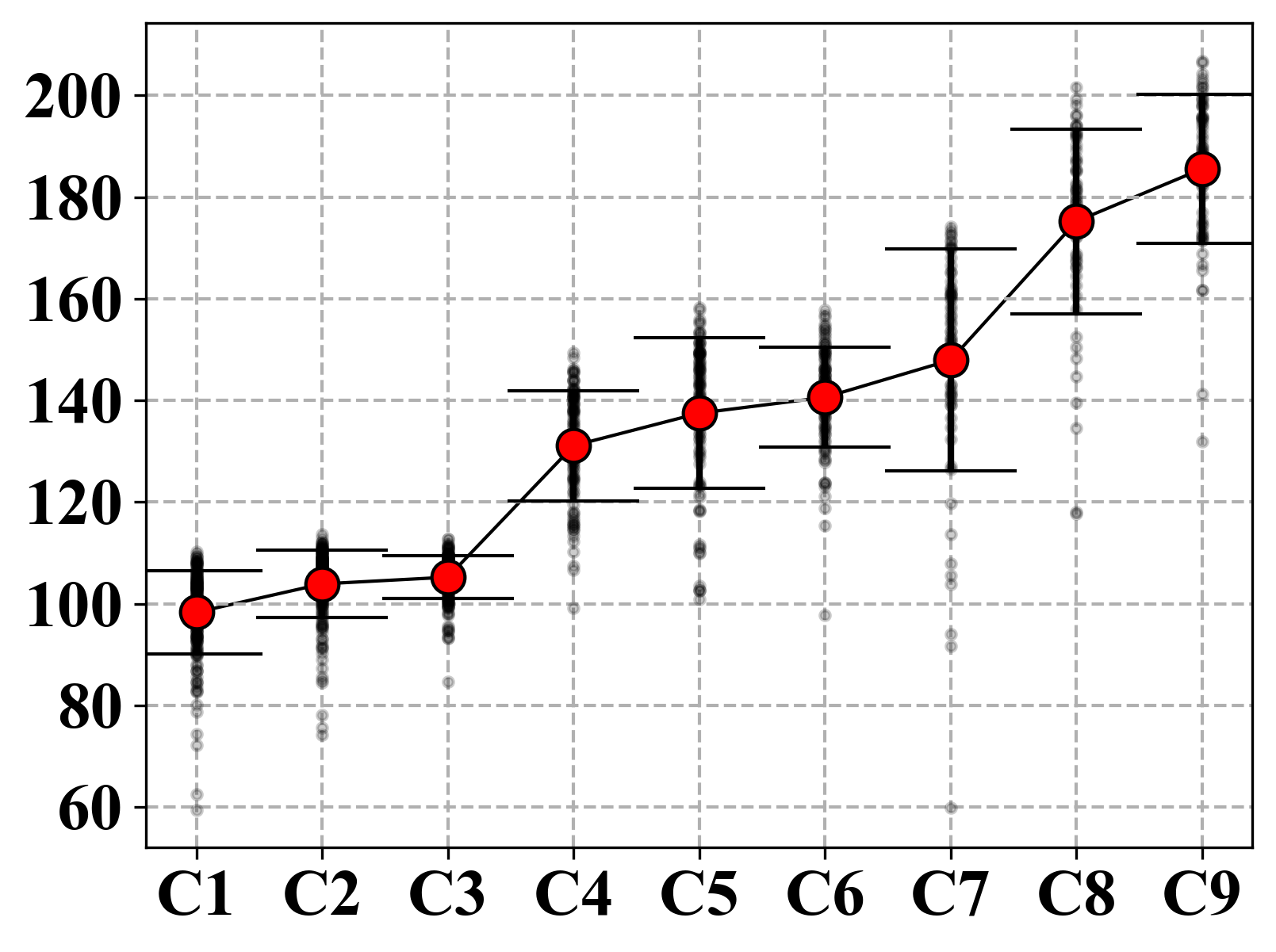}
  \vspace{-5pt}\phantomsection
\end{subfigure}

\begin{subfigure}[b]{0.2\textwidth}
  \centering
  \includegraphics[width=\textwidth]{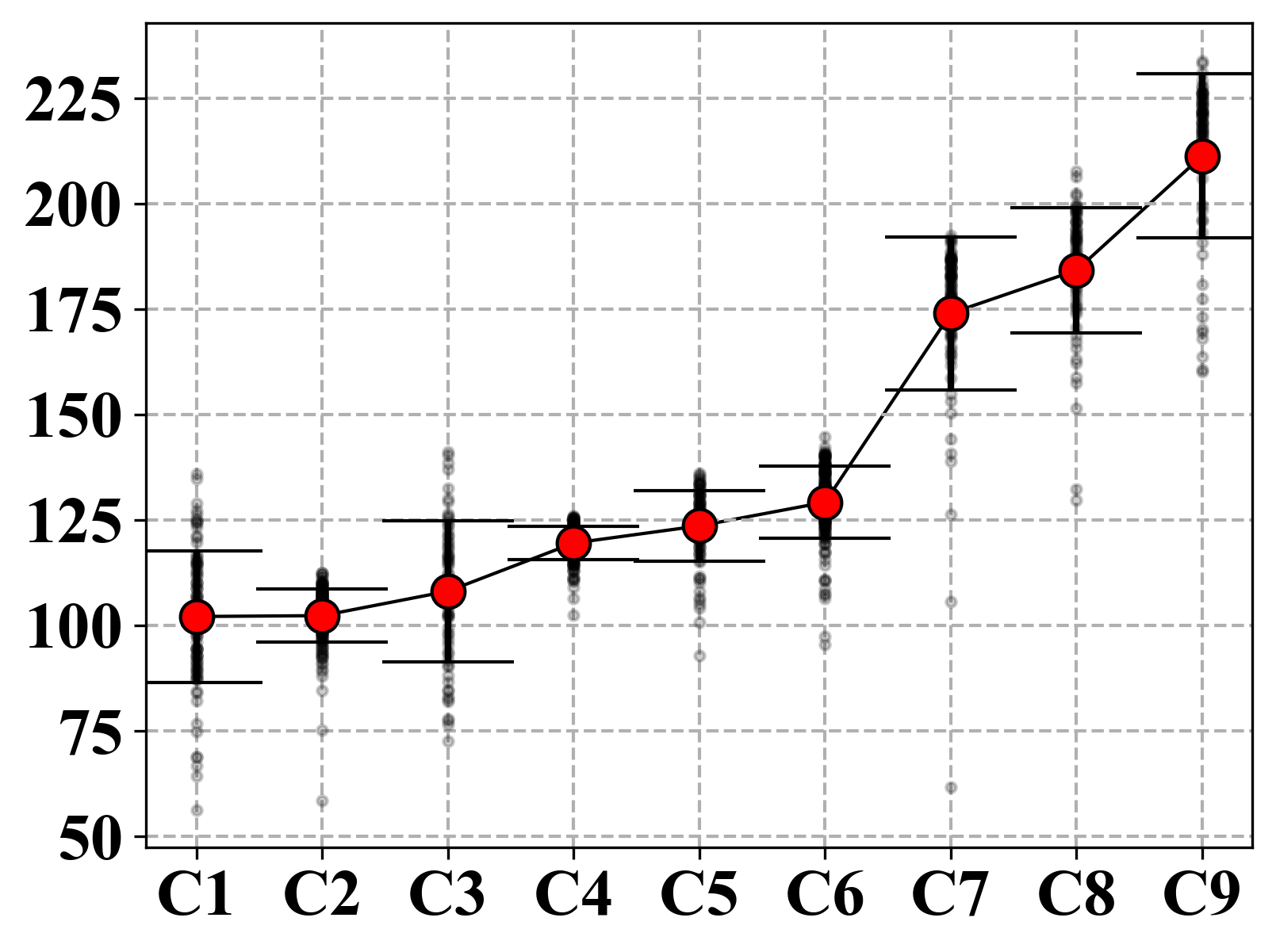}
  \vspace{-5pt}\phantomsection
\end{subfigure}%
\begin{subfigure}[b]{0.2\textwidth}
  \centering
  \includegraphics[width=\textwidth]{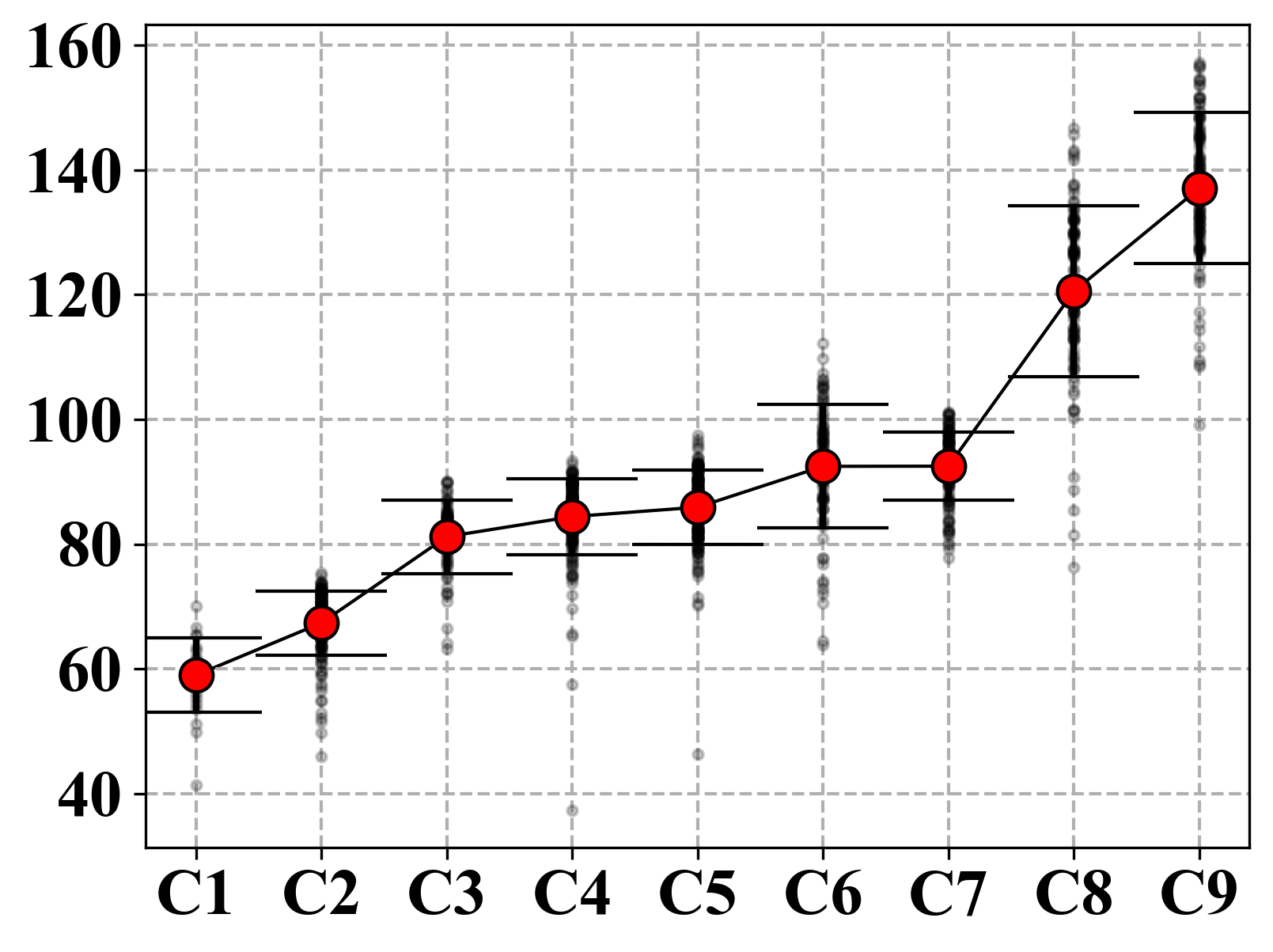}
  \vspace{-5pt}\phantomsection
\end{subfigure}%
\begin{subfigure}[b]{0.2\textwidth}
  \centering
  \includegraphics[width=\textwidth]{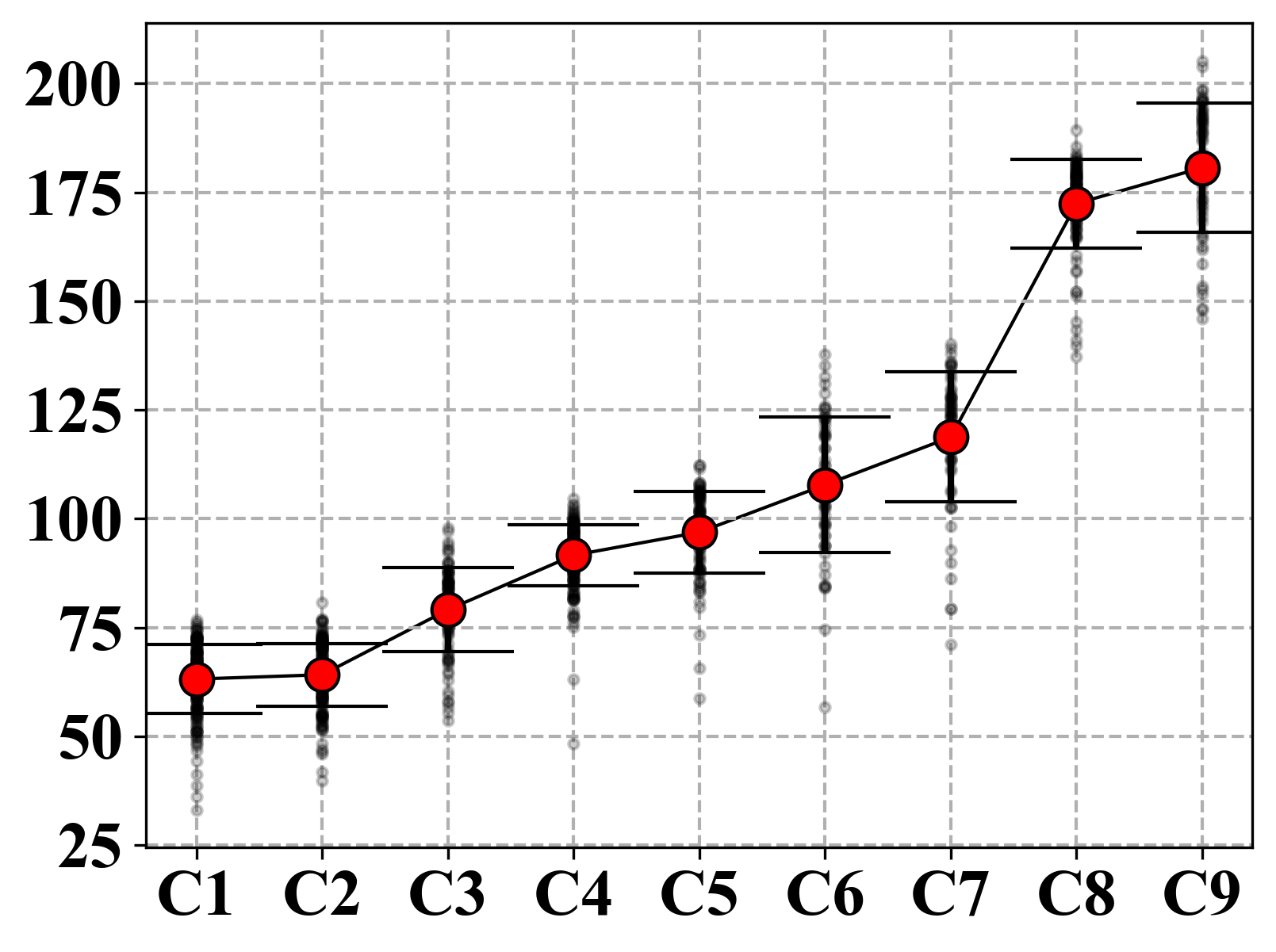}
  \vspace{-5pt}\phantomsection
\end{subfigure}%
\begin{subfigure}[b]{0.2\textwidth}
  \centering
  \includegraphics[width=\textwidth]{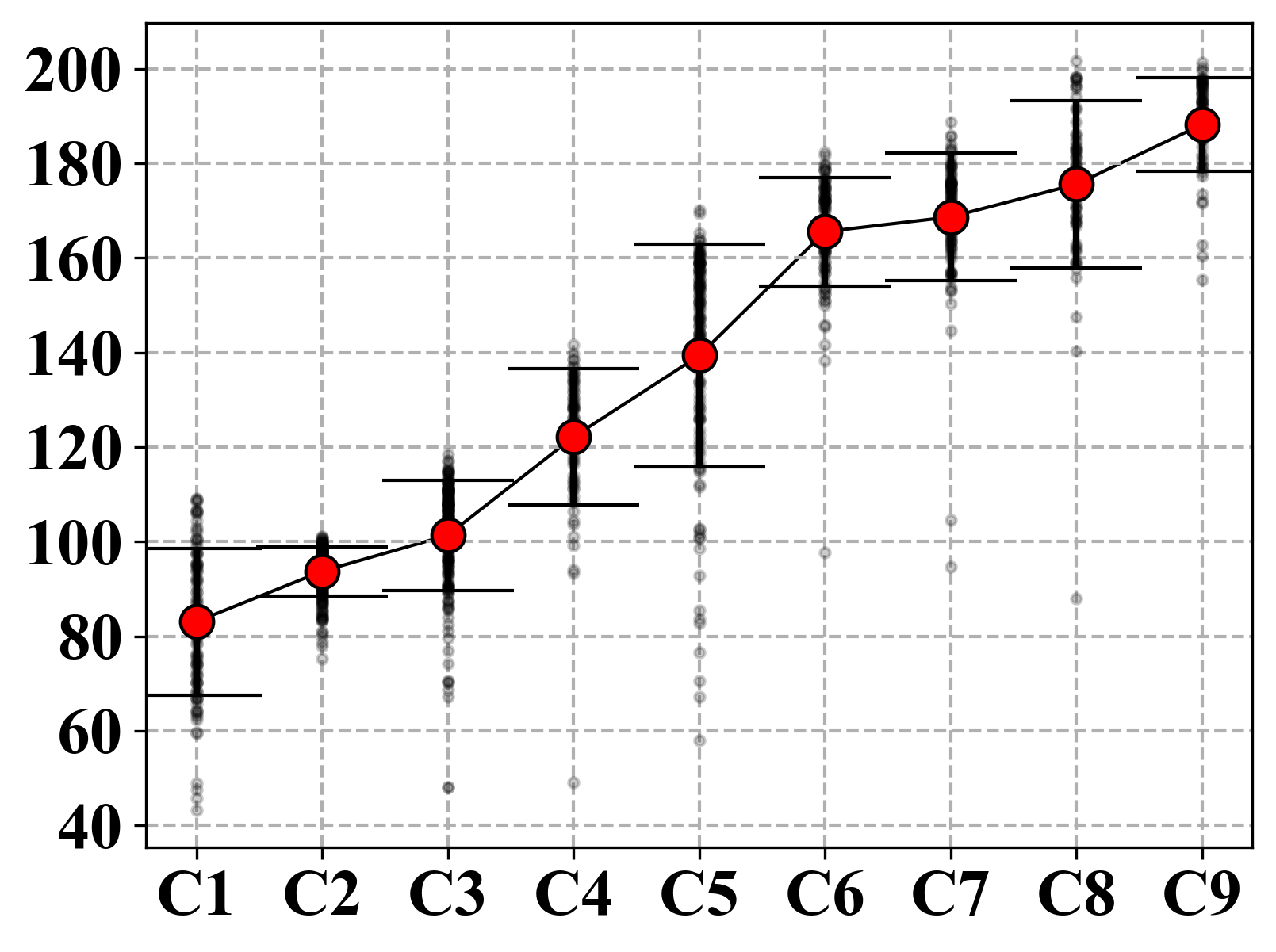}
  \vspace{-5pt}\phantomsection
\end{subfigure}%
\begin{subfigure}[b]{0.2\textwidth}
  \centering
  \includegraphics[width=\textwidth]{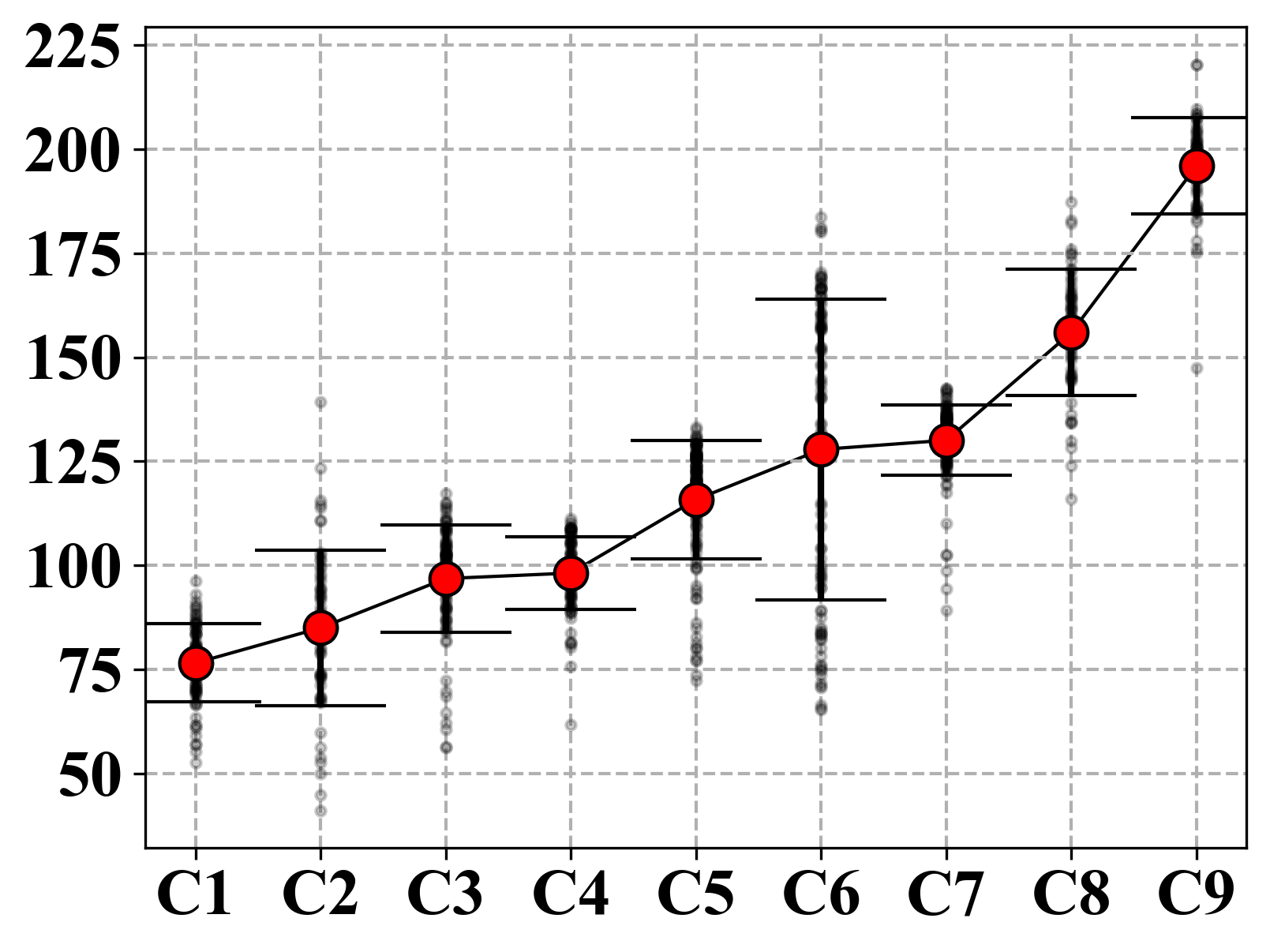}
  \vspace{-5pt}\phantomsection
\end{subfigure}

\begin{subfigure}[b]{0.2\textwidth}
  \centering
  \includegraphics[width=\textwidth]{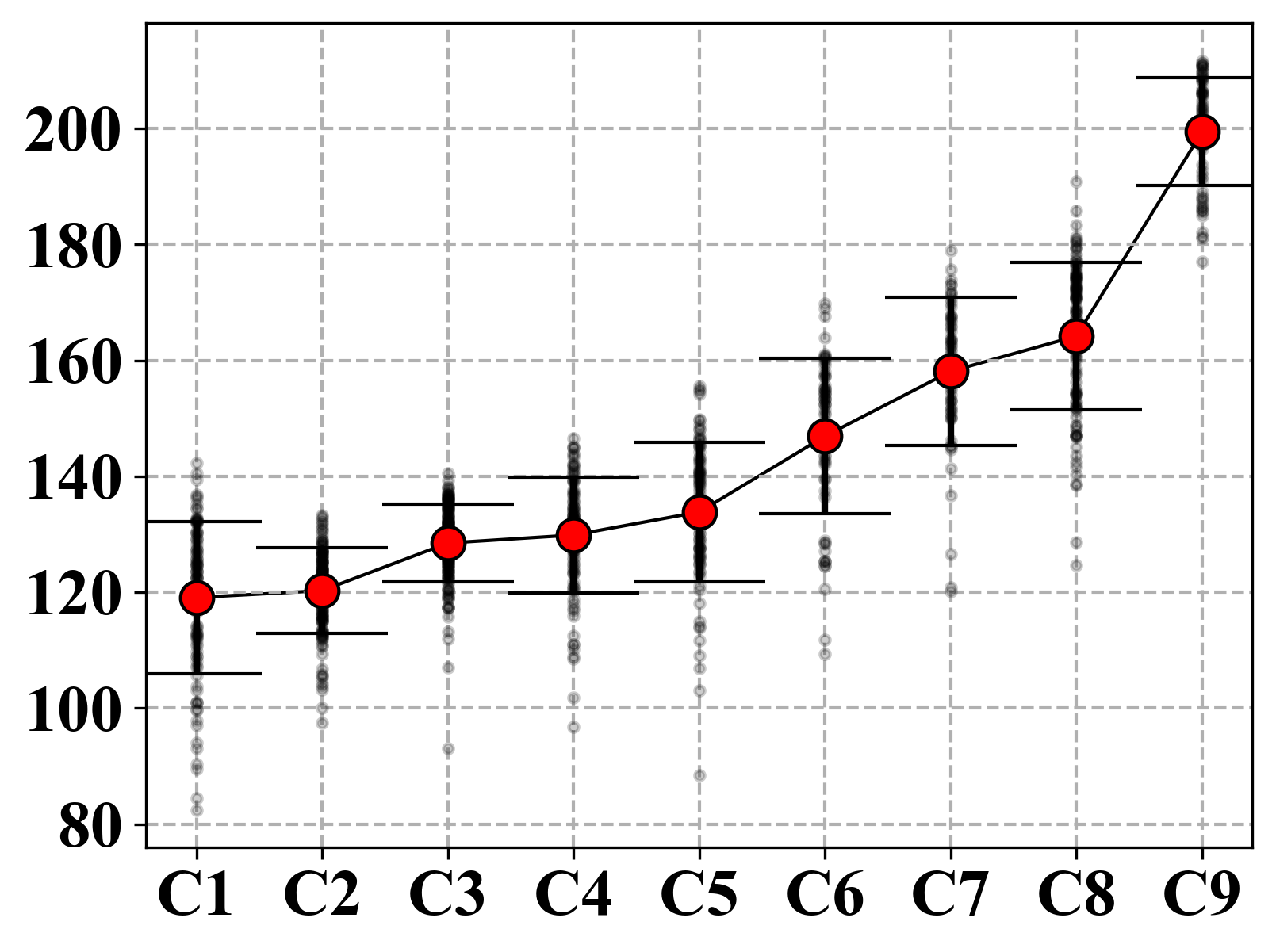}
  \vspace{-5pt}\phantomsection
\end{subfigure}%
\begin{subfigure}[b]{0.2\textwidth}
  \centering
  \includegraphics[width=\textwidth]{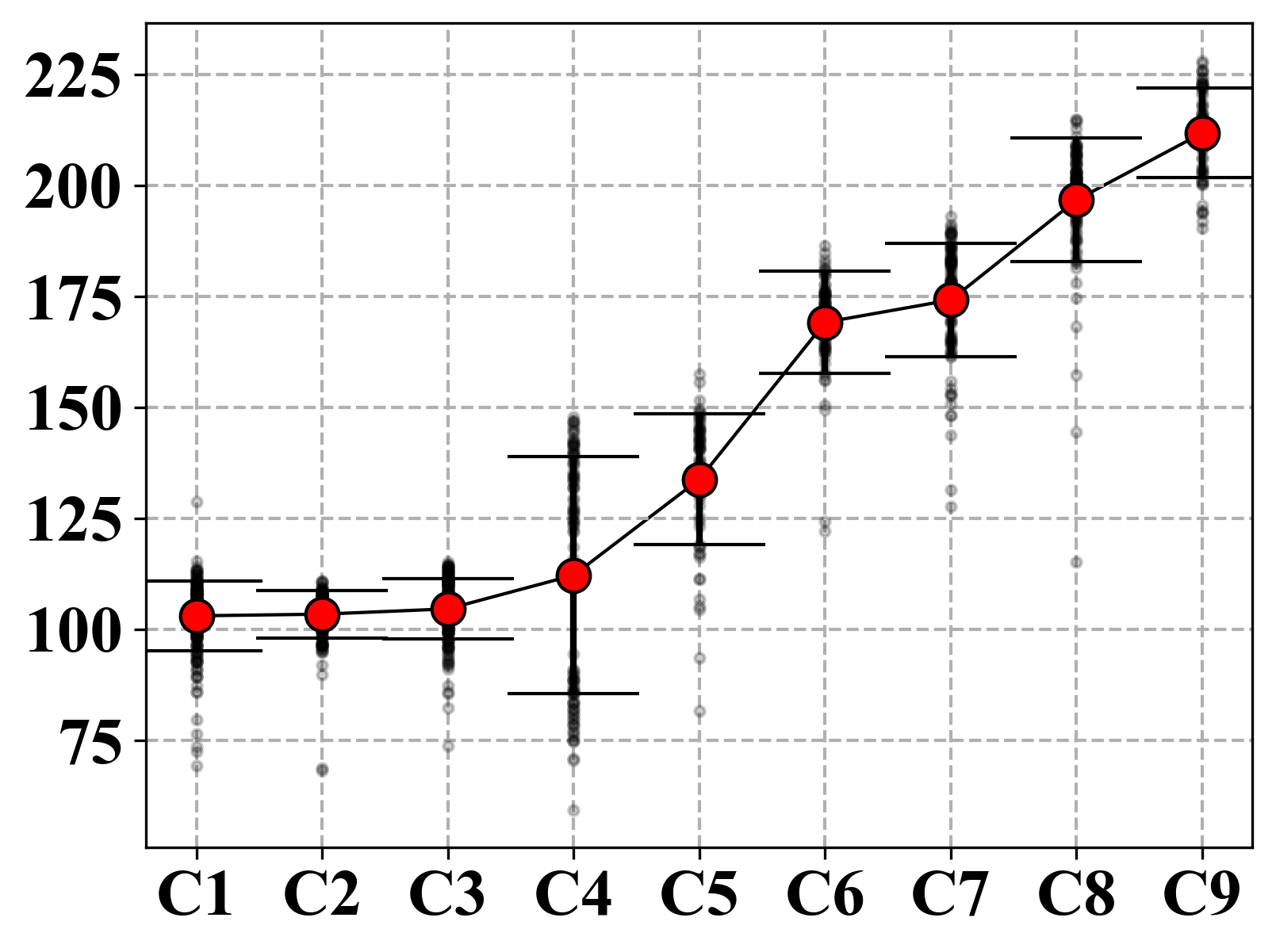}
  \vspace{-5pt}\phantomsection
\end{subfigure}%
\begin{subfigure}[b]{0.2\textwidth}
  \centering
  \includegraphics[width=\textwidth]{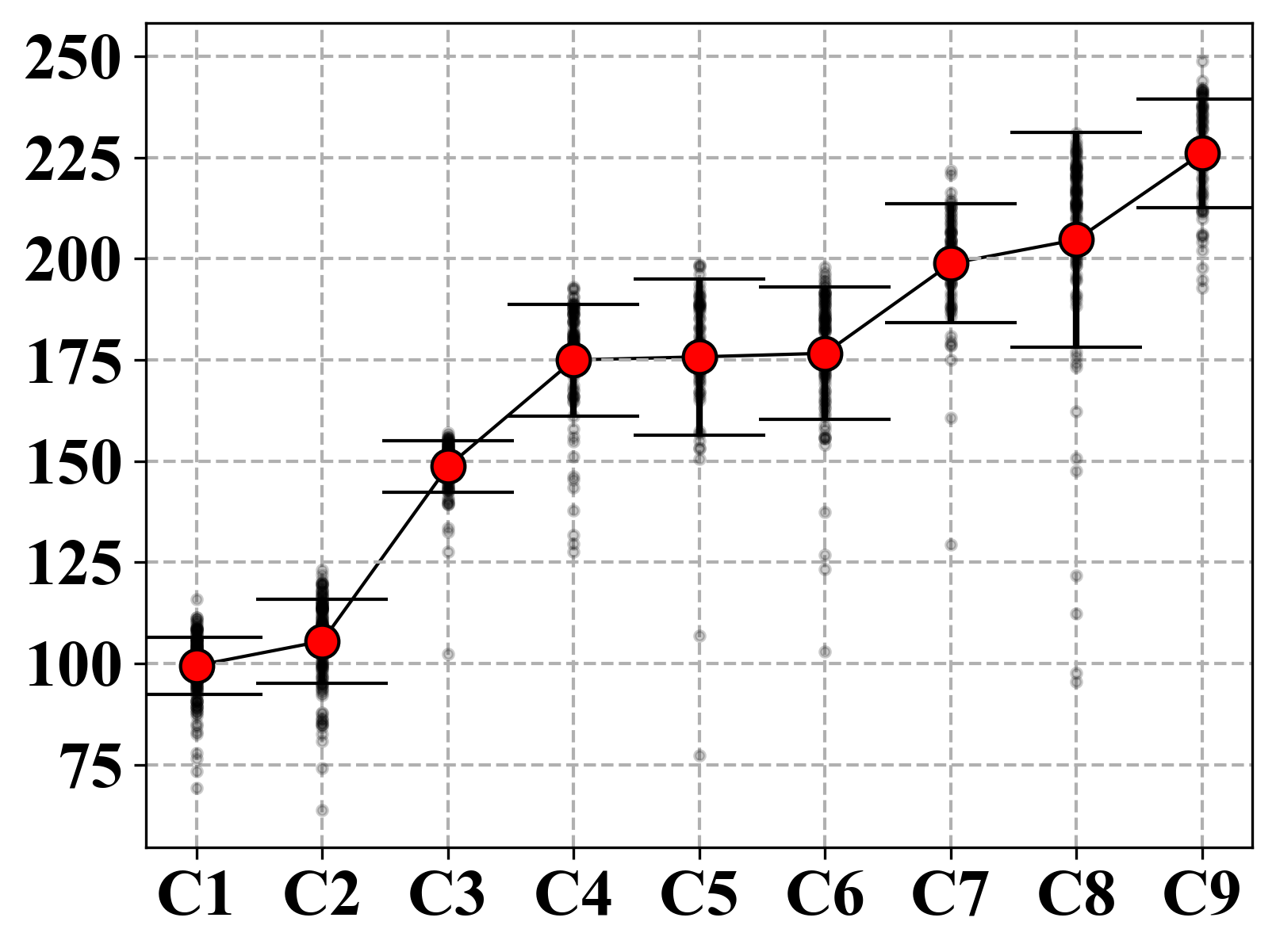}
  \vspace{-5pt}\phantomsection
\end{subfigure}%
\begin{subfigure}[b]{0.2\textwidth}
  \centering
  \includegraphics[width=\textwidth]{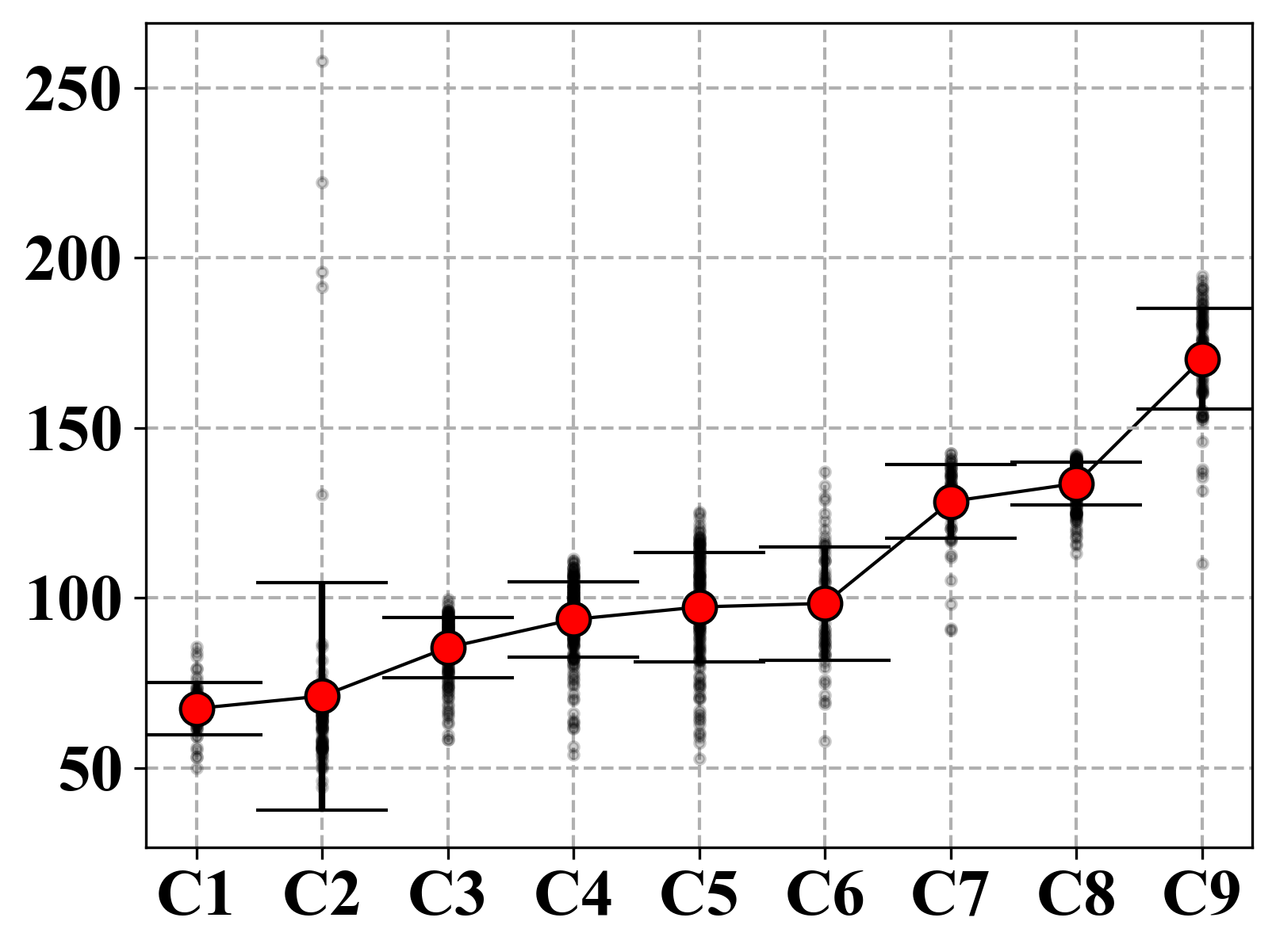}
  \vspace{-5pt}\phantomsection
\end{subfigure}%
\begin{subfigure}[b]{0.2\textwidth}
  \centering
  \includegraphics[width=\textwidth]{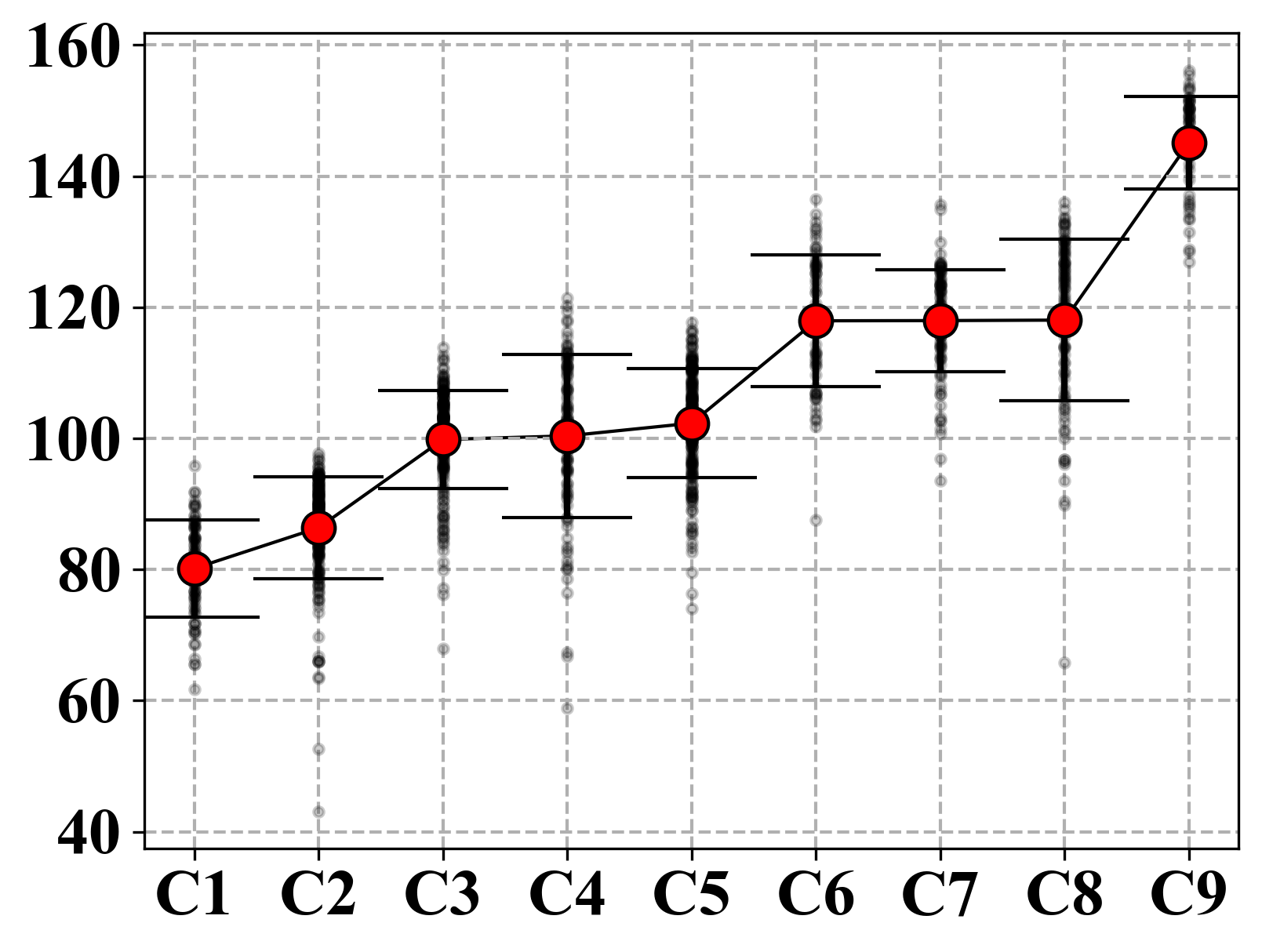}
  \vspace{-5pt}\phantomsection
\end{subfigure}

\begin{subfigure}[b]{0.2\textwidth}
  \centering
  \includegraphics[width=\textwidth]{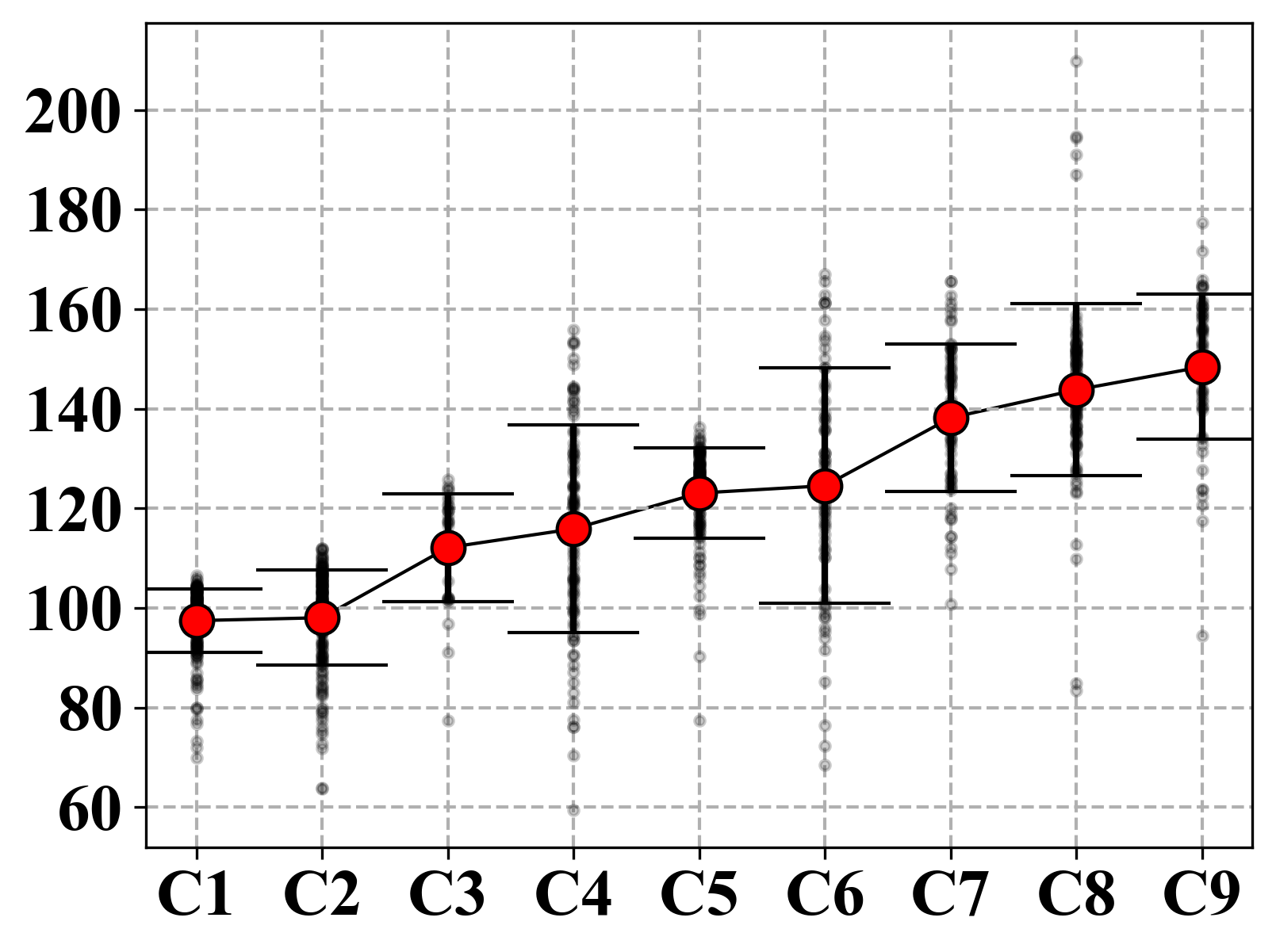}
  \vspace{-5pt}\phantomsection
\end{subfigure}%
\begin{subfigure}[b]{0.2\textwidth}
  \centering
  \includegraphics[width=\textwidth]{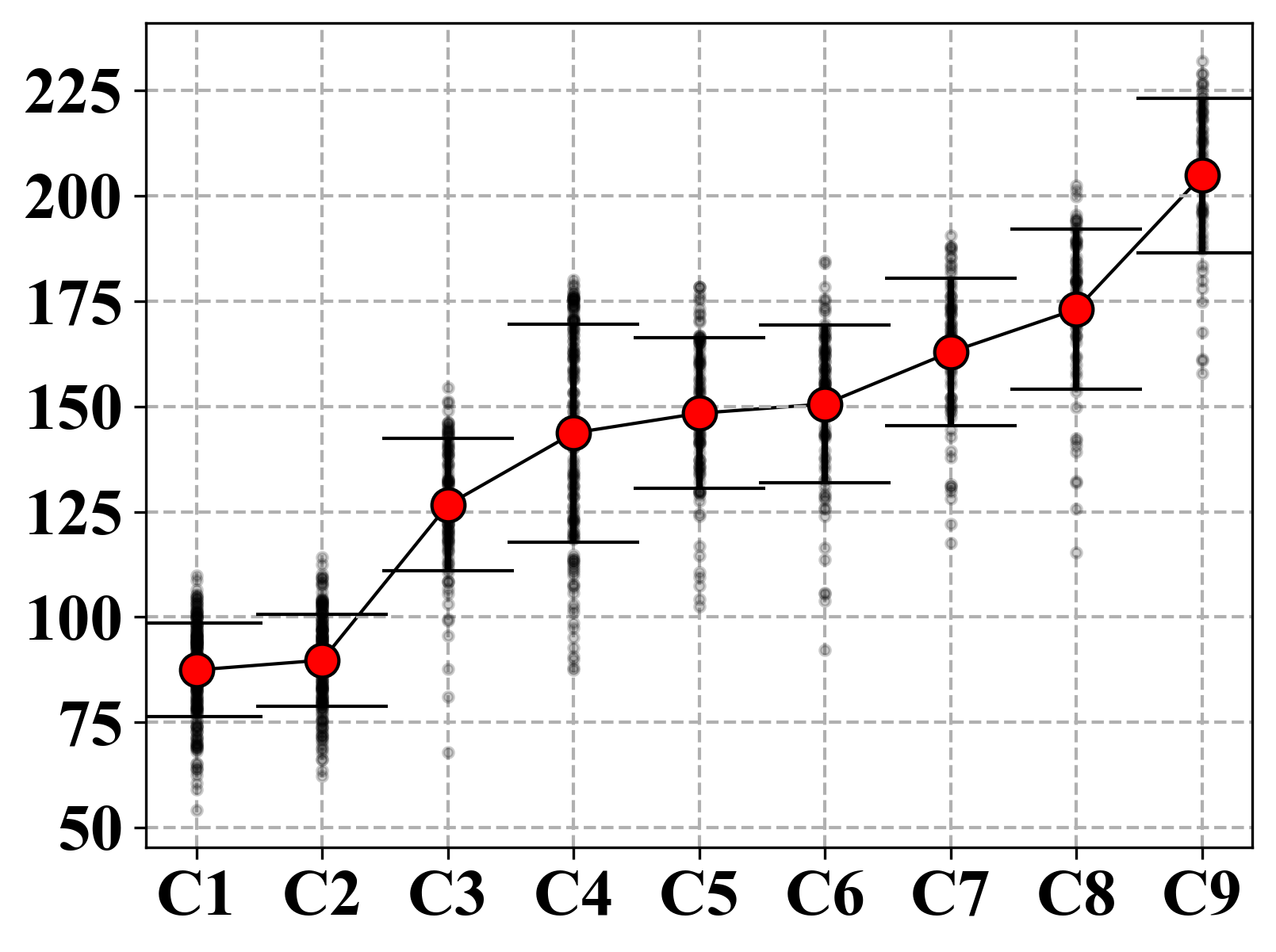}
  \vspace{-5pt}\phantomsection
\end{subfigure}%
\begin{subfigure}[b]{0.2\textwidth}
  \centering
  \includegraphics[width=\textwidth]{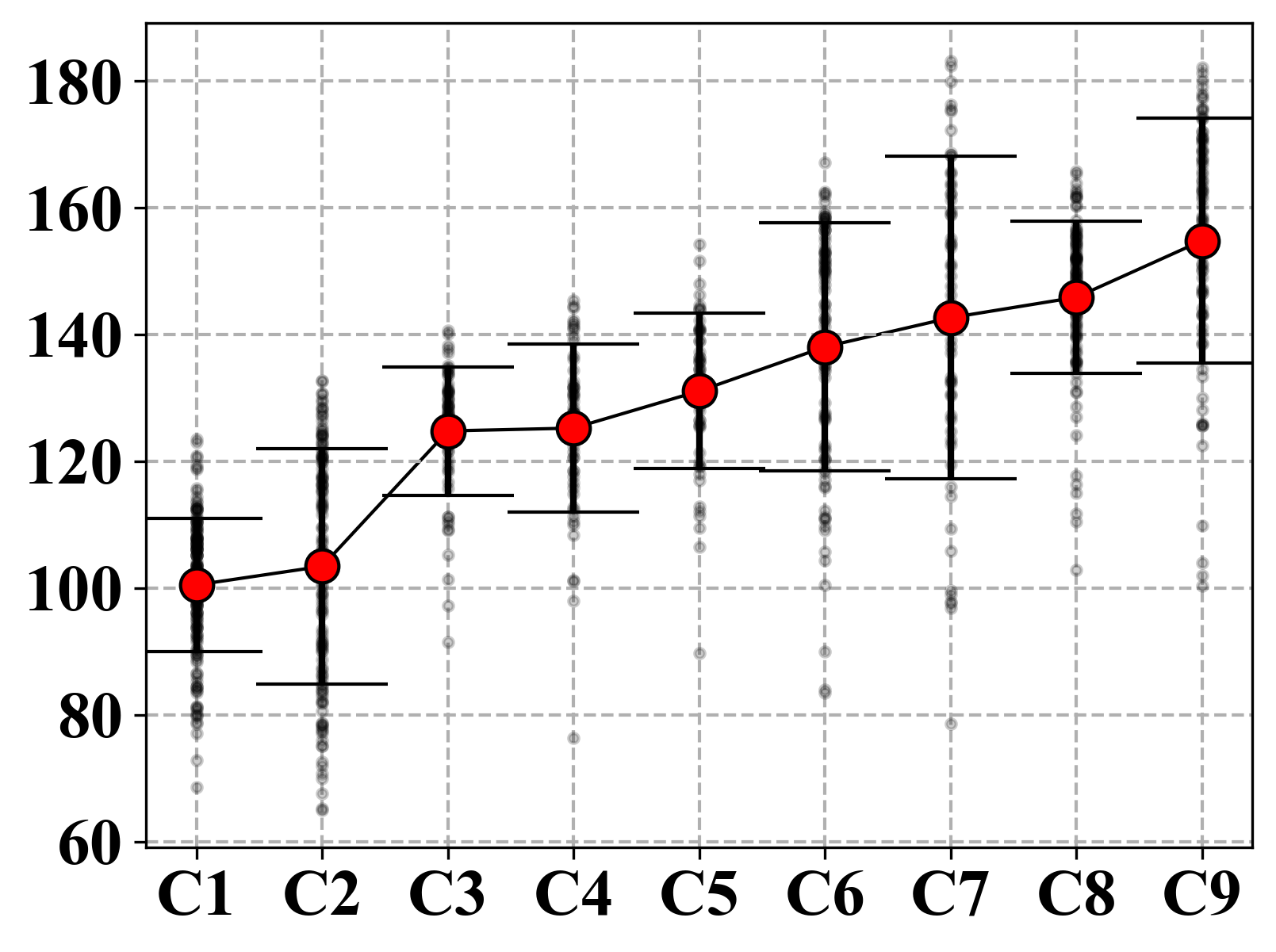}
  \vspace{-5pt}\phantomsection
\end{subfigure}%
\begin{subfigure}[b]{0.2\textwidth}
  \centering
  \includegraphics[width=\textwidth]{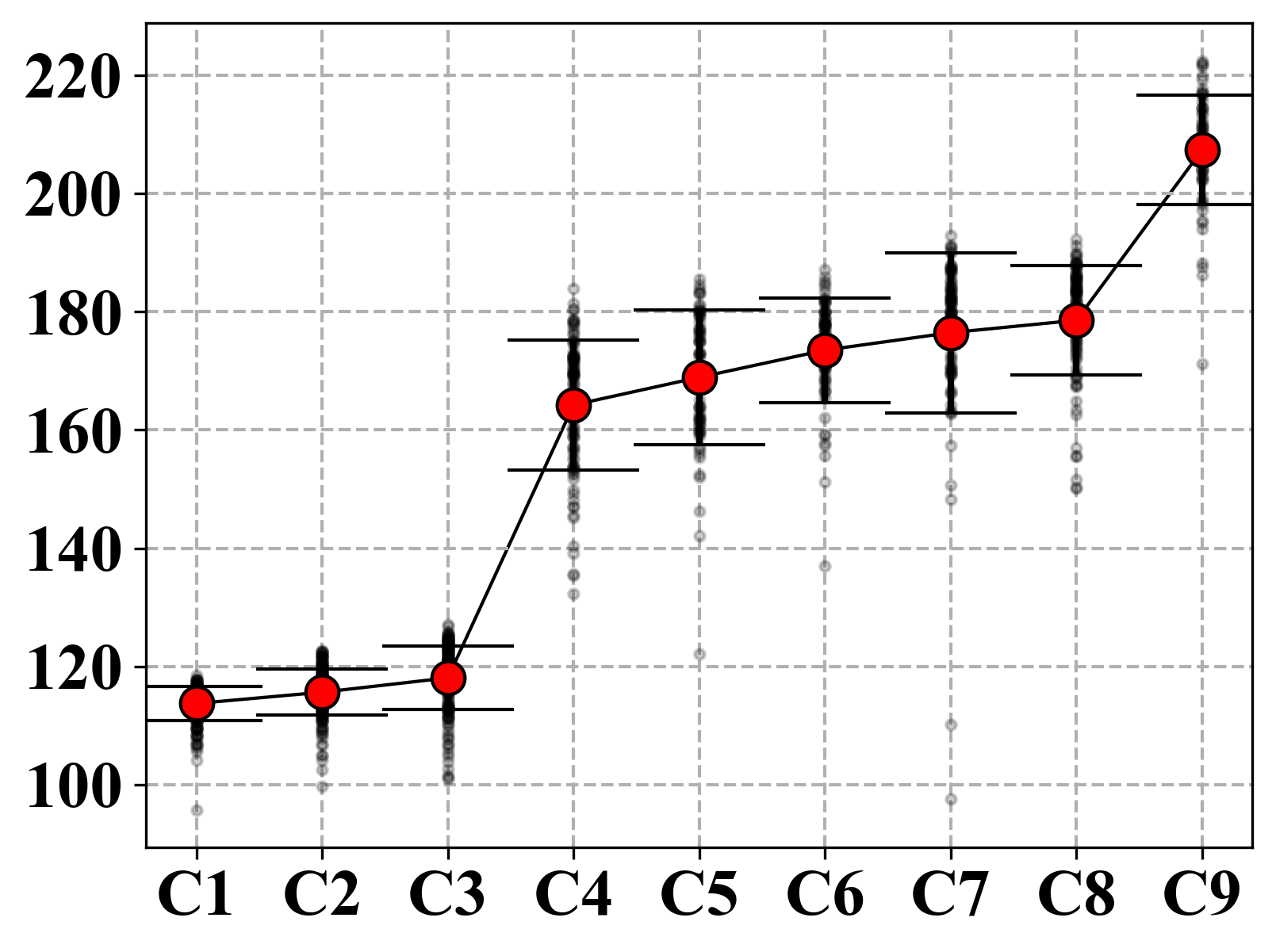}
  \vspace{-5pt}\phantomsection
\end{subfigure}%
\begin{subfigure}[b]{0.2\textwidth}
  \centering
  \includegraphics[width=\textwidth]{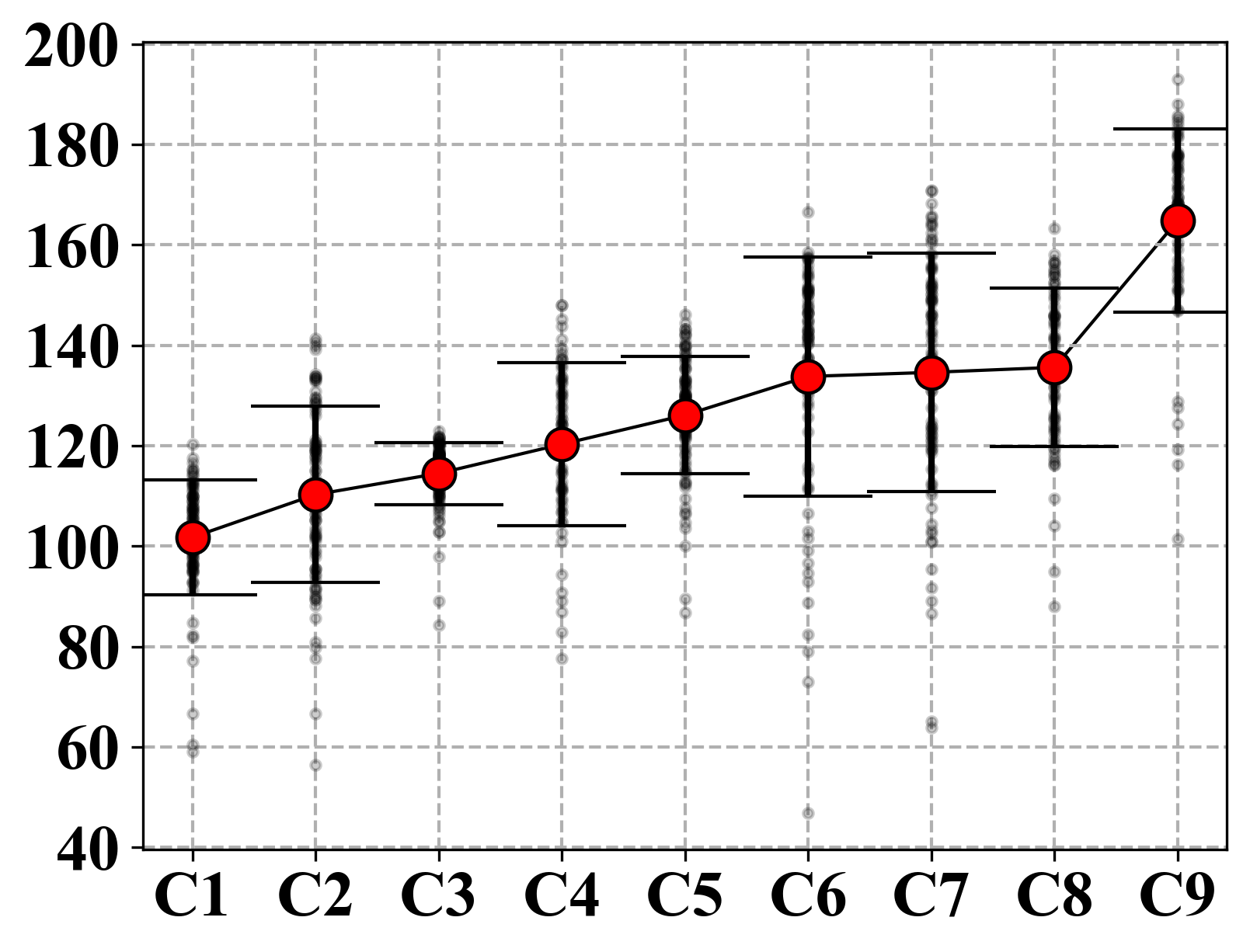}
  \vspace{-5pt}\phantomsection
\end{subfigure}

\begin{subfigure}[b]{0.2\textwidth}
  \centering
  \includegraphics[width=\textwidth]{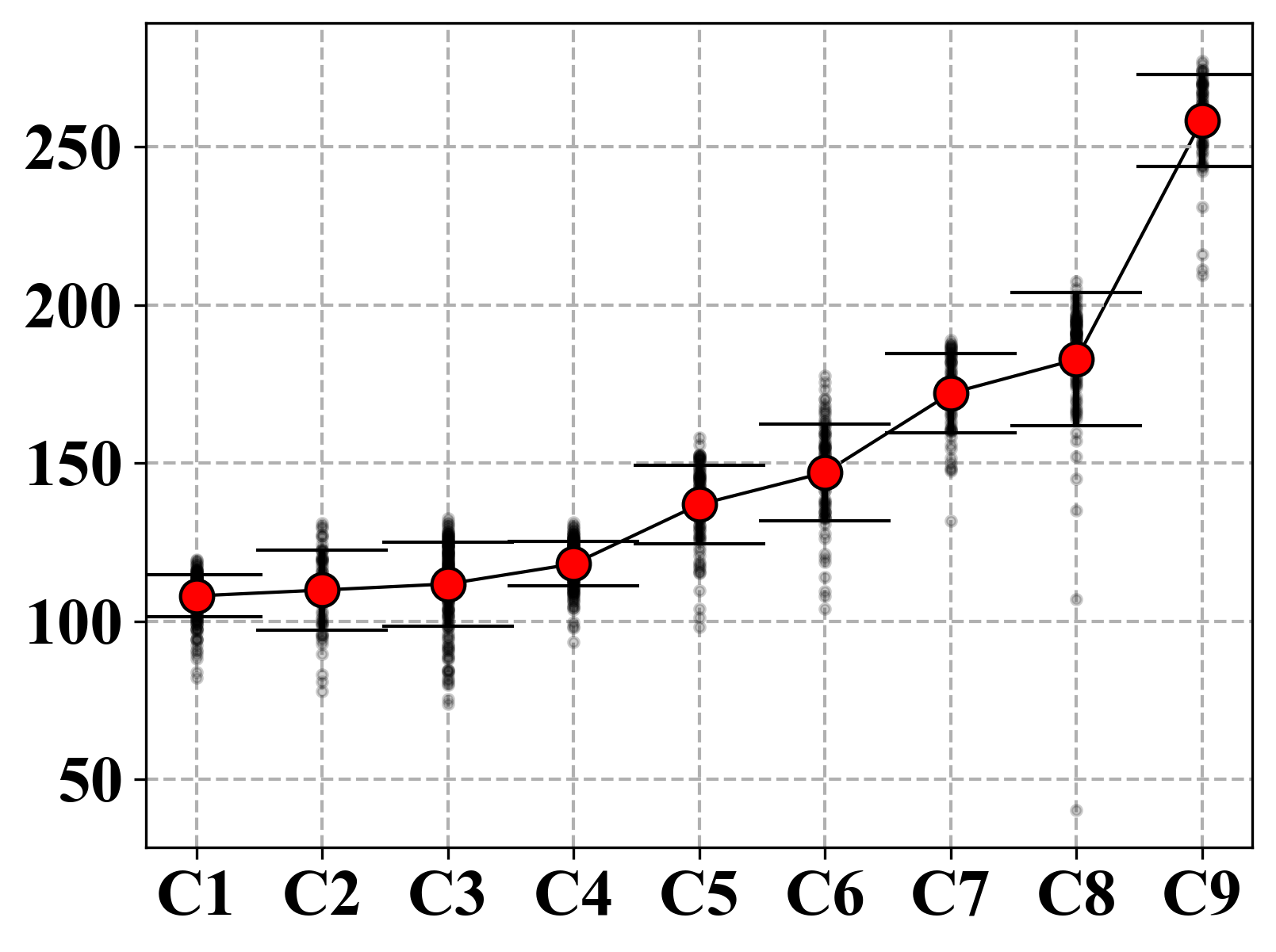}
  \vspace{-5pt}\phantomsection
\end{subfigure}%
\begin{subfigure}[b]{0.2\textwidth}
  \centering
  \includegraphics[width=\textwidth]{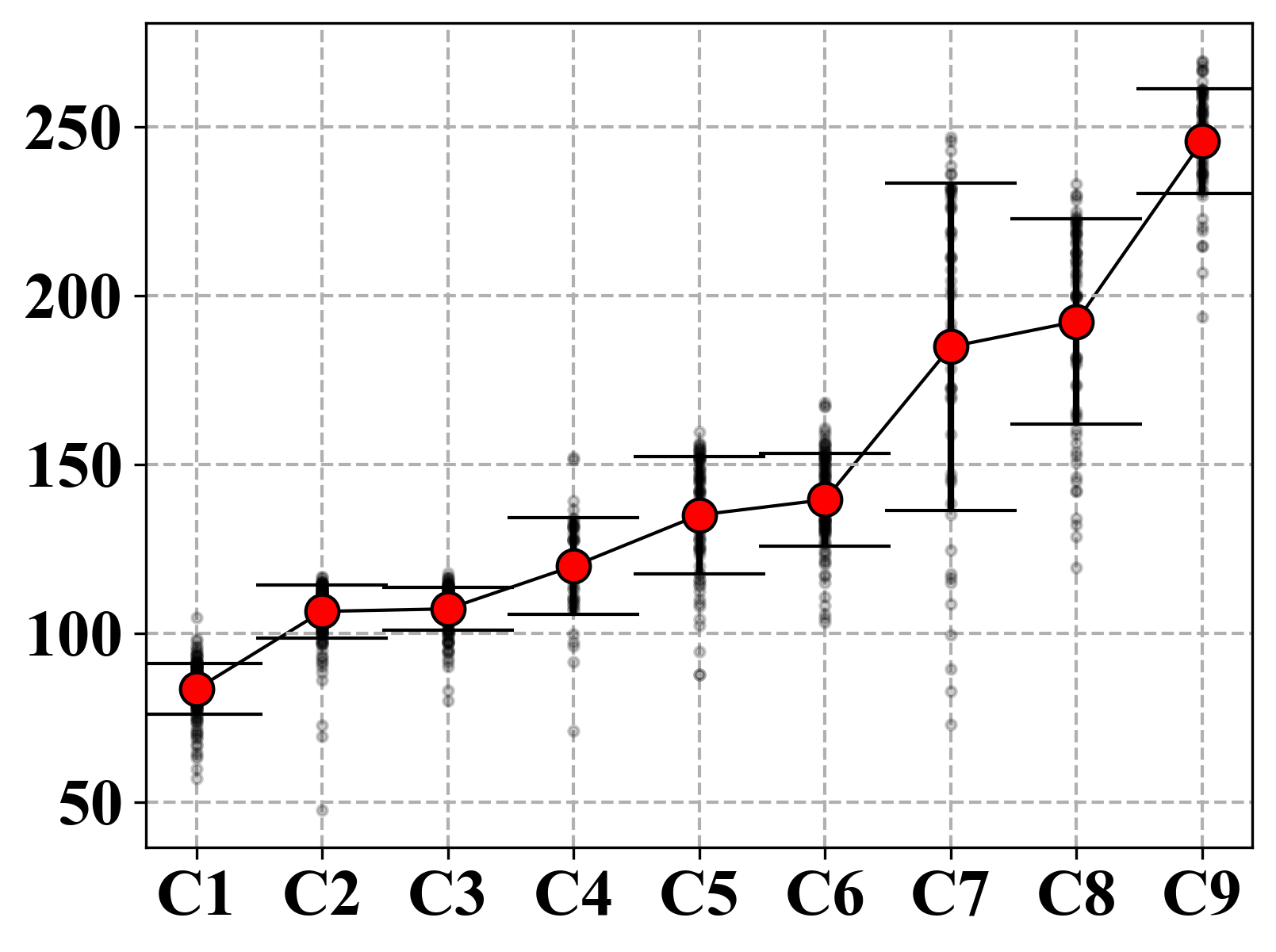}
  \vspace{-5pt}\phantomsection
\end{subfigure}%
\begin{subfigure}[b]{0.2\textwidth}
  \centering
  \includegraphics[width=\textwidth]{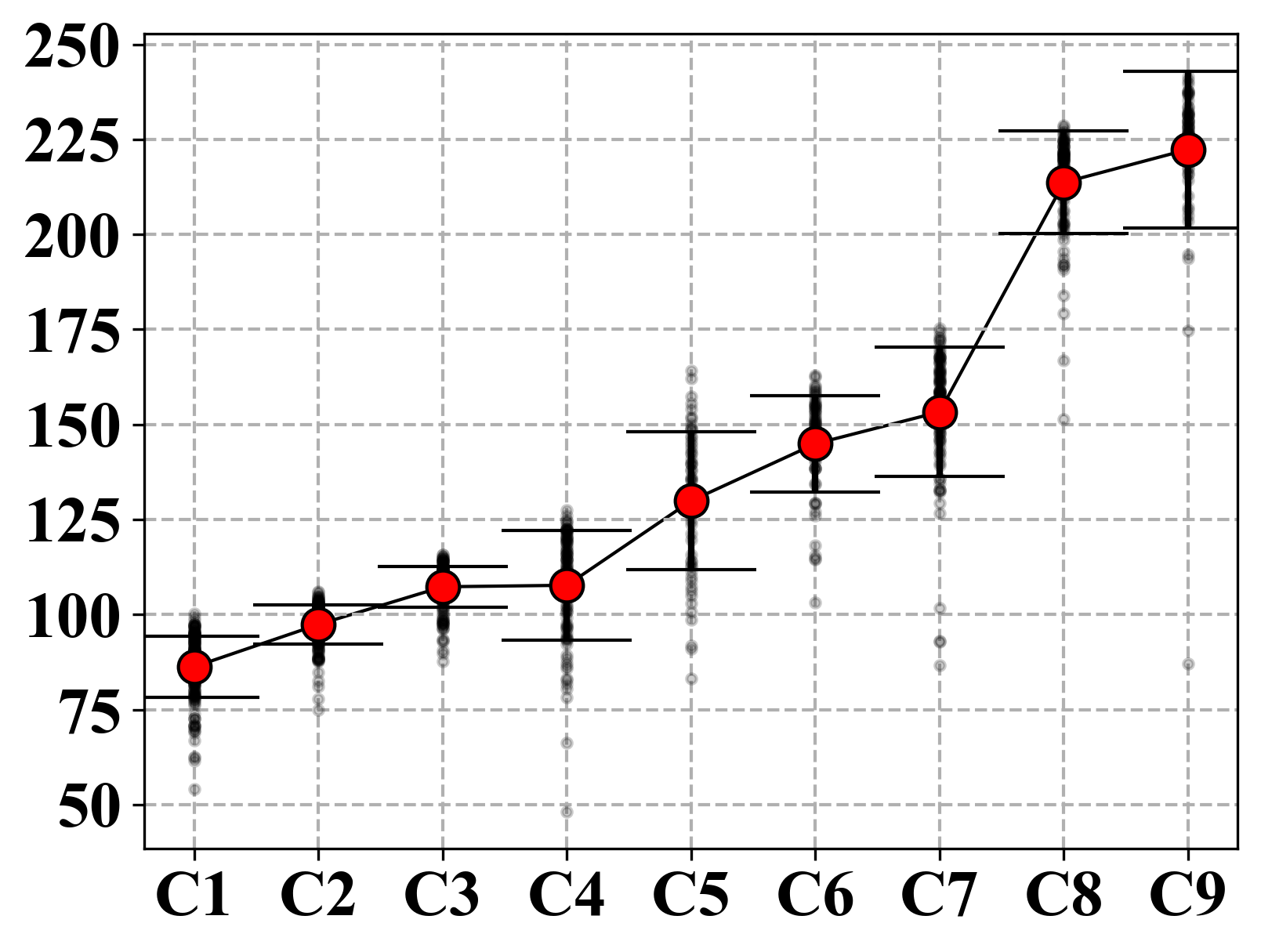}
  \vspace{-5pt}\phantomsection
\end{subfigure}%
\begin{subfigure}[b]{0.2\textwidth}
  \centering
  \includegraphics[width=\textwidth]{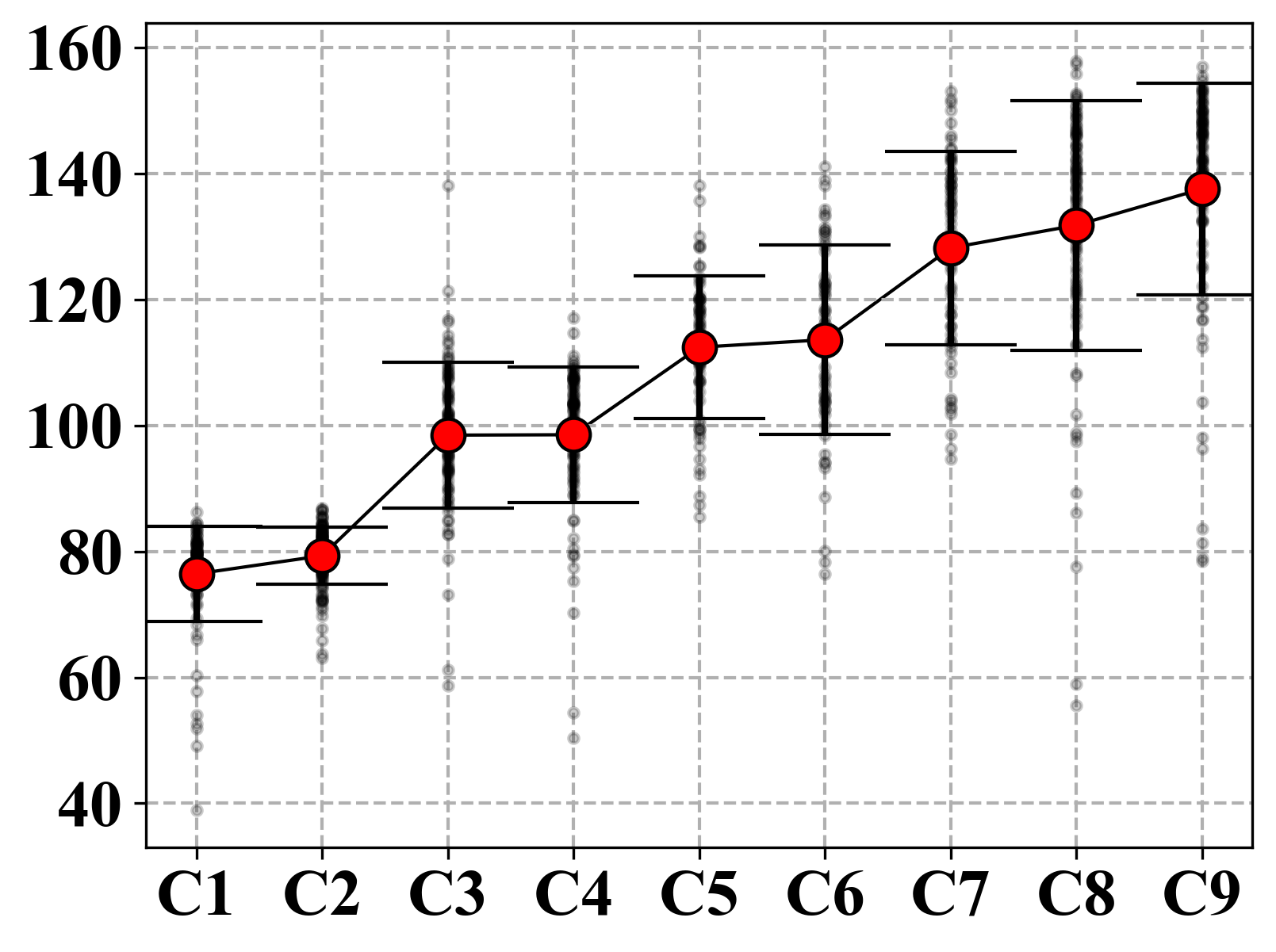}
  \vspace{-5pt}\phantomsection
\end{subfigure}%
\begin{subfigure}[b]{0.2\textwidth}
  \centering
  \includegraphics[width=\textwidth]{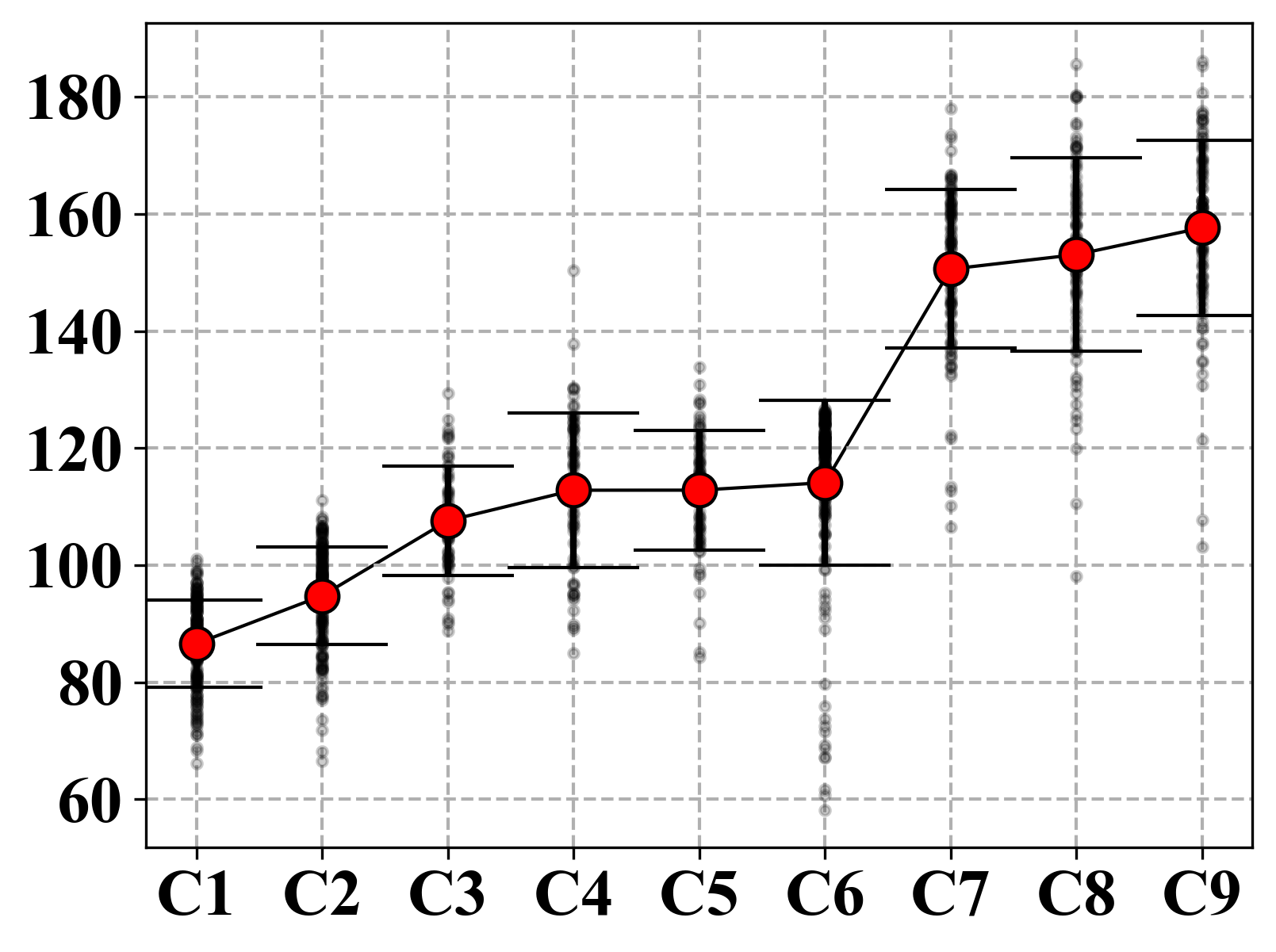}
  \vspace{-5pt}\phantomsection
\end{subfigure}
\caption{\small Similar to Fig.~\ref{CLUSTERFIGURE}, we extract the density ratios of each cluster for the subject, corresponding to one movement and one session, rank their means, and plot them along with their standard deviations. It can be seen that for each subject, the density ratio value for each cluster is different, indicating that our measures provide a unique measurement for each movement of the subject. This figure is extended from Fig.~\ref{CLUSTERFIGURE} and shows that this observation applies to all participants.}
\label{ratio_figure_variation}
\end{figure}%
\vspace{-15pt}\begin{figure}[H]
\centering
\begin{minipage}{\textwidth}
\centering
{\large\textbf{\texttt{TSNE, REACH, GRASP, TWIST, SUB1$\sim$SUB10}}}
\end{minipage}
\begin{subfigure}[b]{0.33\textwidth}
  \centering
  \includegraphics[width=\textwidth]{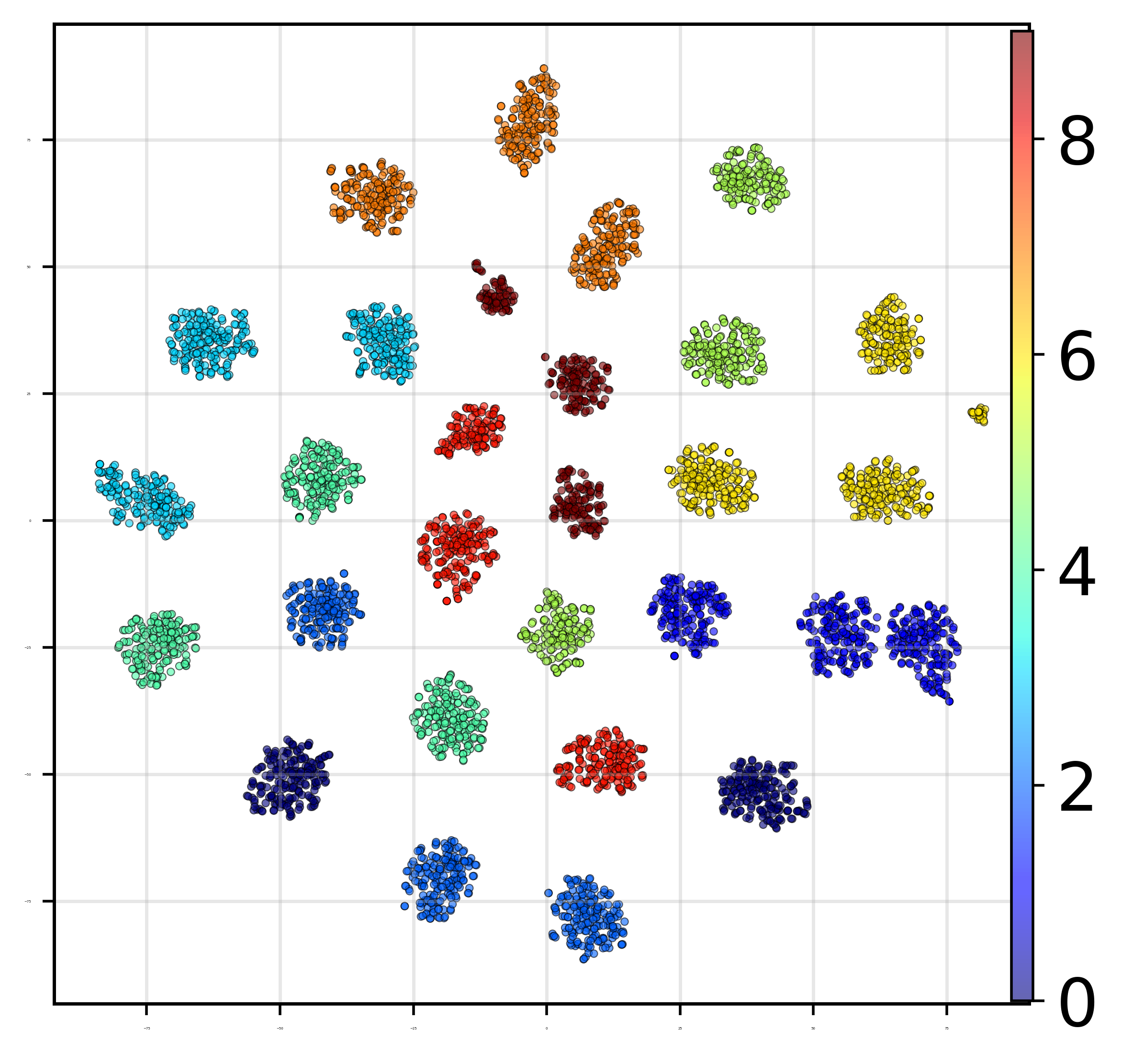}
  \vspace{-5pt}\phantomsection
\end{subfigure}%
\begin{subfigure}[b]{0.33\textwidth}
  \centering
  \includegraphics[width=\textwidth]{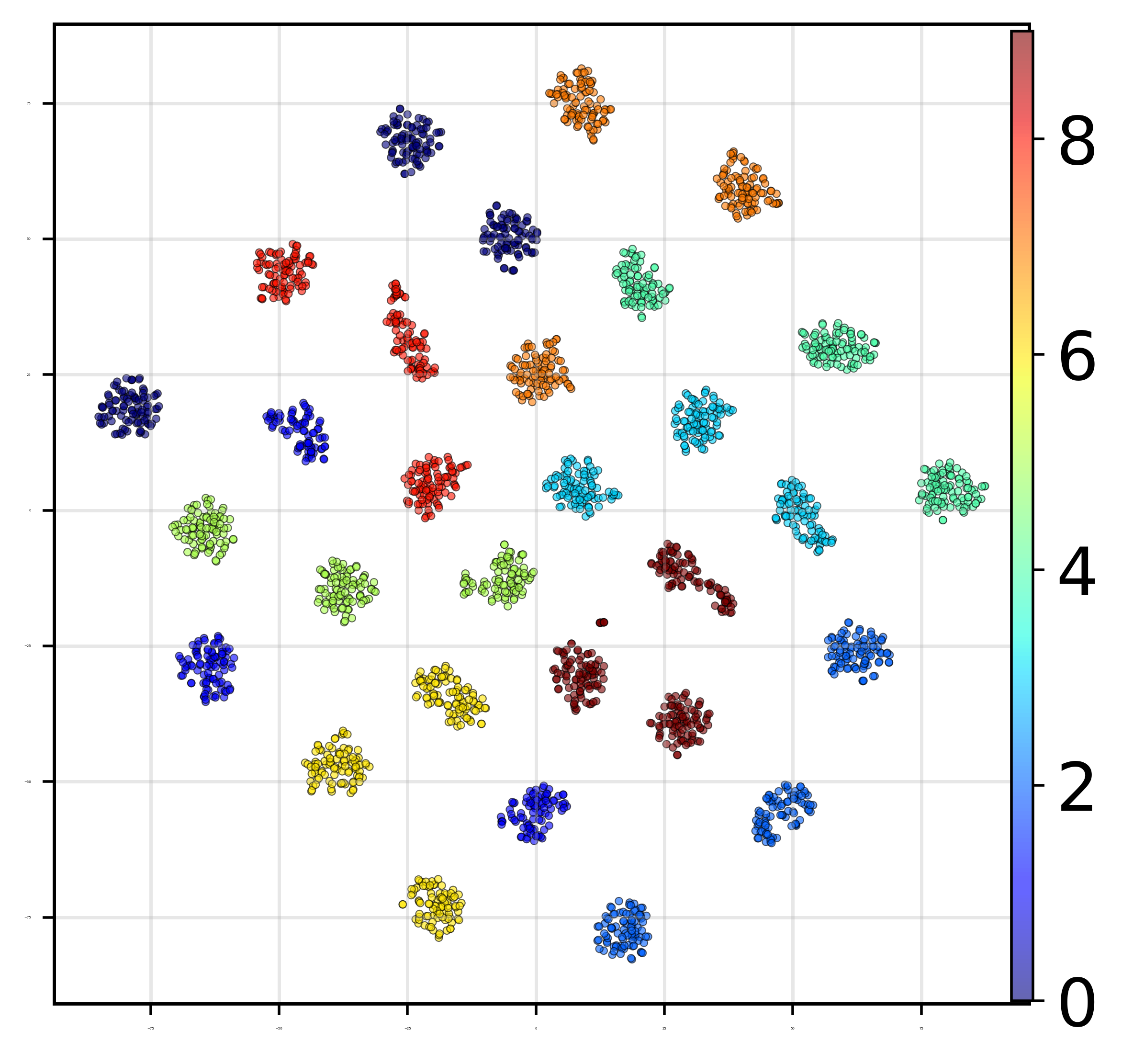}
  \vspace{-5pt}\phantomsection
\end{subfigure}%
\begin{subfigure}[b]{0.33\textwidth}
  \centering
  \includegraphics[width=\textwidth]{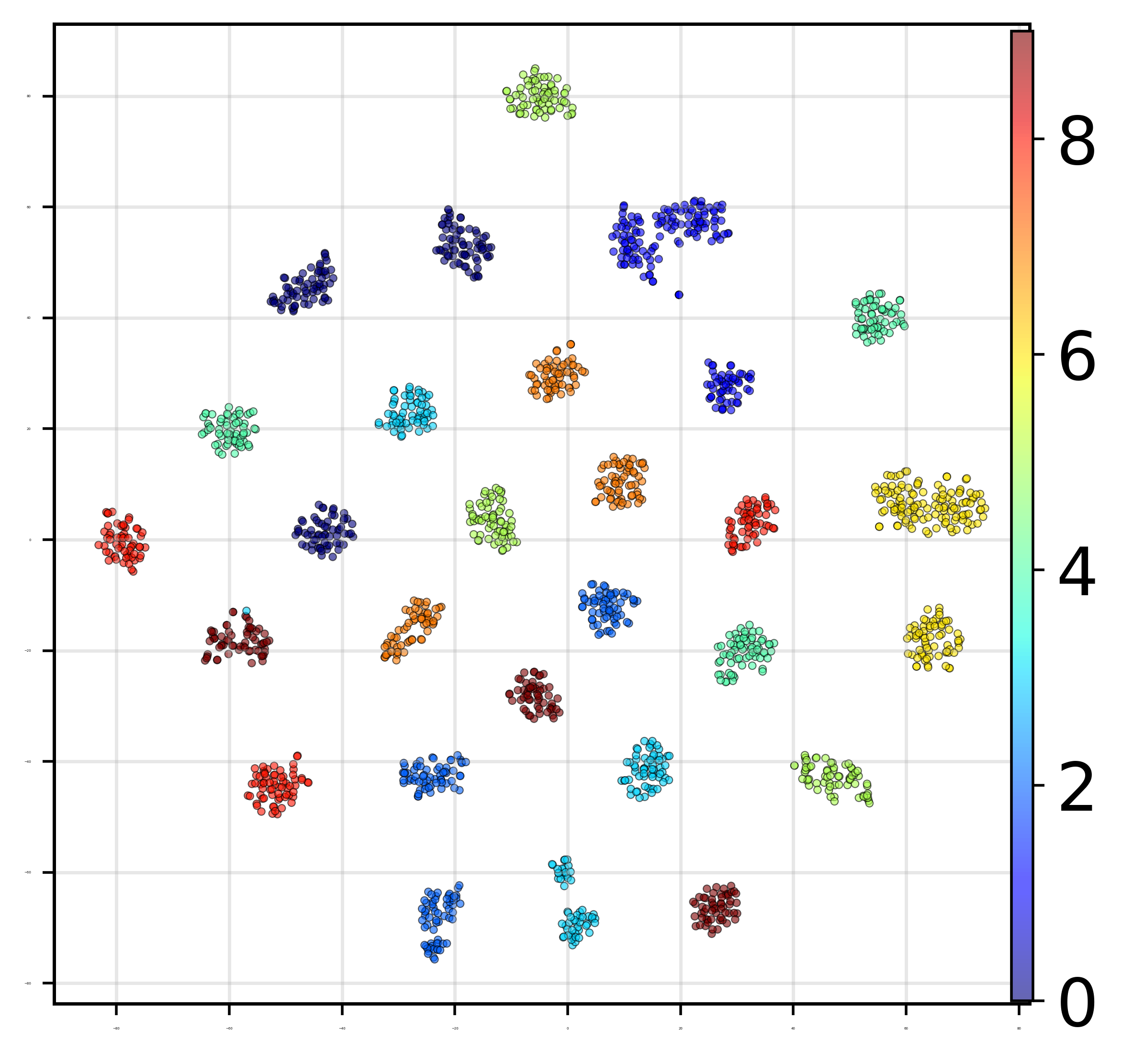}
  \vspace{-5pt}\phantomsection
\end{subfigure}%
\vspace{-7pt}\caption{\small t-SNE projection of trials from ten participants. Subplots from left to right show the projected EEG eigenfunctions from reaching, grasping, and twisting movements, respectively. In all movements, participant-specific information is consistently captured in the clusters of the EEG's eigenfunctions.}
\label{ratio_figure_subject}
\end{figure}

\textbf{Temporal-level dependence.} Extending the analysis in Fig.~\ref{EEG_EIG_PROJ} of the main paper, where the temporal-level dependence was shown for a single subject, we randomly select another seven trials from subject \textbf{\texttt{SUB3}}'s reaching movement (\textbf{\texttt{C1}}) and visualize the temporal-level dependence in {Fig.~\ref{temporal}}. We also plot the average temporal dependence across all trials in Fig.~\ref{temporal:h}.
We find consistent activations of fronto-central (FC) channel during the 4-second movement in each trial. When we compare each trial's result with the averaged one, still we observe that the activation patterns are highly similar. This indicates that our dependence measure captures information that is consistent and generalizable across subjects.
\vspace{-20pt}
\begin{figure}[H]
\centering
{\large\textbf{\texttt{}}
}\\
\begin{subfigure}[b]{0.45\textwidth}
  \centering
  \includegraphics[width=\textwidth]{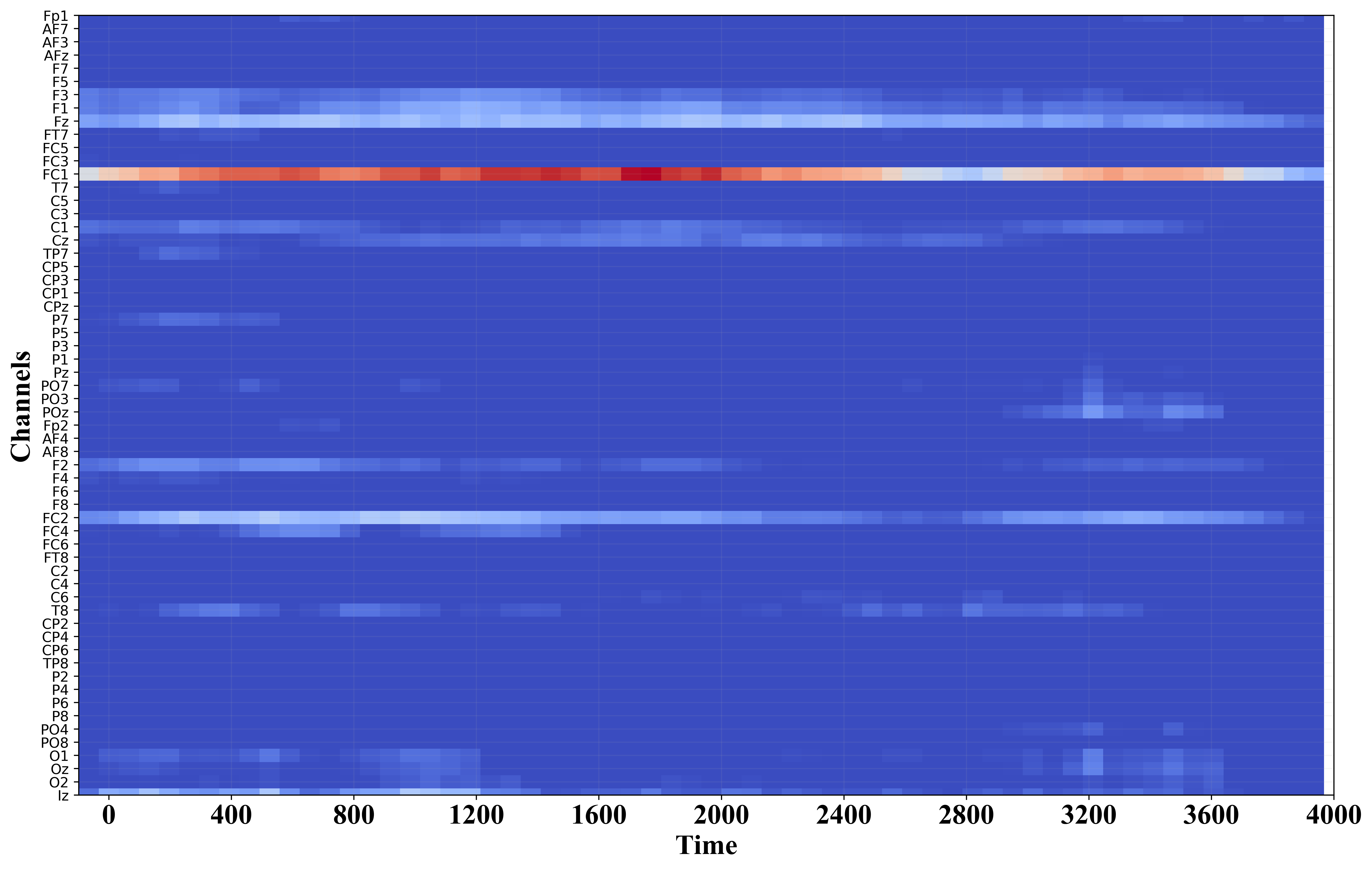}
  \caption{\texttt{\texttt{SUB3, C1, T1}}}
\end{subfigure}%
\begin{subfigure}[b]{0.45\textwidth}
  \centering
  \includegraphics[width=\textwidth]{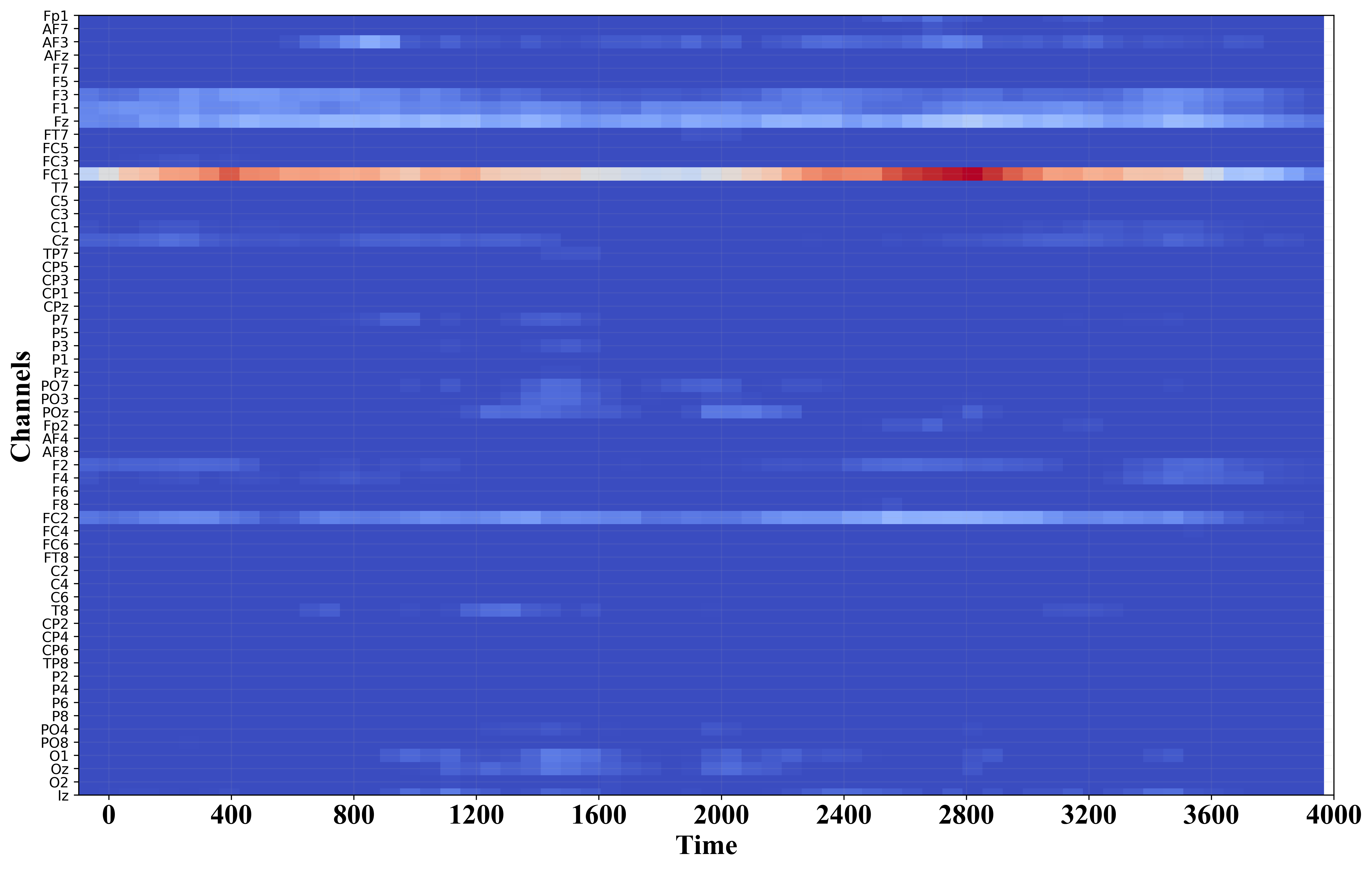}
 \caption{\texttt{SUB3, C1, T2}}
\end{subfigure}\\

\begin{subfigure}[b]{0.45\textwidth}
  \centering
  \includegraphics[width=\textwidth]{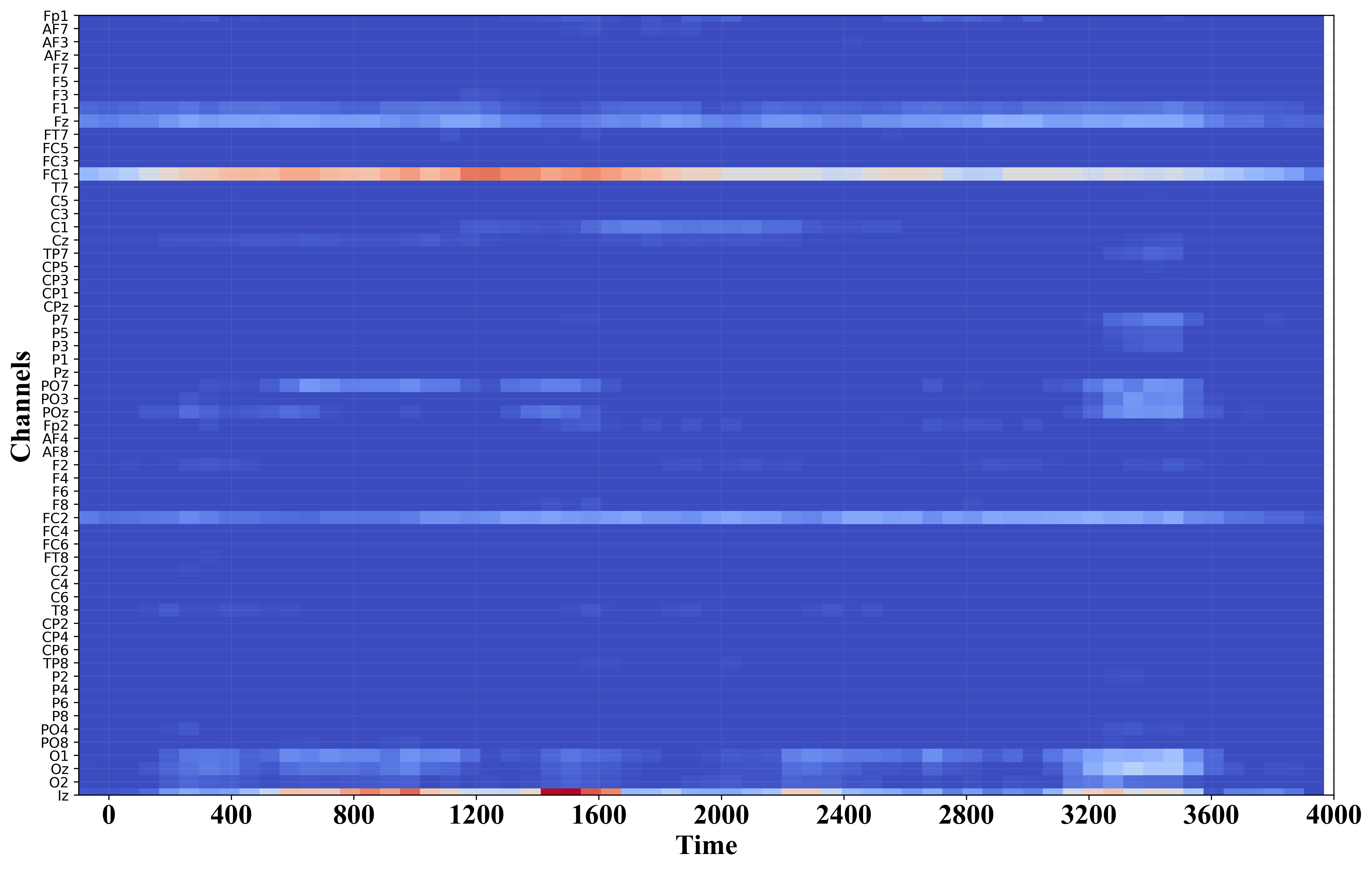}
 \caption{\texttt{SUB3, C1, T3}}
\end{subfigure}%
\begin{subfigure}[b]{0.45\textwidth}
  \centering
  \includegraphics[width=\textwidth]{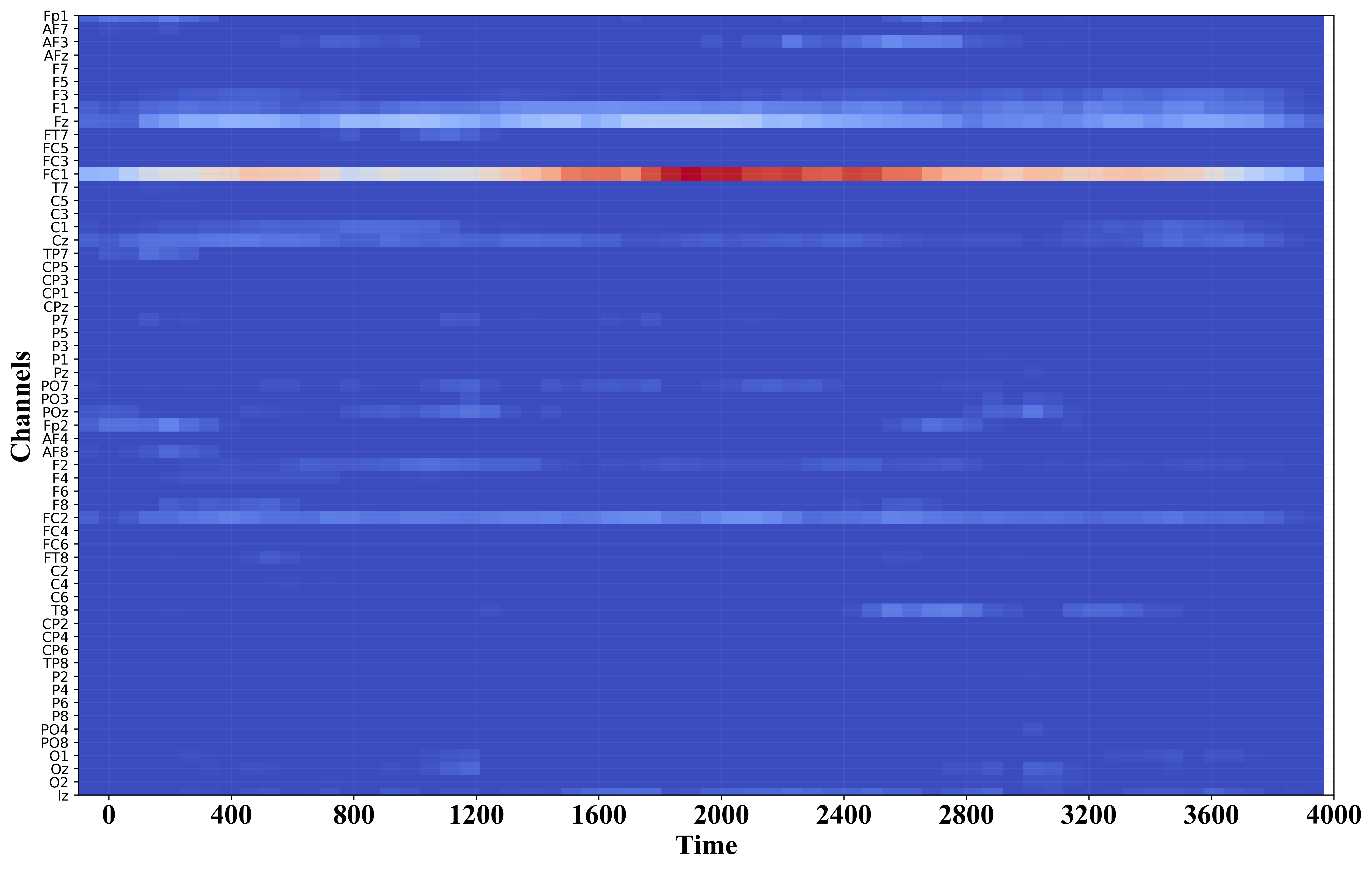}
 \caption{\texttt{SUB3, C1, T4}}
\end{subfigure}\\

\begin{subfigure}[b]{0.45\textwidth}
  \centering
  \includegraphics[width=\textwidth]{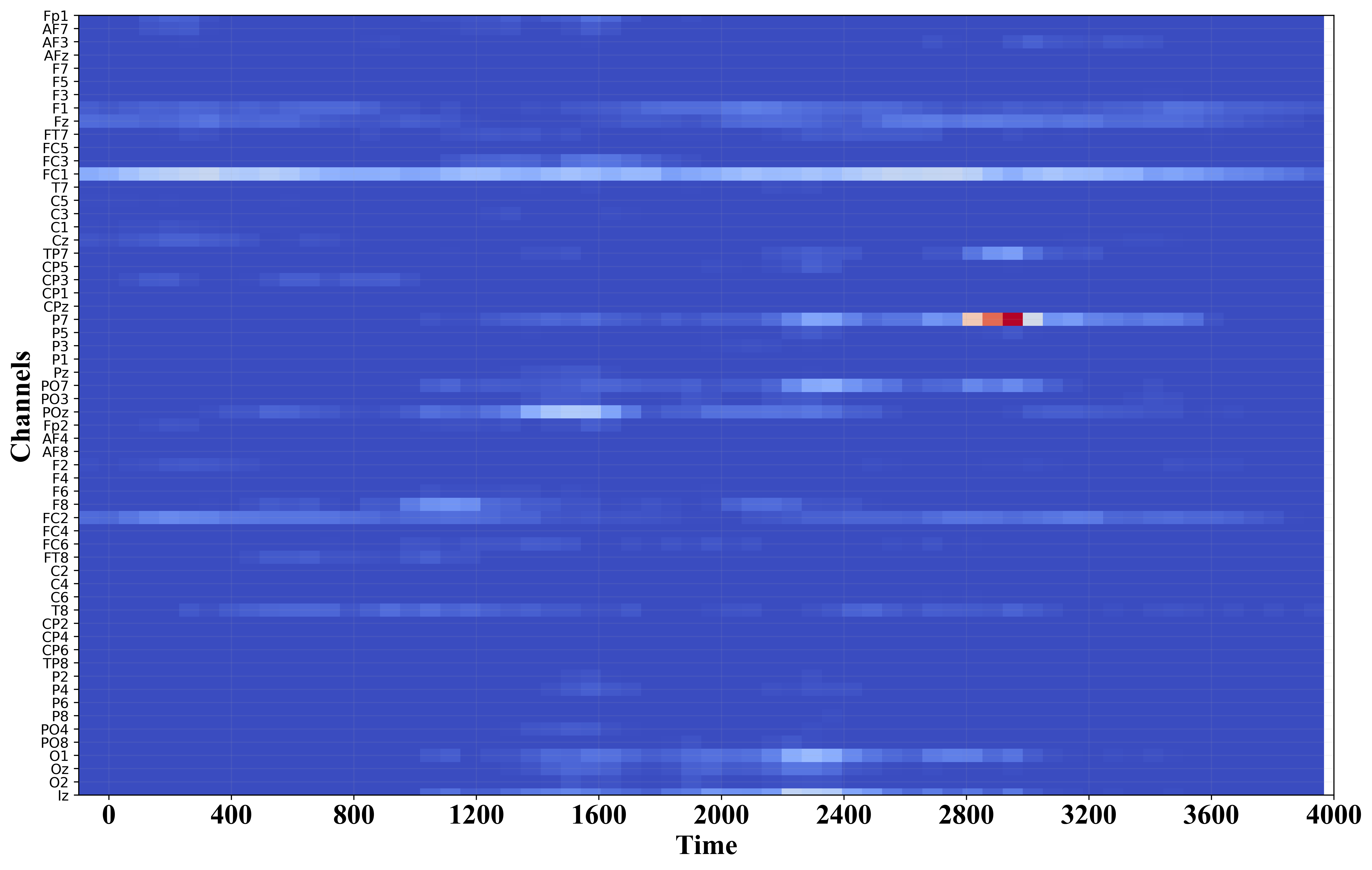}
 \caption{\texttt{SUB3, C1, T5}}
\end{subfigure}%
\begin{subfigure}[b]{0.45\textwidth}
  \centering
  \includegraphics[width=\textwidth]{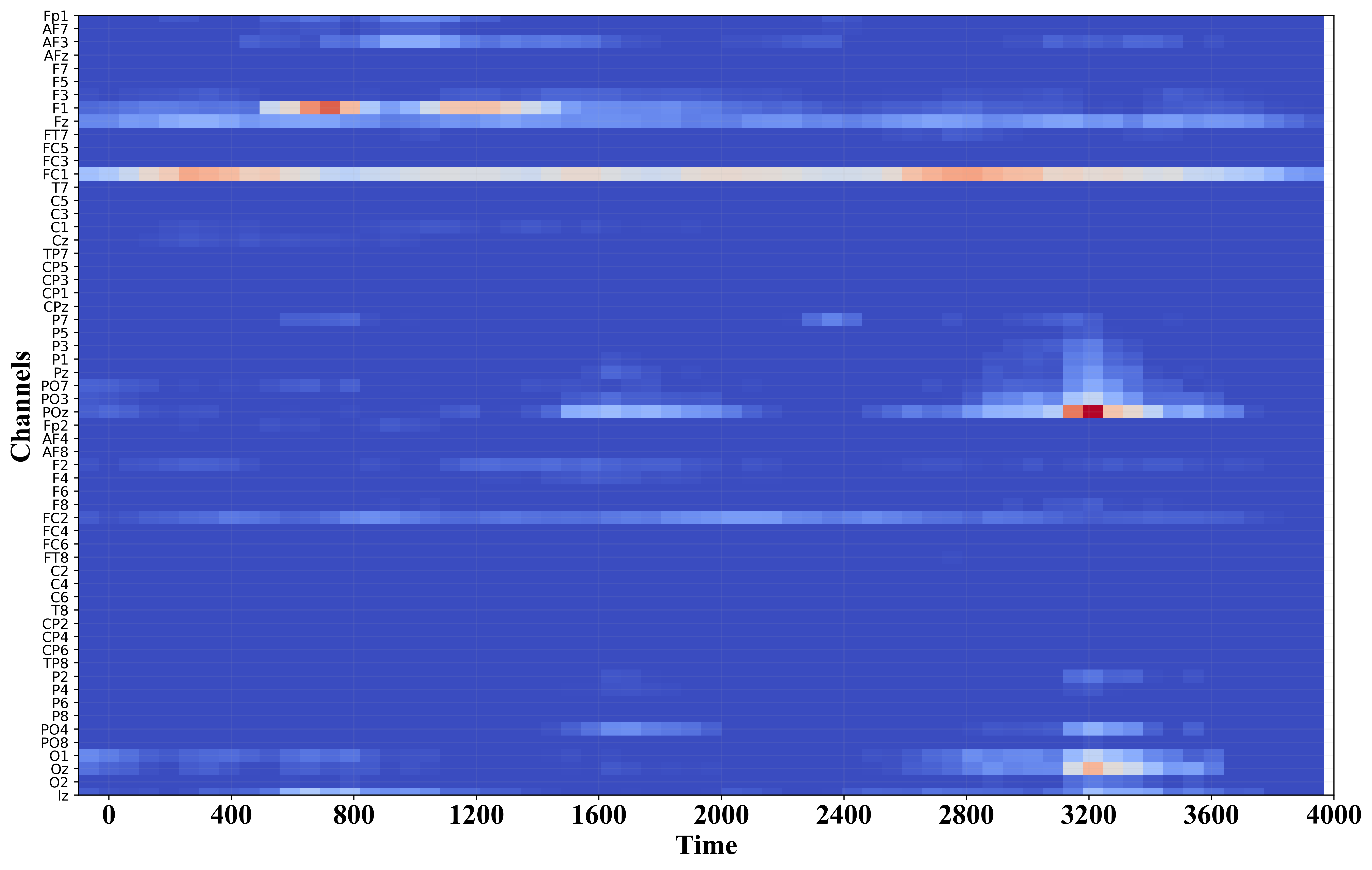}
   \caption{\texttt{SUB3, C1, T6}}
\end{subfigure}\\

\begin{subfigure}[b]{0.45\textwidth}
  \centering
  \includegraphics[width=\textwidth]{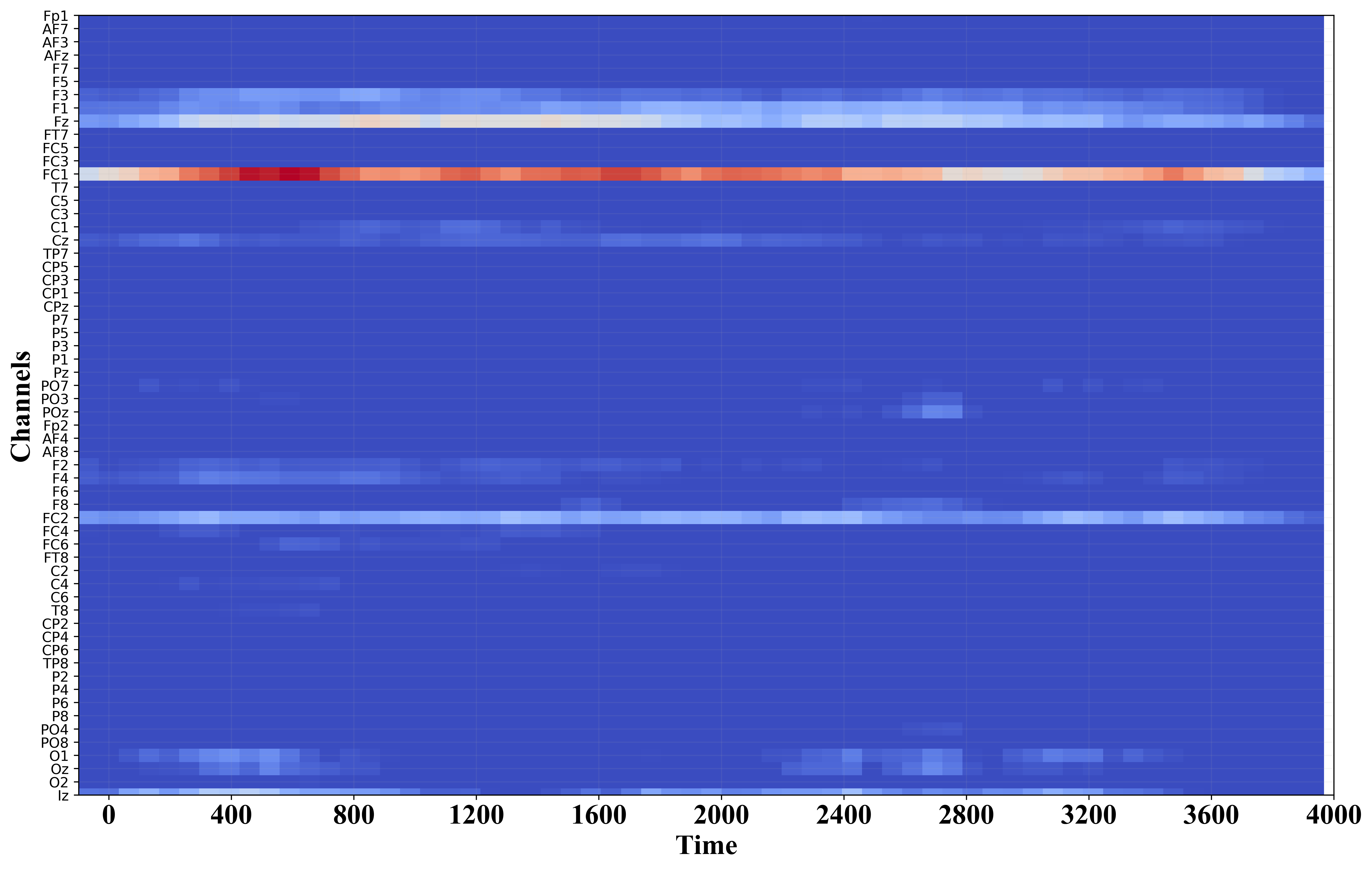}
   \caption{\texttt{SUB3, C1, T7}}
\end{subfigure}%
\begin{subfigure}[b]{0.45\textwidth}
  \centering
  \includegraphics[width=\textwidth]{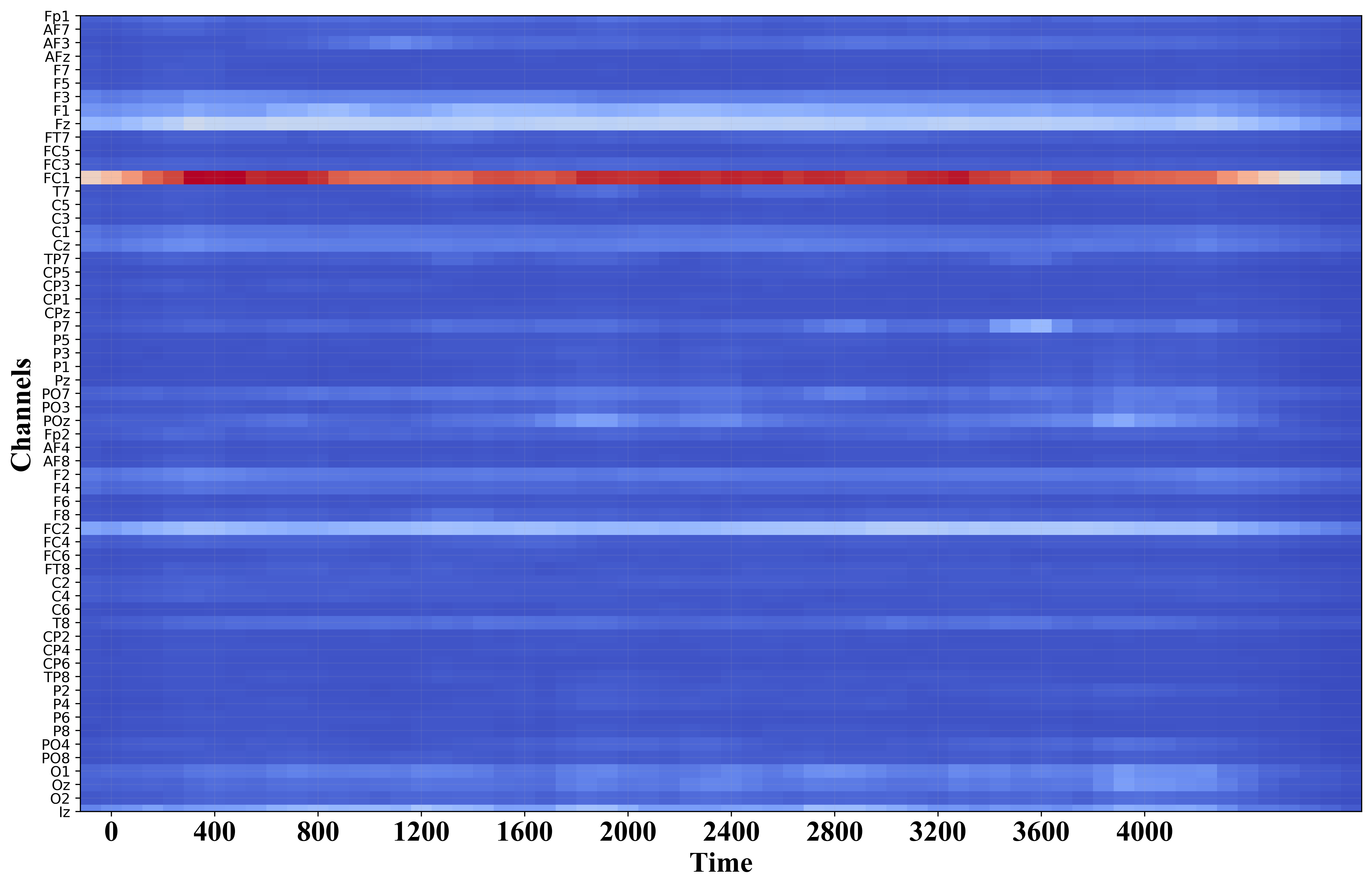}
   \caption{\texttt{SUB3, C1, Averaged}}\label{temporal:h}
\end{subfigure}%
\caption{Temporal-level dependence for nine trials from \textbf{\texttt{SUB3}}'s reaching movement, confirming consistent activation of the frontal brain region and stable temporal dependence over the movement.}
\label{temporal}
\end{figure}

\textbf{Channel-level dependence.} Similar to Fig.~\ref{EEG_EIG_PROJ} of the main paper, we also quantify the channel-level statistical dependence from other subjects, not just \textbf{\texttt{SUB3}}.  
We randomly select clusters from subjects and visualize the results in Fig.~\ref{channel_dependence}. 
We find a consistent pattern across subjects that the FC channels are activated most strongly.
It can also be observed that within each cluster (each subplot), channel activations are highly similar across trials. This indicates that our dependence measure robustly captures the cortical-muscular connectivity of the same movement.

\begin{figure}[H]
\centering
\begin{subfigure}[b]{0.25\textwidth}
  \centering
  \includegraphics[width=\textwidth]{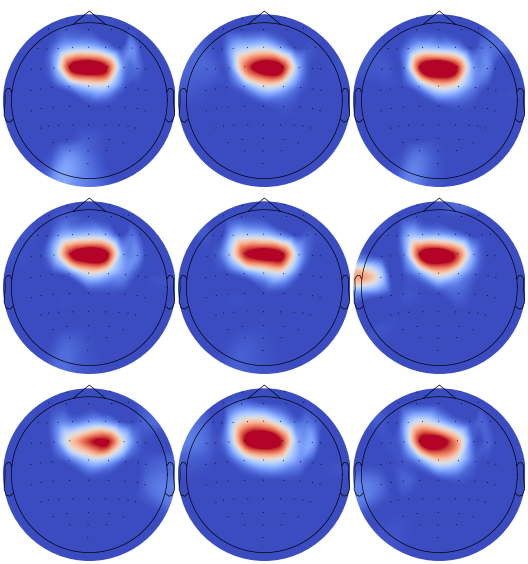}
  \caption{\texttt{\texttt{SUB6, C1}}}
\end{subfigure}\hspace{10pt}
\begin{subfigure}[b]{0.25\textwidth}
  \centering
  \includegraphics[width=\textwidth]{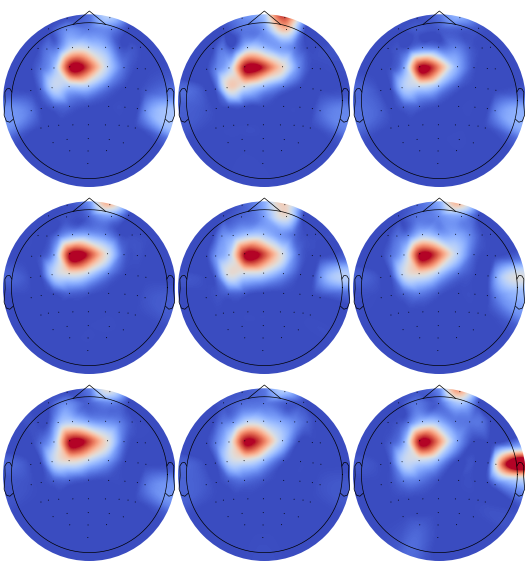}
 \caption{\texttt{SUB13, C9}}
\end{subfigure}\hspace{10pt}
\begin{subfigure}[b]{0.25\textwidth}
  \centering
  \includegraphics[width=\textwidth]{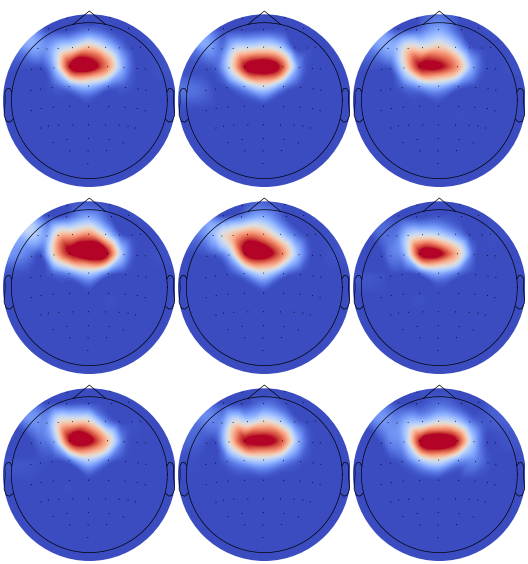}
 \caption{\texttt{SUB17, C8}}
\end{subfigure}
\\
\begin{subfigure}[b]{0.25\textwidth}
  \centering
  \includegraphics[width=\textwidth]{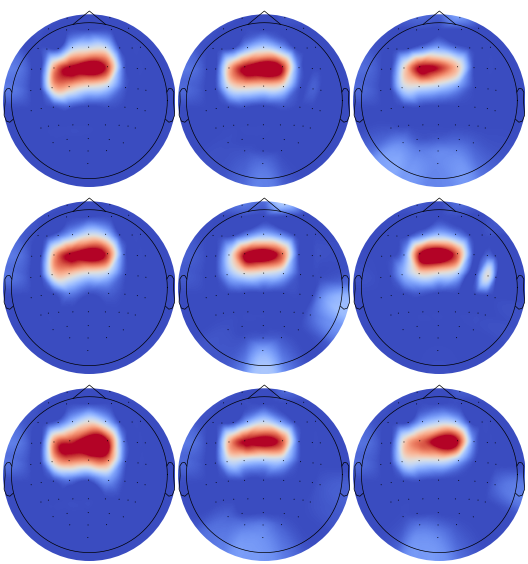}
 \caption{\texttt{SUB17, C4}}
\end{subfigure}\hspace{10pt}
\begin{subfigure}[b]{0.25\textwidth}
  \centering
  \includegraphics[width=\textwidth]{./reference_figure_1/temporal/output_5_0.png}
 \caption{\texttt{SUB6, C1}}
\end{subfigure}\hspace{10pt}
\begin{subfigure}[b]{0.25\textwidth}
  \centering
  \includegraphics[width=\textwidth]{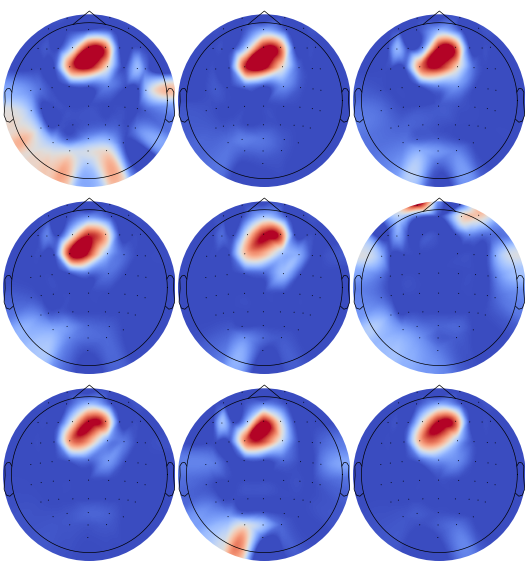}
 \caption{\texttt{SUB24, C2}}
\end{subfigure}\vspace{-5pt}
\caption{\small Channel-level dependence for visualizing  activation patterns. We find that there are strong activations around the FC area across subjects, not just \textbf{\texttt{SUB3}}. We show that multiple clusters from various subjects demonstrate similar activation channels in the FC area. This suggests that these channels are overall the most important to classifying movements and contribute the most to connectivity.}\label{channel_dependence}
\end{figure}




\textbf{Learning curve comparison.} We show the smoothness of the training stage of FMCA-T, comparing its learning curve with the learning curve of MINE when applied both on EEG-EMG-Fusion. As shown in {Fig.~\ref{MINE}}, FMCA-T demonstrates superior stability, whereas MINE suffers from greater instability even when smoothing windows are applied to estimate the gradient of the variational cost in the Donsker-Varadhan representation.

\begin{figure}[H]
\centering
\begin{subfigure}[b]{0.333\textwidth}
  \centering
  \includegraphics[width=\textwidth]{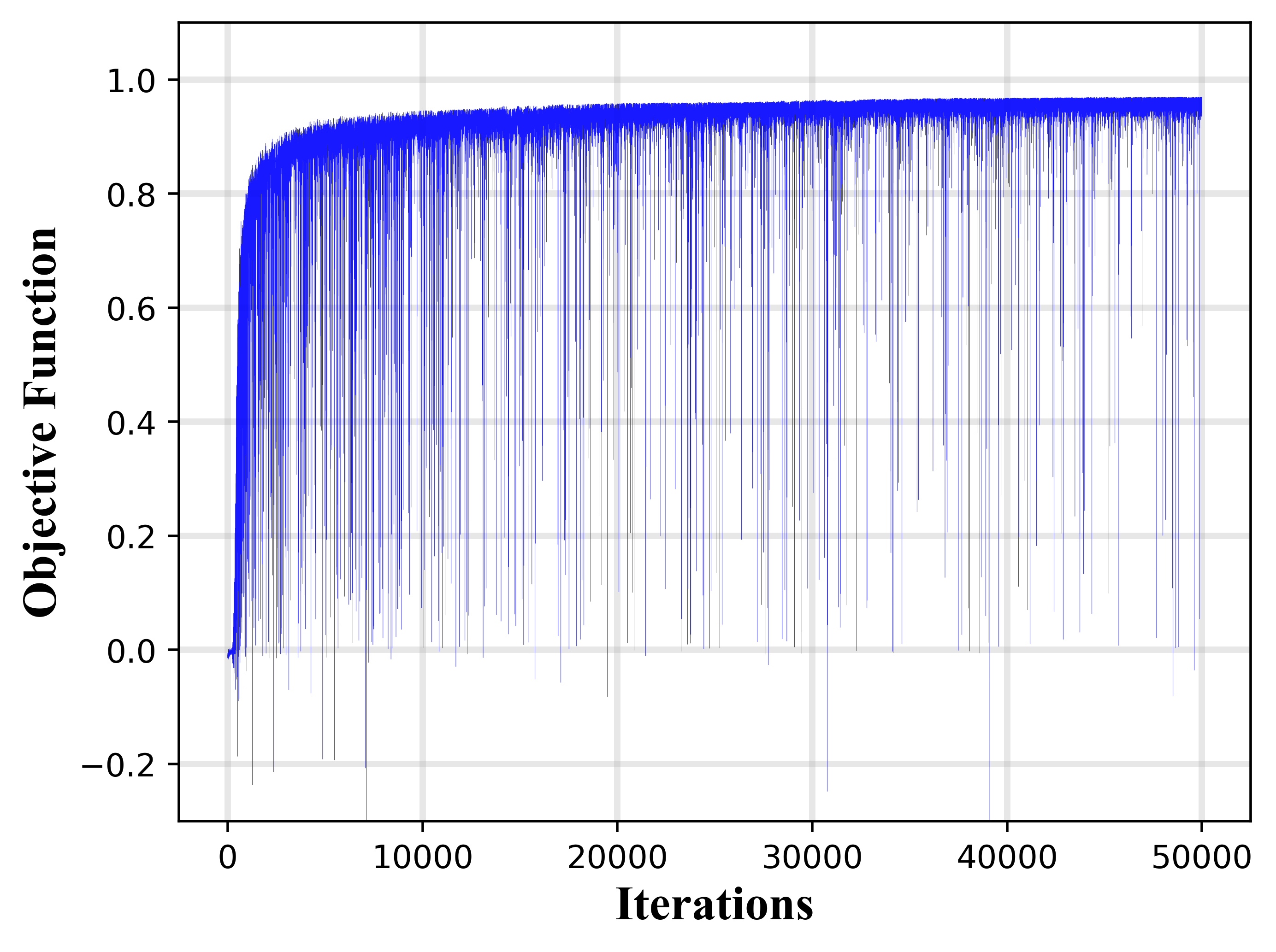}
   \caption{\small \texttt{MINE}}
\end{subfigure}\hfill%
\hfill\begin{subfigure}[b]{0.333\textwidth}
  \centering
  \includegraphics[width=\textwidth]{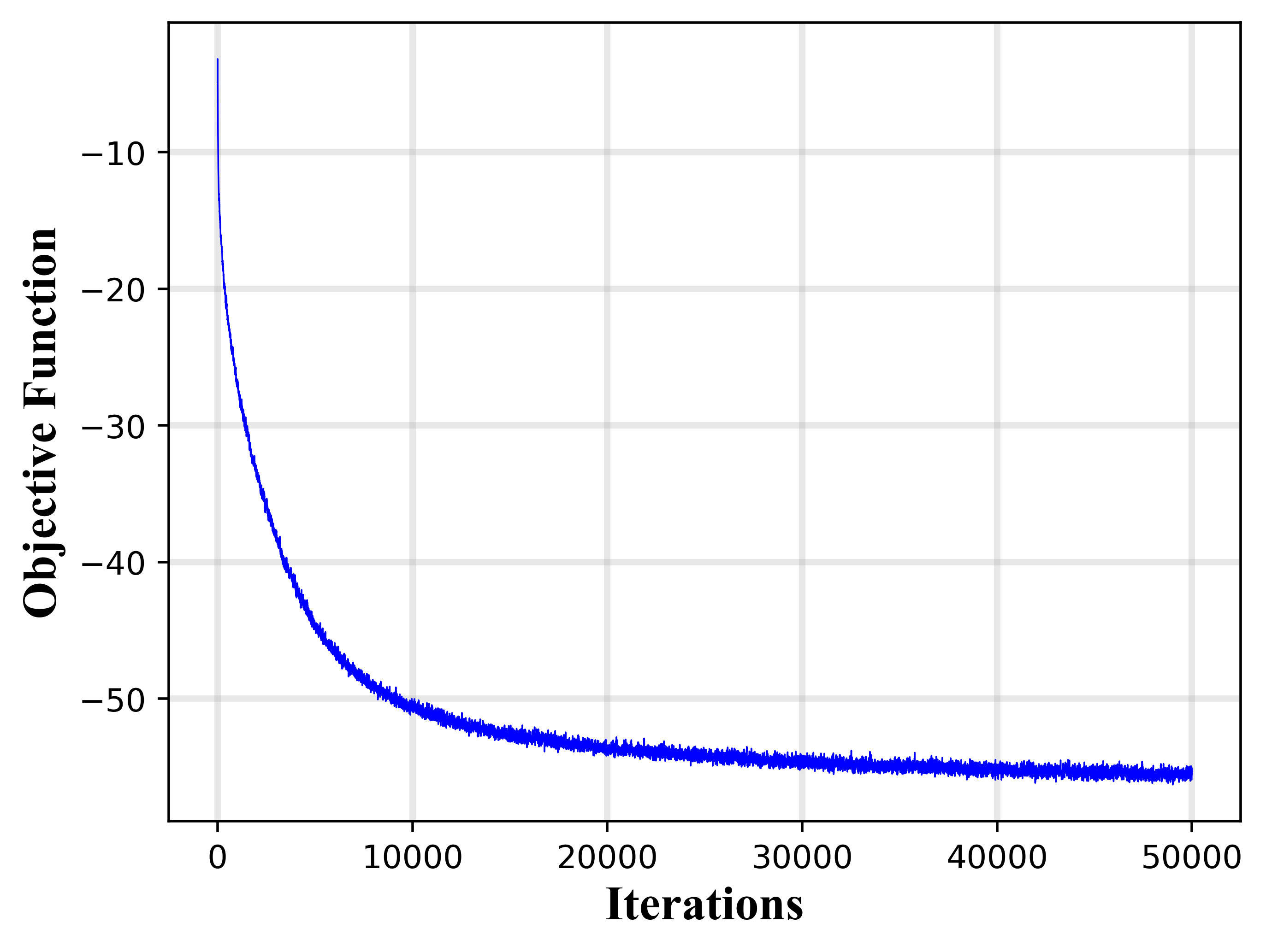}
   \caption{\small \texttt{FMCA-T}}
\end{subfigure}\hfill%
\hfill\begin{subfigure}[b]{0.333\textwidth}
  \centering
  \includegraphics[width=\textwidth]{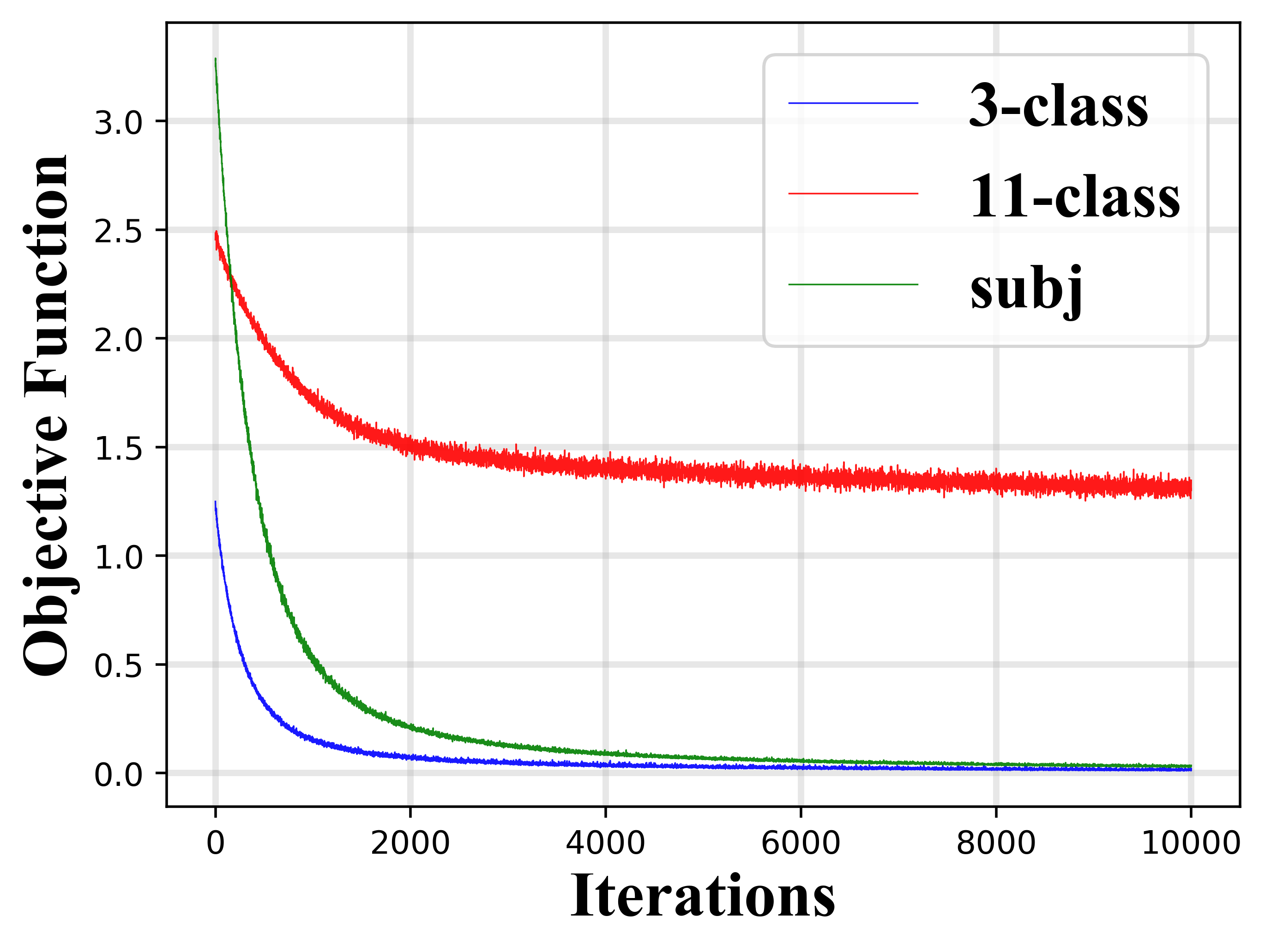}
   \caption{\texttt{Eigenfunction Classifiers}}
\end{subfigure}\hfill%
\caption{\small Comparison of learning curves on EEG-EMG-Fusion dataset. MINE: Variational cost; FMCA-T: Matrix trace cost; Eigenfunction Classifiers: Training errors. MINE suffers from high noise even when smoothing windows are applied to estimate the gradient of the variational cost. MINE is unable to produce stable and comparable results on SinWav. The classifier of eigenfunctions is trained separately from the rest of the networks.}
\label{MINE}
\end{figure}
\textbf{Temporal-level dependence on SinWav.} We analyze the temporal activations of the learned dependence measure on SinWav by visualizing its localized density ratios. This is performed by computing the density ratios between adjacent layers of feature projectors. The layer-wise density ratios are then aggregated for visualization. We find that the localized density ratios exhibit higher activations at the hills and valleys, and correctly capture the period and phase of the sinusoids when there is an increasing delay between the two sinusoids (Fig. \ref{fig:sine_heatmap_delay}). As shown in Fig. \ref{fig:sine_heatmap_gaussian}, when Gaussian noise with a standard deviation equal to 1.0 is added to the clean sinusoidal signal (signal-to-noise ratio less than 0 dB), the density ratio can still correctly identify the hills and valleys. This indicates that our proposed dependence measure is robust to random noise and delay by filtering out trivial factors like noise while focusing on the primary signals.


\begin{figure}[h]
    \centering
    \includegraphics[width=1\columnwidth]{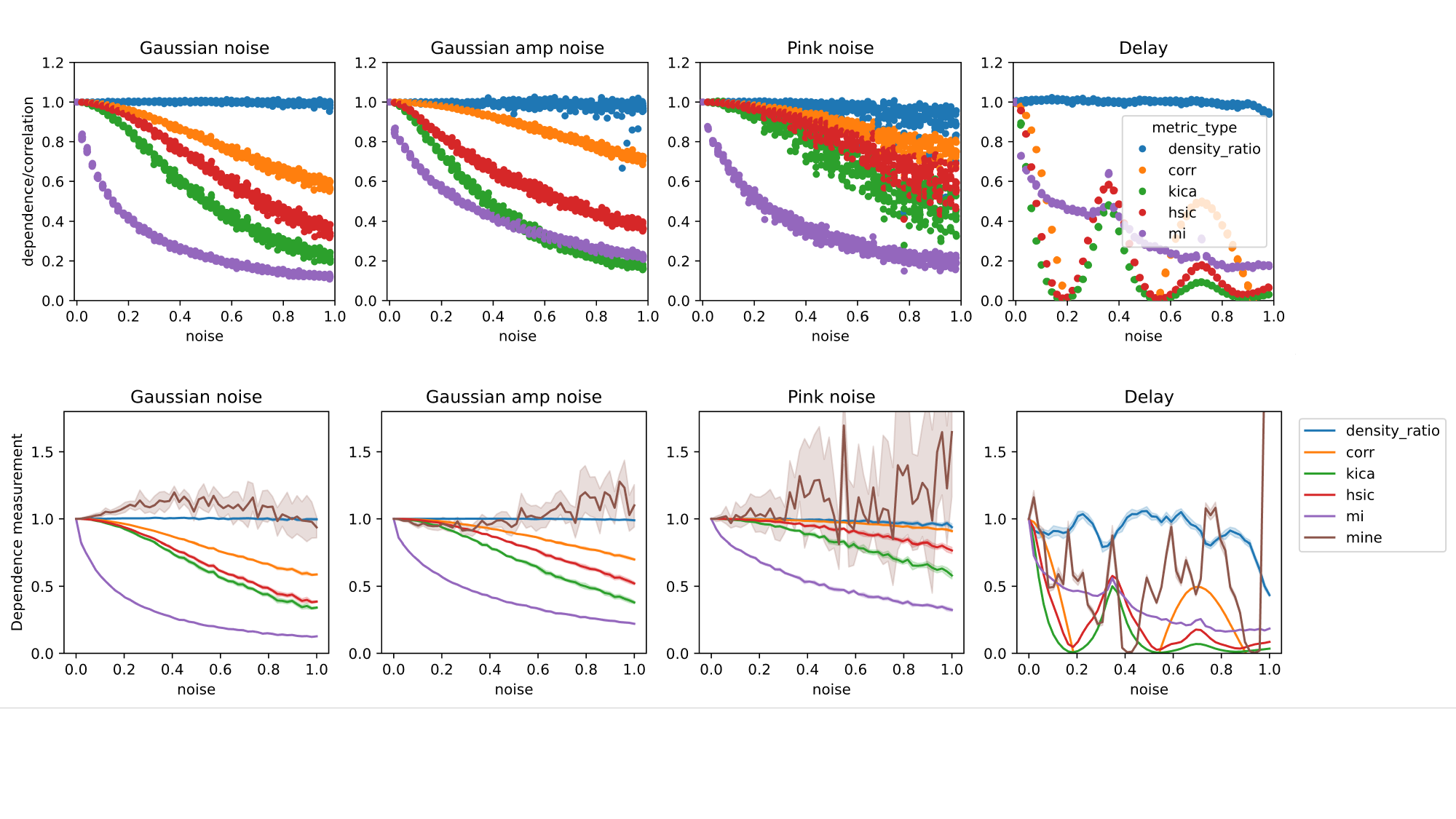}
    \caption{Visualization of localized density ratios for samples from SinWav under three delay levels. The first row shows localized density ratios based on the eigenfunctions of the input (clean) signal and the second row shows the localized density ratios based on the eigenfunctions of the delayed signal. The density ratio successfully captures the period and phase of the two signals.}
    \label{fig:sine_heatmap_delay}
\end{figure} \vspace{-5pt}

\begin{figure}[h]
    \centering
    \includegraphics[width=1\columnwidth]{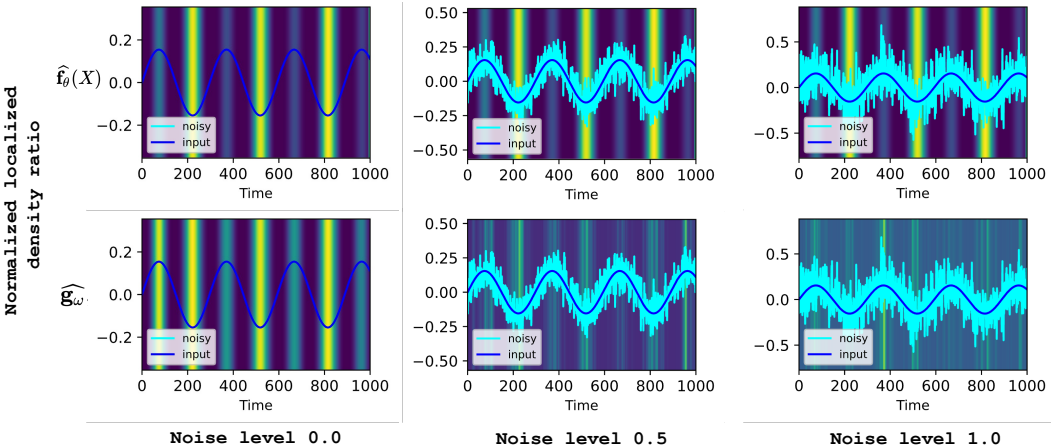}
    \caption{Visualization of localized density ratios for samples from SinWav under three noisy levels. The first row shows localized density ratios based on the eigenfunctions of the input (clean) signal and the second row shows the localized density ratios based on the eigenfunctions of the noisy signal (Gaussian noise). The localized density ratio successfully captures the period and phase of the noisy signals even when the noise level reaches 1.0, where the signal-to-noise ratio is less than $0$ dB.}
    \label{fig:sine_heatmap_gaussian}
\end{figure}

\section{Implementation Details}\label{appendix:implementation}
This section includes details of data preprocessing, implementation of baselines, network structures, and training configurations.
Our code is available at \url{https://github.com/bohu615/corticomuscular-eigen-encoder}.



\subsection{EEG and EMG preprocessing}
We performed standard preprocessing procedures for the 60-channel EEG signals, including 1) Down-sampling from 2500 Hz to 1000 Hz, 2) Band-pass filtering at 1-48 Hz, 3) Removal of signal artefacts with independent component analysis (ICA), and 4) Segmenting EEG at 0-4s of the onset of the movement cue for each trial. 7-channel EMG signals were preprocessed according to the following procedure: 1) Down-sampling from 2500 Hz to 1000 Hz, 2) High-pass filtering at 5 Hz, 3) Baseline correction, and 4) Segmenting EMG at 0-4s of the onset of the movement cue for each trial. The raw dataset (\url{http://gigadb.org/dataset/100788}) is distributed under a CC0 license.

\subsection{Baseline implementations}

\noindent\textbf{CMC baselines.} Cortico-muscular coherence (CMC) measures the linear synchronization between sensorimotor rhythms (present in EEG) and muscular activities (reflected in EMG) to analyze brain-muscle coupling. Given $N$ trials, for each $i$-th EEG channel $X_{1:T}(n, i)$ and $j$-th EMG channel $Y_{1:T}(n, j)$, signals are filtered using a Butterworth filter and segmented into $\widehat{X_{t:t+\tau}}(n, i)$ and $\widehat{Y_{t:t+\tau}}(n, j)$. Coherence is computed as the correlation coefficient ($\mathbf{cc}$) between the spectral densities of the windows, $\mathbf{cc}(FX_{t:t+\tau}, FY_{t:t+\tau})$, averaged over $t$ to $t+\tau$. 

It can be seen that $\mathbf{cc}$ is the correlation coefficients between two variables. We replace this linear measure by nonlinear measures, such as measures from KICA and mutual information estimated by KNN (MIR), producing CMC-KICA and CMC-MIR. \vspace{5pt}


\noindent\textbf{MINE implementation~\cite{belghazi2018mine}.} We adapt MINE with the same topology as ours but operating on the joint space of EEG and EMG $\mathcal{X}\times \mathcal{Y}$, producing a one-dimensional output. We find that using a a sigmoid activation function stabilizes training. MINE computes $h_\theta(\mathbf{X}, \mathbf{Y})$ for samples from the joint distribution and $h_\theta(\mathbf{X}', \mathbf{Y}')$ for samples from respective marginal distributions, and minimizes the variational cost $\min_{\theta} \mathbb{E}[h_\theta(\mathbf{X}, \mathbf{Y})] - \log \mathbb{E}[e^{h_\theta(\mathbf{X}', \mathbf{Y}')}\!+\!10^{-5}]$. The network is optimized by a Adam optimizer with a learning rate of $10^{-4}$, $\beta_1 = 0.5$, and $\beta_2 = 0.9$. MINE's trial-level solution follows the Donsker-Varadhan representation: $\log\rho(X,Y)+\gamma$, where $\gamma$ can be any constant. \vspace{5pt}

\noindent\textbf{KICA and HSIC~\cite{gretton2005measuring, bach2002kernel}.} KICA and HSIC are implemented in the following steps. First, individual Gram matrices $\mathbf{R}_{X}$ and $\mathbf{R}_{Y}$ are estimated for given Gaussian kernel $\mathcal{K}(X_i, X_j) = \mathcal{N}(X_i-X_j;\delta)$ with $\delta$ the standard deviation. Then, we construct the normalization matrix $\mathbf{N}_{i,j}$ and normalize the two Gram matrices as $\widehat{\mathbf{R}_{X}} = \mathbf{N}\mathbf{R}_{X}\mathbf{N}$ and $\widehat{\mathbf{R}_{Y}} = \mathbf{N}\mathbf{R}_{Y}\mathbf{N}$. For KICA-KGV, matrices $\mathbf{A}$ and $\mathbf{B}$ are constructed:

\begin{equation}
\begin{gathered}
\mathbf{A} = \begin{bmatrix}\mathbf{A}_1 & 0\\ 0 & \mathbf{A}_2 \end{bmatrix}, \quad \mathbf{A}_1 = \widehat{\mathbf{R}_{X}}\,\widehat{\mathbf{R}_{Y}}, \quad \mathbf{A}_2 = \widehat{\mathbf{R}_{Y}}\,\widehat{\mathbf{R}_{X}}, \\
\mathbf{B} = \begin{bmatrix}\mathbf{B}_1 & 0\\ 0 & \mathbf{B}_2 \end{bmatrix}, \quad \mathbf{B}_1 = (\widehat{\mathbf{R}_{X}}+\epsilon \mathbf{I}) (\widehat{\mathbf{R}_{X}}+\epsilon \mathbf{I}), \quad \mathbf{B}_2 = (\widehat{\mathbf{R}_{Y}}+\epsilon \mathbf{I}) (\widehat{\mathbf{R}_{Y}}+\epsilon \mathbf{I}),
\end{gathered}
\end{equation}

Then, solve the generalized eigenvalue problem for KICA $\mathbf{A} \mathbf{v}_i = \sigma_i \mathbf{B} \mathbf{v}_i$, where $i=1, \cdots, 2N$. This generalized eigenproblem generates $2N$ eigenvalues that are symmetric over the real line. Only $N$ positive eigenvalues of them are used to compute the measure, obtaining KICA's Kernel Generalized Variance (KGV) measure. 
For HSIC, we construct matrix $\mathbf{C}$:
\begin{equation}
\mathbf{C} = \mathbf{B}_1^{-\frac{1}{2}} \mathbf{A}_1 \mathbf{B_2}^{-\frac{1}{2}},
\end{equation}

and solve the eigenvalue problem $\mathbf{C}\mathbf{v}_i = \sigma_i \mathbf{v}_i$, where $i=1,\cdots, N$. Compute $T_{HSIC} = Trace(\mathbf{C})$. HSIC's measure is named the Normalized Cross-Covariance Operator (NOCCO). Hyperparameters are set as kernel size $\delta=0.1$ and regularization constant $\epsilon=0.1$.

\vspace{5pt}

\noindent\textbf{Self-Supervised baselines:} Self-Supervised Learning (SSL) methods are also implemented to compare the classification performance, including Barlow Twins~\cite{zbontar2021barlow}, SimCLR~\cite{chen2020simple}, and VICReg~\cite{bardes2021vicreg}. We mainly use their cost functions. SSL experiments use a window size of $1,000$ with windows from the same trial as positive pairs, and windows from different trials as negative pairs. The cost, hyper-parameters, and implementations follow the Lightly package~\cite{susmelj2020lightly}.
\vspace{5pt}

\textbf{EEGNet~\cite{lawhern2018eegnet}.} EEGNet is a convolutional neural network commonly used for EEG signal classification. After random search hyperparameter optimization, we find optimal performance with settings close to the original paper's recommendations. Validation set is split from the training set to enable early stopping regularization. As in the original paper, we use the Adam optimizer with a learning rate of $0.0001$ and a decay factor of $0.1$ every $20$ epochs. We train and test EEGNet on the same train-test splits as the proposed algorithm for all inter-subject and cross-subject experiments.
\vspace{5pt}

\textbf{CSP-RLDA~\cite{ang2012filter}.} Common Spatial Pattern (CSP) has been proven to effectively discriminate two classes of EEG by constructing optimal spatial filters. We use CSP to extract features and Regularized Linear Discriminant Analysis (RLDA) as a classifier. If class sample volumes are unbalanced, the CSP-based classifier may be biased towards the larger sample volume category. Thus particularly for CSP, one participant with one wrist-twisting session with bad EEG quality was discarded from the analysis. We use two pairs of CSP filters and extract four feature dimensions. CSP is trained and tested in both inter-subject and cross-subject settings. To achieve three-class classification with CSP, three classifiers are trained for each pair of the three classes, and a voting strategy determines the results.

\subsection{Network structures}


The structure of the temporal network is illustrated in Table~\ref{advanced_1d_layers}, which consists of four convolutional blocks and max pooling. Each of these blocks and max pooling are treated as a layer for computing localized density ratio responses. We apply the temporal network to each channel of the signal, obtaining $\mathbf{Z}_{1,9}, \mathbf{Z}_{2,9},\cdots, \mathbf{Z}_{C,9}$, where the total number of channels is $C=60$ for EEG and $C=7$ for EMG and 9 indicates the ninth layer. The output of the temporal network for each channel is a vector with dimension $K=128$. In the paper, we use the localized density ratios of $\mathbf{Z}_{c, 6}$ to visualize the temporal resolution. 

The channel network is a three-layer MLP that takes $[\mathbf{Z}_{1,9}, \mathbf{Z}_{2,9},\cdots, \mathbf{Z}_{C,9}]^\intercal$ (dimension of $K\times C$) as input and also produces an output of dimension $K=128$. Each layer uses BN and ReLU with $2,000$ units per layer. The classifier used for eigenfunctions is also a three-layer MLP with $500$ units per layer.

\begin{table}[H]
\centering
\begin{tabular}{@{}lcccccc@{}}
\toprule
\textbf{Layer} & \textbf{In Ch.} & \textbf{Out Ch.} & \textbf{Kernel Size} & \textbf{Padding} & \textbf{Output}\\ \midrule
\rowcolor[HTML]{EFEFEF}
Conv, BN, ReLu & 1 & 32 & 11 & 5 & $\mathbf{Z}_{c, 1}$ \\ 
Maxpool & 32 & 32 & 4 & - & $\mathbf{Z}_{c, 2}$ \\
\rowcolor[HTML]{EFEFEF} 
Conv, BN, ReLu & 32 & 64 & 11 & 5 & $\mathbf{Z}_{c, 3}$ \\
Maxpool & 64 & 64 & 4 & - & $\mathbf{Z}_{c, 4}$ \\
\rowcolor[HTML]{EFEFEF}
Conv, BN, ReLu & 64 & 128 & 11 & 5 & $\mathbf{Z}_{c, 5}$\\
Maxpool  & 128 & 128 & 4 & - & $\mathbf{Z}_{c, 6}$ \\
\rowcolor[HTML]{EFEFEF}
Conv, BN, ReLu & 128 & 256 & 11 & 5 & $\mathbf{Z}_{c, 7}$ \\
Maxpool  & 256 & 256 & 4 & - & $\mathbf{Z}_{c, 8}$\\
\rowcolor[HTML]{EFEFEF}
Linear BN, ReLu & 256 $\times$ 15 & 1024 & - & - & \\
\rowcolor[HTML]{EFEFEF}
Linear BN, ReLu & 1024 & 512 & - & - & - \\
\rowcolor[HTML]{EFEFEF}
Linear, Sigmoid & 512 & $K$ & - & - & $\mathbf{Z}_{c, 9}$ \\
\bottomrule
\end{tabular}\vspace{-5pt}
\caption{Architecture of Temporal Network.}
\label{advanced_1d_layers}
\end{table}\vspace{-10pt}






\subsection{Training configurations} 

SinWav experiments were conducted on an NVIDIA GeForce RTX 3090. EEG-EMG-Fusion experiments were conducted on an NVIDIA GeForce A5000. Both SinWav and EEG-EMG-Fusion used an Adam optimizer with $\beta_1 = 0.5$ and $\beta_2 = 0.9$ for network optimization.

\section{Limitation}\label{sec:limitation}
We demonstrate the effectiveness of using FMCA to learn the dependence between EEG and EMG signals and have shown that the learned eigenfunctions embed subject and movement information after optimization. We conducted the experiments on a public EEG and EMG dataset, which only contains 11 discrete upper extremity movements from 25 subjects. What we leave in the future is to use the meaningful eigenfunctions for regression tasks, i.e., continuously predicting the kinematics and contraction forces during the movement. Another limitation of the study is that we did not include patients' data due to a dearth of such large datasets that collect patients' multi-modal bio-signals. We hope that the promise offered by using our dependence measurement to evaluate cortico-muscular connectivity will further stimulate experimental research in this direction. Last, from a technical point of view, we only used convolutional neural networks with a concatenated MLP as the backbone of our networks in this study. Although we suppose that a CNN model is sufficient for the current scope as the temporal information is processed with CNN and spatial information can be leveraged by the final MLP projection layer, more advanced network structures, such as the attention in the transformer, could be potentially useful when aiming for more complex tasks.

\section{Broader Impact}\label{appendix:impact}
This paper proposes to use the statistical dependence between the densities of neural data to evaluate cortico-muscular connectivity. Compared with the traditional method (CMC) that computes the linear correlation between EEG and EMG spectra, our measurement shows less variance within movement and subject while is more distinguishable across movements and subjects. This helps in exploring the neural information pathways from the brain to the muscle, which could further be used in the field of patient rehabilitation. Moreover, the learned eigenfunctions can be used as ``common information'' decoders, for example to decode movement and subject, and are more robust to distribution shift when tested on unseen subjects. This could be very useful for developing brain-machine interfaces under inter-subject conditions.

All datasets in this paper are publicly available and are not associated with any privacy or security concerns. Usages of the datasets strictly follow the corresponding licenses.



\clearpage

\section*{NeurIPS Paper Checklist}

\begin{enumerate}

\item {\bf Claims}
    \item[] Question: Do the main claims made in the abstract and introduction accurately reflect the paper's contributions and scope?
    \item[] Answer: \answerYes{}
    \item[] Justification: In the abstract and introduction, we claim that 1) learning the dependence between EEG and EMG signals by using FMCA produces a robust dependence measure of cortico-muscular connectivity, 2) the learned eigenfunctions embed rich information of movement and subject, and 3) the temporal and spatial dependence can be visualised by computing the localized density ratios. The claims are well justified in the Sec.~\ref{sec:method} with theoretical support and experimentally demonstrated in the Sec.~\ref{sec:main_results}, where we show that our dependence measure is more consistent than other dependence measurements and classifiers fed with EEG eigenfunctions outperform other baseline methods. We also show temporal and channel activations that indicate the temporal and spatial distribution of the density ratio in the Sec.~\ref{sec:main_results}.

\item {\bf Limitations}
    \item[] Question: Does the paper discuss the limitations of the work performed by the authors?
    \item[] Answer: \answerYes{}
    \item[] Justification: See Conclusion and App.~\ref{sec:limitation} for discussions on limitations of this study.

\item {\bf Theory Assumptions and Proofs}
    \item[] Question: For each theoretical result, does the paper provide the full set of assumptions and a complete (and correct) proof?
    \item[] Answer: \answerYes{}
    \item[] Justification: We provide the proof in the Sec.~\ref{density_ratio_decompose} to demonstrate the eigenfunctions decomposed from the density ratio form a linear span. in the Sec.~\ref{algorithm}, we provide proof to derive the matrix trace cost.

    \item {\bf Experimental Result Reproducibility}
    \item[] Question: Does the paper fully disclose all the information needed to reproduce the main experimental results of the paper to the extent that it affects the main claims and/or conclusions of the paper (regardless of whether the code and data are provided or not)?
    \item[] Answer: \answerYes{}
    \item[] Justification: We describe matrix trace cost in detail in the Sec.~\ref{algorithm}, which is the critical component to reproduce our work. Network structures and training configurations are described in the App.~\ref{appendix:implementation}. Pseudocodes for the algorithms are provided in the App.~\ref{appendix:algorithm}. We include an anonymous link (see App.~\ref{appendix:implementation}) that provides the source codes with all implementation details and implementation of baselines.

\item {\bf Open access to data and code}
    \item[] Question: Does the paper provide open access to the data and code, with sufficient instructions to faithfully reproduce the main experimental results, as described in supplemental material?
    \item[] Answer: \answerYes{}
    \item[] Justification: We include an anonymous link (see App.~\ref{appendix:implementation}) that provides the source codes with all implementation details and implementation of baselines. Details of data preparation are provided in the App.~\ref{appendix:implementation}.

\item {\bf Experimental Setting/Details}
    \item[] Question: Does the paper specify all the training and test details (e.g., data splits, hyperparameters, how they were chosen, type of optimizer, etc.) necessary to understand the results?
    \item[] Answer: \answerYes{}
    \item[] Justification: Details of experimental settings are provided in the App.~\ref{appendix:implementation}.

\item {\bf Experiment Statistical Significance}
    \item[] Question: Does the paper report error bars suitably and correctly defined or other appropriate information about the statistical significance of the experiments?
    \item[] Answer: \answerYes{}
    \item[] Justification: We run experiments for five times and report the average value with standard deviation. See Table~\ref{tab:sctable}, Figure~\ref{fig:sine_comp}, Figure~\ref{CLUSTERFIGURE} in the main text and App.~\ref{appendix:exp} for more details.

\item {\bf Experiments Compute Resources}
    \item[] Question: For each experiment, does the paper provide sufficient information on the computer resources (type of compute workers, memory, time of execution) needed to reproduce the experiments?
    \item[] Answer: \answerYes{}
    \item[] Justification: See App.~\ref{appendix:implementation} for information on computer resources.

\item {\bf Code Of Ethics}
    \item[] Question: Does the research conducted in the paper conform, in every respect, with the NeurIPS Code of Ethics \url{https://neurips.cc/public/EthicsGuidelines}?
    \item[] Answer: \answerYes{}

\item {\bf Broader Impacts}
    \item[] Question: Does the paper discuss both potential positive societal impacts and negative societal impacts of the work performed?
    \item[] Answer: \answerYes{}
    \item[] Justification: See App.~\ref{appendix:impact} for discussion on broader impacts of the work.

\item {\bf Safeguards}
    \item[] Question: Does the paper describe safeguards that have been put in place for responsible release of data or models that have a high risk for misuse (e.g., pretrained language models, image generators, or scraped datasets)?
    \item[] Answer: \answerNA{}

\item {\bf Licenses for existing assets}
    \item[] Question: Are the creators or original owners of assets (e.g., code, data, models), used in the paper, properly credited and are the license and terms of use explicitly mentioned and properly respected?
    \item[] Answer: \answerYes{}
    \item[] Justification: We used one public EEG and EMG dataset in this study. We cited the original paper and provided the URL and the license of the asset in the App.~\ref{appendix:implementation}.

\item {\bf New Assets}
    \item[] Question: Are new assets introduced in the paper well documented and is the documentation provided alongside the assets?
    \item[] Answer: \answerNA{}
    \item[] Justification: We used one toy dataset that generates pairs of sinusoids for a preliminary study. We do not regard this toy dataset as a new asset. However, we described the details of the dataset in the Sec.~\ref{Sec_dataset} and provided the codes via the link in the Appendix.

\item {\bf Crowdsourcing and Research with Human Subjects}
    \item[] Question: For crowdsourcing experiments and research with human subjects, does the paper include the full text of instructions given to participants and screenshots, if applicable, as well as details about compensation (if any)? 
    \item[] Answer: \answerNA{}

\item {\bf Institutional Review Board (IRB) Approvals or Equivalent for Research with Human Subjects}
    \item[] Question: Does the paper describe potential risks incurred by study participants, whether such risks were disclosed to the subjects, and whether Institutional Review Board (IRB) approvals (or an equivalent approval/review based on the requirements of your country or institution) were obtained?
    \item[] Answer: \answerNA{}
    \item[] Justification: We used one public EEG and EMG dataset in this study. We have acknowledged the secondary use of this public human dataset and included the IRB approval of the original public dataset (approved by the Institutional Review Board at Korea University, 1040548-KU-IRB-17-181-A-2) in the manuscript.

\end{enumerate}

\end{document}